%% file: main.tex
\documentclass[sigconf,edbt]{acmart}

\setcopyright{rightsretained}

\usepackage{booktabs}
\usepackage{enumitem}
\usepackage{algpseudocode}
\usepackage{algorithm}
\usepackage{graphicx}
\usepackage{graphicx}
\usepackage{balance}  % for  \balance command ON LAST PAGE  (only there!)
\usepackage{tikz}
\usepackage{amsmath}
\usepackage{amsfonts}
\usepackage{dsfont}
\usepackage{mathtools}
\usepackage{tikz-qtree}
\usepackage{listings}
\usepackage{xcolor}
\usepackage{colortbl}
\usepackage{multirow}
\usepackage{hhline}
\usepackage{float}
\usepackage{times}  %error
\usepackage{textcomp}
\usepackage{adjustbox}
\usepackage[disable]{todonotes}
\newcounter{todocounter}
\usepackage{subcaption}
\usepackage{comment}
\usepackage{import}
\usepackage{pseudo}
\usepackage{mdframed}
\usepackage{transparent}
\usepackage{arydshln} % for dash line in table
\usepackage[acronym]{glossaries}
\usepackage{glossaries-extra}
\input{sections/acronym}
\makeatletter
\protected\def\adl@@vlineL#1#2#3#4{\adl@ivline#4\@nil{#1}{#2}%
        \xdef\adl@colsL{\adl@colsL
                \@elt{#3}{\number\@tempcnta}{\number\@tempcntb}%
                        {\adl@dashcolor}{\adl@gapcolor}}}
\protected\def\adl@@vlineR#1#2#3#4{\adl@ivline#4\@nil{#1}{#2}%
        \xdef\adl@colsR{%
                \@elt{#3}{\number\@tempcnta}{\number\@tempcntb}%
                        {\adl@dashcolor}{\adl@gapcolor}%
                \adl@colsR}}
\let\adl@act@@vlineL\adl@@vlineL
\let\adl@act@@vlineR\adl@@vlineR
\makeatother

\usepackage{hyperref}
\usepackage{cleveref}

\Crefname{exam}{Example}{Examples}
\Crefname{lem}{Lemma}{Lemmas}
\Crefname{theo}{Theorem}{Theorems}

\usepackage{xspace}

\newbool{Techreport}
\newrobustcmd{\ifnottechreport}[1]{\ifbool{Techreport}{}{#1}}
\newrobustcmd{\iftechreport}[1]{\ifbool{Techreport}{#1}{}}
\booltrue{Techreport} % in main paper
\input{./sections/macros.tex}

\acmDOI{}

\acmConference[EDBT 2026]{29th International Conference on Extending Database Technology (EDBT)}{24th March-27th March, 2026}{Tampere, Finland}
\acmYear{2026}

\begin{document}
\input{./lstdefs.tex}
\lstset{style=psqlcolor}
% ****************** TITLE ****************************************

%%%%%   \title{In-memory Incremental Maintenance of Provenance Sketches}
%%%%%
%%%%%   \author{Pengyuan Li$^{\ast}$, Boris Glavic$^{\alpha}$, Dieter Gawlick$^{\beta}$, Vasudha Krishnaswamy$^{\beta}$, Zhen Hua Liu$^{\beta}$, Danica Porobic$^{\beta}$, Xing Niu$^{\beta}$}
%%%%%   \affiliation{\institution{Illinois Institute of Technology$^\ast$, University of Illinois Chicago$^{\alpha}$,Oracle$^{\beta}$}}
%%%%%   \email{pli26@hawk.iit.edu, bglavic@uic.edu}
%%%%%   \email{{dieter.gawlick, vasudha.krishnaswamy, zhen.liu, danica.porobic, xing.niu}@oracle.com}

\title{In-memory Incremental Maintenance of Provenance Sketches}

\settopmatter{authorsperrow=4}
%%%%%%%%%%
%%% \author{Pengyuan Li$^{\ast}$, Boris Glavic$^{\alpha}$, Dieter Gawlick$^{\beta}$, Vasudha Krishnaswamy$^{\beta}$, Zhen Hua Liu$^{\beta}$, Danica Porobic$^{\beta}$, Xing Niu$^{\beta}$}
%%% \affiliation{
%%%   \institution{Illinois Institute of Technology$^\ast$, University of Illinois Chicago$^{\alpha}$,Oracle$^{\beta}$}
%%%   \country{}
%%% }
%%% \email{pli26@hawk.iit.edu, bglavic@uic.edu, {dieter.gawlick, vasudha.krishnaswamy, zhen.liu, danica.porobic, xing.niu}@oracle.com}
\author{Pengyuan Li}
\affiliation{%
  \institution{Illinois Institute of Technology}
  % \streetaddress{P.O. Box 1212}
  % \city{Chicago}
  % \state{Illinois}
  \country{USA}
}
\email{pli26@hawk.iit.edu}

%%%%%%%%%%
\author{Boris Glavic}
\affiliation{%
  \institution{University of Illinois Chicago}
  % \streetaddress{P.O. Box 1212}
  % \city{Chicago}
  % \state{Illinois}
  \country{USA}
}
\email{bglavic@uic.edu}

%%%%%%%%%%
\author{Dieter Gawlick}
\affiliation{%
  \institution{Oracle Corporation}
  % \streetaddress{P.O. Box 1212}
  % \state{California}
  \country{USA}
}
\email{dieter.gawlick@oracle.com}

%%%%%%%%%%
\author{Vasudha Krishnaswamy}
\affiliation{%
  \institution{Oracle Corporation}
  % \streetaddress{P.O. Box 1212}
  % \state{California}
  \country{USA}
}
\email{vasudha.krishnaswamy@oracle.com}

%%%%%%%%%%
\author{Zhen Hua Liu}
\affiliation{%
  \institution{Oracle Corporation}
  % \streetaddress{P.O. Box 1212}
  % \state{California}
  \country{USA}
}
\email{zhen.liu@oracle.com}

%%%%%%%%%%
\author{Danica Porobic}
\affiliation{%
  \institution{Oracle Corporation}
  % \streetaddress{P.O. Box 1212}
  % \state{California}
  \country{USA}
}
\email{danica.porobic@oracle.com}

%%%%%%%%%%
\author{Xing Niu}
\affiliation{%
  \institution{Oracle Corporation}
  % \streetaddress{P.O. Box 1212}
  % \state{California}
  \country{USA}
}
\email{xing.niu@oracle.com}

\renewcommand{\shortauthors}{}
%%%%%%%%%%%%%%%%%%%%%%%%%%%%%%%%%%%%%%%%%%%%%%%%%%%%%%%%%%%%%%%%%%%%%%%%%%%%%%%%
% ABSTRACT

\begin{abstract}
    \input{./sections/abstract.tex}
\end{abstract}

\maketitle

%%%%%%%%%% VLDB block
%%%%%   \pagestyle{\vldbpagestyle}
%%%%%   \begingroup\small\noindent\raggedright\textbf{PVLDB Reference Format:}\\
%%%%%   \vldbauthors. \vldbtitle. PVLDB, \vldbvolume(\vldbissue): \vldbpages, \vldbyear.\\
%%%%%   \href{https://doi.org/\vldbdoi}{doi:\vldbdoi}
%%%%%   \endgroup
%%%%%   \begingroup
%%%%%   \renewcommand\thefootnote{}\footnote{\noindent
%%%%%   This work is licensed under the Creative Commons BY-NC-ND 4.0 International License. Visit \url{https://creativecommons.org/licenses/by-nc-nd/4.0/} to view a copy of this license. For any use beyond those covered by this license, obtain permission by emailing \href{mailto:info@vldb.org}{info@vldb.org}. Copyright is held by the owner/author(s). Publication rights licensed to the VLDB Endowment. \\
%%%%%   \raggedright Proceedings of the VLDB Endowment, Vol. \vldbvolume, No. \vldbissue\ %
%%%%%   ISSN 2150-8097. \\
%%%%%   \href{https://doi.org/\vldbdoi}{doi:\vldbdoi} \\
%%%%%   }\addtocounter{footnote}{-1}\endgroup
%%%%%   %%% VLDB block end %%%
%%%%%
%%%%%   %%% do not modify the following VLDB block %%
%%%%%   %%% VLDB block start %%%
%%%%%   \ifdefempty{\vldbavailabilityurl}{}{
%%%%%   \vspace{.3cm}
%%%%%   \begingroup\small\noindent\raggedright\textbf{PVLDB Artifact Availability:}\\
%%%%%   The source code, data, and/or other artifacts have been made available at \url{\vldbavailabilityurl}.
%%%%%   \endgroup
%%%%%   }
%%%%%   %%% VLDB block end %%%

\input{./sections/introduction}
\input{./sections/overview}

\input{./sections/related_work}
\input{./sections/background}
\input{./sections/incremental_data_model}
\input{./sections/problem_definition}
\input{./sections/incremental_operator_rules}
%\input{./sections/incremental_state_data}
\input{./sections/proof_appendix.tex}
% \input{./sections/process_workloads}
% \input{./sections/implementation_sql}
\input{./sections/implementation_in_memory}

\input{./sections/experiments}
\iftechreport{
}
\input{./sections/conclusion}
\input{./sections/artifacts}
%\input{./sections/scratch_writing_test.tex}
\bibliographystyle{ACM-Reference-Format}
\bibliography{./sections/main}

% \newpage
\iftechreport{
\begin{appendix}
\input{./sections/appendix_qlist.tex}
\input{./sections/implementation_sql.tex}
% \input{./sections/proof_appendix.tex}
% \input{./sections/exp_figs.tex}
% \input{./sections/compressCDB}
% \input{./sections/updateCDB}
% \input{./sections/updatePS_AUDB}
% \input{./sections/proof_append}
\end{appendix}
}
% \clearpage
\end{document}
\endinput

%% file: sections/acronym.tex
\setabbreviationstyle[acronym]{long-short}

\newacronym{pbds}{PDBS}{provenance-based data skipping}
\newacronym{ivm}{IVM}{incremental view maintenance}
\newacronym{mp}{IM}{incremental maintenance procedure}
\newacronym{bst}{BST}{balanced search tree}
\newacronym{si}{SI}{snapshot isolation}
\newacronym{oom}{OOM}{orders of magnitude}

%% file: sections/macros.tex
\newcommand{\card}[1]{\vert{#1}\vert}

\newcommand{\mathtext}[1]{\thickspace\text{#1}\thickspace}

\newcommand{\schemaOf}[1]{\textsc{SCH}({#1})}
\newcommand{\relSchema}{\schema{R}}

\newcommand{\att}{\ensuremath{a}\xspace}
\newcommand{\rel}{\ensuremath{R}\xspace}
\newcommand{\db}{\ensuremath{D}\xspace}
\newcommand{\query}{\ensuremath{Q}\xspace}
\newcommand{\tup}{\ensuremath{t}\xspace}

\newcommand{\udom}{\ensuremath{\mathbb{U}}\xspace}
\newcommand{\mult}[2]{{#1}^{#2}}
\newcommand{\concat}{\circ}
\newcommand{\lbbar}{\{\kern-0.5ex|}
\newcommand{\rbbar}{|\kern-0.5ex\}}
\newcommand{\bagopen}{\lbbar}
\newcommand{\bagclose}{\rbbar}
\newcommand{\bag}[1]{\lbbar {#1} \rbbar}

\newcommand{\ATT}{A}

\newcommand{\dbv}[1]{\ensuremath{\db_{#1}}\xspace}
\newcommand{\relv}[1]{\ensuremath{\rel_{#1}}\xspace}

\newcommand{\deltadbv}[1]{\ensuremath{\Delta \db_{#1}}\xspace}

\newcommand{\parti}{F}
\newcommand{\rparti}{F}

\newcommand{\frag}{f}

\newcommand{\domain}[1]{\mathbb{D}(#1)}
\newcommand{\range}{\ensuremath{\rho}\xspace}
\newcommand{\ranges}{\ensuremath{\phi}\xspace}
\newcommand{\dbranges}{\ensuremath{\Phi}\xspace}

\newcommand{\prov}[2]{P(#1,#2)}

\newcommand{\provSketch}{\ensuremath{\mathcal{P}}\xspace}
\newcommand{\instOf}[1]{\db_{#1}}

\newcommand{\schema}[1]{\textsc{Sch}(#1)}

\newcommand{\cp}[3]{\ensuremath{#1_{#2,#3}}\xspace}
\newcommand{\acp}{\ensuremath{\cp{\query}{\rel}{\parti}}\xspace}

\newcommand{\psv}[1]{\ensuremath{\provSketch_{#1}}\xspace}

\colorlet{LightBlue}{blue!30!}
\colorlet{LightRed}{red!20!}
\colorlet{LightGreen}{green!30!}

\colorlet{LRed}{red!60!}
\colorlet{LGreen}{green!70!}
\colorlet{LBlue}{blue!70!}

\definecolor{white}{rgb}{1,1,1}
\definecolor{black}{rgb}{0,0,0}
\definecolor{lgrey}{rgb}{0.9,0.9,0.9}
\definecolor{llgrey}{rgb}{0.93,0.93,0.93}
\definecolor{lllgrey}{rgb}{0.96,0.96,0.96}
\definecolor{grey}{rgb}{0.7,0.7,0.7}
\definecolor{dgrey}{rgb}{0.5,0.5,0.5}
\definecolor{red}{rgb}{1,0,0}
\definecolor{green}{rgb}{0,1,0}
\definecolor{yellow}{rgb}{1.0, 1.0, 0.0}
\definecolor{darkgreen}{rgb}{0,0.5,0}
\definecolor{darkblue}{rgb}{0,0,0.5}
\definecolor{darkpurple}{rgb}{0.5,0,0.5}
\definecolor{darkdarkpurple}{rgb}{0.3,0,0.3}
\definecolor{blue}{rgb}{0,0,1}
\definecolor{shadegreen}{rgb}{0.95,1,0.95}
\definecolor{shadeblue}{rgb}{0.95,0.95,1}
\definecolor{shadered}{rgb}{1,0.85,0.85}
\definecolor{oddRowGrey}{rgb}{0.95,0.95,0.95}
\definecolor{evenRowGrey}{rgb}{0.85,0.85,0.85}
\definecolor{tableHeadGray}{rgb}{0.85,0.85,0.85}
\definecolor{lightgrey}{rgb}{0.88,0.88,0.88}
\definecolor{lightyellow}{rgb}{1.0, 1.0, 0.88}
\definecolor{selectiveyellow}{rgb}{1.0, 0.73, 0.0}
\definecolor{shadepurple}{rgb}{0.95,0.8,0.95}
\definecolor{plcomments}{RGB}{198,224,180}

\newcommand{\parttitle}[1]{\smallskip\noindent\textbf{#1.}}

\newcommand{\projection}{\ensuremath{\Pi}}
\newcommand{\selection}{\ensuremath{\sigma}}
\newcommand{\aggregation}{\ensuremath{\gamma}}
\newcommand{\Aggregation}[2]{\aggregation_{#1;#2}}

\newcommand{\grpatts}{\ensuremath{G}\xspace}
\newcommand{\grp}{\ensuremath{g}\xspace}

\newcommand{\union}{\ensuremath{\cup}}

\newcommand{\difference}{\ensuremath{-}}

\newcommand{\join}{\ensuremath{\bowtie}}
\newcommand{\crossprod}{\ensuremath{\times}}
\newcommand{\duprem}{\ensuremath{\delta}}
\newcommand{\fragmergeop}{\ensuremath{\mu}\xspace}

\newcommand{\releval}[2]{#1(#2)}

\newcommand{\topk}{\ensuremath{\tau_{k, \orderbylist}}\xspace}
\newcommand{\pos}[3]{{\bf{pos}}( #1,#2,#3 )\xspace}

\newcommand{\orderbylist}{\ensuremath{O}\xspace}
\newcommand{\multiOfTupInRel}[2]{#1 (#2) \xspace}

\newcommand{\thead}[1]{\cellcolor{black}\textcolor{white}{#1}}

\numberwithin{theo}{section}
\newtheorem{lem}{Lemma}[section]

\newtheorem{exam}{Example}
\numberwithin{exam}{section}
\newtheorem{defi}{Definition}
\numberwithin{defi}{section}

{% begin code
\noindent \textsc{Proof Sketch.}%
}%
{\qedsymbol}% end code

\crefname{exam}{ex.}{ex.}
\Crefname{exam}{Ex.}{Ex.}
\Crefname{figure}{Fig.}{Fig.}
\Crefname{section}{Sec.}{Sec.}
\Crefname{defi}{Def.}{Def.}
\Crefname{theo}{Thm.}{Thm.}

\newcommand{\BGI}[1]{\todo[inline,color=shadegreen]{\textbf{Boris says:$\,$} #1}}

\newcommand{\provranges}[3]{\provSketch_{#2}(#1,#2,#3)}

\newcommand{\dbInst}[1]{\db_{#1}}

\newcommand{\incremaintain}{\mathcal{I}}
\newcommand{\ieval}[3]{\incremaintain(#1,#2,#3)}
\newcommand{\ievalns}[2]{\incremaintain(#1,#2)}

\newcommand{\deltaOf}[2]{\ensuremath{\Delta(#1,#2)}\xspace}
\newcommand{\deltadb}{\Delta D}
\newcommand{\uniondb}{%
  \mathrel{
  \ooalign{\hss$\cup$\hss\cr%
  \kern0.46ex\raise0.4ex\hbox{\scalebox{0.5}{$\bullet$}}}}}
\newcommand{\biguniondb}{%
  \mathrel{
  \ooalign{\hss$\bigcup$\hss\cr%
  \kern0.46ex\raise0.4ex\hbox{\scalebox{0.5}{$\bullet$}}}}}

\newcommand{\insertdelta}{
  \mathrel{\textcolor{darkgreen}{
  \ooalign{\hss$\Delta$\hss\cr%
  \kern0.4ex\raise0.2ex\hbox{\scalebox{0.45}{+}}}}}}

\newcommand{\deletedelta}{
  \mathrel{\textcolor{red}{
  \ooalign{\hss$\Delta$\hss\cr%
  \kern0.43ex\raise0.1ex\hbox{\scalebox{0.7}{-}}}}}}

\newcommand{\insertdeltadb}{\ensuremath{\insertdelta\kern-0.7ex\db}}
\newcommand{\deletedeltadb}{\ensuremath{\deletedelta\kern-0.7ex\db}}

\newcommand{\insertdeltaprovSketch}{\ensuremath{\insertdelta\kern-0.7ex\provSketch}}
\newcommand{\deletedeltaprovSketch}{\ensuremath{\deletedelta\kern-0.7ex\provSketch}}

\newcommand{\insertdeltatup}{\ensuremath{\insertdelta\kern-0.7ex\tup}}
\newcommand{\deletedeltatup}{\ensuremath{\deletedelta\kern-0.7ex\tup}}

\newcommand{\psFor}[1]{\provSketch[#1]}

\newcommand{\deltaprovSketch}{\ensuremath{\Delta\provSketch}}
\newcommand{\statedata}{\ensuremath{\mathcal{S}}\xspace}
\newcommand{\statedataGBs}{\mathcal{G}}
\newcommand{\pair}[2]{\langle #1, #2 \rangle}
\newcommand{\deltaPair}[2]{\Delta\langle #1, #2 \rangle}
\newcommand{\adeltaPair}{\deltaPair{\tup}{\provSketch}}
\newcommand{\insertDeltaPair}[2]{\insertdelta\,\, \kern-0.8ex \langle #1, #2 \rangle}
\newcommand{\deleteDeltaPair}[2]{\deletedelta\,\, \kern-0.8ex \langle #1, #2 \rangle}

\newcommand{\anno}[2]{{\bf annotate}(#1, #2)}
\newcommand{\annoF}{{\bf annotate}}

\newcommand{\annoRel}{\ensuremath{\mathscr{R}}\xspace}
\newcommand{\makeAnno}[1]{\ensuremath{\mathscr{#1}}\xspace}
\newcommand{\annoDeltaRel}{\ensuremath{\Delta \annoRel}\xspace}
\newcommand{\makeAnnoDelta}[1]{\ensuremath{\Delta \mathscr{#1}}\xspace}

\newcommand{\annoDB}{\ensuremath{\mathscr{D}}\xspace}
\newcommand{\annoDeltaDB}{\ensuremath{\Delta \annoDB}\xspace}
\newcommand{\annoIDeltaDB}{\ensuremath{\insertdelta \annoDB}\xspace}
\newcommand{\annoDDeltaDB}{\ensuremath{\deletedelta \annoDB}\xspace}
\newcommand{\annoQDelta}{\ensuremath{\Delta \query}}

\newcommand{\deltaRel}{\Delta \rel}

\newcommand{\qtopsellers}{\ensuremath{Q_{Top}}\xspace}

\newcommand{\tblSales}{\textbf{sales}}

\newcommand{\ivmAbbr}{IVM\xspace}
\newcommand{\tpchds}{TPC-H\xspace}
\newcommand{\crimeds}{Crime\xspace}
\newcommand{\impAbbr}{IMP\xspace}
\newcommand{\fullmaintenance}{FM\xspace}
\newcommand{\nonsketch}{NS\xspace}

\newcommand{\fragCnt}{\ensuremath{\mathscr{F}}}

\newcommand{\tupCnt}{\textsc{cnt}}
\newcommand{\tupSum}{\textsc{sum}}

\newcommand{\makeagg}[1]{\ensuremath{\mathbf{#1}}\xspace}
\newcommand{\fccnt}{\makeagg{count}}
\newcommand{\fcsum}{\makeagg{sum}}
\newcommand{\fcavg}{\makeagg{avg}}
\newcommand{\fcmin}{\makeagg{min}}
\newcommand{\fcmax}{\makeagg{max}}

\newcommand{\makespecificagg}[1]{\ensuremath{\Aggregation{{#1}(\att)}{\grpatts}(\query)}\xspace}

\newcommand{\sumAgg}{\makespecificagg{\fcsum}}

\newcommand{\minAgg}{\makespecificagg{\fcmin}}

\newcommand{\qhaving}{\ensuremath{Q_{having}}\xspace}
\newcommand{\qgroups}{\ensuremath{Q_{groups}}\xspace}
\newcommand{\qjoin}{\ensuremath{Q_{join}}\xspace}
\newcommand{\qjoinsel}{\ensuremath{Q_{joinsel}}\xspace}
\newcommand{\qsketch}{\ensuremath{Q_{sketch}}\xspace}
\newcommand{\qbloom}{\ensuremath{Q_{bloom}}\xspace}
\newcommand{\qselpd}{\ensuremath{Q_{selpd}}\xspace}
\newcommand{\qspace}{\ensuremath{Q_{space}}\xspace}
\newcommand{\qendtoend}{\ensuremath{Q_{endtoend}}\xspace}

\newcommand{\qtopk}{\ensuremath{Q_{top-k}}\xspace}

\newcommand{\queryclass}[1]{\ensuremath{\mathbb{Q}^{#1}}\xspace}
\newcommand{\tcorrect}{\textbf{tuple correctness}\xspace}
\newcommand{\Tcorrect}{\textbf{Tuple Correctness}\xspace}
\newcommand{\fcorrect}{\textbf{fragment correctness}\xspace}
\newcommand{\Fcorrect}{\textbf{Fragment Correctness}\xspace}

\newcommand{\tupout}[1]{\ensuremath{\tup_{o_{#1}}}\xspace}
\newcommand{\psout}[1]{\ensuremath{\provSketch_{o_{#1}}}\xspace}
\newcommand{\tupsIn}[1]{\ensuremath{\mathbb{T}{(#1)}}\xspace}
\newcommand{\fragsIn}[1]{\ensuremath{\mathbb{F}{(#1)}}\xspace}
\newcommand{\tupins}[1]{\ensuremath{\tup_{i_{#1}}}\xspace}
\newcommand{\tupdel}[1]{\ensuremath{\tup_{d_{#1}}}\xspace}
\newcommand{\psins}[1]{\ensuremath{\provSketch_{i_{#1}}}\xspace}
\newcommand{\psdel}[1]{\ensuremath{\provSketch_{d_{#1}}}\xspace}
\newcommand{\updDB}{\ensuremath{\db'}\xspace}
\newcommand{\annoupdDB}{\ensuremath{\mathscr{D}'}\xspace}
\newcommand{\updDBInst}[1]{\ensuremath{\db'_{#1}}\xspace}

\newcommand{\fragSet}[1]{\ensuremath{\mathbb{S}({#1})}\xspace}

\newenvironment{takeaway}[1]
{\medskip %
  \begin{mdframed}[ %
    linecolor=black, %
    linewidth=0.8pt, %
    roundcorner=10pt, %
    backgroundcolor=black!10, %
    innertopmargin=1.5pt, %
    innerbottommargin=1.5pt, %
    innerrightmargin=2.5pt, %
    innerleftmargin=2.5pt %
    skipbelow=0pt %
    ] %
    \emph{Insights:}\, {\small #1}}
  {\end{mdframed}\par}

\newrobustcmd{\revision}[1]{{#1}}

%% file: lstdefs.tex
%%%%%%%%%%%%%%%%%%%%%%%%%%%%%%%%%%%%%%%%
% Colors
%%%%%%%%%%%%%%%%%%%%%%%%%%%%%%%%%%%%%%%%
\definecolor{lstpurple}{rgb}{0.5,0,0.5}
\definecolor{lstred}{rgb}{1,0,0}
\definecolor{lstreddark}{rgb}{0.7,0,0}
\definecolor{lstredl}{rgb}{0.64,0.08,0.08}
\definecolor{lstmildblue}{rgb}{0.66,0.72,0.78}
\definecolor{lstblue}{rgb}{0,0,1}
\definecolor{lstmildgreen}{rgb}{0.42,0.53,0.39}
\definecolor{lstgreen}{rgb}{0,0.5,0}
\definecolor{lstorangedark}{rgb}{0.6,0.3,0}
\definecolor{lstorange}{rgb}{0.75,0.52,0.005}
\definecolor{lstorangelight}{rgb}{0.89,0.81,0.67}
\definecolor{lstbeige}{rgb}{0.90,0.86,0.45}

% Declare bold typewriter font with Computer Modern
\DeclareFontShape{OT1}{cmtt}{bx}{n}{<5><6><7><8><9><10><10.95><12><14.4><17.28><20.74><24.88>cmttb10}{}

%%%%%%%%%% SQL + proveannce listing settings
\lstdefinestyle{psql}
{
tabsize=2,
basicstyle=\footnotesize\upshape\ttfamily,
language=SQL,
morekeywords={PROVENANCE,BASERELATION,INFLUENCE,COPY,ON,TRANSPROV,TRANSSQL,TRANSXML,CONTRIBUTION,COMPLETE,TRANSITIVE,NONTRANSITIVE,EXPLAIN,SQLTEXT,GRAPH,IS,ANNOT,THIS,XSLT,MAPPROV,cxpath,OF,TRANSACTION,SERIALIZABLE,COMMITTED,INSERT,INTO,WITH,SCN,UPDATED,WINDOW,SET,UPDATE},
extendedchars=false,
keywordstyle=\bfseries,
mathescape=true,
escapechar=@,
sensitive=true
}

%%%%%%%%%% SQL + proveannce listing settings - colorful version
%basicstyle=\footnotesize\upshape\ttfamily,
\lstdefinestyle{psqlcolor}
{
tabsize=2,
basicstyle=\footnotesize\upshape\ttfamily,
language=SQL,
morekeywords={PROVENANCE,BASERELATION,INFLUENCE,COPY,ON,TRANSPROV,TRANSSQL,TRANSXML,CONTRIBUTION,COMPLETE,TRANSITIVE,NONTRANSITIVE,EXPLAIN,SQLTEXT,GRAPH,IS,ANNOT,THIS,XSLT,MAPPROV,cxpath,OF,TRANSACTION,SERIALIZABLE,COMMITTED,INSERT,INTO,WITH,SCN,UPDATED,FOLLOWING,RANGE,UNBOUNDED,PRECEDING,OVER,PARTITION,WINDOW,SUM,SET},
extendedchars=false,
keywordstyle=\bfseries\color{lstpurple},
deletekeywords={count,min,max,avg,sum,lag,first_value,last_value},
keywords=[2]{count,min,max,avg,sum,lag,first_value,last_value,lead,row_number},
keywordstyle=[2]\color{lstblue},
stringstyle=\color{lstreddark},
commentstyle=\color{lstgreen},
mathescape=true,
escapechar=@,
sensitive=true
}

%%%%%%%%%% DATALOG style
\lstdefinestyle{datalog}
{
basicstyle=\footnotesize\upshape\ttfamily,
language=prolog
}

%%%%%%%%%% listings settings for pseudo code
\lstdefinestyle{pseudocode}
{
  tabsize=3,
  basicstyle=\small,
  language=c,
  morekeywords={if,else,foreach,case,return,in,or},
  extendedchars=true,
  mathescape=true,
  literate={:=}{{$\gets$}}1 {<=}{{$\leq$}}1 {!=}{{$\neq$}}1 {append}{{$\listconcat$}}1 {calP}{{$\cal P$}}{2},
  keywordstyle=\color{lstpurple},
  escapechar=&,
  numbers=left,
  numberstyle=\color{lstgreen}\small\bfseries,
  stepnumber=1,
  numbersep=5pt,
}

%%%%%%%%%% XML listings settings
\lstdefinestyle{xmlstyle}
{
  tabsize=3,
  basicstyle=\small,
  language=xml,
  extendedchars=true,
  mathescape=true,
  escapechar=£,
  tagstyle=\color{keywordpurple},
  usekeywordsintag=true,
  morekeywords={alias,name,id},
  keywordstyle=\color{lstred}
}

%% file: sections/abstract.tex
Provenance-based data skipping~\cite{DBLP:journals/pvldb/NiuGLLGKLP21}
compactly over-approximates the provenance of a query using so-called
provenance sketches and utilizes such sketches to speed-up the
execution of subsequent queries by skipping irrelevant data. However,
a sketch captured at some time in the past may become stale if the data
% no longer
% correctly reflect what data is relevant for a query if the data has
has been updated subsequently. Thus, there is a need to maintain
provenance sketches. In this work, we introduce
\textbf{I}n-Memory incremental \textbf{M}aintenance of
\textbf{P}rovenance sketches (\impAbbr), a framework for maintaining
sketches incrementally under updates. At the core of
\impAbbr is an incremental query engine for data annotated with
 sketches that exploits the coarse-grained nature of
 sketches to enable novel optimizations.
 % that trade performance of
% maintenance for sketch size which in turn affects the performance of
% query answering with sketches.
We experimentally demonstrate that \impAbbr significantly reduces the cost of
sketch maintenance, thereby enabling the use of provenance sketches for a broad
range of workloads that involve updates.
%% Removed: we develop strategies for automatically deciding when to
%% maintain a sketch and when to drop or regenerate it from scratch.

% We develop an in-memory engine that applies \impAbbr technique to
% update the sketches for the needs of both immediateness and
% deferment. Using \impAbbr approach, we evaluate experimentally that
% the in memory engine can efficiently generate the delta provenance
% sketches for small numbers of input delta tuples.

% a technique to speed up query processing and it can significantly
% improve query performance. However, sketches may not correctly
% refrect which parts contain provenance information to answer the
% queries when the database changes under updates. In this paper, we
% introduce In-Memory Incremental Maintenance of Provenance Sketch
% (\impAbbr), a technique to maintain the provenance sketches
% incrementally. We develop an in-memory engine that applies \impAbbr
% technique to update the sketches for the needs of both immediateness
% and deferment. Using \impAbbr approach, we evaluate experimentally
% that the in memory engine can efficiently generate the delta
% provenance sketches for small numbers of input delta tuples.

%%% Local Variables:
%%% mode: latex
%%% TeX-master: "../imp"
%%% End:

%% file: sections/introduction.tex
\section{Introduction}%
\label{sec:introduction}
%%%%%%%%%%%%%%%%%%%%%%%%%%%%%%%%%%%%%%%%%%%%%%%%%%%%%%%%%%%%%%%%%%%%%%%%%%%%%%%%
% Techniques to speed up query processing have been studied widely.
Database engines take advantage of physical design such as index structures,
zone maps~\cite{DBLP:conf/vldb/Moerkotte98} and partitioning to prune irrelevant
data as early as possible during query evaluation. \revision{In order to prune
data, database systems need to determine statically (at query compile time) what
data is needed to answer a query and which physical design artifacts to use to
skip irrelevant data. For instance, to answer a query with a \lstinline!WHERE!
clause condition \lstinline!A = 3! filtering the rows of a table \lstinline!R!,
the optimizer may decide to use an index on \texttt{A} to filter out rows that
do not fulfill the condition. However, as was demonstrated
in~\cite{DBLP:journals/pvldb/NiuGLLGKLP21}, for important classes of queries
like queries involving top-k and aggregation with \lstinline!HAVING!, it is not
possible to determine \emph{statically} what data is needed, motivating the use
of \emph{dynamic relevance analysis} techniques that determine during query
execution what data is relevant to answer a query.}
In~\cite{DBLP:journals/pvldb/NiuGLLGKLP21} we introduced such a dynamic
relevance analysis technique called \gls{pbds}. In \gls{pbds}, we encode what
data is relevant for a query as a so-called \emph{provenance sketch}. Given a
range-partition of a table accessed by a query, a provenance sketch records
which fragments of the partition contain provenance. That is, provenance
sketches compactly encode an over-approximation of the provenance of a query.
\cite{DBLP:journals/pvldb/NiuGLLGKLP21} presents safety conditions that ensure a
sketch is \emph{sufficient}, i.e., evaluating the query over the data
represented by the sketch is guaranteed to produce the same result as evaluating
the query over the full database. Thus, sketches are used to speed up queries by
filtering data not in the sketch.

\begin{exam}\label{ex:provenance-sketch-capture}
  Consider the database shown in \Cref{fig:wordload-query-capture-reuse}
  and query \qtopsellers that returns  products whose total sale volume is
  greater than \$5000.  % Evaluating this query over the table shown in
  % \Cref{fig:wordload-query-capture-reuse}
  % produces a single result tuple
  % $(Apple, 5074)$.
  The provenance of the single result tuple $(Apple, 5074)$  % wrt. table \texttt{sales}
  are the two tuples (tuples $s_3$ and $s_4$ shown with purple
  background), as the group for Apple is the only group that fulfills the
  \lstinline!HAVING! clause. To create a provenance sketch for this query, we
   select a range-partition of the \texttt{sales} table that optionally % the range
  % partition does not necessarily have to
  may correspond to the physical storage
  layout of this table. For instance, we may choose to partition on attribute \texttt{price} % $\parti_{price}$ that
  % partitions the table
  based on ranges $\ranges_{price}$:\\[-4mm]
%%%%%%%%%%%%%%%%%%%%%%%%%%%%%%%%%%%%%%%%%%%%%%%%%%%%%%%%%%%%%%%%%%%%%%%%%%%%%%%%
{\small
  \begin{align*}
    % \ranges_{price} = &
                        \{ \range_1 = [1,600], \range_2 = [601,1000], \range_3 = [1001,1500], \range_4 = [1501,10000] \}
  \end{align*}
  }
%%%%%%%%%%%%%%%%%%%%%%%%%%%%%%%%%%%%%%%%%%%%%%%%%%%%%%%%%%%%%%%%%%%%%%%%%%%%%%%%
  In \Cref{fig:wordload-query-capture-reuse}, we show the fragment $\frag_i$
  for the range $\range_i$ each tuple belongs to. Two
  fragments ($\frag_3$ and $\frag_4$ highlighted in red) contain provenance
  and, thus, the provenance sketches for \qtopsellers wrt.
  $\parti_{\tblSales,price}$ is $\provSketch = \{\range_3, \range_4\}$.
  Evaluating the query over the sketch's data is guaranteed to
  produce the same result as evaluation on  the full
  database.\footnote{In general, this is not the case for non-monotone queries.
    The safety check from \cite{DBLP:journals/pvldb/NiuGLLGKLP21} can be used
    to test whether a particular partition for a table is
    safe. The partition used here fulfills this condition.}
\end{exam}
%%%%%%%%%%%%%%%%%%%%%%%%%%%%%%%%%%%%%%%%%%%%%%%%%%%%%%%%%%%%%%%%%%%%%%%%%%%%%%%%

%%%%%%%%%%%%%%%%%%%%%%%%%%%%%%%%%%%%%%%%%%%%%%%%%%%%%%%%%%%%%%%%%%%%%%%%%%%%%%%%
% FIGURE provenance and provenance sketch
\definecolor{fragcolor}{rgb}{0.005,0.10,0.61}
\definecolor{fragcolorhighlight}{rgb}{0.81,0.10,0.1}
\begin{figure}[t]
  \centering
  %%%%%%%%%%%%%%%%%%%%%%%%%%%%%%%%%
  \begin{minipage}{1\linewidth}
    %%%%%%%%%%%%%%%%%%%%%%%%%%%%%%%%% query
    \begin{minipage}{0.73 \textwidth}
      % \centering
      \qtopsellers
\begin{lstlisting}
SELECT brand, SUM(price * numSold) AS rev
FROM sales
GROUP BY brand
HAVING SUM(price * numSold) > 5000
\end{lstlisting}
    \end{minipage}
    %%%%%%%%%%%%%%%%%%%%%%%%%%%%%%%%%%%%%%%%
    \begin{minipage}{0.2 \textwidth}
      \centering
      %%%%%%%%%%%%%%%%%%%%%%%%%%%%%%%%%%%%%%%%
      {\small
      \begin{tabular}{|c|c|}
        \thead{brand} & \thead{rev} \\ \hline
        Apple& 5074\\
        \hline
      \end{tabular}
      }
      %%%%%%%%%%%%%%%%%%%%%%%%%%%%%%%%%%%%%%%%
    \end{minipage}
    %%%%%%%%%%%%%%%%%%%%%%%%%%%%%%%%%
  \end{minipage}
  %%%%%%%%%%%%%%%%%%%%%%%%%%%%%%%%% table sales
  \begin{minipage}{1.0\linewidth}
    \centering
    {\small
    \tblSales                                                                                                           \\[2mm]
    \begin{tabular}{c|c|c|c|c|c|c}
      & \thead{sid} & \thead{brand} & \thead{productName} & \thead{price} & \thead{numSold} \\
      \cline{2-6}
      $s_1$   & 1           & Lenovo        & ThinkPad T14s Gen 2 & 349           & 1 & \textcolor{fragcolor}{$f_1$} \\
      $s_2$   & 2           & Lenovo        & ThinkPad T14s Gen 2 & 449           & 2 & \textcolor{fragcolor}{$f_1$} \\
      \rowcolor{shadepurple}$s_3$ & 3           & Apple         & MacBook Air 13-inch & 1199          & 1 & \textcolor{fragcolorhighlight}{$f_3$} \\
      \rowcolor{shadepurple}$s_4$ & 4           & Apple         & MacBook Pro 14-inch & 3875          & 1 & \textcolor{fragcolorhighlight}{$f_4$} \\
      $s_5$                       & 5           & Dell          & Dell XPS 13 Laptop  & 1345          & 1 & \textcolor{fragcolorhighlight}{$f_3$} \\
      $s_6$   & 6           & HP            & HP ProBook 450 G9   & 999           & 4 & \textcolor{fragcolor}{$f_2$} \\
      $s_7$   & 7           & HP            & HP ProBook 550 G9   & 899           & 1 & \textcolor{fragcolor}{$f_2$} \\
      \cline{2-6}
    \end{tabular}
    }
  \end{minipage}
  \caption{Example query and relevant subsets of the database.}
  \label{fig:wordload-query-capture-reuse}
\end{figure}
% END figure provenance and provenance sketch
%%%%%%%%%%%%%%%%%%%%%%%%%%%%%%%%%%%%%%%%%%%%%%%%%%%%%%%%%%%%%%%%%%%%%%%%%%%%%%%%

%%%%%%%%%%%%%%%%%%%%%%%%%%%%%%%%%%%%%%%%%%%%%%%%%%%%%%%%%%%%%%%%%%%%%%%%%%%%%%%%
%   workflow overview
\begin{figure*}[h]
  \begin{center}
    \makebox[\textwidth]{\includegraphics[width = \textwidth]{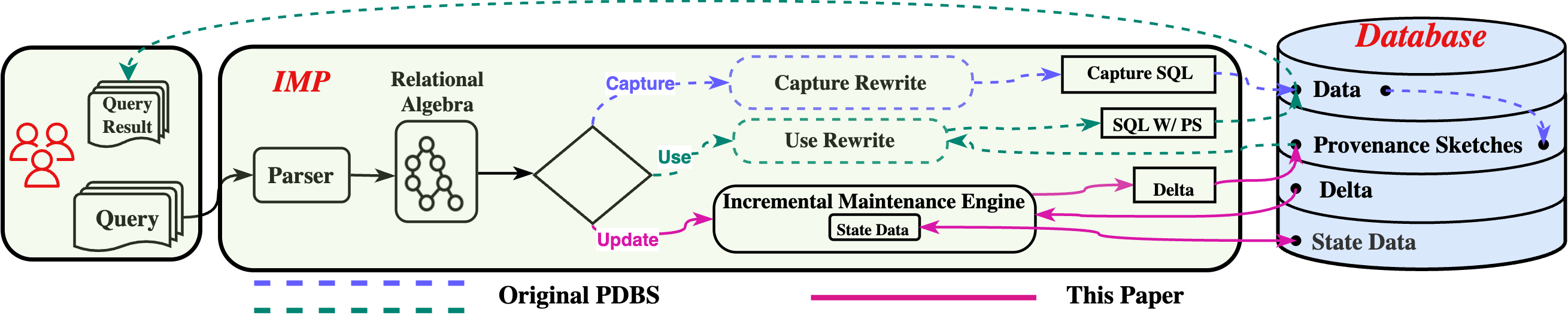}}
  \end{center}
  \vspace{-6mm}
  \caption{\impAbbr manages a set of sketches. For each incoming query, \impAbbr determines whether to (i) capture a new sketches, (ii) use an existing non-stale sketch, or (iii) incrementally maintain a stale sketch and then utilize the updated sketch to answer the query.}
\label{fig:workflow_overview}
\end{figure*}
%%%%%%%%%%%%%%%%%%%%%%%%%%%%%%%%%%%%%%%%%%%%%%%%%%%%%%%%%%%%%%%%%%%%%%%%%%%%%%%%

% Creating (\emph{capturing}) a sketch requires execution of an instrumented version of a query. This cost is
%   amortized over time by using the sketch to answer future queries such as
%   repeated executions of \qtopsellers or queries with the same structure, i.e.,
%   using a different threshold in the \lstinline!HAVING! clause.
As demonstrated in~\cite{DBLP:journals/pvldb/NiuGLLGKLP21}, provenance-based
data skipping can significantly improve query performance --- we pay upfront for
creating sketches  for some of the queries of a workload and then amortize this cost by using  sketches to answer future queries by skipping irrelevant data.
To create, or \emph{capture}, a sketch for a query $\query$ we execute an instrumented version of $\query$. Similarly, to use a sketch for a query $\query$, this query is instrumented to filter out data that does not belong to the sketch.
For instance, consider the sketch for \qtopsellers from
\Cref{ex:provenance-sketch-capture} containing two ranges $\range_3 =
[1001,1500]$ and $\range_4 = [1501,10000]$. To skip irrelevant data, we create a disjunction of conditions testing that each tuple passing
the \lstinline!WHERE! clause belongs to the sketch, i.e., has a price within $\range_3$ or
$\range_4$:\footnote{Note that the conditions for adjacent ranges in a sketch
  can be merged. Thus, the actual instrumentation would be \lstinline!WHERE
  price BETWEEN 1001 AND 10000!.}

\begin{center}
%%%%%%%%%%%%%%%%%%%%%%%%%%%%%%%%%%%%%%%%%%%%%%%%%%%%%%%%%%%%%%%%%%%%%%%%%%%%%%%%
\begin{lstlisting}
WHERE (price BETWEEN 1001 AND 1500)
   OR (price BETWEEN 1501 AND 10000)
\end{lstlisting}
%%%%%%%%%%%%%%%%%%%%%%%%%%%%%%%%%%%%%%%%%%%%%%%%%%%%%%%%%%%%%%%%%%%%%%%%%%%%%%%%
\end{center}$\,$\\[-8mm]

\gls{pbds} enables
databases to exploit physical design for new classes of queries, significantly improving the performance of aggregation queries with \lstinline!HAVING! and top-k queries~\cite{DBLP:journals/pvldb/NiuGLLGKLP21} \revision{and, more generally, any query where only a fraction of the database is relevant for answering the query.}
\revision{For instance, for a top-k query only tuples contributing to a result tuple in the top-k are relevant, but which tuples are in the top-k can only be determined at runtime. A counterexample are queries with selection conditions with low selectivity for which the database can effectively filter the data without \gls{pbds}.} %, e.g., by using an index.}

% \BG{Here may be the place to state what types of queries this works well for?}
However, just like materialized views, a sketch captured in the past may no longer correctly reflect what data is needed
(has become \emph{stale}) when the database is updated. % For instance, consider
% an update to the table \tblSales that inserts a new tuple in
% \Cref{ex:stale-sketches} , which will lead the current sketch invalid to
% correctly answer the query \qtopsellers.
The sketch then has to be maintained
to be valid for the current version of the database. % or we have to drop the sketch.

% While sketches are effective in improving query
% performance,
% \PL{This paragraph comes earlier, put later.\\
% }

% %%%%%%%%%%%%%%%%%%%%%%%%%%%%%%%%%%%%%%%%%%%%%%%%%%%%%%%%%%%%%%%%%%%%%%%%%%%%%%%%
% %   workflow overview
% \begin{figure*}[h]
%   \begin{center}
%     \makebox[\textwidth]{\includegraphics[width = \textwidth]{figs/imp_overview.png}}
%   \end{center}
%   \vspace{-6mm}
%   % \BG{use \impAbbr instead of GProM -from PL:resolved}
% % \BG{@pengyuan: please add the reference to the right place as suggested by the reviewer!}
%   \caption{\impAbbr manages a set of sketches. For each incoming query, \impAbbr determines whether to (i) capture a new sketches, (ii) use an existing non-stale sketch, or (iii) incrementally maintain a stale sketch and then utilize the updated sketch to answer the query.}
% \label{fig:workflow_overview}
% \end{figure*}
% %%%%%%%%%%%%%%%%%%%%%%%%%%%%%%%%%%%%%%%%%%%%%%%%%%%%%%%%%%%%%%%%%%%%%%%%%%%%%%%%

%%%%%%%%%%%%%%%%%%%%%%%%%%%%%%%%%%%%%%%%%%%%%%%%%%%%%%%%%%%%%%%%%%%%%%%%%%%%%%%%
\begin{exam}[Stale Sketches]\label{ex:stale-sketches}
  Continuing with our running example, consider the effect of inserting a new tuple

  \begin{center}
    $s_8$ = \verb!(8, HP, HP ProBook 650 G10, 1299, 1)!
  \end{center}

  into relation \tblSales. Running \qtopsellers over the updated table returns a
  second result tuple \texttt{(HP, 6194)} as the total revenue for HP is now
  above the threshold specified in the \lstinline!HAVING! clause. For the
  updated database, the three tuples for HP also belong to the provenance.
  Thus, the sketch has become stale as it is missing the range $\range_2$ which
  contains these tuples. Evaluating \qtopsellers over the outdated sketch leads
  to an incorrect result that misses the group for HP.
\end{exam}
%%%%%%%%%%%%%%%%%%%%%%%%%%%%%%%%%%%%%%%%%%%%%%%%%%%%%%%%%%%%%%%%%%%%%%%%%%%%%%%%
%%%%%%%%%%%%%%%%%%%%%%%%%%%%%%%%%%%%%%%%
% full maintenance
%% Since the reviewer mentioned that the F and Q is not clear enough, can we add
%% some note int the paper stating "Formal definition is in Section 4.1"
\revision{Consider a
partition $\parti$ of a table $\rel$ accessed by a query $\query$. We use $\acp$ to denote
the \emph{capture query} for $\query$ and $\parti$, generated using the rewrite rules from
\cite{DBLP:journals/pvldb/NiuGLLGKLP21}. Such a query propagates coarse-grained provenance information and ultimately returns a sketch.  A straightforward approach to maintain sketches under updates is
\emph{full maintenance} which means that we rerun the sketch's \emph{capture
  query} $\acp$ to regenerate the sketch.}
% over the new version
% of the database.
Typically, $\acp$ is more expensive than $\query$. % capture queries are more expensive than the queries
% for which sketches are generated.
Thus, frequent execution of capture queries
is not feasible. % Full maintenance is only beneficial if both queries and
% updates are executed in batches such that a maintained sketch can be used by a
% large enough number of queries before it becomes stale, thus, amortizing the
% cost of maintenance.
%%%%%%%%%%%%%%%%%%%%%%%%%%%%%%%%%%%%%%%%
% A sketch may become stale when the database is updated after the
% sketch has been created. In this work, we study the problem of maintaining
% sketches under updates such that a sketch created in the past can be updated to
% be valid for the current state of the database. Towards this goal we develop
% incremental maintenance techniques for sketches. Sketches are captured using
% annotation propagation techniques to instrument a query. In a first step, for
% each input tuple the instrumented query determines which fragment the tuple
% belongs to. The singleton set containing this fragment is the initial sketch
% for the tuple. These sketches are then propagated and merged such that the
% final result of the instrumented query is the query's sketch.
Alternatively, we could employ \gls{ivm} techniques~\cite{DBLP:conf/sigmod/GuptaMS93, DBLP:journals/vldb/KochAKNNLS14} to maintain $\acp$. However, capture queries use specialized data types and functions to efficiently implement common operations related to sketches. For instance, we use
bitvectors to encode sketches compactly and utilize optimized (aggregate) functions and comparison operators for this encoding. To give two concrete examples, a function implementing binary search over the set of ranges for a sketch is used to determine which fragment an input tuple belongs to and an aggregation function that computes the bitwise-or of multiple bitvectors is used to implement the union of a set of partial sketches. To the best of our knowledge these operations are not supported by state-of-the-art \gls{ivm} frameworks.
% Existing view maintenance techniques can be used to incrementally maintain the
% result of $\acp$, the sketch generated by this query.
% Using \gls{ivm} techniques for relational queries would disregard
% the nature of capture queries which propagates partial sketches alongside
% intermediate query results as \emph{annotations} on rows.
% Current view
% maintenance techniques do not support all the operators which means that it
% cannot directly be applied to maintain the provenance sketches.
Furthermore,
sketches are compact over-approximations of the provenance
of a query that are \emph{sound}: evaluating the query over the sketch yields the same
result as evaluating it over the full database. It is often possible to
further over-approximate the sketch, trading improved maintenance performance
for increased sketch size.
Existing \gls{ivm} methods do not support such trade-offs as they have to ensure that incremental maintenance yields the same result as full maintenance.
% Systems like \emph{DBToaster}~\cite{DBLP:journals/vldb/KochAKNNLS14} use annotations to represent incremental maintenance state, e.g., to store the multiplicity of a tuple or an aggregation result. While these types of annotations are related to semiring provenance annotations~\cite{DBLP:journals/sigmod/KarvounarakisG12}, there is no notion of approximation as is the case for provenance sketches.

% The current database systems can do limited work to maintain queries. Only
% simple work can be done in PosqgreQL. For operators like \lstinline!HAVING!,
% \lstinline!ORDER BY! and \lstinline!LIMIT! are not supported. These
% operators are commonly used, and in our study, it is necessary to support them.
% Oracle database system can maintain materialized views as well, while there are
% many restrictions when updating the views, for example, primary keys should be
% in the result when maintaining the join operators. In our study, we should relax
% the constraints and maintain join in all kinds of form.

% Sketches are captured using annotation propagation techniques.
% \emph{DBToaster}~\cite{DBLP:journals/vldb/KochAKNNLS14} supports annotations, while it
% uses the annotations for tracking purpose. This does not align with our needs,
% because during computation of sketches, there are operation like merging needed
% to combine multiple sketches into one for operators like \lstinline!SUM!, which
% is not applicable in DBToaster.

%%%%%%%%%%%%%%%%%%%%%%%%%%%%%%%%%%%%%%%%
In this work, we study the problem of maintaining sketches under updates such
that a sketch created in the past can be updated to be valid for the current
state of the database. Towards this goal we develop an incremental maintenance
framework for sketches that respects the approximate nature of sketches, has specialized data structures for representing data annotated with sketches, and maintenance rules tailored for data annotated with sketches. % \bluetext{Sketches are captured using annotation propagation
% techniques to instrument a query. In a first step, for each input tuple the
% instrumented query determines which fragment the tuple belongs to. The singleton
% set containing this fragment is the initial sketch for the tuple. These sketches
% are then propagated and merged such that the final result of the instrumented
% query is the query’s sketch. }
% \PL{above \bluetext{BLUE}, I want to put this paragraph here, but it seems not
% appropriate. but if we still put it back in the previous position. it also
% seems not appropriate. this paragraph discuss how we capture the PS by
% propagation and merging. If we remove this, it seems we lack something about capturing.}

% annotated data model
We start by introducing a data model where each row is associated with a sketch
and then develop incremental maintenance rules for operators over such annotated
relations. We then present an implementation of these rules in an in-memory
incremental engine called \impAbbr (\textbf{I}ncremental \textbf{M}aintenance of
\textbf{P}rovenance Sketches). The input to this engine is a set of annotated delta
tuples (tuples that are inserted / deleted) that we extract from a backend DBMS.
% The engine features incremental versions of operators with annotation
% semantics.
To maintain a sketch created by a capture query $\acp$ at some point in the past, we extract the
delta between the current version of the database and the database instance at
the original time of capture (or the last time we maintained the sketch) and then
feed this delta as input to our incremental engine to compute a delta for the
sketch. % We demonstrate that applying the sketch delta to the previous version
% of the sketch, yields a valid sketch for the current version of the database.
% To
% incrementally maintain the annotated result of an operator, our engine has to
% maintain state for each operator. Our implementation can persist such state in
% the database and load it into the incremental engine when needed.
\impAbbr % also
% implements SQL implementations of the incremental maintenance rules and can
outsources some of the computation to the backend database. This is in particular
useful for operations like joins where deltas from one side of the join have to
be joined with the full table on the other side similar to the delta rule
$\Delta R \join S$ used in standard incremental view maintenance.
%%%%%%%%%%%%%%%%%%%%%%%%%%%%%%%%%%%%%%%%
% optimizations
Additionally, we present several optimizations of our approach: (i) filtering
deltas determined by the database to prune delta tuples that are guaranteed to
not affect the result of incremental maintenance and (ii) filtering deltas for
joins using bloom filters. \revision{\impAbbr is effective for any query that benefits from sketches, e.g., queries with \lstinline!HAVING!, as long as the cost of maintaining sketches is amortized by using sketches for answering queries.}

%%%%%%%%%%%%%%%%%%%%%%%%%%%%%%%%%%%%%%%%
% extend IMP to maintain provenance.
%\PL{below: add brief discussion how to extended our \impAbbr to maintain
%  provenance}
%\BGI{THIS DOES NOT REALLY WORK HERE, I THINK THIS CAN BE SPLIT INTO A PART THAT GOES INTO RELATED WORK AND A PART THAT GOES INTO FUTURE WORK:
%  Our \impAbbr engine can be extended to maintain the provenance incrementally.
%  Work has been done to maintain the provenance
%  \cite{DBLP:conf/sigmod/ZhouSTLLM10}, while this focuses on Datalog not the SQL.
%  To successfully apply the \impAbbr engine to incrementally maintain the
%  provenance, we should take into account the some differences. First, to maintain
%  the sketches, each row associates with a sketch, while for provenance
%  maintenance, each row will associate with its provenance information. Secondly,
%  the data structures for state data will be the same, but wherever a sketch
%  appears, it should be the provenance instead. To update the state data, we
%  should update the provenance according to the delta tuples. Then, at the end of
%  maintenance, there is a merging operation for sketch to calculate the delta
%  sketches, while this step can be skipped for provenance maintenance because the
%  final output will contain the delta provenance information.}

%%%%%%%%%%%%%%%%%%%%%%%%%%%%%%%%%%%%%%%%
% summary
In summary, we present \impAbbr, the first incremental engine for maintaining
provenance sketches. Our main contributions are:
%%%%%%%%%%%%%%%%%%%%%%%%%%%%%%%%%%%%%%%%
\begin{itemize}[topsep=6pt,leftmargin=8pt,listparindent=5pt]
\item We develop incremental versions of relational algebra operators for
  sketch-annotated data. % and present rules for incrementally maintaining the results of such operators.
\item We implement these operators in \impAbbr, an in-memory engine
  for incremental sketch maintenance. \impAbbr enables \gls{pbds} for any DBMS by acting as a middleware between the user and the database that manages and maintains sketches. % and efficiently maintains them in-memory.
\item We experimentally compare \impAbbr against  full maintenance and against a baseline that does not use \gls{pbds} using \tpchds, real world datasets and synthetic data. % demonstrating its effectiveness.
  \impAbbr outperforms full  maintenance, often by several orders of magnitude. Furthermore, \gls{pbds} with \impAbbr significantly improves the performance of mixed workloads including both queries and updates.
\end{itemize}
%%%%%%%%%%%%%%%%%%%%%%%%%%%%%%%%%%%%%%%%

The remainder of this paper is organized as follows:
\Cref{sec:workflow_overview} presents an overview of \impAbbr. We discuss
related work in \Cref{sec:related_work}.
We formally define incremental maintenance
of sketches and introduce our annotated data model in \Cref{sec:notation_and_background}. In
\Cref{sec:incremental_operator_rules}, we introduce incremental sketch
maintenance rules for relational operators and prove their correctness in \ifnottechreport{\Cref{sec:rule_correctness_proof}}\iftechreport{\Cref{sec:rule_correctness_proof_appendix}}. We discuss
\impAbbr's implementation in \Cref{sec:implementation}, present experiments in
\Cref{sec:experiments}, and conclude in \Cref{sec:conclusion}.

%%% Local Variables:
%%% mode: latex
%%% TeX-master: "../imp"
%%% End:

%% file: sections/overview.tex
\section{Overview of \impAbbr}
\label{sec:workflow_overview}

\Cref{fig:workflow_overview} shows a overview of \impAbbr that operates as a middleware between the user and a DBMS.
\revision{We highlight parts of the system that utilize techniques from \cite{DBLP:journals/pvldb/NiuGLLGKLP21}. The dashed blue pipeline is for capture rewrite and dashed green pipeline is for use rewrite.}  
%\BG{@Pengyuan: update the description of how we indicate this in the figure once you have updated the figure}
Users send SQL queries and updates to \impAbbr that are parsed using \impAbbr's parser and translated into an intermediate representation (relational algebra with update operations). The system stores a set of provenance sketches in the database. For each sketch we store the sketch itself, the query it was captured for, the current state of incremental operators for this query, and the database version it was last maintained at or first captured at for sketches that have not been maintained yet. As sketches are small (100s of bytes), we treat sketches as immutable and retain some or all past versions of a sketch. This has the
advantage that it avoids write conflicts (for updating the sketch) between
concurrent transactions that need to update to different versions of the sketch.
We assume that the DBMS uses snapshot isolation and we can use snapshot identifiers used by the database internally to identify versions of sketches and of the database. For systems that use other concurrency control mechanisms, \impAbbr can maintain version identifiers. Furthermore, the system can persist the state that it maintains for its incremental operators in the database. This enables the system to continue incremental maintenance from a consistent state, e.g., when the database is restarted, or when we are running out of memory and need to evict the operator states for a query. \impAbbr enables \gls{pbds} for workloads with updates on-top of any SQL databases.

\impAbbr supports multiple incremental maintenance strategies. Under \textit{eager} maintenance, the system incrementally maintains each sketch that may be affected by the update (based on which tables are referenced by the sketch's query) by processing the update, retrieving the delta from the database, and running the incremental maintenance. Eager maintenance can be configured to batch updates. If the operator states for a sketch's query are not currently in memory, they will be fetched from the database. The updates to the sketches determined by incremental maintenance are then directly applied. Under \emph{lazy} maintenance, the system passes updates directly to the database. When a sketch is needed to answer a query, this triggers maintenance for the sketch. For that, \impAbbr fetches the delta between the version of the database at the time of the last maintenance for the sketch and the current database state and incrementally maintains the sketch. The result is a sketch that is valid as of the current state of the database. More advanced strategies can be designed on top of these two primitives, e.g., triggering eager maintenance during times of low resource usage or eagerly maintaining sketches for queries with strict response time requirements to avoid slowing down such queries when maintenance is run for a  large batch of updates.

\revision{For queries sent by the user, \impAbbr first determines whether there
  exists a sketch that can be used to answer the query $\query$. For that, it
  applies the mechanism from \cite{DBLP:journals/pvldb/NiuGLLGKLP21} to
  determine whether a sketch captured for a query $\query'$ in the past can be
  safely used to answer \query. If such a sketch \provSketch exists, we
  determine whether \provSketch is stale.}
% which can happen if the lazy maintenance strategy is applied.
\revision{If that is the case, then \impAbbr incrementally maintains the sketch
  (solid red pipeline). Afterwards, the query \query is instrumented to filter input
  data based on sketch \provSketch and then the instrumented query is sent to
  the database and its results are forwarded to the user (dashed green
  pipeline)\cite{DBLP:journals/pvldb/NiuGLLGKLP21}. If no existing provenance
  sketch can be used to answer \query, then \impAbbr creates a capture query for
  \query and evaluates this query to create a new sketch \provSketch (dashed blue
  pipeline) \cite{DBLP:journals/pvldb/NiuGLLGKLP21}. This sketch is then used to
  answer \query (dashed green pipeline) \cite{DBLP:journals/pvldb/NiuGLLGKLP21}.
  \impAbbr is an in-memory engine, exploiting the fact that sketches are small
  and that deltas and state required for incremental maintenance are typically
  small enough to fit into main memory or can be split into multiple batches if
  this is not that case.}

%% file: sections/related_work.tex
\section{Related Work}%
\label{sec:related_work}

%%%%%%%%%%%%%%%%%%%%%%%%%%%%%%%%%%%%%%%%%%%%%%%%%%%%%%%%%%%%%%%%%%%%%%%%%%%%%%%%
% Provenance
\parttitle{Provenance}
Provenance can be captured by annotating data and propagating these annotations
using relational queries or by extending the database
system~\cite{DBLP:journals/sigmod/KarvounarakisG12}
\cite{DBLP:conf/icde/NiuKGGLR17} \cite{DBLP:journals/tkde/NiuKGGLKR19}. Systems
like GProM~\cite{DBLP:journals/debu/ArabFGLNZ17},
Perm~\cite{DBLP:conf/birthday/GlavicMA13},
Smoke~\cite{DBLP:journals/corr/abs-1801-07237}, Smoked
Duck~\cite{mohammed-23-sd}, Links~\cite{DBLP:journals/scp/FehrenbachC18},
ProvSQL\cite{DBLP:journals/pvldb/SenellartJMR18} and
DBNotes\cite{DBLP:journals/vldb/BhagwatCTV05} capture provenance for SQL
queries. % Provenance has been used to speed-up the queries in interactive
% visualization \cite{DBLP:journals/corr/abs-1801-07237}.
In \cite{DBLP:journals/pvldb/NiuGLLGKLP21}, we introduced provenance-based data skipping
(PBDS). The approach captures sketches over-approximating the provenance of a
query and utilizes these sketches to speed-up subsequent queries.
We present the first approach for maintaining sketches under updates,
thus, enabling efficient PBDS for databases that are subject to updates.

%%%%%%%%%%%%%%%%%%%%%%%%%%%%%%%%%%%%%%%%%%%%%%%%%%%%%%%%%%%%%%%%%%%%%%%%%%%%%%%%
% Incremental view maintenance
\parttitle{Incremental View Maintenance (\ivmAbbr)}
View maintenance has been studied
extensively~\cite{DBLP:conf/sigmod/BlakeleyLT86, DBLP:journals/tods/BunemanC79,
DBLP:conf/edbt/GhandeharizadehHJ92, DBLP:conf/sigmod/GuptaMS93,
DBLP:conf/dood/Kuchenhoff91, DBLP:conf/sigmod/ShmueliI84}.
\cite{DBLP:journals/debu/GuptaM95,DBLP:conf/cascon/Vista94} gives an overview of many techniques and applications
of view maintenance. Early work on view maintenance, e.g.,
~\cite{DBLP:conf/sigmod/BlakeleyLT86,DBLP:journals/tods/BunemanC79}, used set
semantics. This was later expanded to bag semantics (e.g.,~ \cite{DBLP:conf/sigmod/GriffinL95, DBLP:conf/icde/ChaudhuriKPS95}).
We consider bag semantics.
% and the \impAbbr engine can process tuples with multiplicity.
Materialization has been studied for Datalog as well
~\cite{DBLP:conf/deductive/GuptaKM92,DBLP:conf/sigmod/GuptaMS93,DBLP:conf/aaai/MotikNPH15a}. % introduced the backward-forward algorithm.
Incremental maintenance algorithms for iterative computations have been studied
in~\cite{abadi-15-fdd, budiu-23-d, MM13c, murray-16-initdpwtd}.
% \cite{abadi-15-fdd,MM13c, murray-16-initdpwtd} has studied incremental
% maintenance for iterative dataflows and recursive computations.
\cite{DBLP:journals/vldb/KochAKNNLS14}  proposed higher-order \ivmAbbr.
\cite{DBLP:journals/vldb/YangW03} maintains aggregate views in temporal
databases. \cite{DBLP:conf/vldb/PalpanasSCP02} proposes a general mechanism for
aggregation functions.
\cite{DBLP:conf/vldb/AgrawalCN00,DBLP:conf/icac/ZilioZLMLCPCGALV04} studied
automated tuning of materialized views and indexes in databases.
As mentioned before, existing view maintenance techniques can not be directly applied for
provenance sketches maintenance, since \cite{DBLP:journals/pvldb/NiuGLLGKLP21}
uses specialized data types and functions to efficiently handle sketches during
capture, which are not supported in start-of-the-art \gls{ivm} systems. Furthermore,
classical \glspl{ivm} solutions have no notion of over-approximating query results and, thus, can not trade sketch accuracy for performance.
% \BGI{Need to recap the discussion why we can't directly use these techniques}
%% \BGI{NOT SURE WE NEED THIS HERE: Our work has in
%% common with self-tuning of physical design that we have to decide when it is
%% worth to pay the overhead of creating a physical design artifact (a provenance
%% sketch in our case) based on predication of future benefits resulting from using
%% the artifact for speeding up operations.}
%% While a full self-tuning solution for
%% sketches is beyond the scope of this work, we briefly discuss a basic cost model
%% for sketch maintenance and its application to tuning sketches in
%% \Cref{sec:cost-estimation}.
%% Our work does not limit usage of any physical designs, and, instead, it
%% encourages the database system to use some index to improve some operators'
%% performance.
Several strategies have been studied for maintaining views eagerly and lazily.
For instance, \cite{DBLP:conf/sigmod/ColbyGLMT96} presented algorithms for
deferred maintenance and
\cite{DBLP:conf/sigmod/BlakeleyLT86,DBLP:conf/vldb/CeriW91,DBLP:conf/sigmod/GuptaMS93}
studied immediate view maintenance. Our approach supports both cases: immediately maintaining sketches after each update or
sketches can be updated lazily  when needed. % according to a delta between the current and a
%past database version.

%%%%%%%%%%%%%%%%%%%%%%%%%%%%%%%%%%%%%%%%%%%%%%%%%%%%%%%%%%%%%%%%%%%%%%%%%%%%%%%%
% PROVENANCE MAINTENANCE
\parttitle{Maintaining Provenance}
\cite{DBLP:conf/sigmod/ZhouSTLLM10} presents a system for maintenance of
provenance in a distributed Datalog engine. In contrast to our work,
\cite{DBLP:conf/sigmod/ZhouSTLLM10} is concerned with efficient distributed
computation and storage for provenance. Provenance maintenance has to deal with
large provenance annotations that are generated by complex queries involving
joins and operations like aggregation that compute a small number of result
tuples based on a large number of inputs. \cite{DBLP:conf/sigmod/ZhouSTLLM10}
addresses this problem by splitting the storage of provenance annotations across
intermediate query results requiring recursive reconstruction at query time. In
contrast, provenance sketches are small and their size is determined upfront
based on the partitioning that is used. Because of this and because of their
coarse-grained nature, sketches enable new optimizations, including trading
accuracy for performance.
% Furthermore, as our goal of using provenance is
% improving query performance, we need to ensure that incremental maintenance is
% fast enough to be amortized by utilizing the maintained sketches.

%%% Local Variables:
%%% mode: LaTeX
%%% TeX-master: "../imp"
%%% End:

%% file: sections/background.tex
\input{./sections/notation_table}

\section{Background and Problem Definition}
\label{sec:notation_and_background}

In this section we introduce necessary background % on relational algebra, and
% provenance sketch, and describe the
and introduce notation used in the following sections. % in this paper.
% %%%%%%%%%%%%%%%%%%%%%%%%%%%%%%%%%%%%%%%%%%%%%%%%%%%%%%%%%%%%%%%%%%%%%%%%%%%%%%%%
% \subsection{The relational data model and algebra}
% \label{subsec:relationalModel}
% A database schema % $\dbSchema = ( \relSchema_1, \ldots, \relSchema_n ) $
% is a
% set of relation schemas.
% A relation schema
%  of arity $n$ consists of a name  and
% a list of attribute names $\att_1$ to $\att_n$.
Let $\udom$ be a domain of
values. An instance $\rel$ of an n-ary relation schema $\schemaOf{\rel} = (\att_1, \ldots, \att_n)$ is a function
$\udom^{n} \rightarrow \mathbb{N}$ mapping tuples to their multiplicity.  We use $\bag{\cdot}$ to denote bags and $\tup ^{n} \in \rel$ to
denote tuple $\tup$ with multiplicity $n$ in relation $\rel$, i.e., $\rel(\tup) = n$. A database $\db$ is a set of relations $\rel_1$ to $\rel_m$. The schema of a database $\schema{\db}$ is the set of relation schemas $\schema{\rel_i}$ for $i \in [1,m]$.
\Cref{fig:rel-algebra-def} shows the bag semantics relational algebra used in
this work. We use $\schema{Q}$ to denote the schema of the query $Q$
and $Q(\db)$ to denote the result of evaluating query $\query$ over database
 $\db$. Selection $\selection_\theta (\rel)$ returns all tuples from
relation $\rel$ which satisfy the condition $\theta$. \revision{Projection
$\projection_\ATT (\rel)$ projects all input tuples on a list of  projection
expressions. Here, $\ATT$ denotes a list of expressions with potential renaming
(denoted by $e \rightarrow \att$) and $\tup.\ATT$ denotes applying these
expressions to a tuple $\tup$. For example, $a + b \rightarrow c$ denotes renaming the result of $a+b$ as $c$.
% \footnote{\revision{Renaming the expression: the expression is $a + b$, and it can be renamed to `$summation$', i.e., $a + b \rightarrow summation$.}}
} $R \crossprod S$ is the cross product for bags.
For convenience we also define join $R \join_\theta S$ and natural join $R
\join S$ in the usual way. %We consider the bag semantics versions of set operations union $R \union S$, intersection $R \intersection S$, and difference $R-S$.
Aggregation $\Aggregation{f(a)}{\grpatts}(\rel)$
groups tuples according to their values in attributes $G$ and computes the
aggregation function $f$ over the bag of values of attribute $\att$ for each
group. We also allow the attribute storing $f(a)$ to be named explicitly, e.g.,
$\Aggregation{f(\att) \to x}{\grpatts}(\rel)$, renames $f(\att)$ as $x$.
Duplicate removal $\duprem (R)$ removes duplicates (definable using aggregation). Top-K $\topk(\rel)$ returns
the first $k$ tuples from the relation $\rel$ sorted on order-by attributes
$\orderbylist$.
We use $<_{\orderbylist}$ to denote the order induced by \orderbylist.
The position of a tuple in \rel ordered on \orderbylist is denoted by
$\pos{\tup}{\rel}{\orderbylist}$ and defined as:
% to be a 0-based position function that return
% the first position of a tuple $\tup$, and $\orderOfTup{\tup}$ to denote the
% order of a tuple $\tup$. Thus, the function can be defined as
$\pos{\tup}{\rel}{\orderbylist} = \sum_{\tup' <_{\orderbylist}  \tup
  \wedge \tup ^{m} \in \rel} m $.
\begin{figure}[t]
  \centering
  \begin{minipage}{1.0\linewidth}
    \begin{align*}
     &  \selection_{\theta} (R) = \bag{t^n|t^n \in R \wedge t \models \theta } \;\;\; \projection_{\ATT} (R) = \bag{t^n|n = \sum_{u.{\ATT} = t} R(u) }  \\
     & \duprem (R) = \bag{ \tup^{1} | \tup  \in \rel } \;\;\; R \crossprod S = \bag{ (t \concat s)^{n*m} |t^n \in R \wedge s^m \in S }\\
     %% & R \union S = \bag{ t^{n+m}|t^n \in R \wedge t^m \in S } \;\; R \intersection S = \bag{ t^{min(n,m)}|t^n \in R \wedge t^m \in S } \\
     %% &  R-S = \bag{ t^{max(n-m,0)}|t^n \in R \wedge t^m \in S }   \\
     &  \Aggregation{f(\att)}{\grpatts} (\rel) = \bag{ (t.\grpatts, f(\grpatts_\tup))^1 | \tup \in R } \;\; \grpatts_{\tup} = \{ (\tup_1.a)^n |{\tup_1}^n \in R \wedge \tup_1.\grpatts = \tup.\grpatts \}  \\
     & \topk (\rel) = \bagopen \tup^m \mid \pos{\tup}{\rel}{\orderbylist} < k \wedge m = min(\multiOfTupInRel{\rel}{\tup}, k - \pos{\tup}{\rel}{\orderbylist}) \bagclose
    \end{align*}\\[-10mm]
  \end{minipage}
  \caption{Bag Relational Algebra}
  \label{fig:rel-algebra-def}
\end{figure}
%%%%%%%%%%%%%%%%%%%%%%%%%%%%%%%%%%%%%%%%%%%%%%%%%%%%%%%%%%%
%
\revision{\Cref{fig:notations} shows an overview of the notations used in this work.}
%%%%%%%%%%%%%%%%%%%%%%%%%%%%%%%%%%%%%%%%%%%%%%%%%%%%%%%%%%%

\subsection{Range-based Provenance Sketches}
We use provenance sketches to concisely represent a superset of the
provenance of a query (a sufficient subset of the input) based on horizontal
partitions of the input relations of the query.
%In this work, we limit the discussion to range-partitioning since it allows us to exploit existing index structures when using a sketch to skip data.
%%%%%%%%%%%%%%%%%%%%%%%%%%%%%%%%%%%%%%%%
\subsubsection{Range Partitioning}
\label{tab:ps_def}
%%%%%%%%%%%%%%%%%%%%%%%%%%%%%%%%%%%%%%%%
%A \emph{horizontal partition} of a relation divides the tuples of the relation into disjoint groups called \emph{fragments}.
Given a set of intervals over the domains of a set of \emph{partition attributes} $A
\subset \relSchema$, \emph{range partitioning} determines membership of tuples to
fragments based on their $A$ values. For simplicity, we
define partitioning for a single attribute $\att$, \revision{but all of our techniques also apply when $\card{A} > 1$.}

%%%%%%%%%%%%%%%%%%%%%%%%%%%%%%%%%%%%%%%%%%%%%%%%%%%%%%%%%%%%%%%%%%%%%%%%%%%%%%%%
\begin{defi}[Range partition] \label{def:range}
Consider a relation $\rel$ and $\att \in \relSchema$. Let $\domain{a}$ denote
the domain of $\att$ and $\ranges = \{\range_1, \ldots, \range_n\}$ be a set
of intervals $[l,u] \subseteq \domain{\att}$ such that $\bigcup_{i=0}^{n}
\range_i = \domain{a}$ and $\range_i \cap \range_j = \emptyset$ for $i \neq j$. The
\emph{range-partition} of $\rel$ on $\att$ according to $\ranges$ denoted as
$\rparti_{\ranges,\att}(R)$ is defined as:
  \begin{align*}
%    \parti_{\ranges,a}(R) & = \{ \frag_1, \ldots, \frag_n \} \hfill
%    &\frag_{i}           & = \{ t^n \mid t^n \in R \wedge t.a \in  \range_i \}
    \parti_{\ranges,a}(R)  = \{ \rel_{{\range}_{1}}, \ldots, \rel_{{\range}_{n}} \} \hspace{2mm}\mathtext{\textbf{where}}\hspace{2mm}
    \rel_{\range}            = \bag{t^n \mid t^n \in \rel \wedge t.a \in  \range}
  \end{align*}
%A \emph{range-based database partitioning} $\dbpart$ is a database partitioning consisting only of range partitions.
%Given a range-based database partitioning $\dbpart = \{\parti_{\ranges_1,a_1}, \ldots, \parti_{\ranges_m,a_m}\}$,
%we define $\dbranges{\dbpart} = \{ \ranges_1, \ldots, \ranges_m \}$ and $\rangesOf{\dbpart} = \bigcup_{i=1}^{m} \ranges_i$.
\end{defi}
%%%%%%%%%%%%%%%%%%%%%%%%%%%%%%%%%%%%%%%%

We will use $\parti$ instead of $\parti_{\ranges,\att}$ if
$\ranges$ and $\att$ are clear from the context and $\frag$, $\frag'$, $\frag_i$, etc. to denote fragments. % , e.g., we may
% write $\parti = \{\frag_1, \ldots, \frag_n\}$.
We also extend range partitioning
to databases. For a database $\db = \{\rel_1, \ldots, \rel_n\}$, we use
$\dbranges$ to denote a set of range - attribute pairs $\{ (\ranges_1, \att_1),
\ldots, (\ranges_n, \att_n) \}$ such that $\parti_{\ranges_i,\att_i}$ is a
partition for $\rel_i$. % Note that this also covers the case where some
% relations
Relations $\rel_i$ that do not have a sketch can be modeled by setting $\ranges_i = \{ [min(\domain{\att_i}),
max(\domain{\att_i})] \}$, a single range covering all domain values.
% Note that we do not require that all relations of $\db$ are associated with a
% sketch. This is modeled by , for relations that are not part of a sketch.

%%%%%%%%%%%%%%%%%%%%%%%%%%%%%%%%%%%%%%%%%%%%%%%%%%%%%%%%%%%%%%%%%%%%%%%%%%%%%%%%
\subsubsection{Provenance Sketches}
\label{sec:prov-sketch-sketches}
Consider a database $\db$,  query $\query$, and a range partition $\parti_{\dbranges}$ of $\db$.
We use $\prov{\query}{\db} \subseteq \db$ to denote the provenance of $\query$ wrt. $\db$. For the purpose of \gls{pbds}, any provenance model that represents the provenance of $\query$ as a subset of $\db$ can be used as long as the model guarantees \emph{sufficiency}\footnote{\revision{Note that our notion of sufficiency aligns with the one from \cite{hu-24-fswcq} which differs slightly from the one used in \cite{G21} that is defined for a single result tuple of $\query$.}}~\cite{G21}: $\query(\prov{\query}{\db}) = \query(\db)$.
A provenance sketch $\provSketch$ for $\query$  according to $\dbranges$ is a
subset of the ranges $\ranges_i$ for each $\ranges_i \in \dbranges$ such that
the fragments corresponding to the ranges in $\provSketch$ fully cover
$\query$'s provenance within each $\rel_i$ in $\db$, i.e.,$\prov{\query}{\db}
\cap \rel_i$. We will write $\range \in \dbranges$ to denote that $\range \in \ranges_i$ for some $\ranges_i \in \dbranges$
% instead of $\range \in
% \ranges_i$ for $\ranges_i \in \dbranges$
and $\dbInst{\range}$ for $\range$
from $\ranges_i$ to denote the subsets of the database where all relations are
empty except for $\rel_i$ which is set to $\rel_{i,\range}$, the fragment for
$\range$. We use $\provranges{\db}{\dbranges}{\query} \subseteq \dbranges$ to
denote the set of ranges whose fragments overlap with the provenance
$\prov{\query}{\db}$:

%%%%%%%%%%%%%%%%%%%%%%%%%%%%%%%%%%%%%%%%
\begin{align*}
  \provranges{\db}{\dbranges}{\query} &= \{ \range \mid \range \in \ranges_i \wedge \exists t \in \prov{\query}{\db}: t \in R_{i,\range} \}
\end{align*}
%%%%%%%%%%%%%%%%%%%%%%%%%%%%%%%%%%%%%%%%

%%%%%%%%%%%%%%%%%%%%%%%%%%%%%%%%%%%%%%%%
\begin{defi}[Provenance Sketch]\label{def:provenance-sketch}
  Let $\query$ be a query, $\db$ a database,  $\rel$ a relation accessed by $\query$, and $\dbranges$ a partition of $\db$.
  We call a subset $\provSketch$ of $\dbranges$ a \textbf{provenance sketch} iff
  $\provSketch \supseteq \provranges{\db}{\dbranges}{\query}$.
  A sketch is \textbf{accurate} if
  % $\provSketch =  \provranges{\db}{\parti_{\ranges,a}(R)}{\query}$.
  $\provSketch = \provranges{\db}{\dbranges}{\query}$. % and noted as $\aProvSketch$.
  The \textbf{instance}  $\dbInst{\provSketch}$ of $\provSketch$ is defined as $\dbInst{\provSketch} = \bigcup_{\range \in \provSketch} \db_{\range}$. A sketch is \textbf{safe} if $\query(\dbInst{\provSketch}) = \query(\db)$. %i.e., the tuples contained in the fragments of  sketch $\provSketch$.
  % We call $\instOf{\provSketch}$ the .
\end{defi}
%%%%%%%%%%%%%%%%%%%%%%%%%%%%%%%%%%%%%%%%
%$\relInst{\provSketch}$

% Given a query $\query$ over relation $\rel$, a provenance sketch $\provSketch$
% is a compact and declarative description of a superset of the provenance of a
% query (the instance $\instOf{\provSketch}$ of $\provSketch$). We call a sketch
% \textit{accurate} if it only contains ranges whose fragments contain
% provenance.
% We use $\psSet$ to denote a set of provenance sketches for a subset of the
% relations in the database. Consider such a set  $\psSet = \{ \provSketch_1,
% \ldots, \provSketch_m \}$ where $\provSketch_i$ is a sketch for relation
% $\rel_i$ and $\rel_i \neq \rel_j$ for $i \neq j$. We use $\instOf{\psSet}$ to
% denote the database derived from database $\db$ by replacing each relation
% $\rel_i$ for $i \in \{1, \ldots, n\}$ with $\relInst{\provSketch_i}$.
% Note that we do not require that all relations of $\db$ are associated with a
% sketch. This is modeled in our framework by associating a relation $\rel_i$
% with a single range covering the whole domain of $\att_i$ as explained above.
% That is, associating each relation with a sketch is just for ease of
% presentation.
Consider the database consisting of a single relation
(\texttt{sales}) from our running example shown in
\Cref{fig:wordload-query-capture-reuse}. According to the partition $\dbranges
= \{ (\ranges_{price}, price) \}$, the accurate provenance sketch $\provSketch$
for the query $\qtopsellers$ according to $\dbranges$ consists of the set of
ranges $\{\range_3,\range_4\}$ (the two tuples in the provenance of this query
highlighted in \Cref{fig:wordload-query-capture-reuse} belong to the fragments
$\frag_3$ and $\frag_4$ corresponding to these ranges). The instance
$\instOf{\provSketch}$, i.e., the data covered by the sketch, consists of all
tuples contained in fragments $\frag_3$ and $\frag_4$ which are:
$\{s_3,s_4,s_5\}$. This sketch is safe. We use the method from \cite{DBLP:journals/pvldb/NiuGLLGKLP21} to determine for an attribute $\att$ and query $\query$ whether a sketch build on any partition of $\rel$ on $\att$ will be safe.

% Abusing notation, we will use $\instOf{\provSketch}$ to denote $\instOf{\{\provSketch\}}$.

%%%%%%%%%%%%%%%%%%%%%%%%%%%%%%%%%%%%%%%%%%%%%%%%%%%%%%%%%%%%%%%%%%%%%%%%%%%%%%%%

\subsection{Updates, Histories, and Deltas}
\label{subsec:database_changing}

For the purpose of incremental maintenance we are interested in the difference
between database states. Given two databases $\db_1$ and $\db_2$ we define the
\emph{delta} between $\db_1$ and $\db_2$ to be the symmetric difference between
$\db_1$ and $\db_2$ where tuples $t$ that have to be inserted into $\db_1$ to
generate $\db_2$ are tagged as $\insertdeltatup$ and tuples that have to be
deleted to derive $\db_2$ from $\db_1$ are tagged as $\deletedeltatup$:
\[
\deltaOf{\db_1}{\db_2} = \bag{ \deletedeltatup \mid t \in \db_1 - \db_2 } \union \bag{\insertdeltatup \mid t \in \db_2 - \db_1 }
\]

For a given delta $\deltadb$, we use $\deletedeltadb$  ($\insertdeltadb$) to
denote $\bag{ \deletedeltatup \mid \deletedeltatup \in \deltadb }$ ($\bag{
  \insertdeltatup \mid \insertdeltatup \in \deltadb }$).
% Given a database history $\dbhist$,
% We will use $\deltadb_i$ to denote
% $\deltaOf{\dbv{i - 1}}{\dbv{i}}$ and $\deltadb_{ij}$ to denote
% $\deltaOf{\dbv{i}}{\dbv{j}}$.
We use $\db \uniondb \deltadb$ to denote \emph{applying} delta $\deltadb$ to database $\db$:
\[
  \db \uniondb \deltadb = \db \difference \bag{ \deletedeltatup | \deletedeltatup \in \deltadb } \union \bag{\insertdeltatup | \insertdeltatup \in \deltadb }
\]

\begin{exam}\label{ex:background_changes_to_db}
  Reconsider the insertion of tuple $s_8$ (also shown below) into \tblSales \xspace as shown in \Cref{ex:stale-sketches}.
  \begin{center}
    $s_8$ = \verb!(8, HP, HP ProBook 650 G10, 1299, 1)!
  \end{center}
  Let us assume that the database before (after) the insertion of this tuple is $\dbv{1}$ ($\dbv{2}$), then we get:
  %\[
  $
\deltadbv{2} = \bag{ \insertdelta s_8 }
  $
  %\]
  % Before any operation to modify the database(assume there is only one table \tblSales ~in it), we define it as $\db_0$. And we define current provenance sketch \verb!'10'! as $\provSketch_0$. Suppose we insert one tuple into the table \tblSales ~ using an update statement shown as $\update$ in \Cref{fig:wordload-query-capture-reuse} and the inserted delta tuple set is shown as $\insertdeltadb$ in \Cref{fig:wordload-query-capture-reuse}. Since there is only one update, the workload so far only contains a single operation: $\workload = (\update)$. The operation will generate a delta tuple set $\deltadb_1 = \insertdeltadb$ and update the database to a new version $\db_1$ such that $\execUpd{\db_0}{1} = \deltadb_1$ and $\db_0 \uniondb \deltadb_1 = \db_1$.
\end{exam}

%%%%%%%%%%%%%%%%%%%%%%%%%%%%%%%%%%%%%%%%%%%%%%%%%%%%%%%%%%%%%%%%%%%%%%%%%%%%%%%%
%provsketch and deltas
% \subsection{Provenance Sketches Capture and Deltas}
% \label{subsec:ps_utilities}

% \subsubsection{Capturing Provenance Sketches}
% \label{subsubsec:capturing_ps}

% \PL{may be \textbf{accurate provenance sketch } is redudant. in 4.1.2, we have
%   already state a sketch is accurate. I am not sure where is the best place to
%   say \textbf{accurate provenance sketch}, probably we can remove the accurate PS
%   here, because we just want to say the Pi Pj here are sketch for database of
%   version i and j}
% \BGI{Double check that we do not need $\psFor{\query, \dbranges, \db}$ for later, I am commenting it out for now}
% Consider a query $\query$ and partitioning ranges $\dbranges$, we use $\psv{i}$ to denote an accurate sketch for \query wrt. \dbranges and $\dbv{i}$.
% we use $\psFor{\query, \dbranges, \db}$ to denote an {\bf {accurate provenance sketch}} $\provSketch$ for $\query$ wrt.
% to $\db$, and ranges $\dbranges$. % $\parti_{\ranges,\att}$:
% \[
%   \capturePS(\query, \dbranges, \db) = \provSketch
% \]
%%%%%%%%%%%%%%%%%%%%%%%%%%%%%%%%%%%%%%%%%%%%%%%%%%%%%%%%%%%%%%%%%%%%%%%%%%%%%%%%
% \subsubsection{Provenance Sketch Deltas}
% \label{subsubsec:combine_two_ps}
We use the same delta notation for sketches, e.g., % $\psv{i}$
% denotes the version of a sketch wrt. database version $\dbv{i}$, and
for two sketch versions $\provSketch_1$ and $\provSketch_2$, $\deltaprovSketch$ is their delta if $\provSketch_2 = \provSketch_1 \uniondb \deltaprovSketch$, where $\uniondb$ on sketches is defined as expected by inserting $\insertdeltaprovSketch$ and deleting $\deletedeltaprovSketch$.

%% file: sections/notation_table.tex
%%%%%%%%%%%%%%%%%%%%%%%%%%%%%%%%%%%%%%%%%%%%%%%%%%%%%%%%%%%%%%%%%%%%%%%%%%%%%%%%
%% Some notations used in the paper
\begin{figure}[t]
{\footnotesize
  \begingroup
  % \color{blue}
  \centering
  \begin{tabular}{|l|l|}
    \hline
    $\db (\deltadb)$ & a database (a delta database)\\
    $\annoDB (\annoDeltaDB)$ & an annotated database (an annotated delta database) \\
    $\rel (\annoRel, \annoDeltaRel)$ & a relation (an annotated relation, an annotated delta relation)\\
    $\pair{\tup}{\provSketch}$ & an annotated tuple: a tuple $t$ associated with its provenance sketch \provSketch\\
    % $\annoDeltaRel $ & an annotated delta relation\\
    $\range$,    $\ranges$,  $\dbranges$ & a range, set of ranges for partitioning relation $\rel$ (database $\db$)\\
    % & a range partitioning of database $\db$\\
    $\statedata $ & State of an incremental relational algebra operator\\
    $\provSketch $ & a provenance sketch\\
    $\incremaintain $ & an incremental maintenance procedure\\
    $\query$ & a query\\
    \hline
  \end{tabular}
  \vspace{-3mm}
  \caption{Glossary}
  \label{fig:notations}
  \endgroup
}
\end{figure}

%%%%%%%%%%%%%%%%%%%%%%%%%%%%%%%%%%%%%%%%%%%%%%%%%%%%%%%
%%% Local Variables:
%%% mode: LaTeX
%%% TeX-master: "../imp"
%%% End:

%% file: sections/incremental_data_model.tex
\subsection{Sketch-Annotated Databases And Deltas}
\label{sec:incremental_data_model}

Our incremental maintenance approach utilizes relations whose
tuples are annotated with sketches. We define an incremental
semantics for maintaining the results of operators over such
annotated relations and demonstrate that this semantics correctly maintains
sketches. % Before discussing this semantics and formalizing the problem solved in
\begin{defi}[Sketch Annotated Relation]\label{def:sketch-annotated-rels}
  A sketch annotated relation $\annoRel$ of arity $m$ for a given set of ranges $\ranges$ over the domain of some attribute $\att \in \relSchema$, is a bag of pairs $\pair{\tup}{\provSketch}$ such that $\tup$ is an $m$-ary tuple and $\provSketch \subseteq \ranges$.
\end{defi}
%%%%%%%%%%%%%%%%%%%%%%%%%%%%%%%%%%%%%%%%%%%%%%%%%%%%%%%%%%%%%%%%%%%%%%%%%%%%%%%%

We next define an operator $\anno{\rel}{\dbranges}$ that annotates each
tuple with the singleton set containing the range its value in attribute
$\att$ belongs to.
This
operator will be used to generate inputs for incremental relational algebra
operators over annotated relations.

%%%%%%%%%%%%%%%%%%%%%%%%%%%%%%%%%%%%%%%%%%%%%%%%%%%%%%%%%%%%%%%%%%%%%%%%%%%%%%%%
\begin{defi}[Annotating Relations]\label{def:annot-operator}
  Given a relation $\rel$, attribute $\att \in \relSchema$ and ranges $\dbranges = \{\ldots, (\ranges,\att), \ldots\}$, i.e., $(\ranges,\att)$ is the partition for $\rel$ in \dbranges, the operator $\annoF$ returns a sketch-annotated relation $\annoRel$ with the same schema as $\rel$:
\[
  \anno{\rel}{\dbranges} = \bag{ \pair{\tup}{\{\range\}} \mid \tup \in \rel \land \tup.\att \in \range \land \range \in \ranges }
\]
\end{defi}
%%%%%%%%%%%%%%%%%%%%%%%%%%%%%%%%%%%%%%%%%%%%%%%%%%%%%%%%%%%%%%%%%%%%%%%%%%%%%%%%

We define annotated deltas as deltas where
each tuple is annotated using the $\annoF$ operator. Consider a delta $\deltaRel$ between two versions $\relv{1}$ and $\relv{2}$ of relation $\rel$. Given % a relation $\rel$ and
ranges
$\ranges$ for attribute $\att \in \relSchema$, % let $\relv{i}$ be the version of $\rel$ in
% $\dbv{i}$. Then we
we define $\annoDeltaRel$ as:
%
%\[
$
\annoDeltaRel = \anno{\deltaRel}{\dbranges}.
$
%\]
$\annoDeltaRel$ contains all tuples from $\rel$ that differ between $\relv{1}$ and $\relv{2}$ tagged with $\insertdelta$ or $\deletedelta$ depending on whether they got inserted or deleted. Each tuple $\tup$ is annotated with the range $\range \in \ranges$ that $\tup.\att$ belongs to. \revision{Analog we use $\annoDB$ to denote the annotated version of database $\db$ and use $\annoDeltaDB$ to denote the annotated version of delta database $\deltadb$.}
\begin{exam}\label{ex:data_model_example}
  Continuing with \Cref{ex:background_changes_to_db}, the annotated version of $\deltadbv{2}$ according to $\ranges_{price}$ is
  %\[
  $
    \bag{ \pair{\insertdelta s_8}{\{\range_3\}} },
  %\]
  $
because $s_8.price$ belongs to $\range_3 = [1001,1500] \in \ranges_{price}$.
  % Recall in \Cref{ex:background_changes_to_db}, we discuss the workload and delta database. The delta database $\deltadb$ only has one inserted delta tuple there, denoted as $\insertdeltatup$. To apply this delta tuple to update the sketches, it should be translated into an instance of data model, which is expressed as $\pair{\insertdeltatup}{\provSketch}$. The tuple part of the instance is (7,4,HP ProBook 450 G9,999,1), and from the partition range information, we can compute that this tuple will be in range [3,4] which is represented as \verb!'01'!. Thus, the instance for this inserted tuple is $\pair{\text{(7,4,HP ProBook 450 G9,999,1)}}{\text{'01'}}$
\end{exam}
%%%%%%%%%%%%%%%%%%%%%%%%%%%%%%%%%%%%%%%%%%%%%%%%%%%%%%%%%%%%%%%%%%%%%%%%%%%%%%%%

%%% Local Variables:
%%% mode: latex
%%% TeX-master: "../imp"
%%% End:

%% file: sections/problem_definition.tex
\subsection{Problem Definition}%
\label{sec:problem_definition}
%%%%%%%%%%%%%%%%%%%%%%%%%%%%%%%%%%%%%%%%%%%%%%%%%%%%%%%%%%%%%%%%%%%%%%%%%%%%%%%%
We are now ready to define \emph{\glspl{mp}} that
maintain provenance sketches. An \gls{mp} takes as input a query $\query$
and an annotated delta $\annoDeltaDB$ for the ranges $\dbranges$ of a provenance
sketch $\provSketch$ and produces a delta $\deltaprovSketch$ for the sketch. \revision{Note that we assume that all attributes used in $\dbranges$ are \emph{safe}. An attribute $\att$ is safe for a query $\query$ if every sketch based on some range partition on $\att$ is safe. We use the 
safety test from~\cite{DBLP:journals/pvldb/NiuGLLGKLP21} to determine safe attributes.
  % \footnote{\revision{The safety check is introducted in ~\cite{DBLP:journals/pvldb/NiuGLLGKLP21}. Basically, a attribute is safe means that any sketch create on this attribute is guaranteed to be sufficient no matter what range partition is based on.}} In this work, we apply the same safety rules from the original paper.
}
\glspl{mp} are allowed to store some state $\statedata$, e.g., information about groups produced by an aggregation operator, %  for a query
% $\query$ and sketch $\provSketch$
to allow for more efficient maintenance. Given
the current state and $\annoDeltaDB$, the \gls{mp} should return a
delta $\deltaprovSketch$ for the sketch $\provSketch$ and an updated state $\statedata'$ such that
$\provSketch \uniondb \deltaprovSketch$ over-approximates an accurate sketch for the updated database.
\begin{defi}[Incremental Maintenance Procedure] \label{def:maintenance_procedure}
  Given a query $\query$, a database $\db$ and a delta $\deltadb$. Let $\provSketch$ be a provenance sketch over $\db$ for $\query$ wrt. some partition $\dbranges$. An {\bf{incremental maintenance procedure}} $\incremaintain$ takes as input a state $\statedata$, the annotated delta $\annoDeltaDB$, and returns an updated state $\statedata'$ and a provenance sketch delta $\deltaprovSketch$:
  \[
    \incremaintain(\query, \dbranges, \statedata, \annoDeltaDB) = (\deltaprovSketch, \statedata')
  \]
  % We call a \gls{mp} {\bf accurate} iff
  % $\provSketch \uniondb \deltaprovSketch = \psFor{\query, \dbranges, \db \union \deltadb}$.
  % state data $\statedata$ of this query, and the annotated delta database $\annoDeltaDB$, , is to take them as inputs and computes a delta provenance sketch $\deltaprovSketch$, which is written as:
\end{defi}
% \PL{if we modify or delete the notation in 'Provenance Sketches Capture and
%   Deltas', we need to modify for \textbf{P[Q, D, D union Delta D] subset P union
%     Delta P}}
%%%%%%%%%%%%%%%%%%%%%%%%%%%%%%%%%%%%%%%%%%%%%%%%%%%%%%%%%%%%%%%%%%%%%%%%%%%%%%%%

\revision{Let $\psFor{\query, \dbranges, \db}$ denote an accurate sketch for $\query$ over $\db$ wrt. $\dbranges$.
Niu et al.~\cite{DBLP:journals/pvldb/NiuGLLGKLP21} demonstrated  that any
over-approximation of a safe sketch is also safe, i.e., evaluating the
query over the over-approximated sketch yields the same result as evaluating the
query over the full database. % Thus, an \gls{mp} always produces
% safe sketches. 
Thus, for a \gls{mp} $\incremaintain$ to be correct, the following condition has to hold:  % any $\query$, database $\db$, delta $\deltadb$, and
for every sketch $\provSketch$ that is valid for $\db$ and delta $\deltadb$,  $\incremaintain$, if provided with the state $\statedata$ for $\db$ and the annotated version $\annoDeltaDB$ of $\deltadb$,
returns an over-approximation of the accurate sketch $\psFor{\query, \dbranges, \db \uniondb \deltadb}$: % the \gls{mp} has to produce $\deltaprovSketch$ such that
  \[
    \psFor{\query, \dbranges, \db \uniondb \deltadb} \subseteq \provSketch \uniondb     \incremaintain(\query, \dbranges, \statedata, \annoDeltaDB)
  \]
}

%% file: sections/incremental_operator_rules.tex
\section{Incremental Annotated Semantics}
\label{sec:incremental_operator_rules}
We now introduce an \gls{mp} that
maintains sketches using annotated and incremental semantics for
 relational algebra operators. Each operator takes as input an annotated delta produced by its inputs (or passed to the \gls{mp} in case of the table access operator), updates its internal state, and
outputs an annotated delta. Together, the states of all such incremental
operators in a query make up the state of our \gls{mp}. For an
operator $O$ (or query $\query$) we use
$\ieval{O}{\dbranges,\annoDeltaDB}{\statedata}$
($\ieval{\query}{\dbranges,\annoDeltaDB}{\statedata}$) to denote the result of
evaluating $O$ (\query) over the annotated delta $\annoDeltaDB$ using the state
$\statedata$. We will often drop $\statedata$ and $\dbranges$. % if they are clear from
% the context and for operators that do not store any state.
Our
\gls{mp} evaluates a query \query expressed in relational algebra
producing an updated state and outputting a delta where each
row is annotated with a partial sketch delta. These partial sketch deltas are then combined into a final result $\deltaprovSketch$. % After

\input{sections/example_whole_process.tex}
%%%%%%%%%%%%%%%%%%%%%%%%%%%%%%%%%%%%%%%%%%%%%%%%%%%%%%%%%%%%%%%%%%%%%%%%%%%%%%%%
% example for all operators.

%%%%%%%%%%%%%%%%%%%%%%%%%%%%%%%%%%%%%%%%%%%%%%%%%%%%%%%%%%%%%%%%%%%%%%%%%%%%%%%%
\begin{exam}\label{ex:op-all-rules-exam}
  \Cref{fig:all-rules-example} shows annotated tables $\mathscr{R}$
  and $\mathscr{S}$, ranges $\ranges_a$ and $\ranges_c$ for attribute $a$ (table $\rel$) and $c$ (table $S$), the delta $\Delta \rel$ and the sketches before the delta has been applied: $\psv{\rel}$ and $\psv{S}$. Consider the following query over $R$ and $S$:
  %%%%%%%%%%%%%%%%%%%%
\begin{lstlisting}
 SELECT a, sum(c) as sc
 FROM (SELECT a, b FROM R WHERE a > 3) JOIN S on (b = d)
 GROUP BY a HAVING SUM(c) > 5
\end{lstlisting}
  %%%%%%%%%%%%%%%%%%%%
  \Cref{fig:all-rules-example} (right table) shows each operator's output. \revision{We will further discuss these in the following when introducing the incremental semantics for individual operators. }
  In this example, a new tuple is inserted into $R$. Leading to a sketch deltas % $\{f_1, g_2\}$ where
$\Delta \psv{R} = \insertdelta \{f_1\}$ and $\Delta \psv{S} = \insertdelta \{g_2\}$. The tuple inserted into $R$ results in the generation of a new group for the aggregation subquery which passes the \lstinline!HAVING! condition and in turn causes the two fragments from the tuple belonging to this group to be added to the sketches.
\end{exam}

%%%%%%%%%%%%%%%%%%%%%%%%%%%%%%%%%%%%%%%%%%%%%%%%%%%%%%%%%%%%%%%%%%%%%%%%%%%%%%%%
\subsection{Merging Sketch Deltas}\label{sec:incremental-maintain-top-merge}
%\BG{Maybe have this at the end of the section? not sure}
% As mentioned above, our \gls{mp} evaluates the query
% using annotated, incremental versions of the operators used in the query and in
% a final step merges the sketch deltas for the query result to produce a delta
% for the current provenance sketch.
\revision{Each incremental algebra operator returns an annotated relation where each tuple is associated with a sketch that is sufficient to produce it. To generate the sketch for a query $\query$ we evaluate the query under our incremental annotated semantics to produce the tuples of $\query(\db)$ each annotated with a partial sketch. We then combine these partial sketches into a single sketch for $\query$.} We now discuss the operator $\fragmergeop$
that implements this final merging step. \revision{To determine whether a change
  to the annotated query result will result in a change to the current sketch,
  this operator maintains as state a map $\statedata: \dbranges \to \mathbb{N}$
  that records for each range $\range \in \dbranges$ the number of result tuples
  for which $\range$ is in their sketch.}
% used by the partitions used in the sketch of the query. used in
% sketches for the query to integers. This integers encode a counter for a
% fragment. Intuitively, this counter indicates the number of result tuples that
% have this fragment in their sketch.
If the counter for a fragment $\range$
reaches $0$ (due to the deletion of tuples), then the fragment needs to be
removed from the sketch. If the counter for a fragment $\range$ changes from $0$
to a non-zero value, then the fragment now belongs to the sketch for the query
% and needs to be added to the previous version of the sketch
(we have to add a
delta inserting this fragment to the sketch).

\[
\ieval{\fragmergeop(\query)}{\annoDeltaDB}{\statedata} = (\deltaprovSketch, \statedata')
\]

We first explain how $\statedata'$, the updated state for the operator, is computed and then explain how to compute
$\deltaprovSketch$ using $\statedata$. We define $\statedata'$
pointwise for a fragment $\range$.
Any newly inserted (deleted)  tuple whose sketch includes $\range$ increases
(decreases) the count for $\range$. That is the total cardinality of such
inserted tuples (of bag $\annoIDeltaDB$ and $\annoDDeltaDB$, respectively) has to
be added (subtracted) from the current count for $\range$.
Depending on the change of the count for $\range$ between $\statedata$ and
$\statedata'$, the operator $\fragmergeop$ has to output a delta for
$\provSketch$. Specifically, if $\statedata[\range] = 0 \neq
\statedata'[\range]$ then the fragment has to be inserted into the sketch and if
$\statedata[\range] \neq 0 = \statedata'[\range]$ then the fragment was part of
the sketch, but no longer contributes and needs to be removed.
%Thus, $\deltaprovSketch$ is computed as shown below:

%%%%%%%%%%%%%%%%%%%%%%%%%%%%%%%%%%%%%%%%%%%%%%%%%%%%%%%%%%%%%%%%%%%%%%%%%%%%%%%%
\begin{align*}
  \statedata'[\range] & = \statedata[\range] + \card{\annoIDeltaDB_{\range}} - \card{\annoDDeltaDB_{\range}} \\
  \annoIDeltaDB_{\range}      & = \bag{ \mult{\insertDeltaPair{\tup}{\provSketch}}{n} \mid \mult{\insertDeltaPair{\tup}{\provSketch}}{n} \in \ievalns{\query}{\annoDeltaDB}\land \range \in \provSketch }\\
  \annoDDeltaDB_{\range}      & = \bag{ \mult{\deleteDeltaPair{\tup}{\provSketch}}{n} \mid \mult{\deleteDeltaPair{\tup}{\provSketch}}{n} \in \ievalns{\query}{\annoDeltaDB}\land \range \in \provSketch }\\
%%%%%%%%%%%%%%%%%%%%%%%%%%%%%%%%%%%%%%%%%%%%%%%%%%%%%%%%%%%%%%%%%%%%%%%%%%%%%%%%
  \deltaprovSketch &= \bigcup_{\range:\, \statedata[\range] = 0 \land \statedata'[\range] \neq 0} \{ \insertdelta \range \}
                     \union
                     \bigcup_{\range:\, \statedata[\range] \neq 0 \land \statedata'[\range] = 0 } \{ \deletedelta \range \}
\end{align*}

%%%%%%%%%%%%%%%%%%%%%%%%%%%%%%%%%%%%%%%%%%%%%%%%%%%%%%%%%%%%%%%%%%%%%%%%%%%%%%%%

% %%%%%%%%%%%%%%%%%%%%%%%%%%%%%%%%%%%%%%%%
% \begin{align*}
%   \deltaprovSketch &= \bigcup_{\range: \statedata[\range] = 0 \land \statedata'[\range] \neq 0} \{ \insertdelta \range \}
%                      \union
%                      \bigcup_{\range: \statedata[\range] \neq 0 \land \statedata'[\range] = 0 } \{ \deletedelta \range \}
% \end{align*}
% %%%%%%%%%%%%%%%%%%%%%%%%%%%%%%%%%%%%%%%%

%%%%%%%%%%%%%%%%%%%%%%%%%%%%%%%%%%%%%%%%%%%%%%%%%%%%%%%%%%%%%%%%%%%%%%%%%%%%%%%%
\begin{exam}\label{ex:merge-final-frag}
  Reconsider our running example from \Cref{ex:provenance-sketch-capture} that partitions based on
  $\ranges_{price}$. Assume that there
are two result tuples $\tup_1$ and $\tup_2$ of a query $\query$ that have
 $\range_2 = [601,1000]$ in their sketch and one result tuple $\tup_3$ that has
$\range_1$ and $\range_2$ in its sketch. Then the current sketch for the
query is $\provSketch = \{ \range_1, \range_2 \}$ and the state of
$\fragmergeop$ is as shown below.   If we are processing a delta
  $\deleteDeltaPair{\tup_3}{\{\range_1,\range_2\}}$ deleting tuple $\tup_3$,
the updated counts $\statedata'$ are:\\[-4mm]
  %%%%%%%%%%%%%%%%%%%%%%%%%%%%%%%%%%%%%%%%
  \begin{align*}
    \statedata[\range_1] &= 1 &\hspace{-5mm}\statedata[\range_2] &=3
    %%%%%%%%%%%%%%%%%%%%
    &\statedata'[\range_1] &= 0 &\hspace{-5mm}\statedata'[\range_2] &=2
  \end{align*}\\[-6mm]
  %%%%%%%%%%%%%%%%%%%%%%%%%%%%%%%%%%%%%%%%
  %%%%%%%%%%%%%%%%%%%%%%%%%%%%%%%%%%%%%%%%
  % \begin{align*}
  %   \statedata'[\range_1] &= 0 &\statedata'[\range_2] &=2
  % \end{align*}
  %%%%%%%%%%%%%%%%%%%%%%%%%%%%%%%%%%%%%%%%
  As there is no longer any justification for $\range_1$ to belong to the sketch
(its count changed to $0$), $\fragmergeop$ returns a delta:
  \(
     \{ \deletedelta \range_1 \}
  \)
\end{exam}

%%%%%%%%%%%%%%%%%%%%%%%%%%%%%%%%%%%%%%%%
%% revision: example for merging
\revision{Consider the merge operator $\fragmergeop$ in \Cref{ex:op-all-rules-exam}. The state data before maintenance contains ranges $f_{2}$ and $g_{1}$. A single tuple annotated with $f_1$ and $g_2$ is added to the input of this operator. Both ranges were not present in \statedata and, thus, in addition adding them to $\statedata'$ the merge operator returns a sketch delta $\insertdelta \{f_1, g_2\}$.}
  % after merging contains ranges $f_{2}$, $g_{1}$, $f_{1}$, and $g_{2}$  in ~\Cref{fig:all-rules-example})
  % . Thus the sketches delta is $\{f_{1}, g_{2}\}$.}

%%%%%%%%%%%%%%%%%%%%%%%%%%%%%%%%%%%%%%%%%%%%%%%%%%%%%%%%%%%%%%%%%%%%%%%%%%%%%%%%
\subsection{Incremental Relational Algebra}\label{sec:inc-rel-algebra}

%%%%%%%%%%%%%%%%%%%%%%%%%%%%%%%%%%%%%%%%%%%%%%%%%%%%%%%%%%%%%%%%%%%%%%%%%%%%%%%%
\subsubsection{Table Access Operator}\label{sec:op-rule-table-access}
The incremental version of the table access operator $\rel$ returns the
annotated delta $\annoDeltaRel$ for $\rel$ passed as part of $\annoDeltaDB$ to the
\gls{mp} unmodified. This operator has no
state.
\[
\ievalns{\rel}{\annoDeltaDB} = \annoDeltaRel
  % \releval{\rel}{\annoDeltaDB} = \annoDeltaRel
\]

%%%%%%%%%%%%%%%%%%%%%%%%%%%%%%%%%%%%%%%%
%% revision: example for table access.
\revision{\Cref{fig:all-rules-example} (top)
  shows the result of annotating the relation $\deltaRel$ from \Cref{ex:op-all-rules-exam}.}
% When a operator needs to access a relation, it will fetch data from delta
% relation instead of original relation. After obtaining all the delta tuples,
% annotated tuples will be constructed using the function $\annoF$. The ranges
% for delta tuples are the same as those for tuples for capturing provenance
% sketch. Although $\annoF$ works on a relation, it will function identically
% for a single tuple because an annotated relation is built through annotating
% every tuple of a relation. To obtain the tuple-sketch pair for a delta tuple,
% $\annoF$ will be applied as follows:
% \begin{gather*}
%    \anno{\deltatup}{\parti_{\ranges, \att}(\rel)} = \bag{ \deltaPair{\tup}{\provSketch}} \\
%    where ~ \deltatup \in \deltaRel \wedge \provSketch = \{\ranges_{\range} \mid \deltatup.\att \in \range\}
% \end{gather*}

% The delta relation may contain tuples both inserted and deleted. In the rule,
% we use $\deltatup$ instead of $\insertdeltatup$ (insertion) or
% $\deletedeltatup$ (deletion) because no matter what type the tuple is, the
% range the tuple belongs to will be calculated in the same way. After
% annotating all tuples from the delta relation, an annotated relation can be
% obtained by unionizing all output bags and it will serve as the input of the
% operator above it.

%%%%%%%%%%%%%%%%%%%%%%%%%%%%%%%%%%%%%%%%%%%%%%%%%%%%%%%%%%%%%%%%%%%%%%%%%%%%%%%%
\subsubsection{Projection}\label{sec:op-rules-projection}
The projection operator does not maintain any state as each output
tuple is produced independently from an input tuple if we consider multiple
duplicates of the same tuple as separate tuples. % The operator projects an input
% delta tuple on a list of expressions $\ATT$.
% Each tuple in the output depends on
% a single tuple in the input under bag semantics.
For each annotated delta tuple $\adeltaPair$, we project $\tup$ on the project expressions $\ATT$ and
 propagate $\provSketch$ unmodified as $\tup.\ATT$ in the result depends on
the same input tuples as $\tup$. % of the projection):
%
%returns expressions from its input tuples. Given a list of expressions, this
%operator produces one value for each expression for every input. To update the
%provenance sketches, we will use the following rule \ruleProj:
%%%%%%%%%%%%%%%%%%%%%%%%%%%%%%%%%%%%%%%%
\[
   % \tag{\ruleProj}
   \ievalns{\projection_{\ATT}(\query)}{\annoDeltaDB} =
   \bag{ \mult{\deltaPair{\tup.\ATT}{\provSketch}}{n} \mid \mult{\deltaPair{\tup}{\provSketch}}{n} \in \ievalns{\query}{\annoDeltaDB}}
\]
%    \begin{cases}
%    \incremaintain(\projection_{\att} , \insertDeltaPair{\tup}{\provSketch}) =  \bag{\insertDeltaPair{\tup{'}}{\provSketch{'}} \mid \tup{'} =  \projection_{\att} (\tup) \wedge \provSketch{'} = \provSketch} \\
%    \incremaintain(\projection_{\att} , \deleteDeltaPair{\tup}{\provSketch}) = \bag{\deleteDeltaPair{\tup{'}}{\provSketch{'}} \mid \tup{'} =  \projection_{\att} (\tup) \wedge \provSketch{'} = \provSketch}
% \end{cases}
%\end{equation}
%%%%%%%%%%%%%%%%%%%%%%%%%%%%%%%%%%%%%%%%
% \BG{Alternative version modeling bags as functions}
% %%%%%%%%%%%%%%%%%%%%%%%%%%%%%%%%%%%%%%%%
% \begin{align*}
%    \ievalns{\projection_{\ATT}(\query)}{\annoDeltaDB}(\deltaPair{\tup}{\provSketch}) &= n
% \end{align*}
% such that
% \[
%   n = \sum_{\deltaPair{u}{\provSketch'} \in \annoDeltaDB: u.\ATT = t} \ievalns{\query}{\annoDeltaDB}(\deltaPair{u}{\provSketch'})
% \]
% and
% \[
%   \provSketch = \biguniondb_{\deltaPair{u}{\provSketch'} \in \annoDeltaDB: u.\ATT = t} \provSketch'
% \]
%%%%%%%%%%%%%%%%%%%%%%%%%%%%%%%%%%%%%%%%

% There is no state data needed for projection operator. The projection operator
% returns an annotated tuple for each input. In \ruleProj, if the input is an
% insertion, the output tuple will be an insertion, and if it is a deletion, a
% deleted annotated tuple will be produced as well. Regardless of the input
% types, the provenance sketch of output tuple will be the same as that of input
% tuple since calculating expressions' values does not affect provenance
% information.

%%%%%%%%%%%%%%%%%%%%%%%%%%%%%%%%%%%%%%%%%%%%%%%%%%%%%%%%%%%%%%%%%%%%%%%%%%%%%%%%
\subsubsection{Selection}\label{sec:op-rules-selection}

The incremental selection operator is stateless and the sketch of an input tuple is sufficient for producing the same tuple in the output of selection.
% maintaining a selection operator does not require state as
% each output tuple depends on a single input tuple and, thus, each annotated
% delta tuple in the input can be dealt with independent of every other delta
% tuple in the input.
% Unlike projection maintenance in which every input will
% generate one output, a tuple will lead to an empty bag if
% it cannot satisfy the selection condition.
%%%%%%%%%%%%%%%%%%%%%%%%%%%%%%%%%%%%%%%%
%% revision: add example for selection
Thus, selection returns all input delta tuples that fulfill the selection
condition unmodified and filters out all other delta tuples. \revision{In our running example (\Cref{fig:all-rules-example}), the single input delta tuple fulfills the condition of
selection $\selection_{a>3}$.}
  % tuple inputted from table access operator whose output can fulfill the
  % selection condition for selection operator in ~\Cref{ex:op-all-rules-exam}.}
%The input will be returned by selection if it satisfies the condition, which
%means that both input and output annotated tuples are identical. We give our
%selection rule \ruleSel based on input type as follows:
%%%%%%%%%%%%%%%%%%%%%%%%%%%%%%%%%%%%%%%%
\begin{align*}
\ievalns{\selection_\theta(\query)}{\annoDeltaDB} &= \bag{ \mult{\adeltaPair}{n} \mid \mult{\adeltaPair}{n} \in \ievalns{\query}{\annoDeltaDB} \land \tup \models \theta}
\end{align*}
\subsubsection{Cross Product}\label{sec:op-rules-cross-product}

The incremental version of a cross product (and join)
$\query_1 \crossprod \query_2$ combines three sets of deltas: (i) joining the
delta of $\query_1$ with the current annotated state of $\query_2$
($\releval{\query_2}{\annoDB}$), (ii) joining the delta of the $\query_2$ with $\releval{\query_1}{\annoDB}$, (iii)
joining the deltas of $\query_1$ and $\query_2$. For (iii) there are four
possible cases depending on which of the two delta tuples being joined is an
insertion or a deletion. For two inserted tuples that join, the joined tuple
$s \concat t$ is inserted into the result of the cross product. For two deleted
tuples, we also have to insert the joined tuple $s \concat t$ into the result.
For a deleted tuple joining an inserted tuple, we should delete the tuple $s
\concat t$. The non-annotated version of these rules have been discussed in
\cite{DBLP:journals/vldb/KochAKNNLS14, DBLP:conf/sigmod/MistryRSR01,
  DBLP:conf/sigmod/GriffinL95, DBLP:conf/sigmod/ColbyGLMT96}. We use $\annoQDelta_{i}$ to denote $\ievalns{\query_i}{\annoDeltaDB}$ for $i \in \{1,2\}$ below.
\noindent\begin{minipage}{\linewidth}
 \begin{align*}
% delta insert
\ievalns{\query_1 \crossprod \query_2}{\annoDeltaDB} =   \hspace{9mm} &\\
  \bagopen \mult{\insertDeltaPair{s \concat t}{\provSketch_1 \uniondb \provSketch_2}}{n \cdot m} \mid\,
  &(\mult{\insertDeltaPair{s}{\provSketch_1}}{n} \in \annoQDelta_1 \land \mult{\insertDeltaPair{t}{\provSketch_2}}{m} \in \annoQDelta_2)\\
  &\hspace{-6mm}\lor (\mult{\deleteDeltaPair{s}{\provSketch_1}}{n} \in \annoQDelta_1 \land \mult{\deleteDeltaPair{t}{\provSketch_2}}{m} \in \annoQDelta_2) \\
  &\hspace{-6mm}\lor (\mult{\insertDeltaPair{s}{\provSketch_1}}{n} \in \annoQDelta_1 \land  \mult{\pair{t}{\provSketch_2}}{m} \in \releval{\query_2}{\annoDB})\\
  &\hspace{-6mm}\lor (\mult{\pair{s}{\provSketch_1}}{n} \in \releval{\query_1}{\annoDB} \land  \mult{\insertDeltaPair{t}{\provSketch_2}}{m} \in \annoQDelta_2)
   \bagclose\\
  \union\hspace{28mm} &\\
% delta delete
  \bagopen \mult{\deleteDeltaPair{s \concat t}{\provSketch_1 \uniondb \provSketch_2}}{n \cdot m} \mid\,
  &(\mult{\deleteDeltaPair{s}{\provSketch_1}}{n} \in \annoQDelta_1 \land \mult{\insertDeltaPair{t}{\provSketch_2}}{m} \in \annoQDelta_2)\\
  &\hspace{-6mm}\lor (\mult{\insertDeltaPair{s}{\provSketch_1}}{n} \in \annoQDelta_1 \land \mult{\deleteDeltaPair{t}{\provSketch_2}}{m} \in \annoQDelta_2)\\
  &\hspace{-6mm}\lor (\mult{\deleteDeltaPair{s}{\provSketch_1}}{n} \in \annoQDelta_1 \land  \mult{\pair{t}{\provSketch_2}}{m} \in \releval{\query_2}{\annoDB})\\
  &\hspace{-6mm}\lor (\mult{\pair{s}{\provSketch_1}}{n} \in \releval{\query_1}{\annoDB} \land \mult{\deleteDeltaPair{t}{\provSketch_2}}{m} \in \annoQDelta_2)
   \bagclose\\
\end{align*}
\end{minipage}
%} %box end

%%%%%%%%%%%%%%%%%%%%%%%%%%%%%%%%%%%%%%%%
%% revision: join example
\revision{Continuing with \Cref{ex:op-all-rules-exam}, as $\makeAnnoDelta{S} = \emptyset$ and $\makeAnnoDelta{R} = \bag{\insertDeltaPair{(5, 8)}{\{f_1\}}}$ only contains insertions, only $\annoDeltaRel \join_{b=d} \mathscr{S}$ returns a non-empty result (the third case above). As $(5, 8)$ only joins with tuple $(7,8)$, a single delta tuple $\insertDeltaPair{(5, 8, 7,
    8)}{\{f_1, g_2\}}$ is returned.}

\iftechreport{
%%%%%%%%%%%%%%%%%%%%%%%%%%%%%%%%%%%%%%%%%%%%%%%%%%%%%%%%%%%%%%%%%%%%%%%%%%%%%%%%
%%%  \subsubsection{Union}\label{sec:op-rule-union}
%%%
%%%  The union operator also does not require any state data. For a union $\query_1 \union \query_2$, the incremental version of this operator propagates deltas from the input unmodified.
%%%
%%%  %%%%%%%%%%%%%%%%%%%%%%%%%%%%%%%%%%%%%%%%
%%%  \begin{align*}
%%%  \ievalns{\query_1 \union \query_2}{\annoDeltaDB} &= \ievalns{\query_1}{\annoDeltaDB} \cup \ievalns{\query_2}{\annoDeltaDB}
%%%  \end{align*}
%%%  %%%%%%%%%%%%%%%%%%%%%%%%%%%%%%%%%%%%%%%%
}

%%%%%%%%%%%%%%%%%%%%%%%%%%%%%%%%%%%%%%%%%%%%%%%%%%%%%%%%%%%%%%%%%%%%%%%%%%%%%%%%
\subsubsection{Aggregation: Sum, Count, and Average}\label{sec:op-rules-agg}

For the aggregation operator, we need to maintain the current aggregation result
for each individual group and record the contribution of fragments from a provenance sketch
towards the aggregation result to be able to efficiently maintain the operator's
result. Consider an aggregation operator
$\Aggregation{\makeagg{f}(\att)}{\grpatts}(\rel)$ where $\makeagg{f}$ is an
aggregation function and \grpatts are the group by attributes ($\grpatts =
\emptyset$ for aggregation without group-by). Given an version $\rel$ of the
input of the aggregation operator, we use $\statedataGBs = \{\tup.\grpatts |
\tup \in \rel\}$  to denote the set of distinct group-by values. %  wrt. to
% \grpatts.

The state data needed for aggregation depends on what aggregation function we
have to maintain. However, for all  aggregation functions the state maintained
for aggregation is a map $\statedata$ from  groups to a per-group state
storing aggregation function results for this group, the sketch for the group,
and a map $\fragCnt_\grp$ recording for each range $\range$ of $\dbranges$ the
number of input tuples belonging to the group with $\range$ in their provenance
sketch. Intuitively, $\fragCnt_{\grp}$ is used in a similar fashion as for
operator $\fragmergeop$ to determine when a range has to be added to or removed
from a sketch for the group. We will discuss aggregation functions \fcsum, \fccnt, and \fcavg that share
 the same state. % data and then will discuss \fcmin and \fcmax.
% To simplify presentation, we will only formalize the case of a single aggregation function.
% The state for multiple aggregation functions is analog.

% We will discuss aggregation functions: \fcsum, \fccnt and \fcavg. First of
% all, we will illustrate how to deal with \fcsum function with its state data
% and the provenance sketches. Then comes the discussions about \fccnt and
% \fcavg. Generally speaking, all the three functions share similarities in the
% rules and state data. We will start from \fcsum.

%%%%%%%%%%%%%%%%%%%%%%%%%%%%%%%%%%%%%%%%%%%%%%%%%%%%%%%%%%%%%%%%%%%%%%%%%%%%%%%%
\parttitle{Sum}\label{sec:op-rules-agg-sum}
Consider an aggregation $\sumAgg$. To be able to incrementally maintain the
aggregation result and provenance sketch for a group \grp, we store the
following state:
\[
  \statedata[\grp] = (\tupSum, \tupCnt, \provSketch, \fragCnt_{g})
\]

 $\tupSum$ and $\tupCnt$ store the sum and count for the group, $\provSketch$ stores
the group's sketch, and $\fragCnt_{\grp}: \dbranges \to
\mathbb{N}$ introduced above  tracks for each range $\range \in \dbranges$
how many input tuples from $\query(\db)$ belonging to the group have $\range$
in their sketch. State $\statedata$ is initialized to $\emptyset$.
% We now first describe how to create the state for the aggregation operator based
% on a database instance. Once the state has been initialized, the operator can
% process annotated deltas.

\iftechreport{
}

%%%%%%%%%%%%%%%%%%%%%%%%%%%%%%%%%%%%%%%%%%%%%%%%%%%%%%%%%%%%%%%%%%%%%%%%%%%%%%%%
\parttitle{Incremental Maintenance}
The operator processes an annotated delta as explained in the following. Consider an
% aggregation $\sumAgg$ and
annotated delta $\annoDeltaDB$.  Let $\Delta \query$
denote $\ievalns{\query}{\annoDeltaDB}$, i.e., the delta produced by incremental
evaluation for $\query$ using $\annoDeltaDB$. We use $\statedataGBs_{\Delta
\query}$ to denote the set of groups present in $\Delta \query$ and
$\Delta \query_{\grp}$ to denote the subset of $\Delta \query$ including all
annotated delta tuples $\adeltaPair$ where $t.\grpatts = \grp$. We now explain
how to produce the output for one such group. The result of the incremental
aggregation operators is then just the union of these results. We first discuss
the case where the group already exists and still exists after applying the
input delta. % Afterwards, we will discuss the two special cases when a new group
% is created or the last input tuple belonging to a group is deleted.

%\BG{In the following we need to add more explanations for why things work the
%way they work}
%%%%%%%%%%%%%%%%%%%%%%%%%%%%%%%%%%%%%%%%%%%%%%%%%%%%%%%%%%%%%%%%%%%%%%%%%%%%%%%%
\parttitle{Updating an existing group}
Assume the current and updated state for $\grp$ as shown below:
%
%%%%%%%%%%%%%%%%%%%%%%%%%%%%%%%%%%%%%%%%
\begin{align*}
   \statedata[\grp] &= (\tupSum, \tupCnt, \provSketch, \fragCnt_\grp)
  &\statedata'[\grp] &= (\tupSum', \tupCnt', \provSketch', \fragCnt_\grp')
\end{align*}
%%%%%%%%%%%%%%%%%%%%%%%%%%%%%%%%%%%%%%%%
% The updated state produced by the operator is:
% %
% \[
%   \statedata'[\grp] = (\tupSum', \tupCnt', \provSketch', \fragCnt_\grp')
% \]
%
The updated sum is produced by adding $\tup.a \cdot n$  for each
inserted input tuple with multiplicity $n$:
$\mult{\insertDeltaPair{\tup}{\provSketch}}{n} \in \Delta \query_{\grp}$ and
subtracting this amount for each deleted tuple:
$\mult{\deleteDeltaPair{\tup}{\provSketch}}{n} \in \Delta \query_{\grp}$. For instance, if % we are maintaining $\fcsum(A)$ and
the delta contains the insertion of 3 duplicates of a tuple with $\att$ value $5$, then the $\tupSum$ will be increased by $3 \cdot 5$.
%
%%%%%%%%%%%%%%%%%%%%%%%%%%%%%%%%%%%%%%%%
\begin{align*}
  % &\tupCnt' =\tupCnt \hspace{2mm}+ \sum_{\mult{\insertDeltaPair{\tup}{\provSketch}}{n} \in \Delta \query_{\grp}} \hspace{-6mm}n \hspace{8mm}- \sum_{\mult{\deleteDeltaPair{\tup}{\provSketch}}{n} \in \Delta \query_{\grp}} \hspace{-6mm}n\\
  &\tupSum' =\tupSum \hspace{2mm}+ \sum_{\mult{\insertDeltaPair{\tup}{\provSketch}}{n} \in \Delta \query_{\grp}} \hspace{-6mm}\tup.a \cdot n \hspace{2mm}- \sum_{\mult{\deleteDeltaPair{\tup}{\provSketch}}{n} \in \Delta \query_{\grp}} \hspace{-6mm}\tup.a \cdot n
\end{align*}
%%%%%%%%%%%%%%%%%%%%%%%%%%%%%%%%%%%%%%%%
%
The update for $\tupCnt$ is computed in the same fashion using $n$ instead of $t.\att \cdot n$.
The updated count in $\fragCnt_\grp'$ is computed for each $\range \in \dbranges$ as:
\[
\fragCnt_{\grp}'[\range] = \fragCnt_{\grp}[\range] + \sum_{\mult{\insertDeltaPair{\tup}{\provSketch}}{n} \in \Delta \query_{\grp} \land \range \in \provSketch} \hspace{-5mm}n \hspace{3mm}- \sum_{\mult{\deleteDeltaPair{\tup}{\provSketch}}{n} \in \Delta \query_{\grp} \land \range \in \provSketch} \hspace{-5mm}n
\]
Based on $\fragCnt_\grp'$ we then determine the updated sketch for the group:
\[
\provSketch' = \{ \range \mid \fragCnt_{\grp}'[\range] > 0 \}
\]
 We then output a pair of annotated delta tuples that deletes the previous result for the group and inserts the updated result:
%
%%%%%%%%%%%%%%%%%%%%%%%%%%%%%%%%%%%%%%%%
\begin{align*}
  &\deleteDeltaPair{\grp \concat (\tupSum)}{\provSketch} &
  &\insertDeltaPair{\grp \concat (\tupSum')}{\provSketch'}
\end{align*}
%%%%%%%%%%%%%%%%%%%%%%%%%%%%%%%%%%%%%%%%

%%%%%%%%%%%%%%%%%%%%%%%%%%%%%%%%%%%%%%%%%%%%%%%%%%%%%%%%%%%%%%%%%%%%%%%%%%%%%%%%
\parttitle{Creating and Deleting Groups}
For groups $\grp$ that are not in $\statedata$, we initialize the state for $\grp$ as shown below:
\(
\statedata'[\grp] = ( 0, 0, \emptyset, \emptyset)
\)
and only output $\insertDeltaPair{\grp \concat (\tupSum')}{\provSketch'}$.
%%%%%%%%%%%%%%%%%%%%%%%%%%%%%%%%%%%%%%%%%%%%%%%%%%%%%%%%%%%%%%%%%%%%%%%%%%%%%%%%
%\parttitle{Deleting an existing group}
%
An existing group gets deleted if $\tupCnt \neq 0$ and $\tupCnt' = 0$. In this case we only output $\deleteDeltaPair{\grp \concat (\tupSum)}{\provSketch}$.

%%%%%%%%%%%%%%%%%%%%%%%%%%%%%%%%%%%%%%%%%%%%%%%%%%%%%%%%%%%%%%%%%%%%%%%%%%%%%%%%
\parttitle{Average and Count}
For average we maintain the same state as for sum. The only difference is that the updated average is computed as $\frac{\tupSum'}{\tupCnt'}$. For count we only maintain the count and output $\tupCnt'$.
\ifnottechreport{ %
    In \cite{techreport}, we also present the incremental annotated semantics for additional aggregation functions ($\fcmin$ and $\fcmax$) which require maintaining a sort order over the input tuples in each group to deal with deletion. % and for a top-k operator.
}

%%%%%%%%%%%%%%%%%%%%%%%%%%%%%%%%%%%%%%%%%%%%%%%%%%%%%%%%%%%%%%%%%%%%%%%%%%%%%%%%
%% revision: example agg
\revision{Continuing with \Cref{ex:op-all-rules-exam}, the output of the join (single delta tuple with group 5) is fed  into the aggregation operator using \fcsum. As no such group is in $\statedata$ we create new entry $\statedata[5]$. After maintaining the state, the output delta produced for this group is $\bag{\insertDeltaPair{(5, 7)}{\{f_1, g_2\}}}$.
This result satisfies \lstinline!HAVING! condition (selection $\selection_{\fcsum(c) > 5}$) and is passed on to the merge operator.}

%%%%%%%%%%%%%%%%%%%%%%%%%%%%%%%%%%%%%%%%%%%%%%%%%%%%%%%%%%%%%%%%%%%%%%%%%%%%%%%%

\iftechreport{
%%%%%%%%%%%%%%%%%%%%%%%%%%%%%%%%%%%%%%%%%%%%%%%%%%%%%%%%%%%%%%%%%%%%%%%%%%%%%%%%
\subsubsection{Aggregation: minimum and maximum}\label{sec:op-rule-agg-min-max}
The aggregation functions \fcmin and \fcmax share the same state. To be able to
efficiently determine the minimum (maximum) value of the aggregation function
\fcmin (\fcmax), we use a data structure like balanced search trees that can provide
efficiently access to the value in sort order.

%% For \fcmin
%% (\fcmax), the key point is to get the next minimum value after update. To
%% accomplish this, we use an ordered map to store all the values of aggregate
%% attribute $\att$. During the incremental maintenance, we insert all aggregate
%% attribute value into the map for all inserted delta tuples and remove values
%% from the map for all deleted annotated input tuples.

\parttitle{Min}\label{sec:oo-rules-agg-min}
Consider an aggregation $\minAgg$. To be able to maintain the aggregation result
and provenance sketch incrementally for a group \grp, we store the following
state:
\[
  \statedata[\grp] = (\tupCnt, \provSketch, \fragCnt_{g})
\]
$\provSketch$ and $\fragCnt_g$ are the same as described in aggregation function
\fcsum.  $\provSketch$ stores the groups' sketch and $\fragCnt_g$ stores for each
range how many tuples in this group have this range in their sketch. $\tupCnt$
is an balanced search tree that record all values of aggregate attribute in sort
order, and for each node in $\tupCnt$, we store the multiplicity of this aggregate value. 

%% where the key of the map is
%% the value of aggregate attribute and value of the map is the multiplicity of
%% each aggregate attribute value. We use $\keyOf{\orderedmap}$ to denote the all
%% keys in the map $\orderedmap$, and use $\mkeyOf{\orderedmap}$ to denote the
%% minimum key (maximum key for \fcmax aggregation function)in the map $\orderedmap$ such that:
%% \[
%%   \mkeyOf{\orderedmap} = \min(\keyOf{ \orderedmap }) = \min_{k \in \keyOf{ \orderedmap }} k
%% \]

\parttitle{Incremental Maintenance}
Consider an aggregation $\minAgg$ and an annotated delta $\annoDeltaDB$. Recall
that $\Delta \query$ denotes the $\ievalns{\query}{\annoDeltaDB}$ and $\Delta
\query_{\grp}$ to denote the subset of $\Delta \query$ including all annotated
delta tuples $\adeltaPair$ where $t.\grpatts = \grp$. We now discuss how to produce the
output for one such group and how the incremental maintenance works for the case
where the group already exists and still exits after maintenance.
% For \fcmin and \fcmax functions, they are different from \fccnt, \fcsum and \fcavg functions, because values in the same group are independent and the functions return new values by comparing each other not by mathematically calculating the current and new input values. Now we focus on \fcmin and \fcmax is similar.

\parttitle{Updating an existing group}
Assume the current and updated state for $\grp$ as shown below:
%%%%%%%%%%%%%%%%%%%%%%%%%%%%%%%%%%%%%%%%%
\begin{align*}
   \statedata[\grp] &= (\tupCnt, \provSketch, \fragCnt_\grp)
  &\statedata'[\grp] &= (\tupCnt', \provSketch', \fragCnt_\grp')
\end{align*}
%%%%%%%%%%%%%%%%%%%%%%%%%%%%%%%%%%%%%%%%%
The $\tupCnt$ is updated as follow: for each annotated tuple in $\Delta
\query_{\grp}$, the multiplicity of aggregate attribute value ($\att$) 
 will be increased by 1 if it is an inserted annotated tuple
($\insertDeltaPair{\tup}{\provSketch}$) and $\tup.\att = \att$. If this value is
new to the tree, we just initialize for this value with a multiplicity 1.
Otherwise, the multiplicity will be decreased by one if the annotated tuple is a
deletion and $\tup.\att = \att$. We remove a node from the tree if the multiplicity 
becomes $0$.
%% \begin{algorithm}[h]
%% \caption{Computing $\orderedmap'$ }\label{alg:omap_update}
%% \begin{algorithmic}
%% \For{$\Delta \pair{\tup}{\provSketch} \in \Delta \query_{\grp}$}
%% \If{$\Delta$ is $\insertdelta$}
%%     \State \orderedmap$[\tup.\att] \gets$ \orderedmap$[\tup.\att] + 1$
%% \ElsIf{$\Delta$ is $\deletedelta$}
%%     \State \orderedmap$[\tup.\att] \gets$ \orderedmap$[\tup.\att] - 1$
%% \EndIf
%% \EndFor
%% \end{algorithmic}
%% \end{algorithm}
%%%%%%%%%%%%%%%%%%%%%%%%%%%%%%%%%%%%%%%%%
\begin{align*}
  \tupCnt[\att]' =\tupCnt[\att] \hspace{2mm}+ \sum_{\mult{\insertDeltaPair{\tup}{\provSketch}}{n} \in \Delta \query_{\grp}\land \tup.\att = \att} \hspace{-6mm}n \hspace{8mm}- \sum_{\mult{\deleteDeltaPair{\tup}{\provSketch}}{n} \in \Delta \query_{\grp}\land \tup.\att = \att} \hspace{-6mm}n
  % &\tupSum' =\tupSum \hspace{2mm}+ \sum_{\mult{\insertDeltaPair{\tup}{\provSketch}}{n} \in \Delta \query_{\grp}} \hspace{-6mm}\tup.a \cdot n \hspace{2mm}- \sum_{\mult{\deleteDeltaPair{\tup}{\provSketch}}{n} \in \Delta \query_{\grp}} \hspace{-6mm}\tup.a \cdot n
\end{align*}
%%%%%%%%%%%%%%%%%%%%%%%%%%%%%%%%%%%%%%%%%
Here, the use of $\tupCnt[\att]$ and $\tupCnt[\att]'$ do not imply $\tupCnt$ is a map but they indicate the multiplicity of $\att$ in $\tupCnt$.

The updated count in $\fragCnt_\grp'$ is computed for each $\range \in \dbranges$ as:
\[
\fragCnt_{\grp}'[\range] = \fragCnt_{\grp}[\range] + \sum_{\mult{\insertDeltaPair{\tup}{\provSketch}}{n} \in \Delta \query_{\grp} \land \range \in \provSketch} \hspace{-5mm}n \hspace{3mm}- \sum_{\mult{\deleteDeltaPair{\tup}{\provSketch}}{n} \in \Delta \query_{\grp} \land \range \in \provSketch} \hspace{-5mm}n
\]
Based on $\fragCnt_\grp'$ we then determine the updated sketch for the group:
\[
\provSketch' = \{ \range \mid \fragCnt_{\grp}'[\range] > 0 \}
\]
 We then output a pair of annotated delta tuples that deletes the previous result for the group and inserts the updated result:
%
%%%%%%%%%%%%%%%%%%%%%%%%%%%%%%%%%%%%%%%%
\begin{align*}
  &\deleteDeltaPair{\grp \concat (min(\tupCnt))}{\provSketch}
  & \insertDeltaPair{\grp \concat (min(\tupCnt'))}{\provSketch'}
\end{align*}
For the group $\grp$, we need to output the minimum value in the balanced
search tree.

%%%%%%%%%%%%%%%%%%%%%%%%%%%%%%%%%%%%%%%%
\parttitle{Creating and Deleting Groups}
For groups $\grp$ that are not in $\statedata$, we initialize the state for $\grp$ as shown below:
\(
\statedata'[\grp] = (\emptyset, \emptyset, \emptyset)
\)
and only output $\insertDeltaPair{\grp \concat
  (min(\tupCnt'))}{\provSketch'}$. An existing group gets deleted if size of the
tree becomes zero from a non-zero value such that:
$|\tupCnt| \neq 0 = |\tupCnt'|$. In this case we only output
$\deleteDeltaPair{\grp \concat (\min{\tupCnt})}{\provSketch}$.

%%%%%%%%%%%%%%%%%%%%%%%%%%%%%%%%%%%%%%%%%%%%%%%%%%%%%%%%%%%%%%%%%%%%%%%%%%%%%%%%
\parttitle{Max}
For max, we maintain the same state as for min. The only difference is that we
output the maximum value from tree $\tupCnt$ and $\tupCnt'$.
}
\subsubsection{Top-k}\label{sec:op-rule-topk}
The top-k operator $\topk$ returns the first $k$ tuples sorted on ${\orderbylist}$.
As we are dealing with bag semantics, the top-k tuples may contain a tuple with multiplicity larger than $1$. As before, we use $\annoQDelta$ to denote $\ievalns{\query}{\annoDeltaDB}$.
% To maintain this operator, the strategy
% is: first we output first $k$ annotated tuples from current state as deleted
% output, then we update the state based on annotated delta input and last we
% propagate the top $k$ annotated tuples as inserted output from the updated
% state. To fast access to the first $k$ smallest tuples among all tuples, we use
% an ordered map as state to store all tuples comes to this operator.

%%%%%%%%%%%%%%%%%%%%%%%%%%%%%%%%%%%%%%%%%%%%%%%%%%%%%%%%%%%%%%%%%%%%%%%%%%%%%%%%
\parttitle{State Data}
To be able to efficiently determine updates to the top-k tuples with sketch annotations we maintain a nested map. The outer map $\statedata$ is order and  should be implemented using a data structure like \glspl{bst} that provide efficient access to entries in sort order on \orderbylist, maps order-by values $o$ to another map $\tupCnt$ which stores multiplicities for each annotated tuple $\pair{\tup}{\provSketch}$ for which $\tup.\orderbylist = o$.
\[
  \statedata[o] = (\tupCnt)
\]
and for any $\pair{\tup}{\provSketch}$ with $\tup.\orderbylist = o$ with $\mult{\pair{\tup}{\provSketch}}{n} \in \annoQDelta$ we store
\[
\tupCnt[\pair{\tup}{\provSketch}] = n
\]
%%%%%%%%%%%%%%%%%%%%
This data structure allows efficient updates to the multiplicity of any annotated tuple  based on the input delta as shown below. Consider such a tuple $\pair{\tup}{\provSketch}$ with $\tup.\orderbylist = o$ with $\mult{\insertDeltaPair{\tup}{\provSketch}}{n} \in \annoQDelta$ and $\mult{\deleteDeltaPair{\tup}{\provSketch}}{m} \in \annoQDelta$.
\[
  \statedata'[o][\pair{\tup}{\provSketch}] = \statedata[o][\pair{\tup}{\provSketch}]
  + n - m
\]

%%%%%%%%%%%%%%%%%%%%%%%%%%%%%%%%%%%%%%%%%%%%%%%%%%%%%%%%%%%%%%%%%%%%%%%%%%%%%%%%
\parttitle{Computing Deltas}
As $k$ is typically relatively small, we select a simple approach for computing deltas by deleting the previous top-k and then inserting the updated top-k. Should the need arise to handle large $k$, we can use a balanced search tree and mark nodes in the tree as modified when updating the multiplicity of annotated tuples based on the input delta and use data structures which enable efficient positional access under updates, e.g., order-statistic trees~\cite{cormen-09-inaled}.
Our simpler technique just fetches the first tuples in sort order from $\statedata$ and $\statedata'$ by accessing the keys stored in the outer map $\statedata$ in sort order. For each $o$ we then iterate through the tuples in $\statedata[o]$ (in an arbitrary, but deterministic order since they are incomparable) keeping track of the total multiplicity $m$ of tuples we have processed so far. As long as $m \leq k$ we output the current tuple and proceed to the next tuple (or order-by key once we have processed all tuples in $\statedata[o]$). Once $m \geq k$, we terminate.
 If the last tuple's multiplicity  exceeds the threshold we output this tuple with the remaining multiplicity. Applied to $\statedata$ this approach produces the tuples to delete and applied to $\statedata'$ it produces the tuples to insert:
\begin{align*}
& \deletedelta \topk(\statedata) & \insertdelta \topk(\statedata')
\end{align*}

\subsection{\revision{Complexity Analysis}}
\label{sec:op_rule_complexity}
\revision{We now analyze the runtime complexity of operators. Let $n$ denote the input delta tuple size  and $p$ denote the number of ranges of the partition on which the sketch is build on. For table access, selection, and projection, we need to iterate
  over these $n$ annotated tuples to generate the output. As for these operations we do not modify the sketches of tuples, the complexity
  is $O(n)$. For aggregation, for each aggregation function we maintain a hashmap that tracks the current aggregation result for each group and a count for each fragment that occurs in a sketch for each tuple in the group. For each input delta tuple, we can update this information in $O(1)$ if we assume that the number of aggregation functions used in an aggregation operator is constant. Thus, the overall runtime for aggregation is $O(n \cdot p)$.
  % since we keep each function a hashmap to
  % store the groups information, it will take $O(n)$ for input data to find all
  % the groups and update the values. For each input data item, we need to update
  % the sketch information as well. Thus it will cost $O(k)$ to update ranges
  % information for each sketch . The total complexity is $O(n * k)$ (typically,
  % the k is small).
  For join operator, we store the bloom filter to pre-filter
  the input data. Suppose the right input table has $m$ tuples. Building such a
  filter incurs a one-time cost of $O(m)$ for scanning the table once.
Consider the part where we join a delta of size $n$ for the left input with the right table producing $o$ output tuples.
  The cost this join depends on what join algorithm is used ranging from $O(n + m + o)$ for a hash join to $O(n \cdot m + o)$ for a nested loop join (in both cases assuming the worst case where no tuples are filtered using the bloom filter).
  % Then
  % when using bloom filters, we need to scan all $n$ input tuples to decided if
  % there are potential join partners which will cost $O(n)$.
  For the top-k operator,
  we assume there are $l$ nodes stored in the balanced search tree. Building this tree
  will cost $O(l \cdot \log l)$ (only built once). An insertion, deletion, or lookup will take $O(\log l)$ time. Thus,
  the runtime complexity of the top-k operator is $O(n \cdot \log l)$ to complete the top-k operator. Regarding space complexity,
  the selection and projection only require constant space. For
  aggregation, the space is linear in the number of groups and in $p$. For
  join, the bloom filter's size is linear in $m$, but for a small constant factor. For top-k operators, we store $l \geq k$ entries in the search tree, each requiring $O(p)$  space. Thus, the overall space complexity for this operator is $O(l \cdot p)$.}
%%% Local Variables:
%%% mode: LaTeX
%%% TeX-master: "../imp"
%%% End:

%% file: sections/example_whole_process.tex
\begin{figure*}[ht!]
 \begin{minipage}{0.3\linewidth}
   \begin{minipage}{0.99\linewidth}
     \captionof*{table}{Table, Ranges and Delta}
     \begin{tabular}{c|c|c|c} \cline{2-3}
       &\thead{$a$} & \thead{$b$} &$\provSketch$                                                                                                        \\ \cline{2-3}
     \multirow{2}{*}{$\mathscr{R}$}  &1           &     7       & \{$f_1$\}                                                                             \\
       &9          &     9       & \{$f_2$\}                                                                                                            \\ \cline{2-3}
     \end{tabular}
   \end{minipage}\\[1mm]
   % \vspace{3mm}
   \begin{minipage}{0.99\linewidth}
     % \captionof*{table}{$\mathscr{S}$}
     \begin{tabular}{c|c|c|c} \cline{2-3}
       &\thead{$c$} & \thead{$d$} &$\provSketch$                                                                                                        \\ \cline{2-3}
     \multirow{2}{*}{$\mathscr{S}$} &6           &    9       & \{$g_1$\}                                                                               \\
       &7          &     8       & \{$g_2$\}                                                                                                            \\ \cline{2-3}
     \end{tabular}
   \end{minipage}\\[1mm]
   % \vspace{2mm}
   \begin{minipage}{0.99\linewidth}
     \begin{tabular}{c}
       $\ranges_{a} = \{ f_1 = [1,5], f_2 = [6,10] \}$                                                                                                  \\
       $\ranges_{c} = \{ g_1 = [1,6], g_2 = [7,15] \}$                                                                                                  \\[2mm]
       $\psv{\rel} = \{f_2\} \;\; \psv{S} = \{g_1\}$                                                                                                    \\[2mm]
       $\Delta \rel = \bag{\insertdelta(5, 8)}$                                                                                                         \\
     \end{tabular}
   \end{minipage}
 \end{minipage}
 \begin{minipage}{0.6\linewidth}
    \captionof*{table}{Output for each incremental operator}
    \vspace{-3.8mm}

    %%%%%%%%%%%%%%%%%%%%%%%%%%%%%%%%%%%%%%%%%%%%%%%%%%%%%%%%%%%%%%%%%%%%%%%%%%%%%%%%
    \begin{tabular}{|rc|}                                                                     \hline
      Table access $R$          & $\bag{\insertDeltaPair{(5, 8)}{\{f_1\}}}$                                                                                \\ \hline
      Selection $\selection_{a >3}$              & $\bag{\insertDeltaPair{(5, 8)}{\{f_1\}}}$                                                                                \\ \hline
      Join $\join_{b=d}$                   & $\bag{\insertDeltaPair{(5, 8, 7, 8)}{\{f_1, g_2\}}}$                                                                     \\ \hline
\multirow{4}{*}{Aggregation $\Aggregation{\fcsum(c)}{a}$} & $\statedata[9] = (\tupSum = 6, \tupCnt = 1, \provSketch = \{f_2, g_1\}, \fragCnt_{9} = \{f_2: 1, g_1: 1\})$ \\  \cline{2-2}
                             & $\statedata'[9] = (\tupSum = 6, \tupCnt = 1, \provSketch = \{f_2, g_1\}, \fragCnt_{9} = \{f_2: 1, g_1: 1\})$             \\
                             & $\statedata'[5] = (\tupSum = 7, \tupCnt = 1, \provSketch = \{f_1, g_2\}, \fragCnt_{5} = \{f_1: 1, g_2: 1\})$             \\ \cline{2-2}
                             & $\bag{\insertDeltaPair{(5, 7)}{\{f_1, g_2\}}}$                                                                           \\ \hline
      Having $\selection_{\fcsum(c) > 5}$                 & $\bag{\insertDeltaPair{(5, 7)}{\{f_1, g_2\}}}$                                                                           \\ \hline
      % Projection             & $\bag{\insertDeltaPair{(5, 7)}{\{f_1, g_2\}}}$                                                                           \\ \hline
\multirow{2}{*}{Merging $\fragmergeop$}     & $\statedata: \{f_2: 1, g_1: 1\}$                                                                                         \\
                             & $\statedata': \{f_1:1, f_2: 1, g_1: 1, g_2: 1\}$                                                                         \\ \hline
       Sketch delta          & $\insertdelta \{f_1, g_2\}$                                                                                              \\ \hline
    \end{tabular}
    %%%%%%%%%%%%%%%%%%%%%%%%%%%%%%%%%%%%%%%%%%%%%%%%%%%%%%%%%%%%%%%%%%%%%%%%%%%%%%%%
   \end{minipage}
%%   \begin{minipage}{0.7 \linewidth}
%%       % \centering
%%     \begin{lstlisting}
%% SELECT a, SUM(c) as sc
%% FROM (SELECT a, b FROM R WHERE a > 3)
%%       JOIN S on (b = d)
%% GROUP BY a
%% HAVING SUM(c) > 5
%%     \end{lstlisting}
%%   \end{minipage}
   \vspace{-4mm}
   \caption{Using our \gls{mp} to evaluate a query under incremental annotated semantics.}
   \label{fig:all-rules-example}
\end{figure*}
%%% Local Variables:
%%% mode: LaTeX
%%% TeX-master: "../imp"
%%% End:

%% file: sections/proof_appendix.tex
\section{Correctness Proof}\label{sec:rule_correctness_proof_appendix}
% \ifnottechreport{\subsection{Correctness Proof}\label{sec:rule_correctness_proof_appendix}}

%%%%%%%%%%%%%%%%%%%%%%%%%%%%%%%%%%%%%%%%%%%%%%%%%%%%%%%%%%%%%%%%%%%%%%%%%%%%%%%
%% main paper
We are now ready to state the main result of this paper, i.e., the incremental operator semantics we have defined is an incremental maintenance procedure. That is, it outputs valid sketch deltas that applied to the safe sketch for the database $\db$ before the update yield an over-approximation of an accurate sketch for the database $\db \uniondb \deltadb$.
\revision{
\begin{theorem}[Correctness]\label{theo:correctness}
  $\incremaintain$ as defined in \Cref{sec:problem_definition} is an incremental maintenance procedure such that it takes as input a state $\statedata$, the annotated delta $\annoDeltaDB$, the ranges $\dbranges$, a query $\query$ and returns an updated state $\statedata'$ and a provenance sketch delta $\deltaprovSketch$: $\incremaintain(\query, \dbranges, \statedata, \annoDeltaDB) = (\deltaprovSketch, \statedata')$. For any query $\query$, sketch $\provSketch$ that is valid for $\db$, and state $\statedata$ corresponding to $\db$ we have:
  \[
    \psFor{\query, \dbranges, \db \uniondb \deltadb} \subseteq \provSketch \uniondb \incremaintain(\query, \dbranges, \statedata, \annoDeltaDB)
  \]
\end{theorem}
}

\revision{
In this section, we will demonstrate the correctness of \Cref{theo:correctness}. Specifically, we introduce two auxiliary notions that
two aspects: 1. Tuple correctness (\tcorrect): the
incremental maintenance procedure can always generate the correct bag of tuples
for the operators it maintains. 2. Fragment correctness
(\fcorrect): the maintenance procedure can output the correct
delta sketches as well for the operators it maintains. Before we present the
proof of the theorem, we will establish several lemmas that used in the proof,
and propose the criteria of \tcorrect and \fcorrect.
}

% D, annoD, and Q(D), Q(annoD)
%%%%%%%%%%%%%%%%%%%%%%%%%%%%%%%%%%%%%%%%%%%%%%%%%%%%%%%%%%%%%%%%%%%%%%%%%%%%%%%%
\subsection{Tuple Correctness, Fragment Correctness, and Auxiliary Results }
\label{sec:auxil-defin-lemm}

\BGI{need to explain for the purpose we have for introducing these definitions
  and lemmas, i.e., give the reader an intuition of how they will be used.
  Basically make clear here first that you will introduce tuple correctness and
  fragment correctness as tools for the proof of \Cref{theo:correctness} and how
  we will use them and how the functions will be used to define these notions.}  

\revision{In this section, we will introduce two functions: tuple extract $\tupsIn{\cdot}$
and fragment extract $\fragsIn{\cdot}$ where the function $\tupsIn{\cdot}$
function specifies the procedure for handing tuples from annotated relations
(database) and the function $\fragsIn{\cdot}$ specifies that for handing
fragments from annotated relations (database). To
demonstrate the \tcorrect and \fcorrect, we will introduce a series of auxiliary
lemmas which present properties of the two function $\tupsIn{\cdot}$  and
$\fragsIn{\cdot}$ and are used during the proof for operators when present the
correctness of tuple and fragment.
Then, we will define the
\tcorrect and \fcorrect as tools for the proof of ~\Cref{theo:correctness}. Next,
we introduce a lemma that for each operator, the ~\Cref{theo:correctness} holds
if we can demonstrate both tuple correctness and fragment correctness hold.} 

We denote $\query^{n}$ to be a query having at most $n$ operators and
$\queryclass{i}$ to be the class of all queries with $i$ operators such
$\query^{i} \in \queryclass{i}$.

Given a database $\db =\bag{\tup_1, \ldots, \tup_n}$ and the annotated database
$\annoDB = \bag{\pair{\tup_{1}}{\provSketch_{1}}, \ldots, \pair{\tup_{n}}{\provSketch_{n}}}$,
the results of running query $\query$ over the database $\db$  and running query
over the annotated database are:
\begin{align*}
\query(\db) &= \bag{\tupout{1}, \ldots, \tupout{m}}
& \query(\annoDB) &= \bag{\pair{\tupout{1}}{\psout{1}}, \ldots, \pair{\tupout{m}}{\psout{m}}}
\end{align*}

% and the annotated result is
% \[
% $\query(\annoDB) = \bag{\pair{\tupout{1}}{\psout{1}}, \ldots, \pair{\tupout{m}}{\psout{m}}}$
% \]
% deltas
Given the delta database $\deltadb = \insertdeltadb \union
\deletedeltadb$ where $\insertdeltadb = \bag{\tupins{1}, \ldots, \tupins{k}}$ and
$\deletedeltadb = \bag{\tupdel{1}, \ldots, \tupdel{x}}$. Let $\updDB$ to
be the updated database such that:
% \[
  $\updDB = \db \uniondb \deltadb$
% \]
. Suppose the results of running query $\query$ over the database and annotated database after updating are:
\begin{align*}
  \query(\updDB) &= \bag{\tupout{1}, \ldots, \tupout{l}}
  & \query(\annoupdDB) &= \bag{\pair{\tupout{1}}{\psout{1}}, \ldots, \pair{\tupout{l}}{\psout{l}}}
\end{align*}

%%%%%%%%%%%%%%%%%%%%%%%%%%%%%%%%%%%%%%%%
% tupIn function: T()
We define $\tupsIn{\cdot}$ (\textbf{extract tuples}), a function that takes as input a bag of annotated
tuples and returns all the tuples from the bag such that:
\[
  \tupsIn{ \bag{ \pair{\tup_{1}}{\provSketch_{1}}, \ldots, \pair{\tup_{y}}{\provSketch_{y}} } }  = \bag{\tup_{1}, \ldots, \tup_{y}}
\]
%%%%%%%%%%%%%%%%%%%%%%%%%%%%%%%%%%%%%%%%

%%%%%%%%%%%%%%%%%%%%%%%%%%%%%%%%%%%%%%%%
%% LEMMA T(D1 U D2) = T(D1) U T(D2)
\begin{lem}\label{lemma:tupsin-union-db}
  Let $\tupsIn{\cdot}$ be the \textbf{extract tuples} function, $\annoDB_{1}$ and $\annoDB_{2}$ be two annotated databases with the same schema. The following holds:\\
  \[
    \tupsIn{ \annoDB_{1} \union \annoDB_{2} } = \tupsIn{\annoDB_{1}} \union \tupsIn{\annoDB_{2}}
  \]
\end{lem}
%% LEMMA PROOF
\begin{proof}
  Suppose $\annoDB_{1} = \bag{\pair{\tup_{1}}{\provSketch_{\tup_{1}}}, \ldots, \pair{\tup_{ m }}{\provSketch_{\tup_{m}}}}$\\
  and $\annoDB_{2} = \bag{\pair{s_{1}}{\provSketch_{s_{n}}}, \ldots, \pair{s_{n}}{\provSketch_{s_{n}}}}$.
  Then $\annoDB_{1} \union \annoDB_{2}$, $\tupsIn{\annoDB_{1}}$ and
  $\tupsIn{\annoDB_{2}}$ are:
  \begin{align*}
    \annoDB_{1} \union \annoDB_{2} & = \bag{\pair{\tup_{1}}{\provSketch_{\tup_{1}}}, \ldots, \pair{\tup_{n}}{\provSketch_{\tup_{n}}}, \pair{s_{1}}{\provSketch_{s_{1}}}, \ldots, \pair{s_{n}}{\provSketch_{s_{n}}} }\\
    \tupsIn{\annoDB_{1}} & = \bag{\tup_{1}, \ldots, \tup_{m}} \;\;\;\;\;\; \tupsIn{\annoDB_{2}} = \bag{s_{2}, \ldots, s_{n}}
  \end{align*}
  % T(D1 U D2)
  We can get that $\tupsIn{\annoDB_{1} \union \annoDB_{2}}$ is:
  \[
    \tupsIn{\annoDB_{1} \union \annoDB_{2}} = \bag{\tup_{1}, \ldots, \tup_{m},s_{1}, \ldots, s_{n} }
  \]
  and
  % T(D1) U T(D2)
  $\tupsIn{\annoDB_{1}} \union \tupsIn{\annoDB_{2}}$ is:
  \[
    \tupsIn{\annoDB_{1}} \union \tupsIn{\annoDB_{2}}=  \bag{\tup_{1}, \ldots, \tup_{m}} \union \bag{\tup_{m},s_{1}, \ldots, s_{n} } = \bag{\tup_{1}, \ldots, \tup_{m},s_{1}, \ldots, s_{n} }
  \]

  Therefore,$\tupsIn{ \annoDB_{1} \union \annoDB_{2} } = \tupsIn{\annoDB_{1}} \union \tupsIn{\annoDB_{2}}$
\end{proof}
%%%%%%%%%%%%%%%%%%%%%%%%%%%%%%%%%%%%%%%%

Analog we can know that $\tupsIn{\annoDeltaDB} =
\tupsIn{\insertdeltadb} \union \tupsIn{\deletedeltadb}$, since $\annoDeltaDB =
\insertdelta \annoDB \union \deletedelta \annoDB$

%%%%%%%%%%%%%%%%%%%%%%%%%%%%%%%%%%%%%%%%
%% LEMMA T(D1 - D2) = T(D1) - T(D2)
\begin{lem}\label{lemma:tupsin-difference-db}
Let  $\annoDB_{1}$ and $\annoDB_{2}$ be two annotated databases over the same schema.
We have:
\[
  \tupsIn{\annoDB_{1} \difference \annoDB_{2}} = \tupsIn{\annoDB_{1}} \difference \tupsIn{\annoDB_{2}}
  \]
\end{lem}
%%%%%%%%%%%%%%%%%%%%
\begin{proof}
  The proof is analog to the proof for \Cref{lemma:tupsin-union-db}.
\end{proof}
%%%%%%%%%%%%%%%%%%%%%%%%%%%%%%%%%%%%%%%%

%%%%%%%%%%%%%%%%%%%%%%%%%%%%%%%%%%%%%%%%
%% LEMMA T(annoDB U. deltaDB) = T(annoDB) U. T(deltaDB)
\begin{lem}\label{lemma:tupsin-udot-annodb-annodeltadb}
  Let $\tupsIn{\cdot}$ be the \textbf{extract tuples} function, $\annoDB$ and $\annoDeltaDB$ be an annotated database and an annotated database delta. The following property holds:\\
  \[
    \tupsIn{ \annoDB \uniondb \annoDeltaDB } = \tupsIn{\annoDB} \uniondb \tupsIn{\annoDeltaDB}
  \]
\end{lem}
%%%%%%%%%%%%%%%%%%%%
%% LEMMA PROOF
\begin{proof}
  Suppose $\annoDB$ and $\annoDeltaDB$ are:
  \begin{align*}
    &\annoDB = \bag{\pair{\tup_{1}}{\provSketch_{\tup_{1}}}, \ldots, \pair{\tup_{m}}{\provSketch_{\tup_{m}}} }\\
    &\annoDeltaDB = \insertdelta \annoDB \union \deletedelta \annoDB = \bag{\deleteDeltaPair{\tupdel{1}}{\psdel{1}}, \ldots, \deleteDeltaPair{\tupdel{j}}{\psdel{j}}} \\
    & \hspace{3mm} \union \bag{\insertDeltaPair{\tupins{1}}{\psins{1}}, \ldots, \insertDeltaPair{\tupins{i}}{\psins{i}} }
  \end{align*}

  Then $\annoDB \uniondb \annoDeltaDB$ is:
  \begin{align*}
      & \annoDB \uniondb \annoDeltaDB\\
    = & \annoDB \difference \deletedelta \annoDB \union \insertdelta \annoDB\\
    = & \bag{\pair{\tup}{\provSketch} \mid \pair{\tup}{\provSketch} \in \annoDB} - \bag{\pair{\tup}{\provSketch} \mid \pair{\tup}{\provSketch} \in \deletedelta \annoDB} \\
      & \hspace{3mm} \union \bag{\pair{\tup}{\provSketch} \mid \pair{\tup}{\provSketch} \in \insertdelta \annoDB}\\
    = & \bag{\pair{\tup_{1}}{\provSketch_{\tup_{1}}}, \ldots, \pair{\tup_{m}}{\provSketch_{\tup_{m}}} } \\
      & \hspace{3mm}- \bag{\deleteDeltaPair{\tupdel{1}}{\psdel{1}}, \ldots, \deleteDeltaPair{\tupdel{j}}{\psdel{j}}} \\
      & \hspace{3mm} \union \bag{\insertDeltaPair{\tupins{1}}{\psins{1}}, \ldots, \insertDeltaPair{\tupins{i}}{\psins{i}} }
  \end{align*}
  Then $\tupsIn{\annoDB \uniondb \annoDeltaDB}$:
  \begin{align*}
      & \tupsIn{\annoDB \uniondb \annoDeltaDB} \\
    = & \bag{\tup_{1}, \ldots, \tup_{m} } \\
      & \hspace{3mm}- \bag{\deletedelta \tupdel{1}, \ldots, \deletedelta\tupdel{j}} \\
      & \hspace{3mm} \union \bag{\insertdelta\tupins{1}, \ldots, \insertdelta\tupins{i}}
  \end{align*}
  $\tupsIn{\annoDB} \uniondb \tupsIn{\annoDeltaDB}$ are:
  \begin{align*}
      & \tupsIn{\annoDB} \uniondb \tupsIn{\annoDeltaDB}\\
    = & \tupsIn{\annoDB} \uniondb \tupsIn{\deletedelta \annoDB \union \insertdelta \annoDB}\\
    = & \tupsIn{\annoDB} \difference \tupsIn{\deletedelta \annoDB} \union \tupsIn{\insertdelta \annoDB}\\
    = & \bag{\tup_{1}, \ldots, \tup_{m} } \\
      & \hspace{3mm} \difference \bag{\deletedelta \tupdel{1}, \ldots, \deletedelta\tupdel{j}} \\
      & \hspace{3mm} \union \bag{\insertdelta\tupins{1}, \ldots, \insertdelta\tupins{i}}
  \end{align*}

  Therefore, $ \tupsIn{ \annoDB \uniondb \annoDeltaDB } = \tupsIn{\annoDB} \uniondb \tupsIn{\annoDeltaDB}$
\end{proof}
%%%%%%%%%%%%%%%%%%%%%%%%%%%%%%%%%%%%%%%%

%% Tuples correctness;
\revision{
% We have shown the function of $\tupsIn{\cdot}$. Thus to prove the tuple
% correctness of the theorem ~\Cref{theo:correctness}, the following property
% should hold:
% \begin{align*}
%   & \tupsIn{\query(\annoDB) \uniondb  \ieval{\query}{\dbranges,\annoDeltaDB}{\statedata}} = \query(\annoupdDB) \tag{\tcorrect}
% \end{align*}
% \BGI{Probably tuple correctness deserves a separate definition}
\begin{defi}[\Tcorrect] \label{def:tuple_correctness}
Consider a query $\query$, a database $\db$ and a delta database $\deltadb$, Let $\updDB$
be the updated database such that $\updDB = \db \uniondb \deltadb$. Let
$\incremaintain$ be an incremental maintenance procedure takes as input a state
$\statedata$, the query $\query$,  the annotated delta $\annoDeltaDB$ and the
range partition $\dbranges$. The tuples in the result of the query running over
the updated database are equivalent to tuples in the result of applying  
query running over the database to incremental maintenance procedure such that: 
\[
  \tupsIn{\query(\annoDB) \uniondb  \ieval{\query}{\dbranges,\annoDeltaDB}{\statedata}} = \query(\updDB) %\tag{\tcorrect}
\]
\end{defi}
}

%%%%%%%%%%%%%%%%%%%%%%%%%%%%%%%%%%%%%%%%%%%%%%%%%%%%%%%%%%%%%%%%%%%%%%%%%%%%%%%%
%% Fragment correctness
Recall $\psFor{\query, \dbranges, \db}$ defines an accurate provenance
sketch $\provSketch$ for $\query$ wrt. to $\db$, and ranges $\dbranges$,
and $\dbInst{\provSketch}$ is an instance of $\provSketch$ which is the
data covered by the sketch.

Now we define a function $\fragsIn{\cdot}$ (\textbf{fragments extracting}) that
takes as input a bag of annotated tuples and return all the sketches from this
bag such that:
\[
  \fragsIn{ \bag{ \pair{\tup_{1}}{\provSketch_{1}}, \ldots, \pair{\tup_{y}}{\provSketch_{y}} } }  = \bag{\provSketch_{1}, \ldots, \provSketch_{y}}
\]
And the $\fragsIn{\cdot}$ function has the following property:

\begin{lem}\label{lemma:frags_annodbU_Uannodb}
  \fragsIn{\annoDB_{1} \union \annoDB_{2}} = \fragsIn{\annoDB_{1}} \union \fragsIn{\annoDB_{2}}
\end{lem}
The proof of this property is similar to \cref{lemma:tupsin-union-db} where for
$\fragsIn{\cdot}$, we focus on fragments instead of tuples. Therefore, $\fragsIn{\annoDeltaDB} = \fragsIn{\insertdelta \annoDB} \union \fragsIn{\deletedelta \annoDB}$, since $\deltadb = \insertdelta \db \union \deletedelta \db$.

%%%%%%%%%%%%%%%%%%%%%%%%%%%%%%%%%%%%%%%%
%% lemma
Like \Cref{lemma:tupsin-udot-annodb-annodeltadb}, the \textbf{extract fragments}
function has the following properties as well.
\begin{lem}\label{lemma:frags-udot-annodb-annodeldb}
  $\fragsIn{\annoDB \uniondb \annoDeltaDB} = \fragsIn{\annoDB} \uniondb \fragsIn{\annoDeltaDB}$
\end{lem}

\begin{proof}
  Suppose $\annoDB$ and $\annoDeltaDB$ are:
  \begin{align*}
    &\annoDB = \bag{\pair{\tup_{1}}{\provSketch_{\tup_{1}}}, \ldots, \pair{\tup_{m}}{\provSketch_{\tup_{m}}} }\\
    &\annoDeltaDB = \insertdelta \annoDB \union \deletedelta \annoDB = \bag{\deleteDeltaPair{\tupdel{1}}{\psdel{1}}, \ldots, \deleteDeltaPair{\tupdel{j}}{\psdel{j}}} \\
    & \hspace{3mm} \union \bag{\insertDeltaPair{\tupins{1}}{\psins{1}}, \ldots, \insertDeltaPair{\tupins{i}}{\psins{i}} }
  \end{align*}

  Then $\annoDB \uniondb \annoDeltaDB$ is:
  \begin{align*}
      & \annoDB \uniondb \annoDeltaDB\\
    = & \annoDB \difference \deletedelta \annoDB \union \insertdelta \annoDB\\
    = & \bag{\pair{\tup}{\provSketch} \mid \pair{\tup}{\provSketch} \in \annoDB} - \bag{\pair{\tup}{\provSketch} \mid \pair{\tup}{\provSketch} \in \deletedelta \annoDB} \\
      & \hspace{3mm} \union \bag{\pair{\tup}{\provSketch} \mid \pair{\tup}{\provSketch} \in \insertdelta \annoDB}\\
    = & \bag{\pair{\tup_{1}}{\provSketch_{\tup_{1}}}, \ldots, \pair{\tup_{m}}{\provSketch_{\tup_{m}}} } \\
      & \hspace{3mm}- \bag{\deleteDeltaPair{\tupdel{1}}{\psdel{1}}, \ldots, \deleteDeltaPair{\tupdel{j}}{\psdel{j}}} \\
      & \hspace{3mm} \union \bag{\insertDeltaPair{\tupins{1}}{\psins{1}}, \ldots, \insertDeltaPair{\tupins{i}}{\psins{i}} }
  \end{align*}
  Then $\fragsIn{\annoDB \uniondb \annoDeltaDB}$:
  \begin{align*}
      & \fragsIn{\annoDB \uniondb \annoDeltaDB} \\
    = & \bag{\provSketch_{1}, \ldots, \provSketch_{m} } \\
      & \hspace{3mm}- \bag{\deletedelta \psdel{1}, \ldots, \deletedelta \psdel{j}} \\
      & \hspace{3mm} \union \bag{\insertdelta \psins{1}, \ldots, \insertdelta \psins{i}}
  \end{align*}
  $\fragsIn{\annoDB} \uniondb \fragsIn{\annoDeltaDB}$ are:
  \begin{align*}
      & \fragsIn{\annoDB} \uniondb   \fragsIn{\annoDeltaDB}\\
    = & \fragsIn{\annoDB} \uniondb   \fragsIn{\deletedelta \annoDB \union \insertdelta \annoDB}\\
    = & \fragsIn{\annoDB} \difference \fragsIn{\deletedelta \annoDB} \union \fragsIn{\insertdelta \annoDB}\\
    = & \bag{\provSketch_{1}, \ldots, \provSketch_{m} } \\
      & \hspace{3mm} \difference \bag{\deletedelta \psdel{1}, \ldots, \deletedelta\psdel{j}} \\
      & \hspace{3mm} \union \bag{\insertdelta\psins{1}, \ldots, \insertdelta\psins{i}}
  \end{align*}

  Therefore, $ \fragsIn{ \annoDB \uniondb \annoDeltaDB } = \fragsIn{\annoDB} \uniondb \fragsIn{\annoDeltaDB}$
\end{proof}
%%%%%%%%%%%%%%%%%%%%%%%%%%%%%%%%%%%%%%%%

We now define $\fragSet{\cdot}$ as a function that takes as input \underline{a bag} of
fragments where each fragment has a multiplicity at least 1, and returns \underline{a set}
of fragments such that the multiplicity of each fragment in input bag is 1 in the output set.
For a query $\query$ running over
a annotated database $\annoDB$, the annotated result is $\query(\annoDB)$. The
fragments in the annotated result are $\fragsIn{ \query(\annoDB) }$. Suppose the
provenance sketch captured for this query given the ranges $\dbranges$ is
$\psFor{\query, \dbranges, \db}$, then we can get that:
\[
  \psFor{\query, \dbranges, \db} = \fragSet{\fragsIn{ \query(\annoDB) }}
\]

A provenance sketch covers relevant data of the database to answer a query such
that:
\[
  \query(\db) = \query(\dbInst{\psFor{\query, \dbranges, \db}}) = \query(\dbInst{\fragSet{\fragsIn{\query(\annoDB)}}})
\]

% When a provenance sketch is used to answer a query, it will be translated into

For a query running over a set of fragments, it is the same as running a bag of
fragments where the ranges are exactly the same as in the set but each one
having different multiplicity. The reason is for a fragment appearing
multiple times in the bag, when using this fragment to answer, they will be
translate into the same expression in the \lstinline!WHERE! clause multiple
times concatenating with \lstinline!OR!. And this expression appearing multiple
times will be treated as a single one when the database engine evaluates the
\lstinline!WHERE!. For example, \lstinline!WHERE (a BETWEEN 10 AND 20)!
\lstinline! OR (a BETWEEN 10 AND 20)! has
the same effect as \lstinline!WHERE (a BETWEEN 10 AND 20)! for a query. Thus,
the following holds:
\[
  \query(\db)=  \query(\dbInst{\fragSet{\fragsIn{\query(\annoDB)}}} = \query(\dbInst{\fragsIn{\query(\annoDB)}})
\]

%%%%%%%%%%%%%%%%%%%%%%%%%%%%%%%%%%%%%%%%
%% frag correctness
\revision{
% Thus, to demonstrate the the fragment correctness of the
% ~\Cref{theo:correctness}, the following property should hold:
% \begin{align*}
%   & \query(\updDB) = \query(\updDBInst{\fragSet{\fragsIn{\query(\annoDB)} \uniondb \fragsIn{\ieval{\query}{\dbranges,\annoDeltaDB}{\statedata}}}}) \tag{\fcorrect}
% \end{align*}
\begin{defi}[\Fcorrect] \label{def:frag_correctness}
Consider a query $\query$, a database $\db$ and a delta database $\deltadb$, Let $\updDB$
be the updated database such that $\updDB = \db \uniondb \deltadb$. Let
$\incremaintain$ be an incremental maintenance procedure takes as input a state
$\statedata$, the query $\query$,  the annotated delta $\annoDeltaDB$ and the
range partition $\dbranges$. The result of the running query over
the updated database is equivalent to the result of   
running query over the data of updated database covered by applying the fragments in current provenance
sketch to fragments generated from the incremental maintenance procedure such that:
\[
  \query(\updDB) = \query(\updDBInst{\fragSet{\fragsIn{\query(\annoDB)} \uniondb \fragsIn{\ieval{\query}{\dbranges,\annoDeltaDB}{\statedata}}}}) %\tag{\fcorrect}
\]
\end{defi}
}

\revision{
  \begin{lem}\label{lemma:tc_fc_hold_theorem_hold}
    For an operator that the sketch is maintained by the incremental procedure, if
    both \tcorrect and \fcorrect hold, then the ~\Cref{theo:correctness} holds
    for this operator.
  \end{lem}
}

\revision{
%We have established that
%\Cref{lemma:tupsin-union-db,lemma:tupsin-difference-db,lemma:tupsin-udot-annodb-annodeltadb,lemma:frags_annodbU_Uannodb,lemma:frags-udot-annodb-annodeldb}
%hold.\BGI{What is their purpose?}
  In the following, we will prove \Cref{theo:correctness} by induction over the
  structure of queries starting with the base case which is the correctness of
  query consisting of a single table access operator followed by inductive steps
  for other operators distinguishing the cases of tuple and fragment
  correctness to show the ~\Cref{lemma:tc_fc_hold_theorem_hold}. 
}

\BGI{Should we add a lemma that states that 6.1 holds as long as both tuple correctness and fragment correctness holds?}

%%%%%%%%%%%%%%%%%%%%%%%%%%%%%%%%%%%%%%%%%%%%%%%%%%%%%%%%%%%%%%%%%%%%%%%%%%%%%%%%%
%% tikz figs show how Q(i + 1) relates Q(i))
\begin{figure}[h!]
  \centering
  \begin{tikzpicture}[
    node distance=1cm and 2cm,
    every node/.style={ minimum size=1cm},
    every path/.style={-, thick}
  ]
    %% Q^{i + 1}
    \node[draw=none] (top) at (-2, 2) {$\query^{i + 1}: \;$};
    % Top row nodes
    \node (t1) at (-0.5, 2) {$\pair{\tup_{1}}{\provSketch_{\tup_{1}}}$};
    \node (t2) at (1, 2) {$\pair{\tup_{2}}{\provSketch_{\tup_{2}}}$};
    \node (t3) at (2.5, 2) {$\ldots$};
    \node (t4) at (4, 2) {$\pair{\tup_{m}}{\provSketch_{\tup_{m}}}$};
    %% define coordinates for t1,2,4
    \coordinate (pt1) at (-0.5, 1.8); % Define point A at (1, 2)
    \coordinate (pt2) at (1, 1.8); % Define point A at (1, 2)
    %% \coordinate (pt4) dots, not needed
    \coordinate (pt4) at (4, 1.8); % Define point A at (1, 2)

    %% Q^{i}
    \node[draw=none] (bottom) at (-2, 0) {$\query^{i}: \;$};
    % Bottom row nodes
    \node (s1) at (-1, 0) {$\pair{s_{1}}{\provSketch_{s_{1}}}$};
    \node (s2) at (0.5, 0) {$\pair{s_{2}}{\provSketch_{s_{2}}}$};
    \node (s3) at (2, 0) {$\ldots$};
    \node (s4) at (3.5, 0) {$\pair{s_{n -1}}{\provSketch_{s_{n -1}}}$};
    \node (s5) at (5.3, 0) {$\pair{s_{n}}{\provSketch_{s_{n}}}$};
    %% define coordinates for s1,2,45
    \coordinate (ps1) at (-1, 0.2); % Define point A at (1, 2)
    \coordinate (ps2) at (0.5, 0.2); % Define point A at (1, 2)
    \coordinate (ps3) at (2, 0.2);
    \coordinate (ps4) at (3.5, 0.2); % Define point A at (1, 2)
    \coordinate (ps5) at (5.3, 0.2); % Define point A at (1, 2)

    \draw (pt1) -- (ps1);
    % \draw (pt2) -- (ps2);
    \draw (pt2) -- (ps4);
    \draw (pt4) -- (ps2);
    \draw (pt4) -- (ps3);
    \draw (pt4) -- (ps4);
    \draw (pt4) -- (ps5);
  \end{tikzpicture}
  \caption{Annotated tuples' relation between $\query^{i}$ and $\query^{i + 1}$}
  \label{fig:anno_tuples_rel_between_two_levels}
\end{figure}
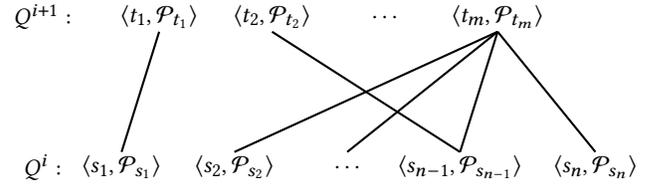

%%%%%%%%%%%%%%%%%%%%%%%%%%%%%%%%%%%%%%%%%%%%%%%%%%%%%%%%%%%%%%%%%%%%%%%%%%%%%%%%
%% inductive proof
\subsection{Proof of \Cref{theo:correctness}}

\revision{Having the properties \tcorrect and \fcorrect defined, we are ready to prove \Cref{theo:correctness} by induction over the structure of a query showing that the tuple correctness and fragment correctness properties hold for supported query which implies the theorem.}

%%%%%%%%%%%%%%%%%%%%%%%%%%%%%%%%%%%%%%%%%%%%%%%%%%%%%%%%%%%%%%%%%%%%%%%%%%%%%%%%
%% base case for a table access operator;
\subsubsection{Base Case}
We start with $\query^{1}$, which is a single table $\rel$. We will show that
for any $\query^{1} \in \queryclass{1}$, the \tcorrect and \fcorrect
properties hold for table access operator.

Suppose the relation and its annotated relation are:
\[
  \rel = \bag{\tup_1, \ldots, \tup_n} \;\;\;\;
  \annoRel = \bag{\pair{\tup_1}{\provSketch_{1}}, \ldots, \pair{\tup_n}{\provSketch_{n}}}
\]
and delta relation and annotated delta relation are:
\[
  \insertdelta \rel = \bag{\tupins{1}, \ldots, \tupins{i}} \;\;\;\; \deletedelta \rel = \bag{\tupdel{1}, \ldots, \tupdel{j}}
 \]
 \[
  \insertdelta \annoRel = \bag{\pair{\tupins{1}}{\psins{1}}, \ldots, \pair{\tupins{i}}{\psins{i}}} \;\; \deletedelta \annoRel = \bag{\pair{\tupdel{1}}{\psdel{1}}, \ldots, \pair{\tupdel{j}}{\psdel{j}}}
\]
\parttitle{\Tcorrect}
\begin{align*}
  & \tupsIn{\query^{1}(\annoDB) \uniondb \ieval{\query^{1}}{\dbranges}{\annoDeltaDB}} \\
= & \tupsIn{\query^{1}(\annoRel) \uniondb \ieval{\query^{1}}{\dbranges}{\annoDeltaRel}} \\
= & \tupsIn{\annoRel \uniondb \Delta \annoRel} \tag{~\Cref{lemma:tupsin-udot-annodb-annodeltadb}}\\
= & \tupsIn{\annoRel \difference \deletedelta \annoRel \union \insertdelta \annoRel} \tag{Applying delta}\\
= & \tupsIn{\annoRel} \difference \tupsIn{\deletedelta \annoRel} \union \tupsIn{\insertdelta \annoRel} \tag{~\Cref{lemma:tupsin-union-db} and ~\Cref{lemma:tupsin-difference-db}}\\
= & \rel \difference \deletedelta \rel \union \insertdelta \rel \tag{\tupsIn{\cdot}}\\
= & \query^{1} (\rel \difference \deletedelta \rel \union \insertdelta \rel) \\
= & \query^{1} (\rel \uniondb \deltaRel) \\
= & \query^{1} (\rel')
\end{align*}

%%%%%%%%%%%%%%%%%%%%%%%%%%%%%%%%%%%%%%%%
\parttitle{\Fcorrect}
First, determine the fragments after the incremental maintenance:
\begin{align*}
  & \fragsIn{\query^{1}(\annoDB)} \uniondb \fragsIn{ \ieval{\query^{1}}{\dbranges}{\annoDeltaDB} }\\
= & \fragsIn{\annoRel} \uniondb \fragsIn{  \annoDeltaRel } \\
= & \fragsIn{\annoRel} \uniondb \fragsIn{ \insertdelta \annoRel \union \deletedelta \annoRel} \\
= & \fragsIn{\annoRel} \difference \fragsIn{ \deletedelta \annoRel} \union {\insertdelta \annoRel} \\
= & \bag{\provSketch_{1}, \ldots, \provSketch_{n}} \difference \deletedelta \bag{\psdel{1}, \ldots, \psdel{j}} \union \insertdelta \bag{\psins{1}, \ldots, \psins{i}}
\end{align*}
Since for every sketch in
$\bag{\provSketch_{1}, \ldots, \provSketch_{n}} \uniondb \deletedelta
\bag{\psdel{1}, \ldots, \psdel{j}} \uniondb \insertdelta \bag{\psins{1}, \ldots,
  \psins{i}} $, it associates a tuple, and the associated tuple will be inserted or deleted as
the sketch does. Thus:
\begin{align*}
  \updDBInst{\fragsIn{\query^{1}(\annoDB)} \uniondb \fragsIn{\ieval{\query^{1}}{\dbranges}{\annoDeltaDB}}} & = \bag{\tup_{1}, \ldots, \tup_{n}} \uniondb \insertdelta \bag{\tupins{1}, \ldots, \tupins{i}}\\
  & \hspace{10mm} \uniondb \deletedelta \bag{\tupdel{1}, \ldots, {\tupdel{j}}}
\end{align*}
Thus,
\begin{align*}
    & \query^{1}(\updDB_{\fragsIn{\query^{1}(\annoDB)} \uniondb \fragsIn{\ieval{\query}{\dbranges,\annoDeltaDB}{\statedata}}}) \\
  = & \bag{\tup_{1}, \ldots, \tup_{n}} \uniondb \insertdelta \bag{\tupins{1}, \ldots, \tupins{i}} \uniondb \deletedelta \bag{\tupdel{1}, \ldots, {\tupdel{j}}} \\
  = & \rel \union \insertdelta \rel \difference \deletedelta \rel
\end{align*}

Since $\query^{1} (\rel') = \query^{1}(\rel \uniondb \deltaRel) = \rel \difference \deletedelta \rel \union \insertdelta \rel$
, then we get for $\query^{1} \in \queryclass{1}$ that:
$\query^{1} (\updDB)  = \query^{1}(\updDB_{\fragsIn{\query^{1}(\annoDB)} \uniondb \fragsIn{\ieval{\query^{1}}{\dbranges,\annoDeltaDB}{\statedata}}}) = \query^{1}(\updDBInst{\fragSet{\fragsIn{\query^{1}(\annoDB)} \uniondb \fragsIn{\ieval{\query^{1}}{\dbranges,\annoDeltaDB}{\statedata}}}})$.

%%%%%%%%%%%%%%%%%%%%%%%%%%%%%%%%%%%%%%%%%%%%%%%%%%%%%%%%%%%%%%%%%%%%%%%%%%%%%%%%
%% General assumption for Q^{i}, the two properties hold.
\subsubsection{Inductive Step}
\revision{
Assume for ~\Cref{theo:correctness}, the two correctness properties hold for $\query^{i} \in
\queryclass{i}$ such that the incremental maintenance procedure can correctly
produce the tuples and provenance sketches:
}
\revision{
\begin{align*}
  & \query^{i}(\updDB) = \tupsIn{\query^{i}(\annoDB) \uniondb  \ieval{\query^{i}}{\dbranges,\annoDeltaDB}{\statedata} } \tag{\tcorrect}\\
  & \query^{i}(\updDB) = \query^{i}(\updDBInst{\fragSet{\fragsIn{\query^{1}(\annoDB)} \uniondb \fragsIn{\ieval{\query^{i}}{\dbranges,\annoDeltaDB}{\statedata}}}}) \tag{\fcorrect}
\end{align*}
}

\revision{
Next we will show that for $\query^{i + 1} \in \queryclass{i + 1}$, where the
$(i + 1)$'s operator is on top of $\query^{i}$, both
properties hold such that:
\begin{align*}
% tup
  & \query^{i + 1}(\updDB) = \tupsIn{\query^{i + 1}(\annoDB) \uniondb \ieval{\query^{i + 1}}{\dbranges,\annoDeltaDB}{\statedata} }\tag{\tcorrect}\\
% frag
  & \query^{i + 1}(\updDB) = \query^{i + 1}(\updDBInst{\fragSet{\fragsIn{\query^{i + 1}(\annoDB)} \uniondb \fragsIn{\ieval{\query^{i + 1}}{\dbranges,\annoDeltaDB}{\statedata}}}})\tag{\fcorrect}
\end{align*}
In the following parts, we will demonstrate the proof for operators distinguishing cases for correctness of tuples and fragments.
}
%%%%%%%%%%%%%%%%%%%%%%%%%%%%%%%%%%%%%%%%%%%%%%%%%%%%%%%%%%%%%%%%%%%%%%%%%%%%%%%%
%%
%% The following is for each operator at level i + 1, we will use the assumption
%% to proof that the properties hold for Q^{i + 1} given Q^{i} holds.
%%
%%%%%%%%%%%%%%%%%%%%%%%%%%%%%%%%%%%%%%%%%%%%%%%%%%%%%%%%%%%%%%%%%%%%%%%%%%%%%%%%

%%%%%%%%%%%%%%%%%%%%%%%%%%%%%%%%%%%%%%%%%%%%%%%%%%%%%%%%%%%%%%%%%%%%%%%%%%%%%%%%
%%  Q^{i + 1} is a SELECTION
\subsubsection{Selection}
Suppose the operator at level $i + 1$ is an selection operator. Then we know
that
\[
  \query^{i + 1}(\db)  = \selection_{\theta} (\query^{i} (\db))
\]

For a selection operator, the annotated tuples' relation between $\query^{i+ 1}$
and $\query^{i}$ is like $\pair{\tup_{1}}{\provSketch_{1}}$ and
$\pair{s_{1}}{\provSketch_{s_{1}}}$ in
\Cref{fig:anno_tuples_rel_between_two_levels} if the annotated tuple
$\pair{s_{1}}{\provSketch_{s_{1}}}$ can satisfy the selection condition of
$\query^{i + 1}$. Otherwise, there is no output for an input annotated tuple.
%% it looks like
%% $\pair{\tup_{1}}{\provSketch_{\tup_{1}}}$ and
%% $\pair{s_{1}}{\provSketch_{s_{1}}}$, where $\pair{s_{1}}{\provSketch_{s_{1}}}$
%% contribute nothing to $\query^{i + 1}$.

%% We assume that the $\query^{i + 1}(\db)$, $\query^{i + 1}(\annoDB)$,
%% $\ieval{\query^{i}}{\dbranges,\annoDeltaDB}{\statedata}$ are as follows:
%% % Q res
%% \[
%%   \query^{i + 1}(\db) = \bag{\tup_{1}, \ldots, \tup_{n}}
%% \]
%% % Q annotated res
%% \[
%%   \query^{i + 1}(\annoDB) = \bag{\pair{\tup_{1}}{\provSketch_{1}}, \ldots, \pair{\tup_{n}}{\provSketch_{n}}}
%% \]
%% % Qi incre res
%% \begin{align*}
%% \ieval{\query^{i}}{\dbranges,\annoDeltaDB}{\statedata} = & \bag{\deleteDeltaPair{\tupdel{1}}{\psdel{1}}, \ldots, \deleteDeltaPair{\tupdel{j}}{\psdel{j}}} \\
%%  &\uniondb \bag{\insertDeltaPair{\tupins{1}}{\psins{1}}, \ldots, \insertDeltaPair{\tupins{i}}{\psins{i}}}
%% \end{align*}
%%%%%%%%%%%%%%%%%%%%%%%%%%%%%%%%%%%%%%%%
%% selection: tuple correctness
\parttitle{\Tcorrect}
Before we demonstrate the tuple correctness, we first show two properties that
will be used for selection operator:

%%%%%%%%%%%%%%%%%%%%%%%%%%%%%%%%%%%%%%%%
%% two properties and proof

%%%%%%%%%%%%%%%%%%%%
%% property 1
\begin{lem}\label{lemma:sel-Tsel-selT}
  The following properties hold:
  \begin{align*}
    &\selection_{\theta}(\tupsIn{\annoDB}) = \tupsIn{\selection_{\theta}(\annoDB)}
    % &\selection_{\theta}(\tupsIn{\annoDB \uniondb \annoDeltaDB}) = \selection_{\theta}(\tupsIn{\annoDB}) \uniondb \selection_{\theta}(\tupsIn{\annoDeltaDB})
  \end{align*}
\end{lem}

\begin{proof}
  For $\selection_{\theta}(\tupsIn{\annoDB})$:
  \begin{align*}
      & \selection_{\theta}(\tupsIn{\annoDB})\\
    = & \selection_{\theta}(\db) \tag{$\tupsIn{\annoDB} = \db$}\\
    = & \bag{\tup \mid \tup \in \db \wedge \tup \models \theta} \tag{$\selection_{\theta}$}
  \end{align*}

  For $\tupsIn{\selection_{\theta}(\annoDB)}$:
  \begin{align*}
      & \tupsIn{\selection_{\theta}(\annoDB})\\
    = & \tupsIn{\bag{\pair{\tup}{\provSketch} \mid \pair{\tup}{\provSketch} \in \annoDB \wedge \tup \models \theta}} \tag{$\selection_{\theta}$}\\
    = & \bag{\tup \mid \tup \in \db \wedge \tup \models \theta}  \tag{$\tupsIn{\annoDB} = \db$}
  \end{align*}

  Therefore, $\selection_{\theta}(\tupsIn{\annoDB}) = \tupsIn{\selection_{\theta}(\annoDB)}$:
\end{proof}

%%%%%%%%%%%%%%%%%%%%
%% property 2
\begin{lem}\label{lemma:sel-annodb-udot-annodel}
  The following properties hold:
  \begin{align*}
    \selection_{\theta}(\tupsIn{\annoDB \uniondb \annoDeltaDB}) = \selection_{\theta}(\tupsIn{\annoDB}) \uniondb \selection_{\theta}(\tupsIn{\annoDeltaDB})
  \end{align*}
\end{lem}

\begin{proof}
  For $ \selection_{\theta}(\tupsIn{\annoDB \uniondb \annoDeltaDB})$:
  \begin{align*}
      & \selection_{\theta}(\tupsIn{\annoDB \uniondb \annoDeltaDB}) \\
    = & \selection_{\theta}(\tupsIn{\annoDB} \uniondb \tupsIn{\annoDeltaDB}) \tag{\cref{lemma:tupsin-udot-annodb-annodeltadb}}\\
    = & \selection_{\theta}(\db \uniondb \deltadb) \tag{$\tupsIn{\annoDB} = \db$}\\
    = & \selection_{\theta}(\db \difference \deletedelta \db \union \insertdelta \db) \tag{$\db \uniondb \deltadb = \db \difference \deletedelta \db \union \insertdelta \db$}\\
    = & \selection_{\theta}(\db) \difference \selection_{\theta}(\deletedelta \db) \union \selection_{\theta}(\insertdelta \db) \tag{$\selection_{\theta}$ and $\union, \difference$}
  \end{align*}
  For $\selection_{\theta}(\tupsIn{\annoDB}) \uniondb \selection_{\theta}(\tupsIn{\annoDeltaDB})$
  \begin{align*}
      & \selection_{\theta}(\tupsIn{\annoDB}) \uniondb \selection_{\theta}(\tupsIn{\annoDeltaDB})\\
    = & \tupsIn{\selection_{\theta}(\annoDB)} \uniondb \tupsIn{\selection_{\theta}(\annoDeltaDB)} \tag{\cref{lemma:sel-Tsel-selT}}\\
    = & \selection_{\theta}(\annoDB) \uniondb \selection_{\theta}(\deltadb) \tag{\tupsIn{\annoDB} = \db}\\
    = & \selection_{\theta}(\annoDB) \uniondb \selection_{\theta}(\deletedelta \db \union \insertdelta \db) \\
    = & \selection_{\theta}(\annoDB) \uniondb \left(\selection_{\theta}(\deletedelta \db) \union \selection_{\theta}(\insertdelta \db)\right) \\
    = & \selection_{\theta}(\annoDB) \difference \selection_{\theta}(\deletedelta \db) \union \selection_{\theta}(\insertdelta \db) \tag{$\uniondb$}\\
  \end{align*}
  Therefore, $\selection_{\theta}(\tupsIn{\annoDB \uniondb \annoDeltaDB}) = \selection_{\theta}(\tupsIn{\annoDB}) \uniondb \selection_{\theta}(\tupsIn{\annoDeltaDB})$
\end{proof}

%%%%%%%%%%%%%%%%%%%%%%%%%%%%%%%%%%%%%%%%
Therefore, the \tcorrect is as following:
\begin{proof}
\begin{align*}
    & \query^{i + 1}(\updDB) \\
  = & \selection_{\theta}(\query^{i}(\updDB)) \tag{$\query^{i + 1} = \selection_{\theta}$}\\
  = & \selection_{\theta}\left(\tupsIn{\query^{i}(\annoDB) \uniondb \ieval{\query^{i}}{\dbranges,\annoDeltaDB}{\statedata} }\right) \tag{$Q^{i}$ holds the property}\\
  = & \selection_{\theta}\left(\tupsIn{\query^{i}(\annoDB)} \uniondb \tupsIn{\ieval{\query^{i}}{\dbranges,\annoDeltaDB}{\statedata}}\right)\tag{\cref{lemma:tupsin-udot-annodb-annodeltadb}}\\
  = & \selection_{\theta}\left(\tupsIn{\query^{i}(\annoDB)}\right) \uniondb \selection_{\theta}\left(\tupsIn{\ieval{\query^{i}}{\dbranges,\annoDeltaDB}{\statedata}}\right) \tag{\cref{lemma:sel-annodb-udot-annodel}}\\
  = & \tupsIn{\selection_{\theta}(\query^{i}(\annoDB))} \uniondb \tupsIn{\selection_{\theta}(\ieval{\query^{i}}{\dbranges,\annoDeltaDB}{\statedata})} \tag{\cref{lemma:sel-Tsel-selT}}\\
  = & \tupsIn{\query^{i + 1}(\annoDB)} \uniondb \tupsIn{\selection_{\theta}(\ieval{\query^{i}}{\dbranges,\annoDeltaDB}{\statedata})} \tag{$\selection_{\theta}(\query^{i})=\query^{i + 1}$}\\
  = & \tupsIn{\query^{i + 1}(\annoDB)} \uniondb \bag{\tup \mid \tup \in \tupsIn{\ieval{\query^{i}}{\dbranges,\annoDeltaDB}{\statedata}} \wedge \tup \models \theta}\tag{ $\selection_{\theta}$ }\\
  = & \tupsIn{ \query^{i +1}(\annoDB) } \uniondb \tupsIn{\ieval{\selection_{\theta}(\query^{i})}{\dbranges,\annoDeltaDB}{\statedata}} \tag{$\incremaintain(\selection_{\theta})$ rule}\\
  = & \tupsIn{ \query^{i +1}(\annoDB) } \uniondb \tupsIn{ \ieval{\query^{i +
      1}}{\dbranges,\annoDeltaDB}{\statedata} } \tag{$\selection_{\theta}(\query^{i})=\query^{i + 1}$}\\
  = & \tupsIn{ \query^{i +1}(\annoDB) \uniondb \ieval{\query^{i + 1}}{\dbranges,\annoDeltaDB}{\statedata}}\tag{\cref{lemma:tupsin-udot-annodb-annodeltadb}}
\end{align*}
\end{proof}
%
%%%%%%%%%%%%%%%%%%%%%%%%%%%%%%%%%%%%%%%%
\parttitle{\Fcorrect}
We have the assumption that:
\[
\query^{i}(\updDB) = \query^{i}(\updDBInst{\fragSet{\fragsIn{\query^{1}(\annoDB)} \uniondb \fragsIn{\ieval{\query^{i}}{\dbranges,\annoDeltaDB}{\statedata}}}})
\]

Then:
\begin{align*}
    & \query^{i + 1}(\updDB)\\
  = & \selection_{\theta}(\query^{i}(\updDB)) \tag{$\query^{i + 1} = \selection_{\theta}(\query^{i})$}\\
  % = & \selection_{\theta}(\query^{i}(\updDB)) \tag{$\tupsIn{\annoDB} = \db$}\\
  = & \selection_{\theta}(\query^{i}(\updDBInst{\fragSet{\fragsIn{\query^{i}(\annoDB)} \uniondb \fragsIn{\ieval{\query^{i}}{\dbranges,\annoDeltaDB}{\statedata}}}})) \tag{$\query^{i}$ holds fragments correctness}\\
  = & \selection_{\theta}(\query^{i}(\updDBInst{\fragsIn{\query^{i}(\annoDB)} \uniondb \fragsIn{\ieval{\query^{i}}{\dbranges,\annoDeltaDB}{\statedata}}})) \tag{$\dbInst{\fragSet{\fragsIn{\cdot}}} = \dbInst{\fragsIn{\cdot}}$}\\
  = & \selection_{\theta}(\query^{i}(\updDBInst{\fragsIn{\query^{i}(\annoDB) \uniondb \ieval{\query^{i}}{\dbranges,\annoDeltaDB}{\statedata}}})) \tag{\Cref{lemma:frags-udot-annodb-annodeldb}}\\
  = & \selection_{\theta}(\tupsIn{\query^{i}(\annoupdDB_{\fragsIn{\query^{i}(\annoDB) \uniondb \ieval{\query^{i}}{\dbranges,\annoDeltaDB}{\statedata}}}})) \tag{$\tupsIn{\annoDB} = \db$}\\
  = & \tupsIn{\selection_{\theta}(\query^{i}(\annoupdDB_{\fragsIn{\query^{i}(\annoDB) \uniondb \ieval{\query^{i}}{\dbranges,\annoDeltaDB}{\statedata}}}))} \tag{\Cref{lemma:sel-Tsel-selT}}\\
\end{align*}

From above, we can get that
\begin{align*}
  & \query^{i + 1}(\updDB) = \tupsIn{\selection_{\theta}(\query^{i}(\annoupdDB_{\fragsIn{\query^{i}(\annoDB) \uniondb \ieval{\query^{i}}{\dbranges,\annoDeltaDB}{\statedata}}}))}
\end{align*}

We now focus on an annotated tuples. For an annotated tuple
$\pair{\tup}{\provSketch} \in$ $\query^{i}(\annoDB{'}_{\fragsIn{\query^{i}(\annoDB) \uniondb \ieval{\query^{i}}{\dbranges,\annoDeltaDB}{\statedata}}})$
, we can get that:
1. $\tup \in \query^{i}(\updDBInst{\fragsIn{\query^{i}(\annoDB) \uniondb \ieval{\query^{i}}{\dbranges,\annoDeltaDB}{\statedata}}})$
, 2. $\provSketch \in \fragsIn{\query^{i}(\annoDB) \uniondb \ieval{\query^{i}}{\dbranges,\annoDeltaDB}{\statedata}}$
and 3. tuple $\tup$ can be obtain by $\query^{i}(\updDBInst{\provSketch})$.
For this annotated tuple, if tuple $\tup$ satisfies the selection
condition,$\tup \models \theta$, of
an selection operator on top of $\query^{i}$, then
$\tup \in \selection_{\theta}(\query^{i}(\updDBInst{\fragsIn{\query^{i}(\annoDB) \uniondb \ieval{\query^{i}}{\dbranges,\annoDeltaDB}{\statedata}}}))$
, $\provSketch \in \fragsIn{\query^{i}(\annoDB) \uniondb \ieval{\query^{i}}{\dbranges,\annoDeltaDB}{\statedata}}$
, and tuple $\tup$ can be obtain by $\selection_{\theta}(\query^{i}(\updDBInst{\provSketch}))$
. Now the tuple $\tup$ is in the result of selection operator, and it can be obtained by
$\selection_{\theta}(\query^{i}(\updDBInst{\provSketch}))$. Every tuple $\tup$
associates with its sketch $\provSketch$, and according to the selection semantics rule, $\provSketch$ is in
$\fragsIn{\query^{i + 1}(\annoDB) \uniondb \ieval{\query^{i+1}}{\dbranges,\annoDeltaDB}{\statedata}}$
, which is
$\pair{\tup}{\provSketch} \in$ $\selection_{\theta}(\query^{i}(\annoDB{'}_{\fragsIn{\query^{i + 1}(\annoDB) \uniondb \ieval{\query^{i+1}}{\dbranges,\annoDeltaDB}{\statedata}}}))$.
Then:
\begin{align*}
                  & \pair{\tup}{\provSketch} \in \selection_{\theta}(\query^{i}(\annoupdDB_{\fragsIn{\query^{i}(\annoDB) \uniondb \ieval{\query^{i}}{\dbranges,\annoDeltaDB}{\statedata}}}))\\
  \Leftrightarrow & \pair{\tup}{\provSketch} \in \query^{i}(\annoupdDB_{\fragsIn{\query^{i}(\annoDB) \uniondb \ieval{\query^{i}}{\dbranges,\annoDeltaDB}{\statedata}}}) \wedge \tup \models \theta \tag{$\selection_{\theta}$ definition}\\
  \Leftrightarrow & \pair{\tup}{\provSketch} \in \selection_{\theta}(\query^{i}(\annoupdDB_{\fragsIn{\query^{i + 1}(\annoDB) \uniondb \ieval{\query^{i + 1}}{\dbranges,\annoDeltaDB}{\statedata}}})) \tag{$\incremaintain(\selection_{\theta})$}\\
  \Leftrightarrow & \pair{\tup}{\provSketch} \in \query^{i + 1}(\annoupdDB_{\fragsIn{\query^{i + 1}(\annoDB) \uniondb \ieval{\query^{i + 1}}{\dbranges,\annoDeltaDB}{\statedata}}}) \tag{$\query^{i+1} = \selection_{\theta}$}
\end{align*}
Thus, for all annotated tuples in
$\selection_{\theta}(\query^{i}(\annoupdDB_{\fragsIn{\query^{i}(\annoDB) \uniondb \ieval{\query^{i}}{\dbranges,\annoDeltaDB}{\statedata}}}))$,
they are in $\query^{i + 1}(\annoupdDB_{\fragsIn{\query^{i + 1}(\annoDB) \uniondb \ieval{\query^{i + 1}}{\dbranges,\annoDeltaDB}{\statedata}}})$.
Therefore,
\begin{align*}
    & \query^{i + 1}(\annoupdDB_{\fragsIn{\query^{i + 1}(\annoDB) \uniondb \ieval{\query^{i + 1}}{\dbranges,\annoDeltaDB}{\statedata}}})\\
  = & \selection_{\theta}(\query^{i}(\annoupdDB_{\fragsIn{\query^{i}(\annoDB) \uniondb \ieval{\query^{i}}{\dbranges,\annoDeltaDB}{\statedata}}}))\\
  = & \tupsIn{\selection_{\theta}(\query^{i}(\annoupdDB_{\fragsIn{\query^{i}(\annoDB) \uniondb \ieval{\query^{i}}{\dbranges,\annoDeltaDB}{\statedata}}}))} \\
  = & \query^{i + 1}(\updDB)
\end{align*}
% \[
   % & \pair{\tup}{\provSketch} \in \query^{i + 1}(\annoupdDB_{\fragsIn{\query^{i + 1}(\annoDB) \uniondb \ieval{\query^{i + 1}}{\dbranges,\annoDeltaDB}{\statedata}}}) \tag{$\query^{i+1} = \selection_{\theta}$}\\
% \]

%%%%%%%%%%%%%%%%%%%%%%%%%%%%%%%%%%%%%%%%%%%%%%%%%%%%%%%%%%%%%%%%%%%%%%%%%%%%%%%%
%% QI{i + 1} is an ROJECTION
\subsubsection{Projection}

Suppose the operator at level $i + 1$ is an projection operator. Then we have
the following:
\[
  \query^{i + 1}(\db)  = \projection_{\ATT} (\query^{i} (\db))
\]

For projection operator, the annotated tuples' relation between $\query^{i+ 1}$
and $\query^{i}$ is like $\pair{\tup_{1}}{\provSketch_{1}}$ and
$\pair{s_{1}}{\provSketch_{s_{1}}}$ in
\Cref{fig:anno_tuples_rel_between_two_levels} such that:
\[
 \pair{\tup_{1}}{\provSketch_{\tup_{1}}} = \projection_{\ATT}(\pair{s_{1}}{\provSketch_{s_{1}}})
\]

Before we demonstrate the tuple correctness, we first show two properties hold
for projection operator
%%%%%%%%%%%%%%%%%%%%%%%%%%%%%%%%%%%%%%%%
%% two properties and proof

%%%%%%%%%%%%%%%%%%%%
%% property 1
\begin{lem}\label{lemma:proj-Tproj-projT}
  The following properties hold:
  \begin{align*}
    &\projection_{\ATT}(\tupsIn{\annoDB}) = \tupsIn{\projection_{\ATT}(\annoDB)}
  \end{align*}
\end{lem}

\begin{proof}
  For $\projection_{\ATT}(\tupsIn{\annoDB})$:
  \begin{align*}
      & \projection_{\ATT}(\tupsIn{\annoDB})\\
    = & \projection_{\ATT}(\db) \tag{$\tupsIn{\annoDB} = \db$}\\
    = & \bag{\tup \mid \tup' \in \db \wedge \tup'.\ATT = \tup} \tag{$\projection_{\ATT}$}
  \end{align*}

  For $\tupsIn{\projection_{\ATT}(\annoDB)}$:
  \begin{align*}
      & \tupsIn{\projection_{\ATT}(\annoDB})\\
    = & \tupsIn{\bag{\pair{\tup}{\provSketch} \mid \pair{\tup'}{\provSketch} \in \annoDB \wedge \tup'.\ATT = \tup}} \tag{$\projection_{\ATT}$}\\
    = & \bag{\tup \mid \tup' \in \db \wedge \tup'.\ATT =\tup}  \tag{$\tupsIn{\annoDB} = \db$}
  \end{align*}

  Therefore, $\projection_{\ATT}(\tupsIn{\annoDB}) = \tupsIn{\projection_{\ATT}(\annoDB)}$:
\end{proof}

%%%%%%%%%%%%%%%%%%%%
%% property 2
\begin{lem}\label{lemma:proj-annodb-udot-annodel}
  The following properties hold:
  \begin{align*}
    \projection_{\ATT}(\tupsIn{\annoDB \uniondb \annoDeltaDB}) = \projection_{\ATT}(\tupsIn{\annoDB}) \uniondb \projection_{\ATT}(\tupsIn{\annoDeltaDB})
  \end{align*}
\end{lem}

\begin{proof}
  For $ \projection_{\ATT}(\tupsIn{\annoDB \uniondb \annoDeltaDB})$:
  \begin{align*}
      & \projection_{\ATT}(\tupsIn{\annoDB \uniondb \annoDeltaDB}) \\
    = & \projection_{\ATT}(\tupsIn{\annoDB} \uniondb \tupsIn{\annoDeltaDB}) \tag{\cref{lemma:tupsin-udot-annodb-annodeltadb}}\\
    = & \projection_{\ATT}(\db \uniondb \deltadb) \tag{$\tupsIn{\annoDB} = \db$}\\
    = & \projection_{\ATT}(\db \difference \deletedelta \db \insertdelta \db) \tag{$\db \uniondb \deltadb = \db \difference \deletedelta \db \union \insertdelta \db$}\\
    = & \projection_{\ATT}(\db) \difference \projection_{\ATT}(\deletedelta \db) \union \projection_{\ATT}(\insertdelta \db) \tag{$\projection_{\ATT}$ and $\union, \difference$}
  \end{align*}
  For $\projection_{\ATT}(\tupsIn{\annoDB}) \uniondb \projection_{\ATT}(\tupsIn{\annoDeltaDB})$
  \begin{align*}
      & \projection_{\ATT}(\tupsIn{\annoDB}) \uniondb \projection_{\ATT}(\tupsIn{\annoDeltaDB})\\
    = & \tupsIn{\projection_{\ATT}(\annoDB)} \uniondb \tupsIn{\projection_{\ATT}(\annoDeltaDB)} \tag{\cref{lemma:sel-Tsel-selT}}\\
    = & \projection_{\ATT}(\annoDB) \uniondb \projection_{\ATT}(\deltadb) \tag{\tupsIn{\annoDB} = \db}\\
    = & \projection_{\ATT}(\annoDB) \uniondb \projection_{\ATT}(\deletedelta \db \union \insertdelta \db) \\
    = & \projection_{\ATT}(\annoDB) \uniondb \left(\projection_{\ATT}(\deletedelta \db) \union \projection_{\ATT}(\insertdelta \db)\right) \\
    = & \projection_{\ATT}(\annoDB) \difference \projection_{\ATT}(\deletedelta \db) \union \projection_{\ATT}(\insertdelta \db) \tag{$\uniondb$}\\
  \end{align*}
  Therefore, $\projection_{\ATT}(\tupsIn{\annoDB \uniondb \annoDeltaDB}) = \projection_{\ATT}(\tupsIn{\annoDB}) \uniondb \projection_{\ATT}(\tupsIn{\annoDeltaDB})$
\end{proof}

%%%%%%%%%%%%%%%%%%%%%%%%%%%%%%%%%%%%%%%%
%% projection tuple correctness;
\parttitle{\Tcorrect}
\begin{proof}
\begin{align*}
    & \query^{i + 1}(\updDB) \\
  = & \projection_{\ATT}(\query^{i}(\updDB))\tag{$\query^{i + 1} = \projection_{\ATT}$}\\
  = & \projection_{\ATT}\left(\tupsIn{\query^{i}(\annoDB) \uniondb \ieval{\query^{i}}{\dbranges,\annoDeltaDB}{\statedata} }\right) \tag{$Q^{i}$ holds the property}\\
  = & \projection_{\ATT}\left(\tupsIn{\query^{i}(\annoDB)} \uniondb \tupsIn{\ieval{\query^{i}}{\dbranges,\annoDeltaDB}{\statedata}}\right)\tag{\cref{lemma:tupsin-udot-annodb-annodeltadb}}\\
  = & \projection_{\ATT}\left(\tupsIn{\query^{i}(\annoDB)}\right) \uniondb \projection_{\ATT}\left(\tupsIn{\ieval{\query^{i}}{\dbranges,\annoDeltaDB}{\statedata}}\right) \tag{\cref{lemma:proj-annodb-udot-annodel}}\\
  = & \tupsIn{\projection_{\ATT}(\query^{i}(\annoDB)} \uniondb \tupsIn{\projection_{\ATT}(\ieval{\query^{i}}{\dbranges,\annoDeltaDB}{\statedata}}\tag{\cref{lemma:proj-Tproj-projT}} \\
  = & \tupsIn{\query^{i + 1}( \annoDB )} \uniondb \tupsIn{\projection_{\ATT}(\ieval{\query^{i}}{\dbranges,\annoDeltaDB}{\statedata}}\tag{$\projection_{\ATT}(\query^{i})=\query^{i + 1}$}\\
  = & \tupsIn{ \query^{i +1}(\annoDB) } \uniondb \bag{\tup \mid \tup' \in \tupsIn{\ieval{\query^{i}}{\dbranges,\annoDeltaDB}{\statedata}} \wedge \tup'{.}{\ATT} = \tup}\tag{$\projection_{\ATT}$}\\
  = & \tupsIn{ \query^{i +1}(\annoDB) } \uniondb \tupsIn{\ieval{\projection_{\ATT}(\query^{i})}{\dbranges,\annoDeltaDB}{\statedata}}\tag{$\incremaintain(\projection_{\ATT})$ rule}\\
  = & \tupsIn{ \query^{i +1}(\annoDB) } \uniondb \tupsIn{ \ieval{\query^{i + 1}}{\dbranges,\annoDeltaDB}{\statedata} }\tag{$\projection_{\ATT}(\query^{i})=\query^{i + 1}(\query^{i})$}\\
  = & \tupsIn{ \query^{i +1}(\annoDB) \uniondb \ieval{\query^{i + 1}}{\dbranges,\annoDeltaDB}{\statedata}} \tag{\cref{lemma:tupsin-udot-annodb-annodeltadb}}
\end{align*}
\end{proof}

\parttitle{\Fcorrect}
We have the assumption that:
\[
\query^{i}(\updDB) = \query^{i}(\updDBInst{\fragSet{\fragsIn{\query^{1}(\annoDB)} \uniondb \fragsIn{\ieval{\query^{i}}{\dbranges,\annoDeltaDB}{\statedata}}}})
\]

Then for a projection operator above $\query^{i}$, we have the following proof:
\begin{proof}
\begin{align*}
    & \query^{i + 1}(\updDB)\\
  = & \projection_{\ATT}(\query^{i}(\updDB)) \tag{$\query^{i + 1} = \projection_{\ATT}(\query^{i})$}\\
  % = & \projection_{\ATT}(\query^{i}(\updDB)) \tag{$\tupsIn{\annoDB} = \db$}\\
  = & \projection_{\ATT}(\query^{i}(\updDBInst{\fragSet{\fragsIn{\query^{i}(\annoDB)} \uniondb \fragsIn{\ieval{\query^{i}}{\dbranges,\annoDeltaDB}{\statedata}}}})) \tag{$\query^{i}$ holds fragments correctness}\\
  = & \projection_{\ATT}(\query^{i}(\updDBInst{\fragsIn{\query^{i}(\annoDB)} \uniondb \fragsIn{\ieval{\query^{i}}{\dbranges,\annoDeltaDB}{\statedata}}})) \tag{$\dbInst{\fragSet{\fragsIn{\cdot}}} = \dbInst{\fragsIn{\cdot}}$}\\
  = & \projection_{\ATT}(\query^{i}(\updDBInst{\fragsIn{\query^{i}(\annoDB) \uniondb \ieval{\query^{i}}{\dbranges,\annoDeltaDB}{\statedata}}})) \tag{\Cref{lemma:frags-udot-annodb-annodeldb}}\\
  = & \projection_{\ATT}(\tupsIn{\query^{i}(\annoupdDB_{\fragsIn{\query^{i}(\annoDB) \uniondb \ieval{\query^{i}}{\dbranges,\annoDeltaDB}{\statedata}}}})) \tag{$\tupsIn{\annoDB} = \db$}\\
  = & \tupsIn{\projection_{\ATT}(\query^{i}(\annoupdDB_{\fragsIn{\query^{i}(\annoDB) \uniondb \ieval{\query^{i}}{\dbranges,\annoDeltaDB}{\statedata}}}))} \tag{\Cref{lemma:sel-Tsel-selT}}\\
\end{align*}

From above, we can get that
\begin{align*}
  & \query^{i + 1}(\updDB) = \tupsIn{\projection_{\ATT}(\query^{i}(\annoupdDB_{\fragsIn{\query^{i}(\annoDB) \uniondb \ieval{\query^{i}}{\dbranges,\annoDeltaDB}{\statedata}}}))}
\end{align*}

For an annotated tuple
$\pair{\tup}{\provSketch} \in$ $\query^{i}(\annoDB{'}_{\fragsIn{\query^{i}(\annoDB) \uniondb \ieval{\query^{i}}{\dbranges,\annoDeltaDB}{\statedata}}})$
, the following holds:
1. $\tup \in \query^{i}(\updDBInst{\fragsIn{\query^{i}(\annoDB) \uniondb \ieval{\query^{i}}{\dbranges,\annoDeltaDB}{\statedata}}})$
, 2. $\provSketch \in \fragsIn{\query^{i}(\annoDB) \uniondb \ieval{\query^{i}}{\dbranges,\annoDeltaDB}{\statedata}}$
and 3. tuple $\tup$ can be obtain by $\query^{i}(\updDBInst{\provSketch})$.
For this annotated tuple, if expressions in $\ATT$ are projected from tuple
$\tup$ of an projection operator on top of $\query^{i}$, then
\[
\tup.\ATT \in \projection_{\ATT}(\query^{i}(\updDBInst{\fragsIn{\query^{i}(\annoDB) \uniondb \ieval{\query^{i}}{\dbranges,\annoDeltaDB}{\statedata}}}))
\]
, $\provSketch \in \fragsIn{\query^{i}(\annoDB) \uniondb \ieval{\query^{i}}{\dbranges,\annoDeltaDB}{\statedata}}$
, and  $\tup.\ATT$ can be obtain by $\projection_{\ATT}(\query^{i}(\updDBInst{\provSketch}))$
. Now the  $\tup.\ATT$ is in the result of projection operator, and it can be obtained by
$\projection_{\ATT}(\query^{i}(\updDBInst{\provSketch}))$. Every $\tup.\ATT$
associates with its sketch $\provSketch$, and according to the projection semantics rule, $\provSketch$ is in
$\fragsIn{\query^{i + 1}(\annoDB) \uniondb \ieval{\query^{i+1}}{\dbranges,\annoDeltaDB}{\statedata}}$
, which is
\[
\pair{\tup}{\provSketch} \in\projection_{\ATT}(\query^{i}(\annoDB{'}_{\fragsIn{\query^{i + 1}(\annoDB) \uniondb \ieval{\query^{i+1}}{\dbranges,\annoDeltaDB}{\statedata}}}))
\]
Then:
\begin{align*}
                  & \pair{\tup}{\provSketch} \in \projection_{\ATT}(\query^{i}(\annoupdDB_{\fragsIn{\query^{i}(\annoDB) \uniondb \ieval{\query^{i}}{\dbranges,\annoDeltaDB}{\statedata}}}))\\
  \Leftrightarrow & \pair{\tup}{\provSketch} \in \tup' \in \query^{i}(\annoupdDB_{\fragsIn{\query^{i}(\annoDB) \uniondb \ieval{\query^{i}}{\dbranges,\annoDeltaDB}{\statedata}}}) \wedge \tup'.\ATT=\tup \tag{$\projection_{\ATT}$ definition}\\
  \Leftrightarrow & \pair{\tup}{\provSketch} \in \projection_{\ATT}(\query^{i}(\annoupdDB_{\fragsIn{\query^{i + 1}(\annoDB) \uniondb \ieval{\query^{i + 1}}{\dbranges,\annoDeltaDB}{\statedata}}})) \tag{$\incremaintain(\projection_{\ATT})$}\\
  \Leftrightarrow & \pair{\tup}{\provSketch} \in \query^{i + 1}(\annoupdDB_{\fragsIn{\query^{i + 1}(\annoDB) \uniondb \ieval{\query^{i + 1}}{\dbranges,\annoDeltaDB}{\statedata}}}) \tag{$\query^{i+1} = \projection_{\ATT}(\query^{i})$}
\end{align*}
Thus, for all annotated tuples in
$\projection_{\ATT}(\query^{i}(\annoupdDB_{\fragsIn{\query^{i}(\annoDB) \uniondb \ieval{\query^{i}}{\dbranges,\annoDeltaDB}{\statedata}}}))$,
they are in $\query^{i + 1}(\annoupdDB_{\fragsIn{\query^{i + 1}(\annoDB) \uniondb \ieval{\query^{i + 1}}{\dbranges,\annoDeltaDB}{\statedata}}})$.

Therefore,
\begin{align*}
    & \query^{i + 1}(\annoupdDB_{\fragsIn{\query^{i + 1}(\annoDB) \uniondb \ieval{\query^{i + 1}}{\dbranges,\annoDeltaDB}{\statedata}}})\\
  = & \projection_{\ATT}(\query^{i}(\annoupdDB_{\fragsIn{\query^{i}(\annoDB) \uniondb \ieval{\query^{i}}{\dbranges,\annoDeltaDB}{\statedata}}}))\\
  = & \tupsIn{\projection_{\ATT}(\query^{i}(\annoupdDB_{\fragsIn{\query^{i}(\annoDB) \uniondb \ieval{\query^{i}}{\dbranges,\annoDeltaDB}{\statedata}}}))} \\
  = & \query^{i + 1}(\updDB)
\end{align*}
\end{proof}
%%%%%%%%%%%%%%%%%%%%%%%%%%%%%%%%%%%%%%%%%%%%%%%%%%%%%%%%%%%%%%%%%%%%%%%%%%%%%%%%
%% Q^{i + 1} is a Cross product%%%%%%%%%%%%%%%%%%%%%%%%%%%%%%%%%%%%%%%%%%%%%%%%%%%%%%%%%%%%%%%%%%%%%%%%%%%%%%%%%
%% tikz figs show how Q(i + 1) relates Q(i))
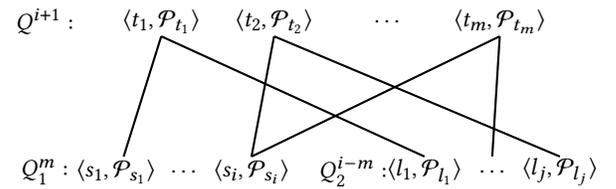
\begin{figure}[h!]
  \centering
  \begin{tikzpicture}[
    node distance=1cm and 2cm,
    every node/.style={ minimum size=1cm},
    every path/.style={-, thick}
  ]
    %% Q^{i + 1}
    \node[draw=none] (top) at (-2, 2) {$\query^{i + 1}: \;$};
    % Top row nodes
    \node (t1) at (-0.5, 2) {$\pair{\tup_{1}}{\provSketch_{\tup_{1}}}$};
    \node (t2) at (1, 2) {$\pair{\tup_{2}}{\provSketch_{\tup_{2}}}$};
    \node (t3) at (2.5, 2) {$\ldots$};
    \node (t4) at (4, 2) {$\pair{\tup_{m}}{\provSketch_{\tup_{m}}}$};
    %% define coordinates for t1,2,4
    \coordinate (pt1) at (-0.5, 1.8); % Define point A at (1, 2)
    \coordinate (pt2) at (1, 1.8); % Define point A at (1, 2)
    %% \coordinate (pt4) dots, not needed
    \coordinate (pt4) at (4, 1.8); % Define point A at (1, 2)

    %% Q^{i}
    \node[draw=none] (bottom) at (-2, 0) {$\query_{1}^{m}: \;$};
    % \node[draw=none] (bottom) at (3.2, 0){$\pair{l_{n -1}}{\provSketch_{l_{n}}}$};
    % Bottom row nodes
    \node (s1) at (-1.1, 0) {$\pair{s_{1}}{\provSketch_{s_{1}}}$};
    \node (s2) at (-0.2, 0) {$\ldots$};
    \node (s3) at (0.7, 0) {$\pair{s_{i}}{\provSketch_{s_{i}}}$};
    \node (l0) at (2.1, 0) {$\query_{2}^{i-m}: \;$};
    \node (l1) at (3, 0) {$\pair{l_{1}}{\provSketch_{l_{1}}}$};
    \node (l2) at (3.9, 0) {$\ldots$};
    \node (l3) at (4.8, 0) {$\pair{l_{j}}{\provSketch_{l_{j}}}$};
    %% define coordinates for s1,2,45
   \coordinate (ps1) at (-1, 0.2); % Define point A at (1, 2)
   \coordinate (psi) at (0.7, 0.2); % Define point A at (1, 2)
   \coordinate (pl1) at (3, 0.2); % Define point A at (1, 2)
   \coordinate (pl2) at (3.9, 0.2); % Define point A at (1, 2)
   \coordinate (plj) at (4.8, 0.2); % Define point A at (1, 2)
%%    \coordinate (ps3) at (2, 0.2);
%%    \coordinate (ps4) at (3.5, 0.2); % Define point A at (1, 2)
%%    \coordinate (ps5) at (5.3, 0.2); % Define point A at (1, 2)
%%
   \draw (pt1) -- (ps1);
   \draw (pt1) -- (pl1);
   \draw (pt2) -- (psi);
   \draw (pt2) -- (plj);
   \draw (pt4) -- (psi);
   \draw (pt4) -- (pl2);
   % \draw (pt2) -- (ps2);
%%    \draw (pt2) -- (ps4);
%%    \draw (pt4) -- (ps2);
%%    \draw (pt4) -- (ps3);
%%    \draw (pt4) -- (ps4);
%%    \draw (pt4) -- (ps5);
  \end{tikzpicture}
  \caption{Annotated tuples' relation between $\query^{i}$, $\query_{1}^{m}$ and $\query_{2}^{i-m}$}
  \label{fig:anno_tuples_rel_crossprod}
\end{figure}
%%%%%%%%%%%%%%%%%%%%%%%%%%%%%%%%%%%%%%%%%%%%%%%%%%%%%%%%%%%%%%%%%%%%%%%%%%%%%%%%

\subsubsection{Cross Product}
Suppose the operator at level $i + 1$ is a cross product (join) operator, then
we have the following:
\[
  \query^{i + 1}(\db) = \query^{m}_{1}(\db) \crossprod \query^{i - m}_{2}(\db)
\]

For cross product, the annotated tuples' relation between $\query^{i+ 1}$
and $\query_{1}^{i}$ and $\query_{2}^{i - m}$ is like $\pair{\tup_{2}}{\provSketch_{2}}$ and
$\pair{s_{i}}{\provSketch_{s_{i}}}$ and $\pair{l_{j}}{\provSketch_{l_{j}}}$ in
\Cref{fig:anno_tuples_rel_crossprod}, such that:
\[
  \pair{\tup_{2}}{\provSketch_{2}} = \pair{(s_{i} \concat l_{j})}{\bag{\provSketch_{s_{i}}, \provSketch_{l_{j}}}}
\]
Each annotated tuple in the result of $\query^{i + 1}$ is the product of two
annotated tuples, each one from one side.

%%%%%%%%%%%%%%%%%%%%%%%%%%%%%%%%%%%%%%%%
%% cross prod: tuple correctness
\parttitle{\Tcorrect}
Before we demonstrate the tuple correctness, we first show the property hold
for cross product operator
%%%%%%%%%%%%%%%%%%%%%%%%%%%%%%%%%%%%%%%%
%% two properties and proof

%%%%%%%%%%%%%%%%%%%%
%% property 1
\begin{lem}\label{lemma:cross-Tcross-crossT}
  The following properties hold:
  \begin{align*}
    &\tupsIn{ \annoDB_{1} \crossprod \annoDB_{2} } = \tupsIn{\annoDB_{1}} \crossprod \tupsIn{\annoDB_{2}}
  \end{align*}
\end{lem}

\begin{proof}
  For $\tupsIn{ \annoDB_{1} \crossprod \annoDB_{2} } $:
  \begin{align*}
      & \tupsIn{ \annoDB_{1} \crossprod \annoDB_{2} }\\
    = & \tupsIn{ \bag{(\pair{\tup}{\provSketch_{\tup}} \concat \pair{s}{\provSketch_{s}})^{m * n} \mid \pair{\tup}{\provSketch_{\tup}}^{m} \in \annoDB_{1} \wedge \pair{s}{\provSketch_{s}}^{n} \in \annoDB_{2}}}\\
    = & \bag{(\tup \concat s)^{m * n} \mid \tup^{m} \in \db_{1} \wedge s^{n} \in \db_{2}}\\
    = & \db_{1} \crossprod \db_{2}
  \end{align*}

  For $\tupsIn{\annoDB_{1}} \crossprod \tupsIn{\annoDB_{2}}$:
  \begin{align*}
      & \tupsIn{\annoDB_{1}} \crossprod \tupsIn{\annoDB_{2}}\\
    = & \db_{1} \crossprod \db_{2}\\
  \end{align*}

  Therefore, $\tupsIn{ \annoDB_{1} \crossprod \annoDB_{2} } = \tupsIn{\annoDB_{1}} \crossprod \tupsIn{\annoDB_{2}}$
\end{proof}
%%%%%%%%%%%%%%%%%%%%%%%%%%%%%%%%%%%%%%%%

Thus the following is the \tcorrect:

\begin{proof}
\begin{align*}
    & \query^{i + 1}(\updDB) \\
  = & \query^{m}_{1}(\updDB) \crossprod \query^{i - m}_{2}(\updDB) \tag{$\query^{i + 1}(\db) = \query^{m}_{1}(\db) \crossprod \query^{i - m}_{2}(\db)$}\\
  = & \tupsIn{\query^{m}_{1}(\annoDB) \uniondb \ieval{\query^{m}_{1}}{\dbranges,\annoDeltaDB}{\statedata} } \\
    & \hspace{3mm}\crossprod \tupsIn{\query^{i - m}_{2}(\annoDB) \uniondb \ieval{\query^{i - m}_{2}}{\dbranges,\annoDeltaDB}{\statedata}} \tag{property hold for $\query^{m}_{1}(\updDB)$ and $\query^{i - m}_{2}(\updDB)$}\\
  = & (\tupsIn{\query^{m}_{1}(\annoDB)} \uniondb \tupsIn{\ieval{\query^{m}_{1}}{\dbranges,\annoDeltaDB}{\statedata} } )\\
    & \hspace{3mm}\crossprod (\tupsIn{\query^{i - m}_{2}(\annoDB)} \uniondb \tupsIn{\ieval{\query^{i - m}_{2}}{\dbranges,\annoDeltaDB}{\statedata}}) \tag{\cref{lemma:tupsin-udot-annodb-annodeltadb}}\\
  = & (\tupsIn{\query^{m}_{1}(\annoDB)} \crossprod \tupsIn{\query^{i - m}_{2}(\annoDB)} )\\
    & \hspace{3mm}\uniondb (\tupsIn{\query^{m}_{1}(\annoDB)} \crossprod \tupsIn{\ieval{\query^{i - m}_{2}}{\dbranges,\annoDeltaDB}{\statedata}}) \\
    & \hspace{3mm}\uniondb (\tupsIn{\ieval{\query^{m}_{1}}{\dbranges,\annoDeltaDB}{\statedata} \crossprod \tupsIn{\query^{i - m}_{2}(\annoDB)}}) \\
    & \hspace{3mm}\uniondb(\tupsIn{\ieval{\query^{m}_{1}}{\dbranges,\annoDeltaDB}{\statedata}} \crossprod \tupsIn{\ieval{\query^{i-m}_{2}}{\dbranges,\annoDeltaDB}{\statedata}}) \tag{\cref{lemma:cross-Tcross-crossT}}\\
  = & \tupsIn{\query^{m}_{1}(\annoDB) \crossprod \query^{i - m}_{2}(\annoDB)}\\
    & \hspace{3mm}\uniondb \tupsIn{\query^{m}_{1}(\annoDB) \crossprod \ieval{\query^{i - m}_{2}}{\dbranges,\annoDeltaDB}{\statedata}} \\
    & \hspace{3mm}\uniondb \tupsIn{\ieval{\query^{m}_{1}}{\dbranges,\annoDeltaDB}{\statedata} \crossprod \query^{i - m}_{2}(\annoDB)} \\
    & \hspace{3mm}\uniondb \tupsIn{\ieval{\query^{m}_{1}}{\dbranges,\annoDeltaDB}{\statedata} \crossprod \ieval{\query^{i-m}_{2}}{\dbranges,\annoDeltaDB}{\statedata}}\tag{\cref{lemma:tupsin-udot-annodb-annodeltadb}} \\
  = & \tupsIn{\query^{m}_{1}(\annoDB) \crossprod \query^{i - m}_{2}(\annoDB)}\\
    & \hspace{3mm}\uniondb \mathbb{T}\left( \query^{m}_{1}(\annoDB) \crossprod \ieval{\query^{i - m}_{2}}{\dbranges,\annoDeltaDB}{\statedata} \right. \\
    & \hspace{9mm}\union \ieval{\query^{m}_{1}}{\dbranges,\annoDeltaDB}{\statedata} \crossprod \query^{i - m}_{2}(\annoDB) \\
    & \hspace{9mm}\union \left. \ieval{\query^{m}_{1}}{\dbranges,\annoDeltaDB}{\statedata} \crossprod \ieval{\query^{i-m}_{2}}{\dbranges,\annoDeltaDB}{\statedata} \right)  \\
  = & \tupsIn{\query^{i + 1}(\annoDB)} \uniondb \tupsIn{\ieval{\query^{i + 1}}{\dbranges,\annoDeltaDB}{\statedata}} \tag{$\incremaintain(\crossprod)$ rule}\\
  = & \tupsIn{\query^{i + 1}(\annoDB) \uniondb \ieval{\query^{i + 1}}{\dbranges,\annoDeltaDB}{\statedata}} \tag{\cref{lemma:tupsin-udot-annodb-annodeltadb}}\\
\end{align*}
\end{proof}
%%%%%%%%%%%%%%%%%%%%%%%%%%%%%%%%%%%%%%%%
%% cross prod: frag correctness
\parttitle{\Fcorrect}
From the \tcorrect, we get that
\begin{align*}
    & \query^{i + 1}(\updDB) \\
  = & \query^{m}_{1}(\updDB) \crossprod \query^{i - m}_{2}(\updDB) \tag{$\query^{i + 1}= \query^{m}_{1} \crossprod \query^{i - m}_{2}$}\\
  = & \query^{m}_{1}(\updDBInst{\fragSet{\fragsIn{ \query_{1}^{m}(\annoupdDB) } \uniondb \fragsIn{\ieval{\query^{m}_{1}}{\dbranges, \annoDeltaDB}{\statedata}}}}) \\
    & \hspace{3mm}\crossprod \query^{i - m}_{2}(\updDBInst{\fragSet{\fragsIn{\query_{2}^{i-m}(\annoupdDB) } \uniondb \fragsIn{\ieval{\query^{i-m}_{2}}{\dbranges, \annoDeltaDB}{\statedata}}}})\tag{$\query^{m}_{1}$ and $\query^{i - m}_{2}$ hold the fragments correctness}\\
  = & \query^{m}_{1}(\updDBInst{\fragsIn{ \query_{1}^{m}(\annoupdDB) } \uniondb \fragsIn{\ieval{\query^{m}_{1}}{\dbranges, \annoDeltaDB}{\statedata}}}) \\
    & \hspace{3mm}\crossprod \query^{i - m}_{2}(\updDBInst{\fragsIn{\query_{2}^{i-m}(\annoupdDB)}  \uniondb \fragsIn{\ieval{\query^{i-m}_{2}}{\dbranges, \annoDeltaDB}{\statedata}}})\tag{$\dbInst{\fragSet{\fragsIn{\cdot}} = \dbInst{\fragsIn{\cdot}}}$ }\\
  = & \query^{m}_{1}(\updDBInst{\fragsIn{ \query_{1}^{m}(\annoupdDB)  \uniondb \ieval{\query^{m}_{1}}{\dbranges, \annoDeltaDB}{\statedata}}}) \\
    & \hspace{3mm}\crossprod \query^{i - m}_{2}(\updDBInst{\fragsIn{\query_{2}^{i-m}(\annoupdDB)  \uniondb \ieval{\query^{i-m}_{2}}{\dbranges, \annoDeltaDB}{\statedata}}})\tag{\Cref{lemma:frags-udot-annodb-annodeldb}} \\
  = & \tupsIn{\query^{m}_{1}(\annoupdDB_{\fragsIn{ \query_{1}^{m}(\annoupdDB)  \uniondb \ieval{\query^{m}_{1}}{\dbranges, \annoDeltaDB}{\statedata}}})} \\
    & \hspace{3mm}\crossprod \tupsIn{\query^{i - m}_{2}(\annoupdDB_{\fragsIn{\query_{2}^{i-m}(\annoupdDB)  \uniondb \ieval{\query^{i-m}_{2}}{\dbranges, \annoDeltaDB}{\statedata}}})}\tag{\tupsIn{\annoDB = \db}} \\
\end{align*}

From above, we know that:
\[
    \begin{cases}
      \query^{m}_{1}(\updDB) =\tupsIn{\query^{m}_{1}(\annoupdDB_{\fragsIn{ \query_{1}^{m}(\annoupdDB) \uniondb \ieval{\query^{m}_{1}}{\dbranges, \annoDeltaDB}{\statedata}}})} \\
      \query^{i-m}_{2}(\updDB) = \tupsIn{\query^{i-m}_{2}(\annoupdDB_{\fragsIn{\query_{2}^{i-m}(\annoupdDB) \uniondb \ieval{\query^{i-m}_{2}}{\dbranges, \annoDeltaDB}{\statedata}}})}
    \end{cases}
\]

We now focus two annotated  tuples $\pair{\tup}{\provSketch_{\tup}}$ and
$\pair{s}{\provSketch_{s}}$ such that:
\[
  \begin{cases}
    & \pair{\tup}{\provSketch_{\tup}} \in \tupsIn{\query^{m}_{1}(\annoupdDB_{\fragsIn{ \query_{1}^{m}(\annoupdDB) \uniondb \ieval{\query^{m}_{1}}{\dbranges, \annoDeltaDB}{\statedata}}})} \\
    & \pair{s}{\provSketch_{s}} \in \tupsIn{\query^{i-m}_{2}(\annoupdDB_{\fragsIn{ \query_{2}^{i-m}(\annoupdDB) \uniondb \ieval{\query^{i-m}_{2}}{\dbranges, \annoDeltaDB}{\statedata}}})}
  \end{cases}
\]

If $\pair{\tup}{\provSketch_{\tup}}$ is a non-delta annotated tuple
and $\pair{s}{\provSketch_{s}}$ is a non-delta annotated from ,
$\pair{(\tup \concat s)}{\bag{\provSketch_{\tup}, \provSketch_{s}}}$ is a
non-delta annotated tuple in
${\query^{i + 1}(\annoupdDB)}$, and $(\tup \concat s) $ can be obtain by $\query^{m}_{1}(\updDB_{
  \provSketch_{\tup} }) \crossprod \query^{i -
  m}_{2}(\updDB_{\provSketch_{s}})$. Thus $\provSketch_{\tup}$and
$\provSketch_{s}$ are in $\fragsIn{\query^{i + 1}(\annoupdDB)}$.
If one of $\pair{\tup}{\provSketch_{\tup}}$ and
$\pair{s}{\provSketch_{s}}$ is a delta annotated tuple or both are annotated
tuples, then $(\tup \concat s)$ in $\Delta \query^{i + 1}(\db)$, and the
fragments $\provSketch_{\tup}$and $\provSketch_{s}$  are in $\Delta \fragsIn{\query^{i
    + 1}(\annoupdDB)}$. Therefore, for any $\pair{\tup}{\provSketch_{\tup}}$ and
$\pair{s}{\provSketch_{s}}$,  $(\tup \concat s)$ in $\Delta \query^{i +
  1}(\updDB)$, and $\provSketch_{\tup}$ and $\provSketch_{s}$ are in
$\fragsIn{\query^{i + 1}(\annoupdDB)} \uniondb \Delta \fragsIn{\query^{i + 1}(\annoupdDB)}$.
Then for all annotated tuples from $\query^{m}_{1}(\annoupdDB)$ and
$\query^{i-m}_{2}(\annoupdDB)$, the cross product result is a bag of annotated tuples that is
the same as $\query^{i + 1}(\annoupdDB)$ and all the tuples of $\query^{i +
  1}(\updDB)$ are the same as $\query^{i+1}(\updDBInst{\fragsIn{\query^{i + 1}(\annoupdDB)} \uniondb \Delta \fragsIn{\query^{i + 1}(\annoupdDB)}})$
. From the incremental semantics $\Delta \fragsIn{ \query^{i + 1}(\annoupdDB) } =
\ieval{\query^{i + 1}}{\dbranges, \annoDeltaDB}{\statedata}$. Then,  $\query^{i +
  1}(\updDB) = \query^{i+1}(\updDBInst{\fragsIn{\query^{i + 1}(\annoupdDB)} \uniondb \Delta \fragsIn{\query^{i + 1}(\annoupdDB)}})$.

% \begin{align*}
    % & \query^{i+1}(\updDBInst{\fragsIn{\query^{i + 1}(\annoupdDB)} \uniondb \ieval{\query^{i + 1}}{\dbranges, \annoDeltaDB}{\statedata}})\\
  % = & \query^{i+1}(\updDBInst{\fragsIn{\query^{i + 1}(\annoupdDB)} \uniondb \Delta \fragsIn{\query^{i + 1}(\annoupdDB)}}) \tag{$\updDB = \db \uniondb \deltadb$}\\
  % = & \tupsIn{\query^{m}_{1}(\annoupdDB) \crossprod \query^{i-m}_{2}(\annoupdDB)}\\

% \end{align*}

%
%%%%%%%%%%%%%%%%%%%%%%%%%%%%%%%%%%%%%%%%%%%%%%%%%%%%%%%%%%%%%%%%%%%%%%%%%%%%%%%%
%% Q^{i + 1} is an Aggregation
\subsubsection{Aggregation}
Suppose the operator at level $i + 1$ is an aggregation function (any one of
\fcsum, \fccnt, \fcavg, \fcmin and \fcmax), Then we have the following:
\[
  \query^{i + 1}(\db)  = \Aggregation{f(\att)}{\grpatts} (\query^{i} (\db))
\]
We have the assumption that for $\query^{i} \in \queryclass{i}$, the \tcorrect
and \fcorrect hold. We will show that these properties still hold for
$\query^{i + 1} \in \queryclass{i + 1}$ when $\query^{i + 1}$ is an aggregation
function.

To show the correctness of tuples and fragments, we focus on one group $\grp$.
Let $\query^{i + 1}_{\grp}(\query(\db))$ to be an aggregation function works on
$\query^{i}(\db)$ and only focus on groups $\grp$,
$ \forall \tup \in \query^{i}(\db) \colon \tup.\grpatts = \grp$. Thus, the two
properties for $\grp$ will be:
\begin{align*}
  & \query_{\grp}^{i + 1}(\updDB) = \tupsIn{\query_{\grp}^{i + 1}(\annoDB) \uniondb \ieval{\query_{\grp}^{i + 1}}{\dbranges,\annoDeltaDB}{\statedata} } \tag{\tcorrect} \\
  & \query_{\grp}^{i + 1}(\updDB) = \query_{\grp}^{i+1}(\updDBInst{\fragSet{\fragsIn{\query^{i + 1}(\annoDB)} \uniondb \fragsIn{\ieval{\query_{\grp}^{i+1}}{\dbranges,\annoDeltaDB}{\statedata}}}}) \tag{\fcorrect}
\end{align*}

%% For each annotated tuple in $\ieval{\query_{g}^{i}}{\dbranges, \annoDeltaDB}{\statedata}$
%% , it belongs to and can only belong to one group. And it contributes to and only contributes to
%% one group when computing the new aggregation value and group sketch. Then for
%% every group $\grp \in \grpatts$, the property holds such that for aggregation
%% function
%% at level $i + 1$:
%% \begin{align*}
%%   & \query_{g}^{i + 1}(\updDB) = \tupsIn{\query_{g}^{i + 1}(\annoDB) \uniondb \ieval{\query_{g}^{i + 1}}{\dbranges,\annoDeltaDB}{\statedata} }\\
%%   & \query^{i + 1}(\updDB) = \query_{\grp}^{i+1}(\updDBInst{\fragSet{\fragsIn{\query^{i  + 1}(\annoDB)} \uniondb \fragsIn{\ieval{\query_{\grp}^{i+1}}{\dbranges,\annoDeltaDB}{\statedata}}}})
%% \end{align*}

%%%%%%%%%%%%%%%%%%%%%%%%%%%%%%%%%%%%%%%%
%% Aggregation: tuple correctness
\parttitle{\Tcorrect}
For one group $\grp$, based on our rule, the annotated tuple before and after applying
delta annotated tuples are:

\begin{align*}
  &\query_{g}^{i + 1}(\annoDB) = \bag{\deletedelta \pair{\grp \concat (f(\att))}{\provSketch}} \\
  &\query_{g}^{i + 1}(\annoupdDB) = \bag{\insertdelta \pair{\grp \concat (\widehat{f(\att)})}{\widehat{\provSketch}}}
\end{align*}

Since $\query_{g}^{i + 1}(\annoDB) = \bag{ \pair{\deletedelta \grp \concat (f(\att))}{\provSketch}}$, then, we can get that:
\begin{align*}
    & \tupsIn{\query_{g}^{i + 1}(\annoDB)} \uniondb \tupsIn{\bag{\deletedelta \pair{\grp \concat (f(\att))}{\provSketch}}} = \emptyset \tag{$\emptyset_{\grp}$}
\end{align*}
Therefore, the tuple correctness for $\query^{i + 1}_{g}(\updDB)$ is shows:
\begin{align*}
    & \query_{\grp}^{i + 1}(\updDB)\\
  = & \tupsIn{\query_{\grp}^{i + 1 }(\annoupdDB) } \tag{$\tupsIn{\annoDB} = \db$}\\
  = & \tupsIn{\bag{\insertdelta \pair{\grp \concat (\widehat{f(\att)})}{\widehat{\provSketch}}}}\\
  = & \emptyset \uniondb \tupsIn{\bag{\insertdelta \pair{\grp \concat (\widehat{f(\att)})}{\widehat{\provSketch}}}}\\
  = & \tupsIn{\query_{g}^{i + 1}(\annoDB)} \uniondb \tupsIn{ \bag{\deletedelta \pair{\grp \concat (f(\att))}{\provSketch}}}\uniondb \tupsIn{\bag{\insertdelta \pair{\grp \concat (\widehat{f(\att)})}{\widehat{\provSketch}}}}\tag{$\emptyset_{\grp}$}\\
  = & \tupsIn{\query_{g}^{i + 1}(\annoDB)} \uniondb \left(\tupsIn{\bag{\deletedelta \pair{\grp \concat (f(\att))}{\provSketch}}\uniondb \bag{\insertdelta \pair{\grp \concat (\widehat{f(\att)})}{\widehat{\provSketch}}}}\right)\tag{\cref{lemma:tupsin-udot-annodb-annodeltadb}}\\
\end{align*}

Based on the incremental rule of aggregation, it will delete the current group's
$(\grp \concat (f(\att)))$ and insert an tuple $(\grp \concat (\widehat{f(\att)}))$. If we do not apply the $\uniondb$ but keep them as
two independent tuples, which is
$\tupsIn{\bag{\deletedelta \pair{\grp \concat (f(\att))}{{\provSketch}} } \union \bag{\insertdelta \pair{\grp \concat (f(\att))}{\widehat{\provSketch}}}}$.
And it is the output of incremental procedure which is
$\tupsIn{\ieval{\query^{i + 1}_{g}}{\dbranges, \annoDeltaDB}{\statedata}}$

Therefore, for one group $\grp$, the tuple correctness hold such that:
\begin{align*}
  & \query_{\grp}^{i + 1}(\updDB)\\
= & \tupsIn{\query_{\grp}^{i + 1}(\annoDB)} \uniondb \tupsIn{\ieval{\query_{\grp}^{i + 1}}{\dbranges,\annoDeltaDB}{\statedata} }\\
= & \tupsIn{\query_{\grp}^{i + 1}(\annoDB) \uniondb \ieval{\query_{\grp}^{i + 1}}{\dbranges,\annoDeltaDB}{\statedata} }
\end{align*}

\parttitle{Creating or deleting a group}
If we create a new group, then $\query_{g}^{i + 1}(\annoDB) = \bag{\deletedelta \pair{\grp
   \concat (f(\att))}{\provSketch}} = \emptyset$. Then $\bag{\insertdelta \pair{\grp \concat
   (\widehat{f(\att)})}{\widehat{\provSketch}}}$ is the only output of
$\ieval{\query^{i}}{\dbranges, \annoDeltaDB}{\statedata}$.

Therefore,
$\query_{\grp}^{i + 1}(\updDB) = \tupsIn{\query_{\grp}^{i + 1}(\annoDB) \uniondb
 \ieval{\query_{\grp}^{i + 1}}{\dbranges,\annoDeltaDB}{\statedata}}$. If we
delete a group, then $\query_{g}^{i + 1}(\updDB) = \emptyset$. From the incremental
semantics, we just output $\bag{\deletedelta \pair{\grp \concat
   (\widehat{f(\att)})}{\provSketch}}$, then from $\emptyset_{\grp}$, then the result is
   empty as well. Thus, for deleting a group, the results are empty. Therefore,
   for a group $\grp$, the property holds such that:
\[
  \query_{\grp}^{i + 1}(\updDB) = \tupsIn{\query_{\grp}^{i + 1}(\annoDB) \uniondb \ieval{\query_{\grp}^{i + 1}}{\dbranges,\annoDeltaDB}{\statedata} }
\]

%% Therefore, for $\query^{i + 1} \in \queryclass{i + 1}$ and $\query^{i + 1}$ is
%% group-by aggregation. The tuple correctness holds such that:
%% \[
%%   \query^{i + 1}(\updDB) = \tupsIn{\query^{i + 1}(\annoDB) \uniondb \ieval{\query^{i + 1}}{\dbranges,\annoDeltaDB}{\statedata} }
%% \]

%%%%%%%%%%%%%%%%%%%%%%%%%%%%%%%%%%%%%%%%
%% Aggregation: frag correctness
\parttitle{\Fcorrect}
For one group $\grp$, based on the rule, the annotated tuple before and after applying
delta annotated tuples are:

\begin{align*}
  &\query_{g}^{i + 1}(\annoDB) = \bag{\deletedelta \pair{\grp \concat (f(\att))}{\provSketch}} \\
  &\query_{g}^{i + 1}(\annoupdDB) = \bag{\insertdelta \pair{\grp \concat (\widehat{f(\att)})}{\widehat{\provSketch}}}
\end{align*}

Since $\query_{g}^{i + 1}(\annoDB) = \bag{\deletedelta \pair{\grp \concat (f(\att))}{\provSketch}}$, then, we can get that:
\begin{align*}
    & \fragsIn{\query_{g}^{i + 1}(\annoDB)} \uniondb \fragsIn{\bag{\deletedelta \pair{\grp \concat (f(\att))}{\provSketch}}} = \emptyset \tag{$\emptyset_{\grp}$}
\end{align*}
Since $\query_{g}^{i + 1}(\annoupdDB) = \bag{\insertdelta \pair{\grp \concat (\widehat{f(\att)})}{\widehat{\provSketch}}}$, then $\query_{g}^{i + 1}(\updDB) = \query^{i+1}_{g}(\updDBInst{\widehat{\provSketch}}) = \query^{i+1}_{g}(\updDBInst{\fragSet{\fragsIn{\bag{\insertdelta \pair{\grp \concat (\widehat{f(\att)})}{\widehat{\provSketch}}}}}})$.
Therefore, the fragment correctness for $\query^{i + 1}_{g}(\db)$ is shows:
\begin{align*}
    & \query_{\grp}^{i + 1}(\updDB)\\
  = & \query^{i+1}_{g}(\updDBInst{\fragSet{\fragsIn{\bag{\insertdelta \pair{\grp \concat (\widehat{f(\att)})}{\widehat{\provSketch}}}}}})\\
  = & \query^{i+1}_{g}(\updDBInst{\fragsIn{\bag{\insertdelta \pair{\grp \concat (\widehat{f(\att)})}{\widehat{\provSketch}}}}})\\
  = & \query^{i+1}_{g}(\updDBInst{\emptyset \uniondb \fragsIn{\bag{\insertdelta \pair{\grp \concat (\widehat{f(\att)})}{\widehat{\provSketch}}}}}) \tag{$\emptyset \uniondb \fragsIn{\cdot} = \fragsIn{\cdot}$}\\
  = & \query^{i+1}_{g}(\updDBInst{\fragsIn{\query_{g}^{i + 1}(\annoDB)} \uniondb \fragsIn{\bag{\deletedelta \pair{\grp \concat (f(\att))}{\provSketch}}} \uniondb \fragsIn{\bag{\insertdelta \pair{\grp \concat (\widehat{f(\att)})}{\widehat{\provSketch}}}}}) \tag{$\emptyset_{\grp}$}\\
  = & \query^{i+1}_{g}(\updDBInst{\fragsIn{\query_{g}^{i + 1}(\annoDB)} \uniondb \fragsIn{\bag{\deletedelta \pair{\grp \concat (f(\att))}{\provSketch}} \uniondb \bag{\insertdelta \pair{\grp \concat (\widehat{f(\att)})}{\widehat{\provSketch}}}}}) \tag{\Cref{lemma:frags-udot-annodb-annodeldb}}
\end{align*}

Based on the incremental semantics of aggregation, it will delete the current
group's fragments $\fragsIn{\bag{\deletedelta \pair{\grp \concat (f(\att))}{\provSketch}}}$
and insert fragments $\fragsIn{\bag{\insertdelta \pair{\grp \concat
      (\widehat{f(\att)})}{\widehat{\provSketch}}}}$. As tuple correctness, we do not apply the $\uniondb$ but keep then as two fragments bags which is
\[
  \fragsIn{\bag{\deletedelta \pair{\grp \concat (f(\att))}{\provSketch}} \union \bag{\insertdelta \pair{\grp \concat \widehat{(f(\att))}}{\widehat{\provSketch}}} }
\]

And
$\fragsIn{\bag{\deletedelta \pair{\grp \concat (f(\att))}{\provSketch}} \union \bag{\insertdelta \pair{\grp \concat \widehat{(f(\att))}}{\widehat{\provSketch}}} }$
is the fragments output from incremental maintenance of aggregation for group $\grp$ which is $\fragsIn{\ieval{\query^{i + 1}_{g}}{\dbranges, \annoDeltaDB}{\statedata}}$
\begin{align*}
    & \query^{i+1}_{g}(\updDB) \\
  = & \query^{i+1}_{g}(\updDBInst{\fragsIn{\query_{g}^{i + 1}(\annoDB)} \uniondb \fragsIn{\ieval{\query^{i + 1}_{g}}{\dbranges, \annoDeltaDB}{\statedata}}})\\
  = & \query^{i+1}_{g}(\updDBInst{\fragSet{\fragsIn{\query_{g}^{i + 1}(\annoDB)} \uniondb \fragsIn{\ieval{\query^{i + 1}_{g}}{\dbranges, \annoDeltaDB}{\statedata}}}})\\
\end{align*}

\parttitle{Create or deleting a group}
If the group $\grp$ is newly created, then there is no previous sketch for this
group, and the
$\fragsIn{\query_{\grp}^{i + 1}(\annoupdDB)} = \fragsIn{ \bag{\insertdelta \pair{\grp \concat (\widehat{f(\att)})}{\widehat{\provSketch}}} }$
. Thus
$\query_{\grp}^{i+1} (\updDB) = \query_{\grp}^{i + 1}(\updDBInst{\fragsIn{\bag{\insertdelta \pair{\grp \concat (\widehat{f(\att)})}{\widehat{\provSketch}} } }})$. If current group $\grp$ is
deleted. Then after maintenance, this group does not exist anymore, and from the
incremental semantics, there is no fragments related to this group.

So for a group $\grp$, the property holds such that:
\[
  \query_{\grp}^{i + 1}(\updDB) = \query_{\grp}^{i+1}(\updDBInst{\fragSet{\fragsIn{\query^{i + 1}(\annoDB)} \uniondb \fragsIn{\ieval{\query_{\grp}^{i+1}}{\dbranges,\annoDeltaDB}{\statedata}}}}) \tag{\fcorrect}
\]

%%%%%%%%%%%%%%%%%%%%%%%%%%%%%%%%%%%%%%%%%%%%%%%%%%%%%%%%%%%%%%%%%%%%%%%%%%%%%%%%
%% all groups
\parttitle{All groups}
We have shown that for one group $\grp$, the correctness of
~\Cref{theo:correctness} holds, therefore, for all groups the theorem holds for
group-by aggregation query $\query^{i + 1} \in \queryclass{i + 1}$ such that:
\begin{align*}
  & \query^{i + 1}(\updDB) = \tupsIn{\query^{i + 1}(\annoDB) \uniondb \ieval{\query^{i + 1}}{\dbranges,\annoDeltaDB}{\statedata} } \tag{\tcorrect} \\
  & \query^{i + 1}(\updDB) = \query^{i+1}(\updDBInst{\fragSet{\fragsIn{\query^{i + 1}(\annoDB)} \uniondb \fragsIn{\ieval{\query^{i+1}}{\dbranges,\annoDeltaDB}{\statedata}}}}) \tag{\fcorrect}
\end{align*}

%%%%%%%%%%%%%%%%%%%%%%%%%%%%%%%%%%%%%%%%%%%%%%%%%%%%%%%%%%%%%%%%%%%%%%%%%%%%%%%%
%% Q^{i + 1} is a topk
\subsubsection{Top-K}
Suppose the operator at level $i + 1$ is a top-k operator, we have
% \[
% \query^{i + 1}(\db) = \topk(\query^{i}(\db))
% \]

%%%%%%%%%%%%%%%%%%%%%%%%%%%%%%%%%%%%%%%%
%% topk: tuple correctness;
\parttitle{\Tcorrect}
Based on the rule, the annotated tuple before and after applying
delta annotated tuples are:
\begin{align*}
  &\query^{i + 1}(\annoDB) =  \deletedelta \topk(\statedata)\\
  &\query^{i + 1}(\annoupdDB) = \insertdelta \topk(\statedata')\\
\end{align*}

Since $\query^{i + 1}(\annoDB) = \topk(\statedata)$, then, we can get that:
\begin{align*}
  \tupsIn{\query^{i + 1}(\annoDB)} \uniondb \tupsIn{\deletedelta \topk(\statedata)} =
  \emptyset \tag{$\emptyset_{\topk}$}
\end{align*}

Therefore, the tuple correctness for $\query^{i + 1}(\updDB)$ is shown:
\begin{align*}
    & \query^{i}(\updDB)\\
  = & \tupsIn{\query^{i + 1 }(\annoupdDB) } \tag{$\tupsIn{\annoDB} = \db$}\\
  = & \tupsIn{\topk(\statedata')}\\
  = & \emptyset \uniondb \tupsIn{\topk(\statedata')}\\
  = & \tupsIn{\query^{i + 1}(\annoDB)} \uniondb \tupsIn{\deletedelta \topk(\statedata)} \uniondb  \tupsIn{\insertdelta \topk(\statedata')}\tag{$\emptyset_{\topk}$}\\
  = & \tupsIn{\query^{i + 1}(\annoDB)} \uniondb \tupsIn{\deletedelta \topk(\statedata) \uniondb  \insertdelta \topk(\statedata')}\tag{\cref{lemma:tupsin-udot-annodb-annodeltadb}}\\
\end{align*}
From the $\topk$ incremental rule, it will delete the a bag of $k$ annotated
tuples which is $\deletedelta \topk(\statedata)$, and insert a bag of $k$ updated annotated
tuples which is  $\insertdelta \topk(\statedata')$. Like tuple correctness of
aggregation, we can keep they as two independent bags of annotated tuples. Then,
they are output of incremental procedure which is:
$\tupsIn{\ieval{\query^{i + 1}}{\dbranges, \annoDeltaDB}{\statedata} }$
Therefore:
\begin{align*}
    & \query^{i}(\updDB)\\
  = & \tupsIn{\query^{i + 1}(\annoDB)} \uniondb \tupsIn{\ieval{\query^{i + 1}}{\dbranges, \annoDeltaDB}{\statedata}}\\
  = & \tupsIn{\query^{i + 1}(\annoDB) \uniondb \ieval{\query^{i + 1}}{\dbranges, \annoDeltaDB}{\statedata}}\\
\end{align*}
%%%%%%%%%%%%%%%%%%%%%%%%%%%%%%%%%%%%%%%%
%% topk: frag Correctness
\parttitle{\Fcorrect}
Based on the rule, the annotated tuple before and after applying
delta annotated tuples are:
\begin{align*}
  &\query^{i + 1}(\annoDB) =  \deletedelta \topk(\statedata)\\
  &\query^{i + 1}(\annoupdDB) = \insertdelta \topk(\statedata')\\
\end{align*}

For $\fragsIn{\deletedelta \topk(\statedata)}$, the fragments are the
same as from $\fragsIn{\query^{i + 1}(\annoDB)}$. Then
\begin{align*}
  & \fragsIn{ \query^{i + 1}(\annoDB) } \uniondb \fragsIn{\deletedelta \topk(\statedata)} = \emptyset
\end{align*}
  For $\fragsIn{\insertdelta \topk(\statedata')}$, since the state
$\statedata'$ contains all annotated tuples corresponding to $\annoupdDB$, then,
the $\fragSet{\fragsIn{\topk(\statedata')}}$ contains all fragments of $\updDB$
to get $\query^{i+1}(\updDB)$ due to the association of tuple and its provenance
sketch. Therefore:
\begin{align*}
    & \query^{i + 1}(\updDB) \\
  = & \query^{i + 1}(\updDBInst{\fragSet{\fragsIn{\insertdelta \topk(\statedata')}}})\\
  = & \query^{i + 1}(\updDBInst{\fragsIn{\insertdelta \topk(\statedata')}})\\
  = & \query^{i + 1}(\updDBInst{\emptyset \uniondb \fragsIn{\insertdelta \topk(\statedata')}})\\
  = & \query^{i + 1}(\updDBInst{\fragsIn{\topk( \query^{i}(\annoDB) )} \uniondb \fragsIn{\deletedelta \topk(\statedata)} \uniondb \fragsIn{\insertdelta \topk(\statedata')}})\\
  = & \query^{i + 1}(\updDBInst{\fragsIn{\topk( \query^{i}(\annoDB) )} \uniondb \fragsIn{\deletedelta \topk(\statedata) \uniondb \insertdelta \topk(\statedata')}})\\
\end{align*}

For $\fragsIn{\deletedelta \topk(\statedata) \uniondb \insertdelta \topk(\statedata')}$, we do not apply $\uniondb$ but keep them two independent bags of annotated tuples. Then they are the output of incremental procedure which is: $\fragsIn{\ieval{\query^{i + 1}}{\dbranges, \annoDeltaDB}{\statedata}}$. Therefore:

\begin{align*}
    & \query^{i + 1}(\updDB) \\
  = & \query^{i + 1}(\updDBInst{\fragsIn{\query^{i + 1}(\annoDB)} \uniondb \fragsIn{\ieval{\query^{i + 1}}{\dbranges, \annoDeltaDB}{\statedata}}})\\
  = & \query^{i + 1}(\updDBInst{\fragSet{\fragsIn{\query^{i + 1}(\annoDB)} \uniondb \fragsIn{\ieval{\query^{i + 1}}{\dbranges, \annoDeltaDB}{\statedata}}}})\\
\end{align*}

\subsection{Conclusion}
\revision{
  We have shown that both \tcorrect and \fcorrect hold for every operator the
  \impAbbr supports such that the the following holds:
\begin{align*}
% tup
  & \query(\updDB) = \tupsIn{\query(\annoDB) \uniondb \ieval{\query}{\dbranges,\annoDeltaDB}{\statedata} } \tag{\tcorrect}\\
% frag
  & \query(\updDB) = \query(\updDBInst{\fragSet{\fragsIn{\query(\annoDB)} \uniondb \fragsIn{\ieval{\query}{\dbranges,\annoDeltaDB}{\statedata}}}})\tag{\fcorrect}
\end{align*}
Then for every operator the \impAbbr supports to maintain its sketches, the
~\Cref{theo:correctness} holds.
}

%%% Local Variables:
%%% mode: LaTeX
%%% TeX-master: "../techreport_edbt"
%%% End:

%% file: sections/implementation_in_memory.tex
\section{The IMP System}
\label{sec:implementation}
%
%%%%%%%%%%%%%%%%%%%%%%%%%%%%%%%%%%%%%%%%%%%%%%%%%%%%%%%%%%%%%%%%%%%%%%%%%%%%%%%%
%% sql implementation
%% \subsection{SQL-based implementation}
%% \label{subsec:sql_implementation}
%
While the semantics from \Cref{sec:incremental_operator_rules} can be
implemented in SQL, an
in-memory implementation can be significantly more efficient as we can utilize
data structures not available in a SQL-based implementation (see  \ifnottechreport{\cite{techreport}}\iftechreport{the discussion in \Cref{sec:sql_implementation_appendix}}).
\iftechreport{As deltas are typically small, having a purely in-memory engine is
 sufficient.}
\impAbbr is implemented as a stand-alone in-memory engine that uses a backend
database for fetching deltas and for evaluating operations (joins) that require
access to large amounts of data.
In \Cref{fig:workflow_overview},
\impAbbr's incremental engine is the pipeline shown in red. % To obtain the latest provenance
% sketch for a query, the user will provide the query and delta database
% information (update statements or delta relations).
\impAbbr executes  $\incremaintain(\query, \statedata, \annoDeltaDB)$ to
generate delta sketches. For joins (and cross products), $\annoDeltaRel \join
\mathscr{S}$ and $\annoRel \join \Delta \mathscr{S}$ are executed by sending
$\annoDeltaRel$ ($\Delta \mathscr{S}$) to the database and evaluating the join
in the database.
% And during maintaining the sketches, stored state
% data will be fetched from the database. At last, new provenance sketches will
% be computed by calculating the original and the delta sketches.

%%%%%%%%%%%%%%%%%%%%%%%%%%%%%%%%%%%%%%%%%%%%%%%%%%%%%%%%%%%%%%%%%%%%%%%%%%%%%%%%
%% Eager vs lazy
% \subsubsection{Eager vs. Lazy Maintenance}
% \label{sec:eager_vs_lazy}
% %
% To determine when to update sketches, we have two strategies: eager maintenance
% and lazy maintenance. In eager maintenance, whenever an update comes, we
% maintain all the provenance sketches. The other extreme is lazy maintenance
% where the delta of each update to the database is
% accumulated and a sketch is determined to be maintained only when it is needed
% to answer a query. Then we fetch the delta between last maintained version of
% the database, $\dbv{j}$, and current version, $\dbv{i}$, which is
% $\deltadbv{ij}$, to refresh the sketch from $\provSketch_{j}$ to
% $\provSketch_{i}$.

%%%%%%%%%%%%%%%%%%%%%%%%%%%%%%%%%%%%%%%%%%%%%%%%%%%%%%%%%%%%%%%%%%%%%%%%%%%%%%%%
%% data chunk and column chunk
\subsection{Storage Layout \& State Data}
\label{sec:data_chunk_and_column_chunk}
We store data in a columnar
representation for horizontal chunks of a table (\textit{data chunks}). % , which means that data in a column of same datatype is kept in a
% vector.
% Then all these vectors are store in a data chunk with metadata. There
% are two types of data chunks:
Annotated inserted / deleted tuples are stored in separate chunks.
% Inserted data chunk and deleted data chunk, which
% stores annotated tuples of insertion or deletion respectively.
The annotations (provenance sketches) of the rows in a data chunk are stored in a separate column as bit sets.
% The data model
% for our incremental processing is to associate each tuple with its provenance
% sketch. When accessing the delta relations, the \impAbbr engine processes each
% tuple and attaches the fragment information to it and stores fragments in its
% own vector. For each operator in a query, the engine uses the data chunks
% propagated up by its child(ren), and generates data chunks with only designated
% columns. Column chunk stores data with meta information in columnar fashion as
% well which is used for expressions evaluation.
% The advantage of using columnar
% fashion is mainly because of efficiency to handle data. For expressions
% evaluation, the engine can process data of columnar caching in memory
% continuously without redirecting to elements if stored in a row-oriented style.

%%%%%%%%%%%%%%%%%%%%%%%%%%%%%%%%%%%%%%%%%%%%%%%%%%%%%%%%%%%%%%%%%%%%%%%%%%%%%%%%
%% state data
\parttitle{Sketch \& State Data}
\label{subsec:sate_data}
\impAbbr stores sketches in a hash-table where the key is a query template for which the sketch was created and the value is the sketch and the state of the incremental operators for this query. Here a query template refers to a version of a query $\query$ where constants in selection conditions are replaced with placeholders such that two queries that only differ in these constants have the same key. This is done to be able to efficiently prefilter candidate sketches to be used for a query as the techniques from \cite{DBLP:journals/pvldb/NiuGLLGKLP21} can determine whether a sketch for query $\query_1$ can be used to answer a query $\query_2$ if these queries share the same template. Furthermore, for each sketch we store a version identifier % (a snapshot identifier for databases that use snapshot isolation)
to record which database version the sketch corresponds to. \impAbbr can persist its state in the database.

% State data of each operator is to best exploit specialized data structures that
% are suited best for operations for data annotated with sketches. When
% maintaining sketches for a query, the engine will first check if this query has
% met before to know weather to build state data or fetch the materialized state
% data. The state data will be updated during maintaining the sketch and be
% retained in the database (the pipeline of bottom in
% \Cref{fig:workflow_overview}) when the engine does not maintain its query
% anymore.

%%%%%%%%%%%%%%%%%%%%%%%%%%%%%%%%%%%%%%%%%%%%%%%%%%%%%%%%%%%%%%%%%%%%%%%%%%%%%%%%
%% agg and top-k
\iftechreport{
\parttitle{Aggregation and Top-k}
\label{sec:aggregation}
For aggregation, we use hashmaps to implement the state for \fcsum, \fccnt, and
\fcavg. \iftechreport{For \fcmin and \fcmax we use red-black trees.} The
  aggregation result for a single group may be updated multiple times when
  processing a delta with multiple tuples. To avoid producing multiple delta
  tuples per group we maintain copies of the previous states of groups before
  an incremental maintenance operations that are created lazily when a group is
  updated for the first time when processing a delta $\annoDeltaDB$.
Once a delta has been processed, we use
the per batch data structure to determine which groups have been updated and
output deltas for the group as described in \ifnottechreport{\Cref{sec:op-rules-agg}.}
\iftechreport{\Cref{sec:op-rules-agg,sec:op-rule-agg-min-max}.}
  We follow the same approach for the top-k operator. We use red-black trees for maintaining state for top-k operators (and for \fcmin and \fcmax).
}
 % \PL{Here we missing discuss red-black trees for min/max and top-k.}

%   Once a
%   delta has been fully processed, we use this data structure to produce two delta tuples per updated group (deleting the previous version and inserting the new version).

%   At the beginning of each
% batch we initialize an empty copy of the state that is used to track which
% groups have been updated as part of the current batch. Whenever a group is
% updated for the first time in a batch (there is no entry for this group in the
% copy yet), then we copy the current state for the group to the per-batch data
% structure. The purpose of this data structure is to track which groups were
% affected by a batch and cache their previous state. We do not output any delta
% results while processing a batch. Instead after processing the batch, we use
% the per batch data structure to determine which groups have been updated and
% output deltas for the group as described in
% \Cref{sec:op-rules-agg,sec:op-rule-agg-min-max}. Note that we do not copy
% $\fragCnt_{\grp}$ to the per batch data structure as we do not need the
% previous state for this data structure. We follow the same approach for the
% top-k operator.

%%%%%%%%%%%%%%%%%%%%%%%%%%%%%%%%%%%%%%%%%%%%%%%%%%%%%%%%%%%%%%%%%%%%%%%%%%%%%%%%
%% optimization
\subsection{Optimizations}
\label{sec:optimization}
Data transfer between \impAbbr and the DBMS can become a bottleneck. We now introduce several optimizations that reduce communication.
% can reduce the efficiency of the
% performance of incremental procedure. To overcome this issue, we have optimized
% operators that need to access database, i.e., table accessing and join
% operators. Both selection-push-down and bloom filter techniques target to
% reduce the amount of data in transferring. Less data in transferring, less
% overhead increasing for incremental procedure.

%%%%%%%%%%%%%%%%%%%%%%%%%%%%%%%%%%%%%%%%%%%%%%%%%%%%%%%%%%%%%%%%%%%%%%%%%%%%%%%%
%% optimization: bloom filter
\parttitle{Bloom Filters For Join}
\label{sec:bloom-filters-join}
For join operators, \impAbbr uses the DBMS to compute the
result of $\annoRel \join \makeAnnoDelta{S}$ and $\makeAnnoDelta{R} \join \makeAnno{S}$ which requires sending $\makeAnnoDelta{R}$ (or $\makeAnnoDelta{S}$) to the database.
\impAbbr maintains bloom filters on the join attributes for both sides of equi-joins that are used to filter out rows from $\makeAnnoDelta{R}$ (and $\makeAnnoDelta{S}$) that do not have any join partners in the other table. If according to bloom filter no rows from the delta have join partners then we can avoid the round trip to the database completely.

% For a table joined
% another one with low selectivity, most tuples do not have join partners, and if
% those tuples can be removed up-front and not transfer to the database, then it
% will reduce the overhead of join performance. To pre-decide which tuples can
% potentially pair others, bloom filter can select them out. We build for each
% attribute in eqi-join conditions a bloom filter to tell the existence of values
% in the other side of join. Before passing delta tuples into database, first
% check each tuple if there exists tuples can matching it. Then send the tuples
% passing bloom filter validation to the database and get back the result. Bloom
% filters is materialized locally and used for subsequently updates.

%%%%%%%%%%%%%%%%%%%%%%%%%%%%%%%%%%%%%%%%%%%%%%%%%%%%%%%%%%%%%%%%%%%%%%%%%%%%%%%%
%% optimization: filter delta
\parttitle{Filtering Deltas Based On Selections}
\label{sec:filt-delta-inputs}
If a query involves a selection and all operators in the subtree rooted at a
selection are stateless, then we can avoid fetching delta tuples from the
database that do not fulfill the selection's condition as such tuples will
neither affect the state of operators downstream from the selection nor will they impact the final maintenance result as their decedents will be filtered by the incremental
 selection. That is, we can push the selection
conditions into the query that retrieves the delta. % from the database.

% For operators accessing a delta relation, a % selection-push-down technique
% is applied to push the conditions to scan % operator. This enables to filter
% out data that is not satisfied the conditions % inside DBMS and transfers back
% only the desired tuples.
\input{sections/fig-end-to-end.tex}

\input{sections/fig-tpch-crime}

%%%%%%%%%%%%%%%%%%%%%%%%%%%%%%%%%%%%%%%%%%%%%%%%%%%%%%%%%%%%%%%%%%%%%%%%%%%%%%%%
%% optimization: space for min/max/top-k
\iftechreport{
\parttitle{Optimizing Minimum, Maximum, and Top-k}
\label{sec:optim-minim-maxim}
If the input to an aggregation with \fcmin or \fcmax or the top-k is large,
then maintaining the sorted map can become a bottleneck. Instead of storing the
full input, we can only store the top / bottom $m$
tuples. By keeping a record of the first $m$ tuples, it is safe-guaranteed to
delete $m$ tuples from its input. This is useful when dealing with deletion.
For example, if we keep first 20 minimum values for a group when we build the state for this group.
The state data can support deletions that removes at
least than 20 tuples. To achieve this, we pass an parameter to \impAbbr
engine when building the state data for these operators, the engine will
only get a certain number of tuples to build the state. This optimization will
help for deletion, because it is necessary to know the next minimum/maximum,
while for insertion, only one tuple per group stored can make incremental
maintenance work, because the only tuple is always the current minimum/maximum
and the new minimum/maximum will always be in the state data if it comes from
inputs.
}

\ifnottechreport{In \cite{techreport} we also discuss how non-serializable isolation levels for concurrency control protocols based on \gls{si} affect  maintenance of sketch versions.}

\iftechreport{
%%%%%%%%%%%%%%%%%%%%%%%%%%%%%%%%%%%%%%%%%%%%%%%%%%%%%%%%%%%%%%%%%%%%%%%%%%%%%%%%
%% transaction and delta
\subsection{Concurrency Control \& Sketch Versions}%
\label{sec:transaction_and_delta_data}
So far we have assumed that the database backend uses snapshot isolation and
that sketch versions are identified by snapshot identifiers. However, in
snapshot isolation, each transaction sees data committed before it started
(identified by a snapshot identifier) and its own changes. Thus, for a
transaction that wants to use a sketch after updating a table accessed by the
query for the sketch, we have to include the transaction's updates when
maintaining the sketch. We can track these updates using standard audit logging
mechanisms supported nativly in databases like Oracle or implemented through
extensibility mechanisms like triggers to keep a history of row versions. For
statement-level snapshot isolation (isolation level \lstinline!READ COMMITTED!
in systems like Postgres or Oracle), we face the challenge that even if we run
the queries for incremental maintenance in the same transaction as the query
that uses the updated sketch, these queries may see different versions of the
database. Thus, supporting statement-level snapshot isolation requires either
deeper integration into the database to run the maintenance query as of the same
snapshot as the query that uses the sketch or use techniques like
reenactment~\cite{AG17c, AG17} to reconstruct such database states.
}

\subsection{\revision{Selecting Sketch Attributes and Ranges}}\label{sec:how-to-range}
\revision{
Both the choice of attribute for a sketch and the choice of ranges can affect the amount of data covered by a sketch. Regarding the choice of attribute, we first identify which attributes are safe using techniques from \cite{DBLP:journals/pvldb/NiuGLLGKLP21}. In \cite{costarxiv} we studied cost-based selection of attributes for sketches. For \impAbbr we employ a heuristic to select attributes based on the insights from \cite{costarxiv}. We select attributes that are important for a query such as group-by attributes or attributes for which efficient access methods, e.g., an index, are available. Regarding the choice of ranges, as long as ranges are fine-granular enough and data is roughly evenly distributed across ranges, the exact choice of ranges has typically neglectable effect on sketch size. We use the bounds of equi-depth histograms maintained by many DBMS as statistics as ranges. Note that we generate ranges to cover the whole domain of an attribute instead of only its active domain. If a significant fraction of the data in a relation is updated, then this can lead to an imbalance in the amount of data per range and in turn to a degradation of the performance of sketches over time. As a significant change in distribution is unlikely to occur frequently, we can simply update the ranges and recapture sketches.
\ifnottechreport{In \cite{techreport} we discuss potential strategies to avoid this.}
\iftechreport{If more frequent changes to distribution are expected, then we could track estimates of the number of tuples per range and split or merge ranges that under- or overflow. If a range $\range$ is split into two ranges $\range_1$ and $\range_2$ then any sketch containing $\range$ would then be updated to contain $\range_1$ and $\range_2$. If two ranges $\range_1$ and $\range_2$ are merged into a single range $\range$, then any sketch containing either $\range_1$ and $\range_2$ is updated to contain $\range$ instead.}
%
%   Our system supports partitioning over the full domain of multiple attributes
% not just the active domain. To get the ranges over the domain, we utilize the
% histograms maintained by the database. Since the database only maintains the
% histograms for current active domain, in order support the full domain in case
% for values exceeding the active domain, our current approach is insert a tuple
% into the first range if the value is less than the minimum value (inserted
% into the last range if the value is greater than the maximum value of the
% active domain). Another approach is the add two extra intervals: one for
% values less than active domain and the other one is for values larger than the
% active domain. This approach will not increase too much for the size of
% provenance sketches. Another interesting question is what if the data
% distribution is changed significantly by some updates. The simple solution is
% to recapture the provenance sketches. Some smarter strategies can be applied:
% merging or splitting current ranges. Under merging, we can merge several
% continuous ranges into one, and add this larger range into provenance sketch
% if only one of these ranges belongs to provenance sketch. If splitting one
% range into several smaller ranges, we should check which sub-range contains
% provenance. This can be done by partially re-capture sketches over the
% original larger range. Or we can delay to validate the sub-ranges and add all
% of them into provenance sketches if the original sketch contains provenance.
}

%% file: sections/fig-end-to-end.tex
%%%%%%%%%%%%%%%%%%%%%%%%%%%%%%%%%%%%%%%%%%%%%%%%%%%%%%%%%%%%%%%%%%%%%%%%%%%%%%%%
%% figs for end to end experiments
\iftechreport{
\begin{figure*}[h]
  %% delta tuples: 1 tuple
  \begin{minipage}{1.0\linewidth}
    %% 1U 5Q
    \begin{subfigure}{0.33\linewidth}
      \includegraphics[width = 1.0\textwidth]{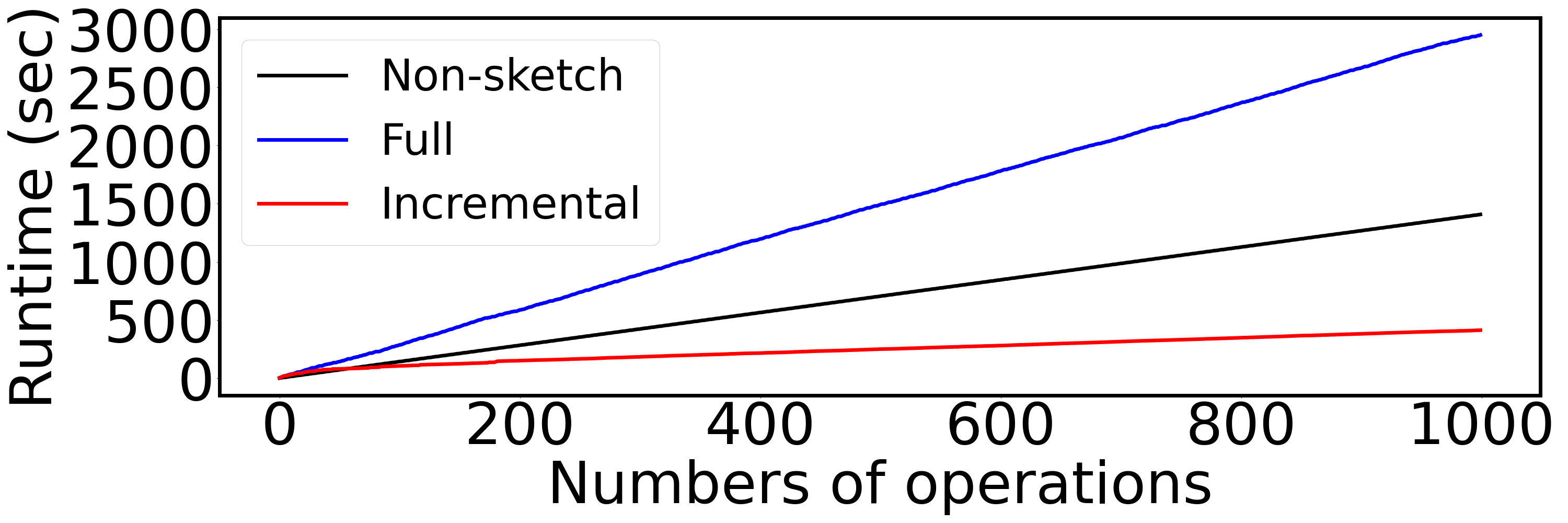}
      \vspace{-6mm}
      \caption{1U5Q delta: 1 tuple}
      \label{fig:endtoend_t1_1u5q}
    \end{subfigure}
    %% 1U 1Q
    \begin{subfigure}{0.33\linewidth}
      \includegraphics[width = 1.0\textwidth]{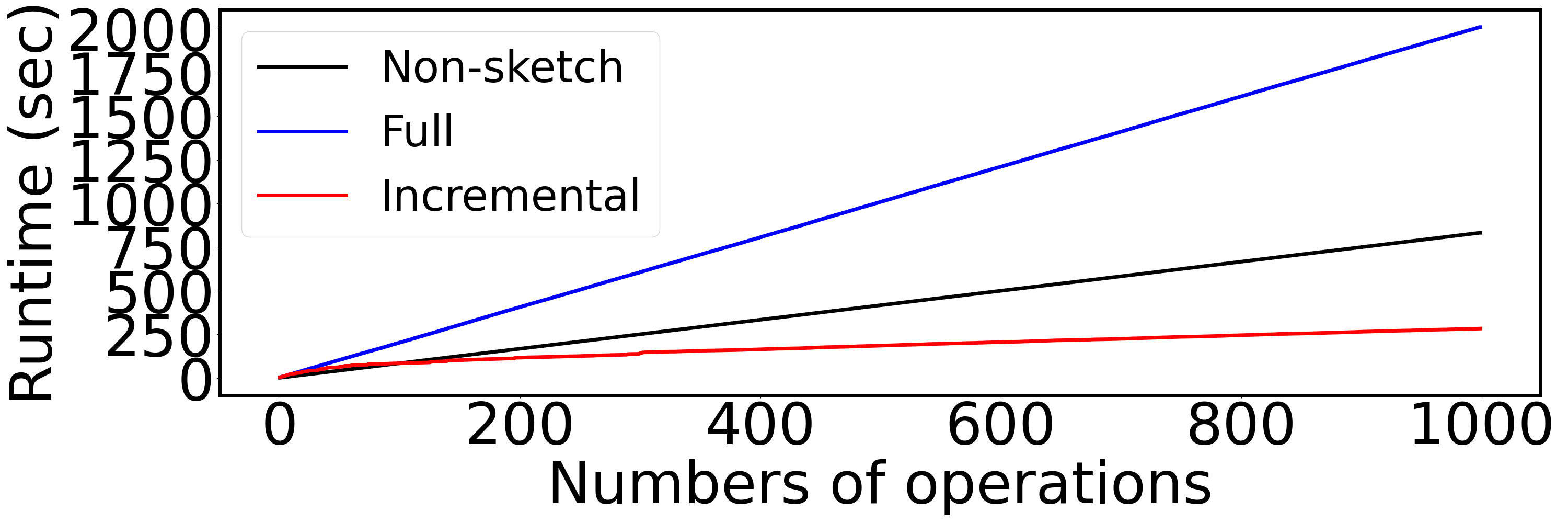}
      \vspace{-6mm}
      \caption{1U1Q delta: 1 tuple}
      \label{fig:endtoend_t1_1u1q}
    \end{subfigure}
    %% 5U 1Q
    \begin{subfigure}{0.33\linewidth}
      \includegraphics[width = 1.0\textwidth]{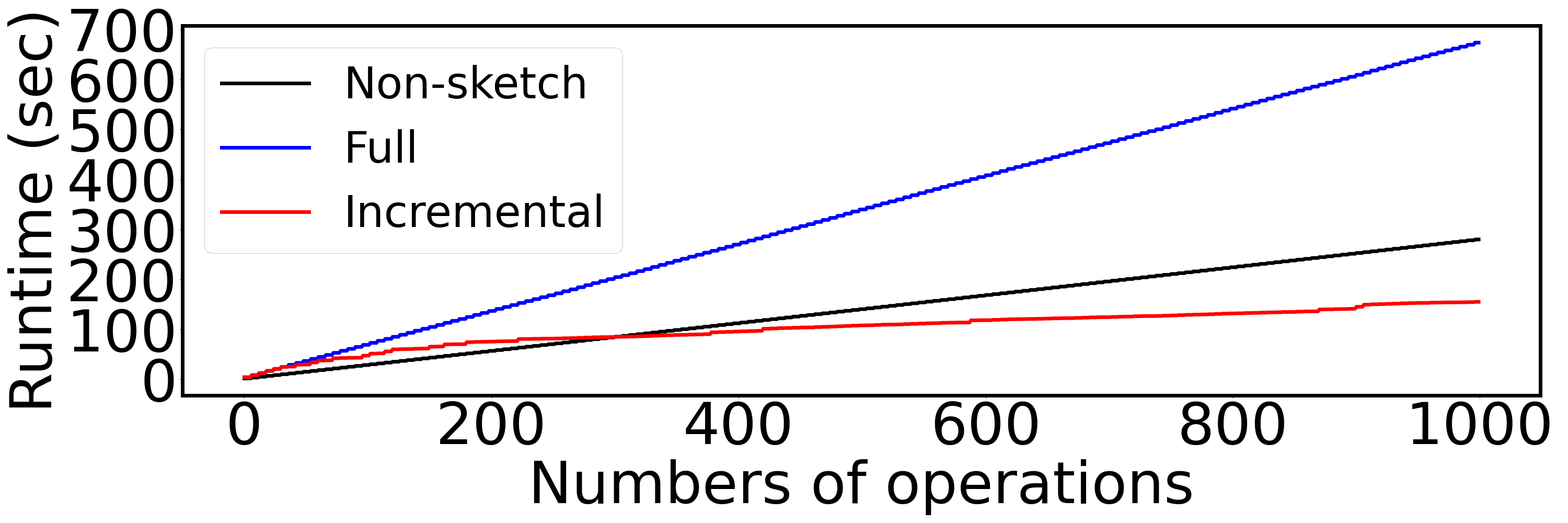}
      \vspace{-6mm}
      \caption{5U1Q delta: 1 tuple}
      \label{fig:endtoend_t1_5u1q}
    \end{subfigure}
  \end{minipage}
  %% delta tuples: ~20 tuples
  \begin{minipage}{1.0\linewidth}
    %% 1U 5Q
    \begin{subfigure}{0.33\linewidth}
      \includegraphics[width = 1.0\textwidth]{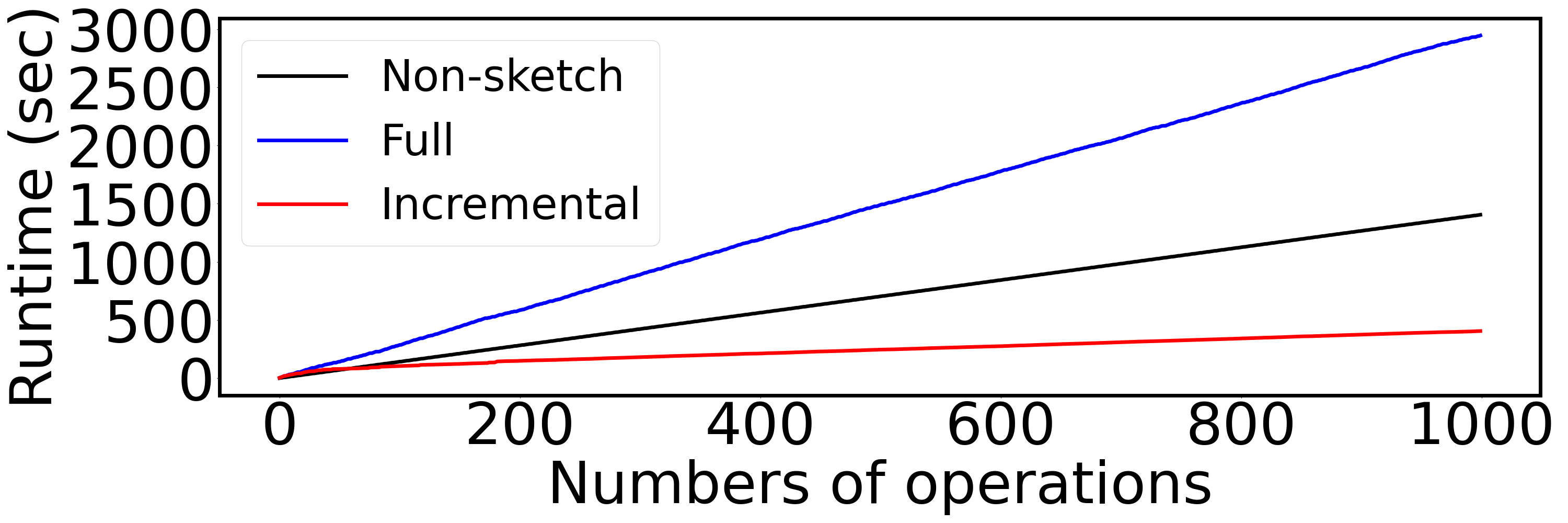}
      \vspace{-6mm}
      \caption{1U5Q delta: 20 tuples}
      \label{fig:endtoend_t20_1u5q}
    \end{subfigure}
    %% 1U 1Q
    \begin{subfigure}{0.33\linewidth}
      \includegraphics[width = 1.0\textwidth]{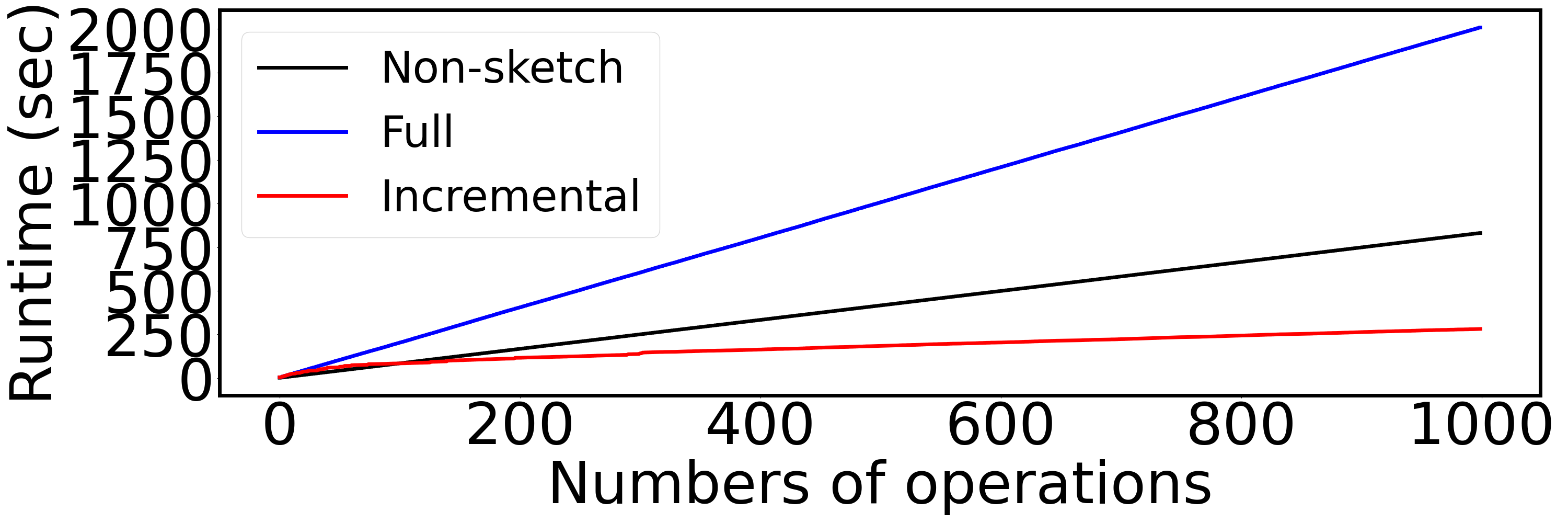}
      \vspace{-6mm}
      \caption{1U1Q delta: 20 tuples}
      \label{fig:endtoend_t20_1u1q}
    \end{subfigure}
    %% 5U 1Q
    \begin{subfigure}{0.33\linewidth}
      \includegraphics[width = 1.0\textwidth]{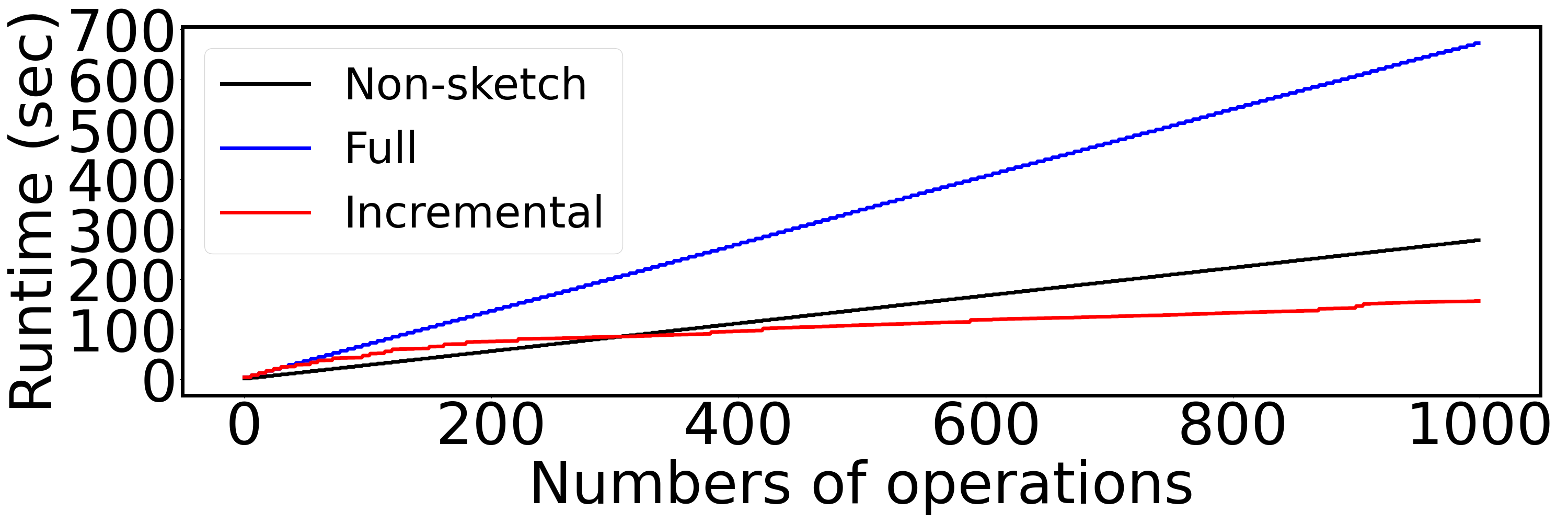}
      \vspace{-6mm}
      \caption{5U1Q delta: 20 tuples}
      \label{fig:endtoend_t20_5u1q}
    \end{subfigure}
  \end{minipage}
  %% delta tuples: ~200 tuples
  \begin{minipage}{1.0\linewidth}
    %% 1U 5Q
    \begin{subfigure}{0.33\linewidth}
      \includegraphics[width = 1.0\textwidth]{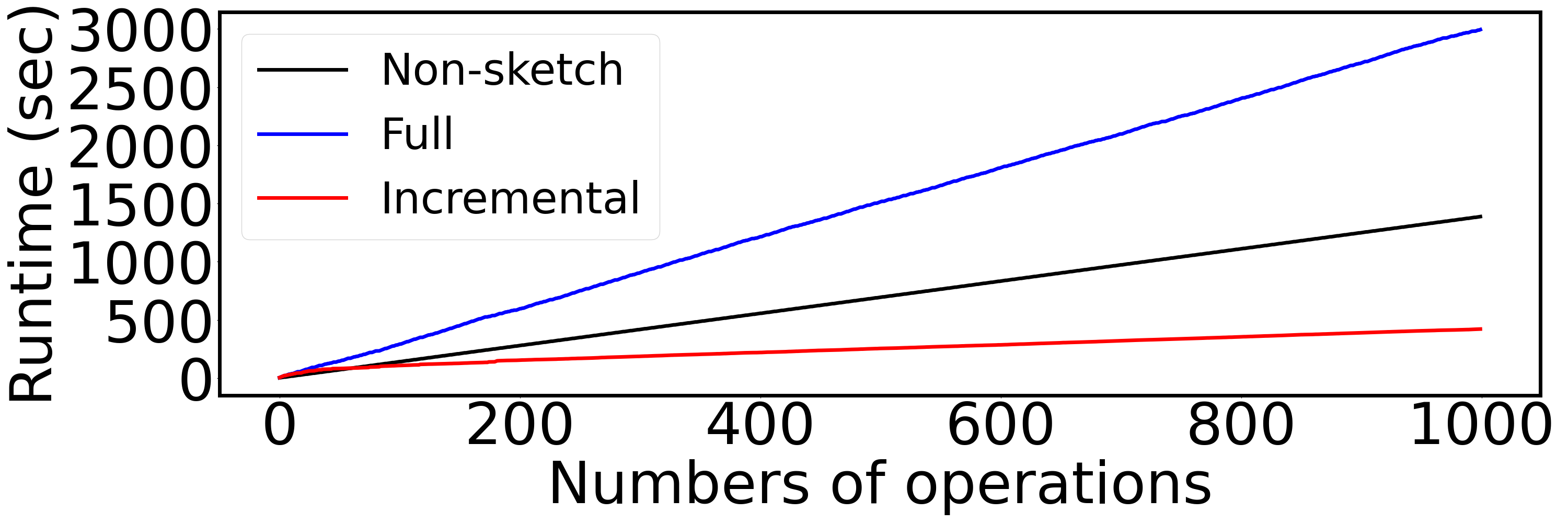}
      \vspace{-6mm}
      \caption{1U5Q delta: 200 tuples}
      \label{fig:endtoend_t200_1u5q}
    \end{subfigure}
    %% 1U 1Q
    \begin{subfigure}{0.33\linewidth}
      \includegraphics[width = 1.0\textwidth]{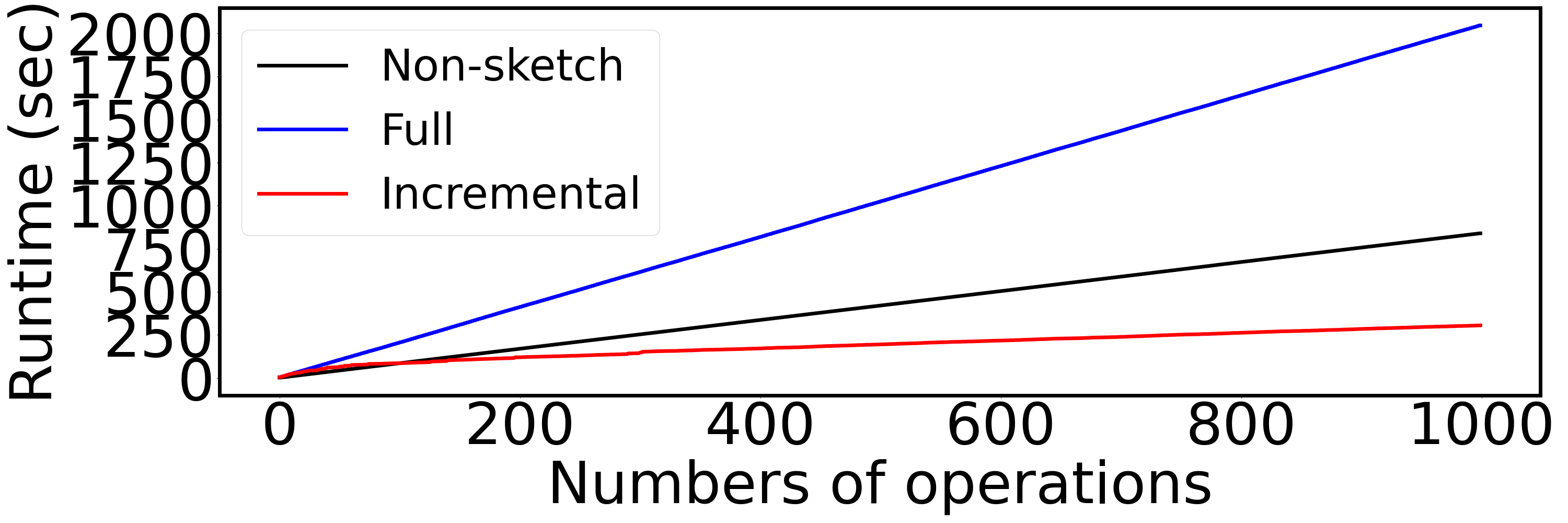}
      \vspace{-6mm}
      \caption{1U1Q delta: 200 tuples}
      \label{fig:endtoend_t200_1u1q}
    \end{subfigure}
    %% 5U 1Q
    \begin{subfigure}{0.33\linewidth}
      \includegraphics[width = 1.0\textwidth]{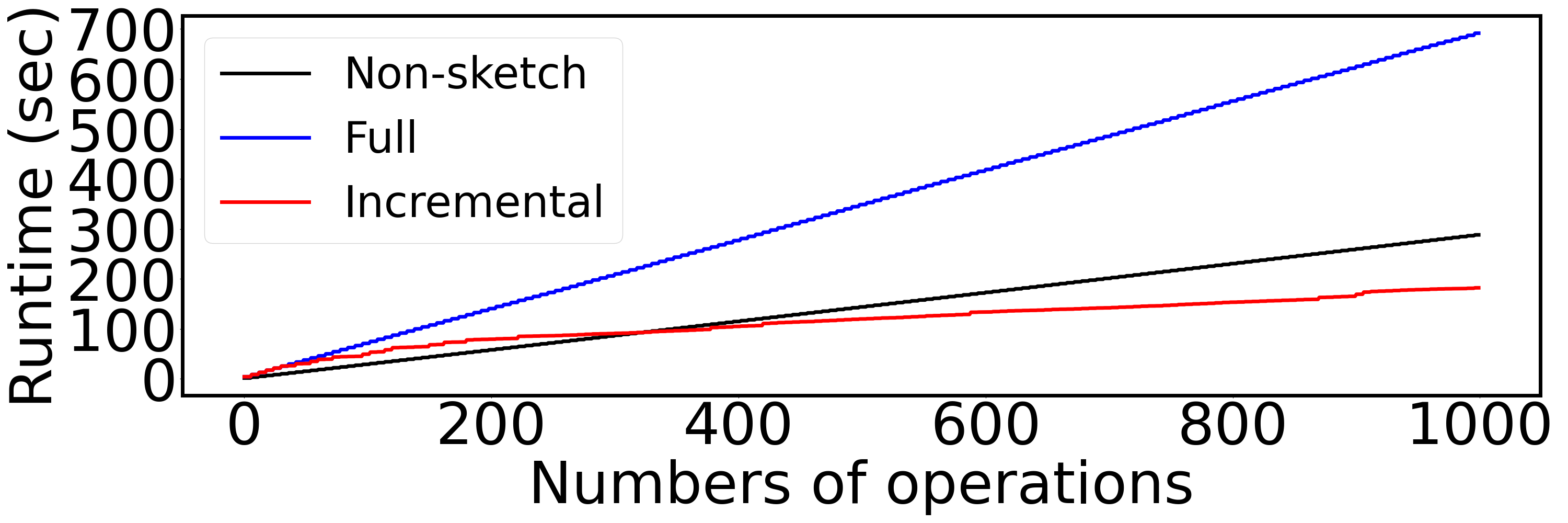}
      \vspace{-6mm}
      \caption{5U1Q delta: 200 tuples}
      \label{fig:endtoend_t200_5u1q}
    \end{subfigure}
  \end{minipage}
  %% delta tuples: ~2000 tuples
  \begin{minipage}{1.0\linewidth}
    %% 1U 5Q
    \begin{subfigure}{0.33\linewidth}
      \includegraphics[width = 1.0\textwidth]{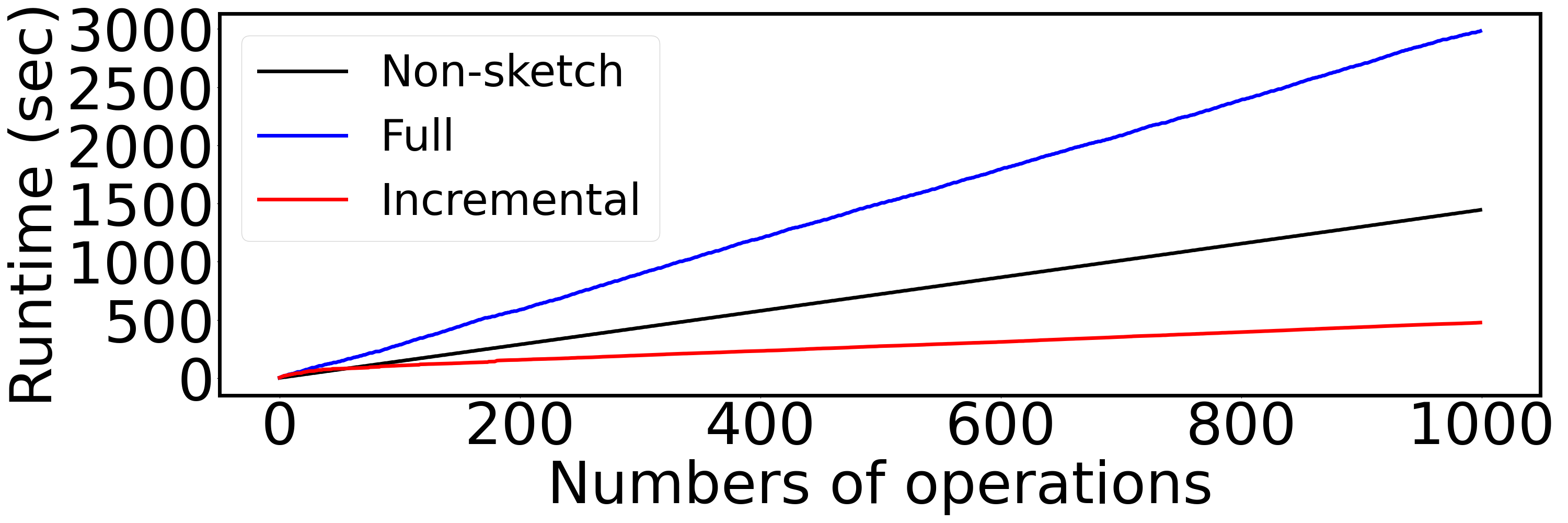}
      \vspace{-6mm}
      \caption{1U5Q delta: 2000 tuples}
      \label{fig:endtoend_t2000_1u5q}
    \end{subfigure}
    %% 1U 1Q
    \begin{subfigure}{0.33\linewidth}
      \includegraphics[width = 1.0\textwidth]{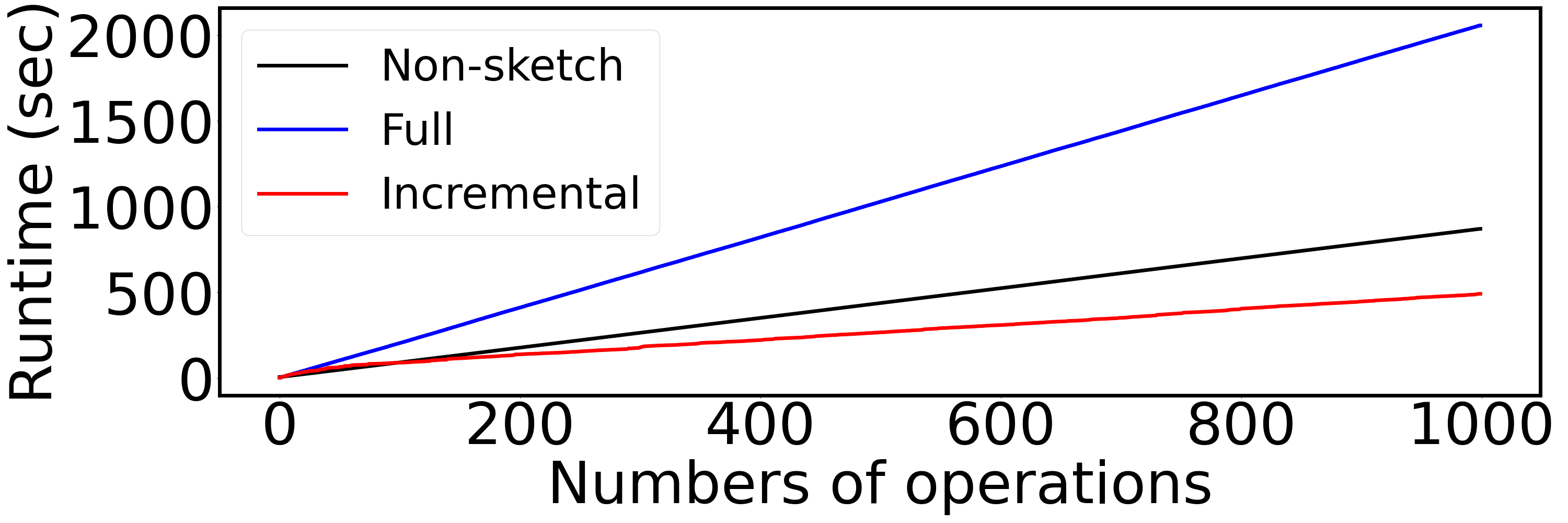}
      \vspace{-6mm}
      \caption{1U1Q delta: 2000 tuples}
      \label{fig:endtoend_t2000_1u1q}
    \end{subfigure}
    %% 5U 1Q
    \begin{subfigure}{0.33\linewidth}
      \includegraphics[width = 1.0\textwidth]{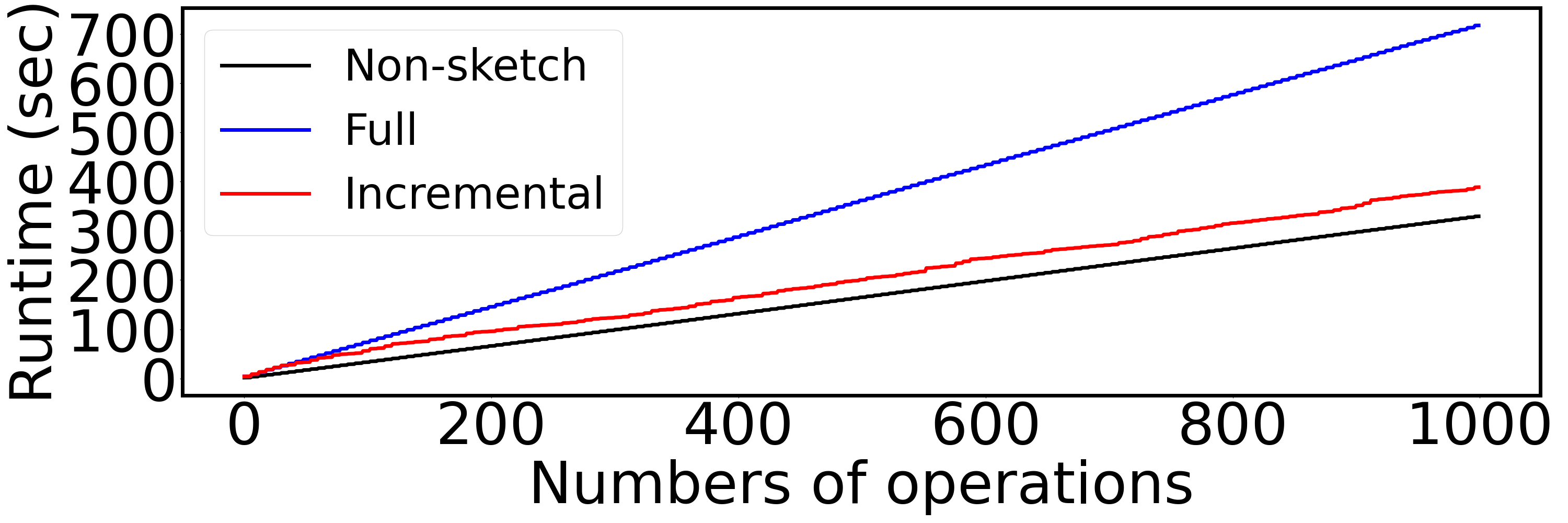}
      \vspace{-6mm}
      \caption{5U1Q delta: 2000 tuples}
      \label{fig:endtoend_t2000_5u1q}
    \end{subfigure}
  \end{minipage}
  \vspace{-4mm}
  % \caption{End to end experiments}
  \caption{Varying delta size to check the performance of non-sketch, full and incremental approaches}
  \label{fig:endtoends}
\end{figure*}
}
%%%%%%%%%%%%%%%%%%%%%%%%%%%%%%
%%  main paper
\ifnottechreport{
 \begin{figure*}[h]
  %% delta tuples: 1 tuple
%  \begin{minipage}{1.0\linewidth}
%    %% 1U 5Q
%    \begin{subfigure}{0.33\linewidth}
%      \includegraphics[width = 1.0\textwidth]{figs/endtoend_T1_1u5q.png}
%      \vspace{-6mm}
%      \caption{1U5Q delta: 1 tuple}
%      \label{fig:endtoend_t1_1u5q}
%    \end{subfigure}
%    %% 1U 1Q
%    \begin{subfigure}{0.33\linewidth}
%      \includegraphics[width = 1.0\textwidth]{figs/endtoend_T1_1u1q.png}
%      \vspace{-6mm}
%      \caption{1U1Q delta: 1 tuple}
%      \label{fig:endtoend_t1_1u1q}
%    \end{subfigure}
%    %% 5U 1Q
%    \begin{subfigure}{0.33\linewidth}
%      \includegraphics[width = 1.0\textwidth]{figs/endtoend_T1_5u1q.png}
%      \vspace{-6mm}
%      \caption{5U1Q delta: 1 tuple}
%      \label{fig:endtoend_t1_5u1q}
%    \end{subfigure}
%  \end{minipage}
  %% delta tuples: ~20 tuples
  \begin{minipage}{1.0\linewidth}
    %% 1U 5Q
    \begin{subfigure}{0.33\linewidth}
      \includegraphics[width = 1.0\textwidth]{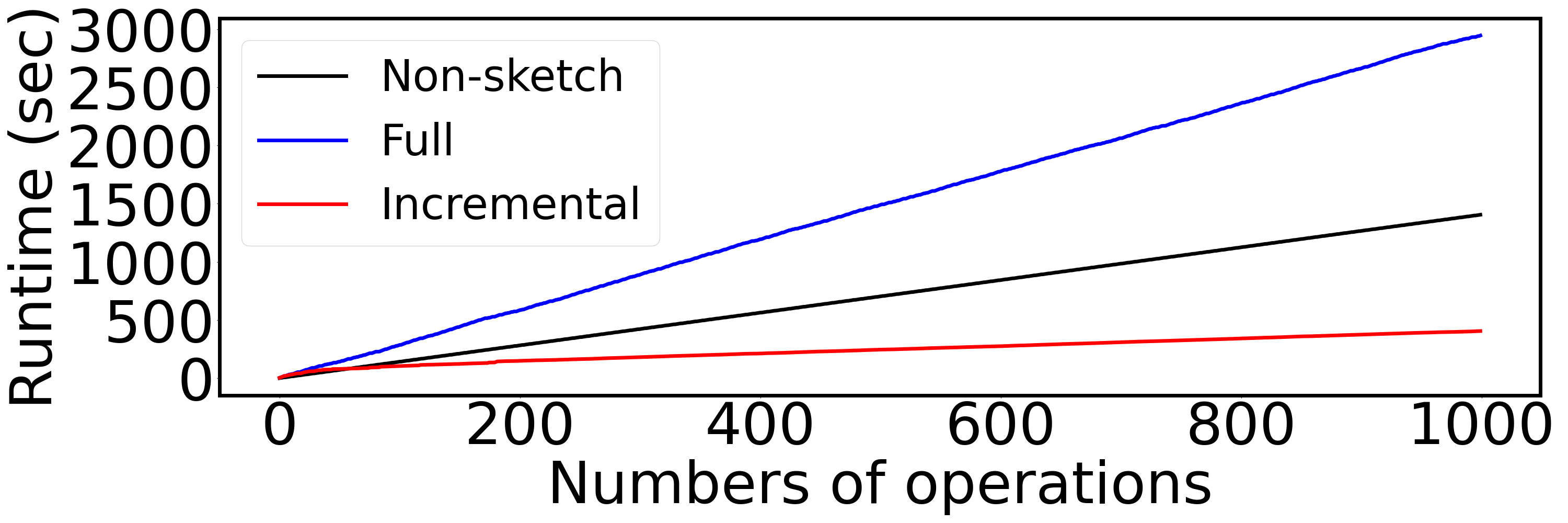}
      \vspace{-6mm}
      \caption{1U5Q, delta size: 20 tuples}
      \label{fig:endtoend_t20_1u5q}
    \end{subfigure}
    %% 1U 1Q
    \begin{subfigure}{0.33\linewidth}
      \includegraphics[width = 1.0\textwidth]{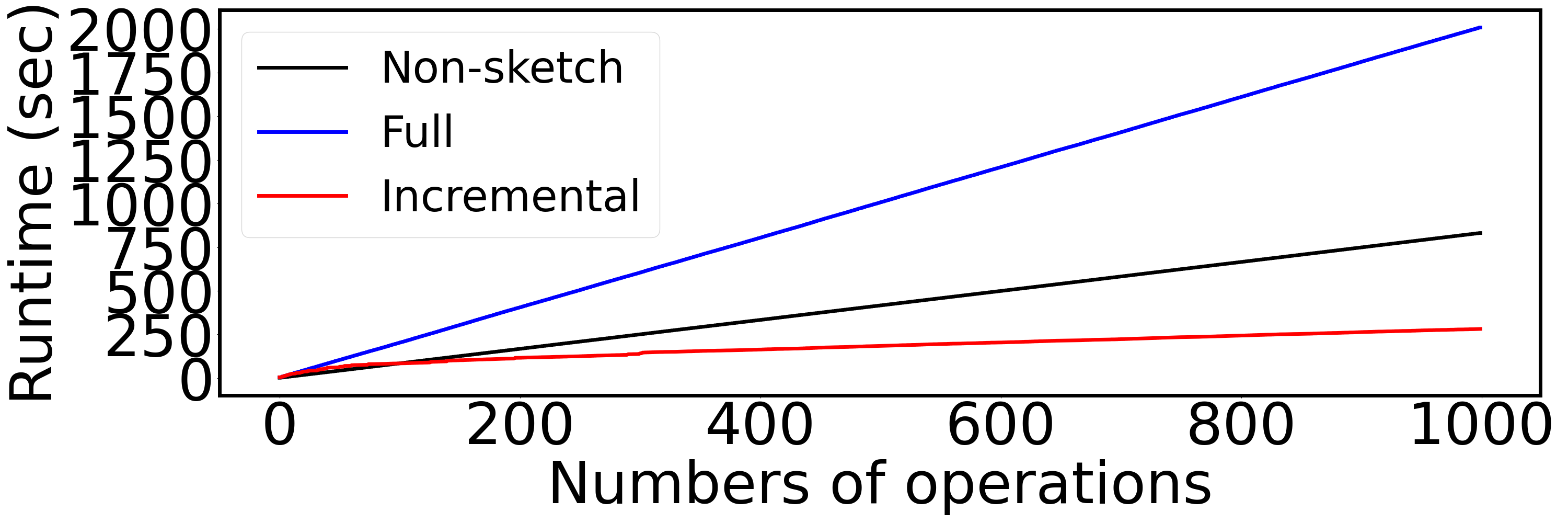}
      \vspace{-6mm}
      \caption{1U1Q, delta size: 20 tuples}
      \label{fig:endtoend_t20_1u1q}
    \end{subfigure}
    %% 5U 1Q
    \begin{subfigure}{0.33\linewidth}
      \includegraphics[width = 1.0\textwidth]{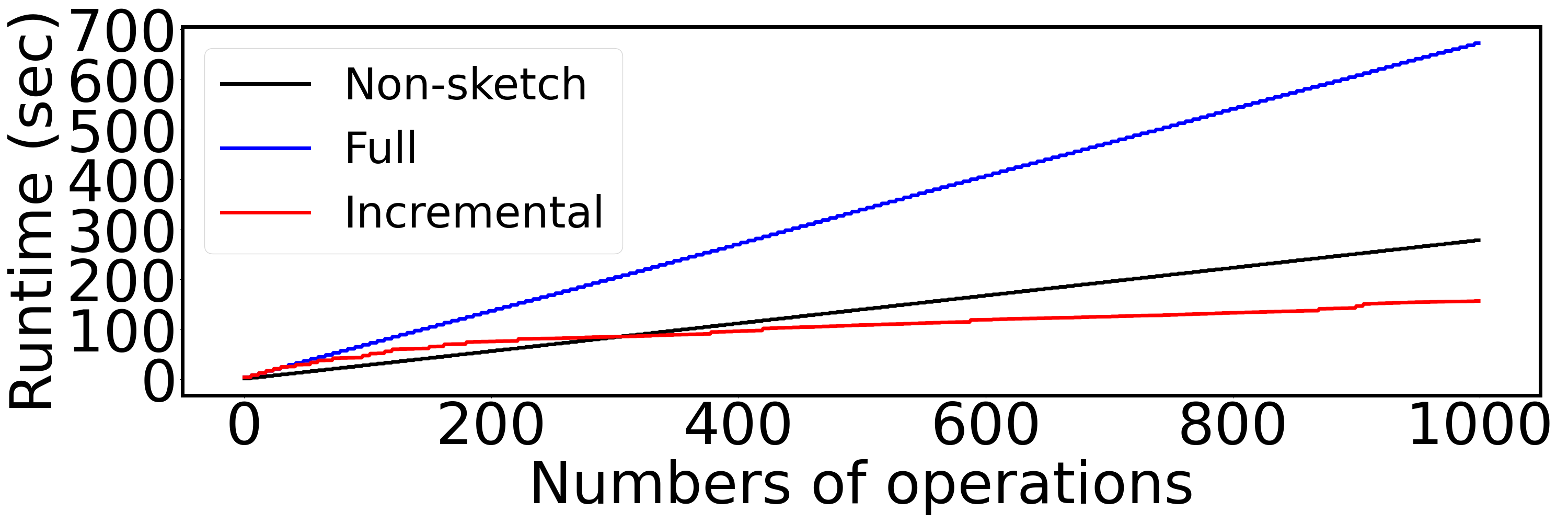}
      \vspace{-6mm}
      \caption{5U1Q, delta size: 20 tuples}
      \label{fig:endtoend_t20_5u1q}
    \end{subfigure}
  \end{minipage}
  %% compare diff delta size using 5U1Q
  %% delta tuples: 1 tuple
  \begin{minipage}{1.0\linewidth}
    %% 5U 1Q: 1 tuple
    \begin{subfigure}{0.33\linewidth}
      \includegraphics[width = 1.0\textwidth]{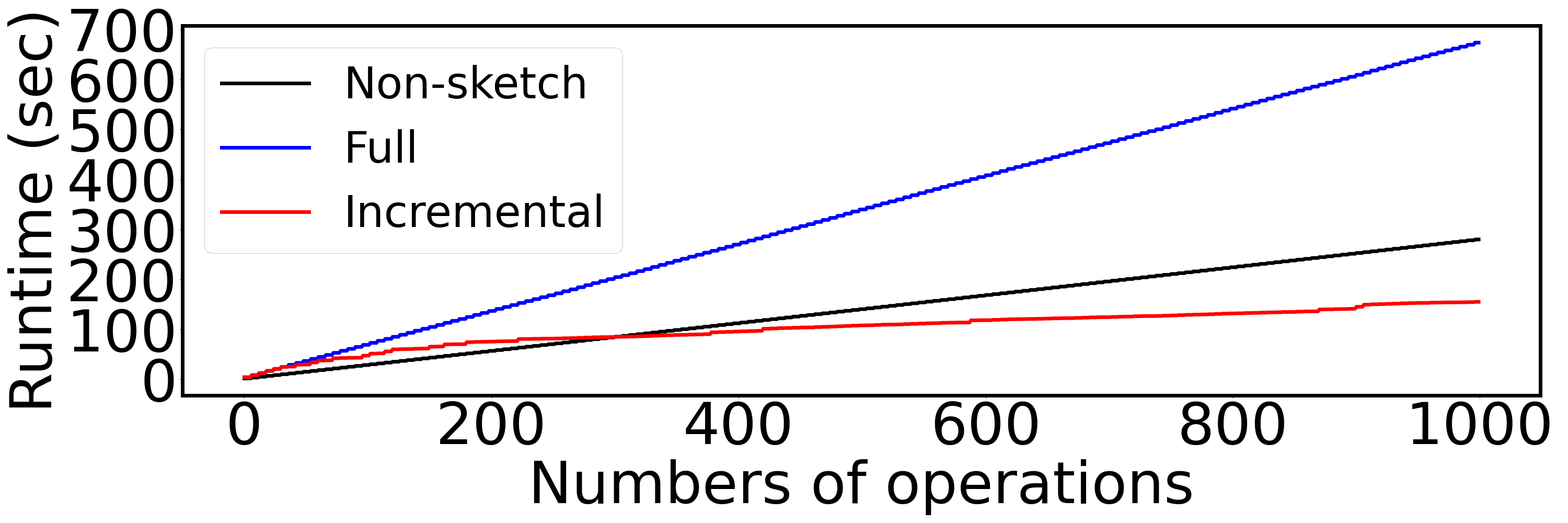}
      \vspace{-6mm}
      \caption{1U5Q, delta size: 1 tuple}
      \label{fig:endtoend_t1_5u1q}
    \end{subfigure}
    %% 5U 1Q: 200 tuples
    \begin{subfigure}{0.33\linewidth}
      \includegraphics[width = 1.0\textwidth]{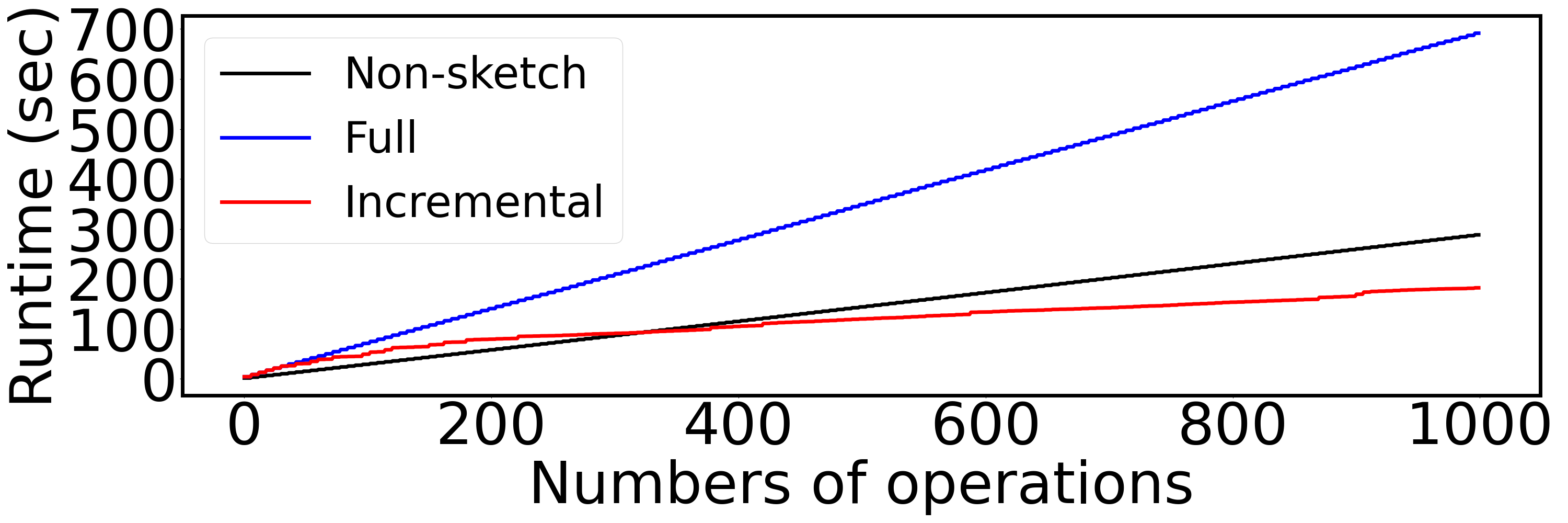}
      \vspace{-6mm}
      \caption{1U1Q, delta size: 200 tuples}
      \label{fig:endtoend_t200_5u1q}
    \end{subfigure}
    %% 5U 1Q
    \begin{subfigure}{0.33\linewidth}
      \includegraphics[width = 1.0\textwidth]{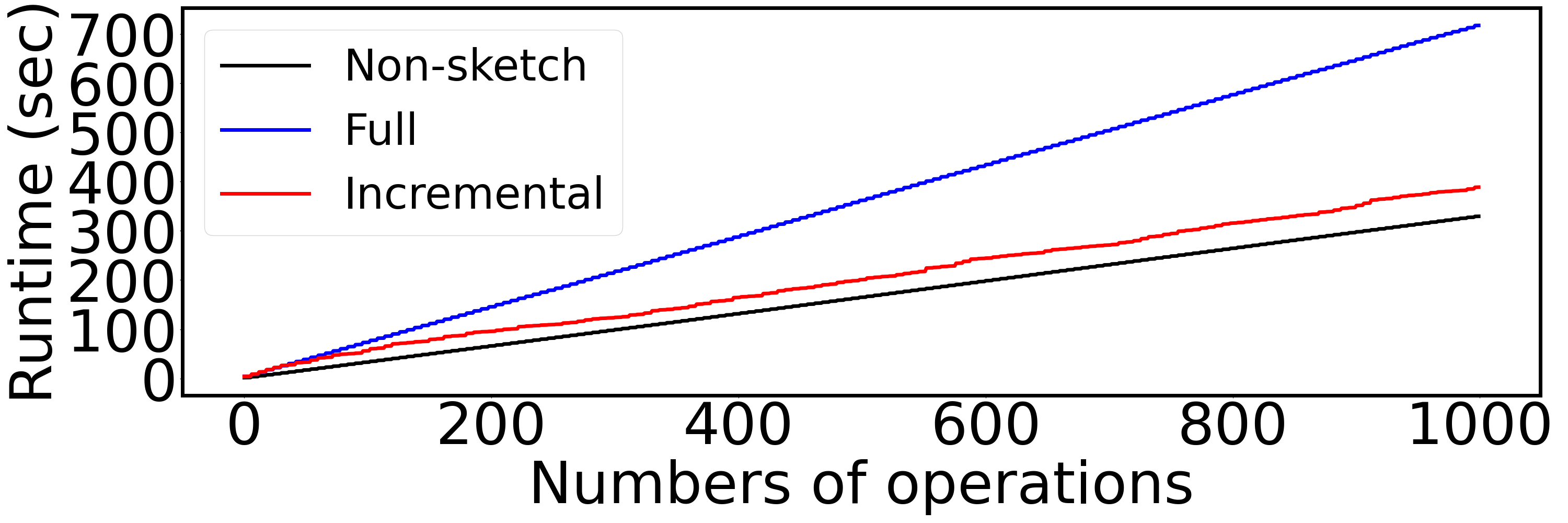}
      \vspace{-6mm}
      \caption{5U1Q, delta size: 2000 tuples}
      \label{fig:endtoend_t2000_5u1q}
    \end{subfigure}
  \end{minipage}
  \vspace{-4mm}
  \caption{Varying the delta size, we measure end-to-end workload runtime for  \nonsketch, \fullmaintenance, and \impAbbr.}
  \label{fig:endtoends}
\end{figure*}
}

%%% Local Variables:
%%% mode: LaTeX
%%% TeX-master: "../imp"
%%% End:

%% file: sections/fig-tpch-crime.tex
%%%%%%%%%%%%%%%%%%%%%%%%%%%%%%%%%%%%%%%%%%%%%%%%%%%%%%%%%%%%%%%%%%%%%%%%%%%%%%%%
%% below figs should be in techreport
\iftechreport{
% tpch figures
\begin{figure*}[h]
  \begin{minipage}{1.0\linewidth}
    \begin{subfigure}{1.0\linewidth}
      \includegraphics[width = 1.0\textwidth]{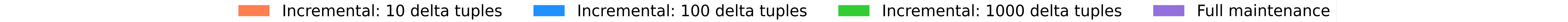}
    \end{subfigure}
  \end{minipage}
  \begin{minipage}{1.0\linewidth}
    \begin{subfigure}{0.33\textwidth}
      \includegraphics[width = 1.0\textwidth]{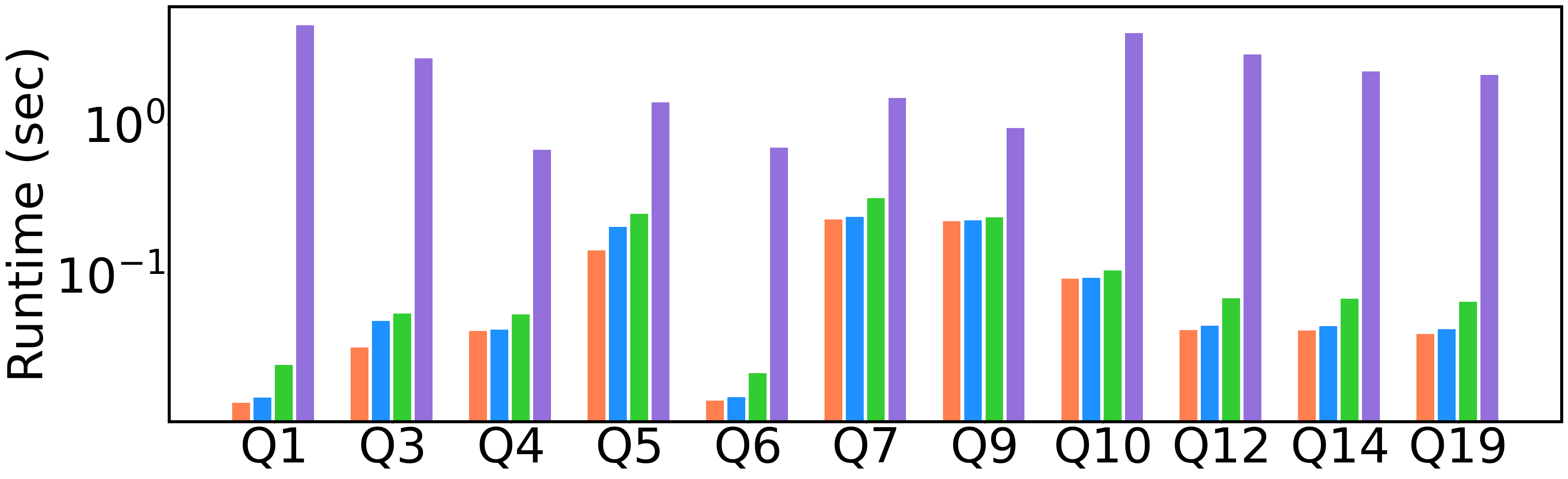}
      \vspace{-6mm}
      \caption{TPC-H 1GB}
      \label{fig:tpch-1g}
    \end{subfigure}
    \begin{subfigure}{0.33\textwidth}
      \includegraphics[width = 1.0\textwidth]{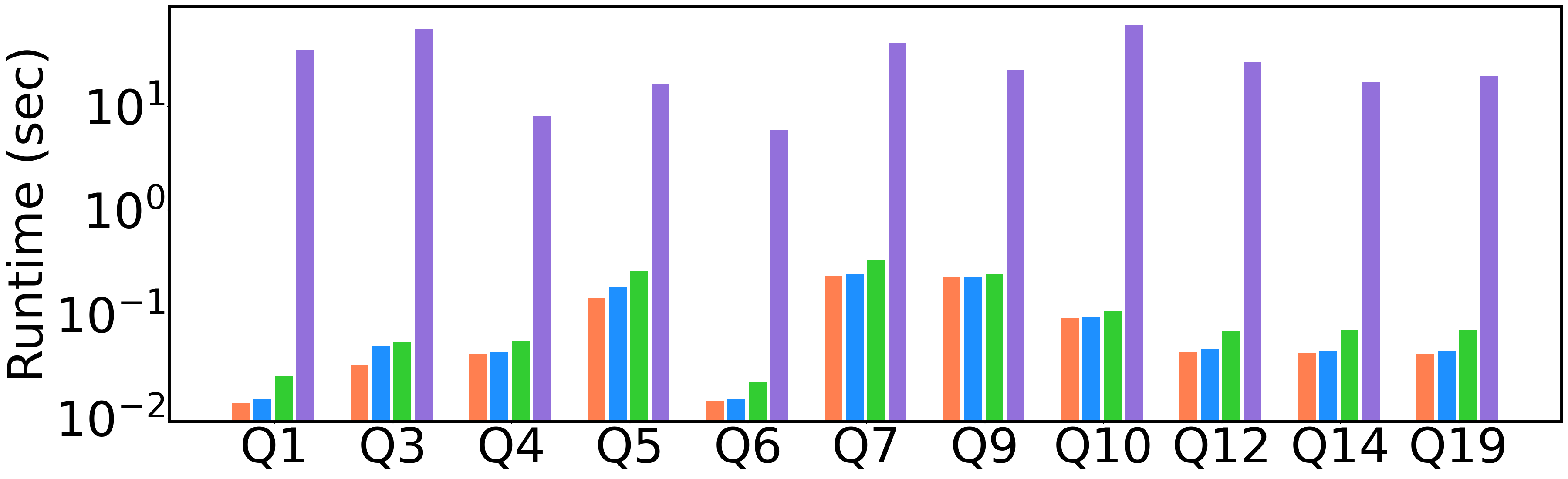}
      \vspace{-6mm}
      \caption{TPC-H 10GB}
      \label{fig:tpch-10g}
    \end{subfigure}
    \begin{subfigure}{0.33\textwidth}
      \includegraphics[width = 1.0\textwidth]{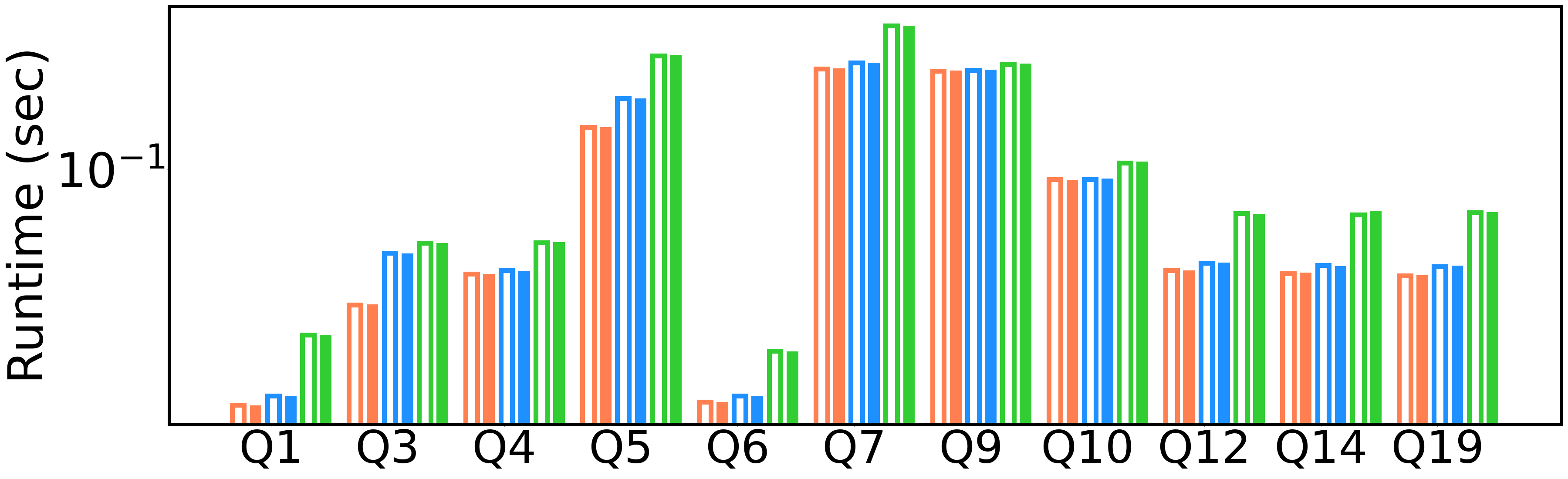}
      \vspace{-6mm}
      \caption{TPC-H 10GB Insert and Delete}
      \label{fig:tpch-10g-ins-del}
    \end{subfigure}
  \end{minipage}
  \vspace{-4mm}
  % \caption{TPC-H Benchmark}
  \caption{Maintaining provenance sketches: incremental versus full maintenance on the \tpchds.}
  \label{fig:tpch}
\end{figure*}
}

%%%%%%%%%%%%%%%%%%%%%%%%%%%%%%%%%%%%%%%%%%%%%%%%%%%%%%%%%%%%%%%%%%%%%%%%%%%%%%%%
%% figs for main paper
%% Since insert vs delete are outsourced to techreport, here we put two tpch figs
%% and one crime fig
\ifnottechreport{
    \begin{figure*}[h]
      \vspace{-4mm}
  \begin{minipage}{1.0\linewidth}
    \begin{subfigure}{1.0\linewidth}
      \includegraphics[width = 1.0\textwidth]{figs/tpchcrime_bar_legends.png}
    \end{subfigure}
  \end{minipage}
  \begin{minipage}{1.0\linewidth}
    \begin{subfigure}{0.395\textwidth}
      \includegraphics[width = 1.0\textwidth]{figs/tpch1g.png}
      \vspace{-6mm}
      \caption{TPC-H 1GB}
      \label{fig:tpch-1g}
    \end{subfigure}
    \begin{subfigure}{0.395\textwidth}
      \includegraphics[width = 1.0\textwidth]{figs/tpch10g.png}
      \vspace{-6mm}
      \caption{TPC-H 10GB}
      \label{fig:tpch-10g}
    \end{subfigure}
    \begin{subfigure}{0.20\textwidth}
      \includegraphics[width = 1.0\textwidth]{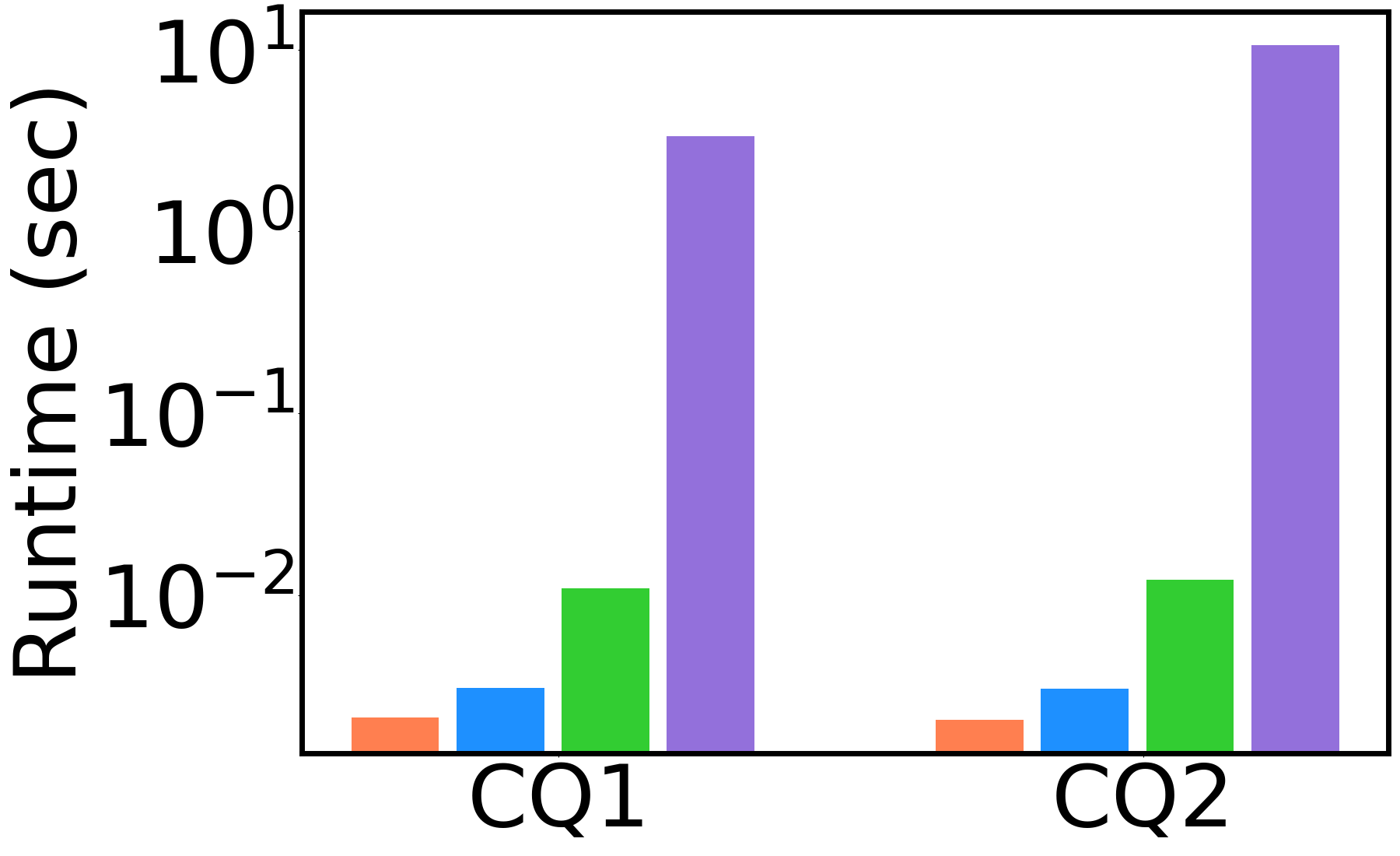}
      \vspace{-6mm}
      \caption{Crimes}
      \label{fig:crimes-overview}
    \end{subfigure}
  \end{minipage}
  \vspace{-4mm}
  \caption{Maintaining provenance sketches: incremental versus full maintenance on the \tpchds and \crimeds dataset.}
  \label{fig:tpch}
\end{figure*}
}
%%%%%%%%%%%%%%%%%%%%%%%%%%%%%%%%%%%%%%%%%%%%%%%%%%%%%%%%%%%%%%%%%%%%%%%%%%%%%%%%
%% figs for techreport
\iftechreport{
\begin{figure}[h]
  \centering
  \begin{minipage}{1\linewidth}
    \begin{subfigure}{0.49\textwidth}
      \includegraphics[width = 1.0\textwidth]{figs/crimes.png}
      \vspace{-6mm}
      \caption{Crimes}
      \label{fig:crimes-overview}
    \end{subfigure}
    \begin{subfigure}{0.49\textwidth}
      \includegraphics[width = 1.0\textwidth]{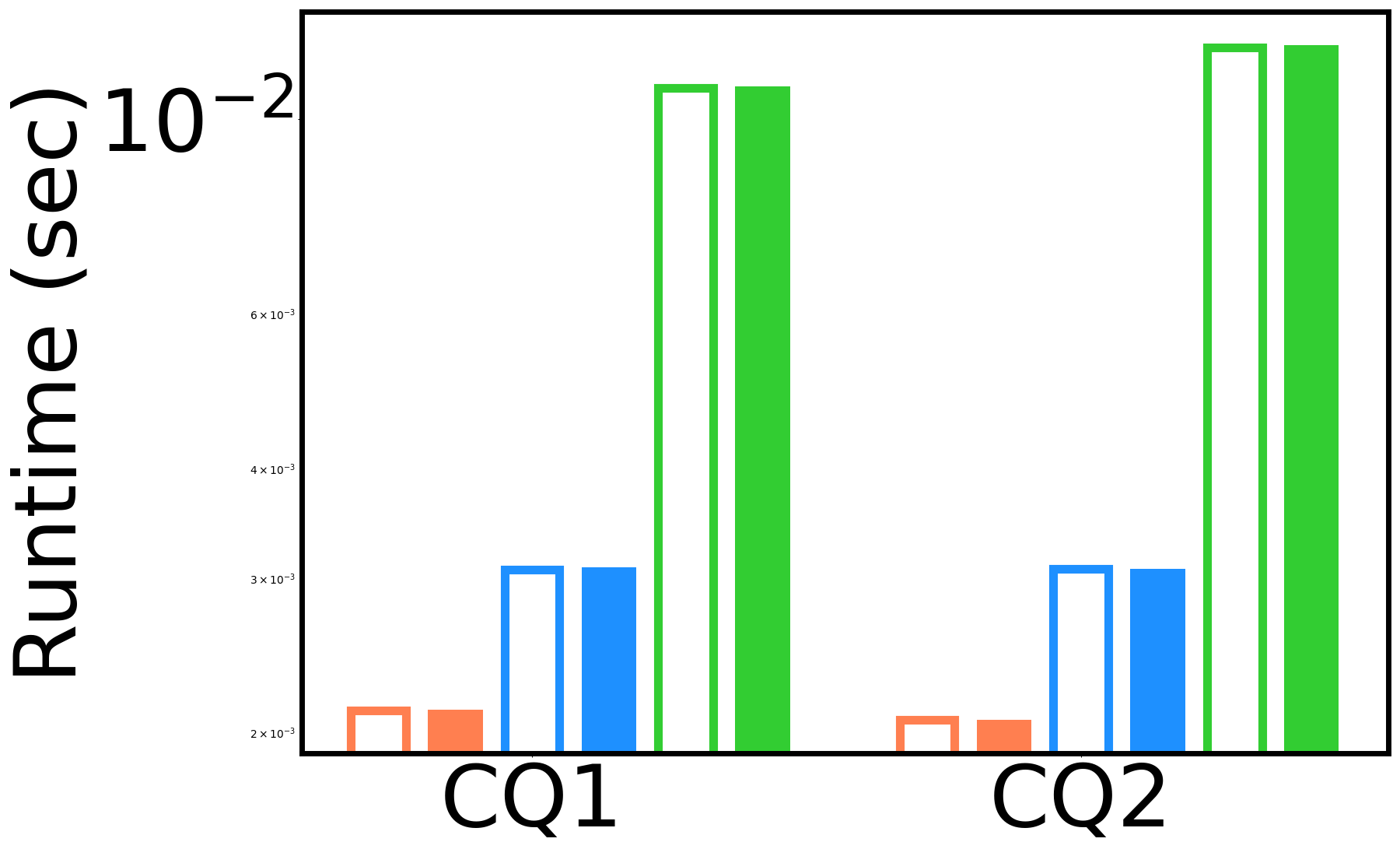}
      \vspace{-6mm}
      \caption{Insertion and deletion}
      \label{fig:crimes-ins-del}
    \end{subfigure}
  \end{minipage}
  \vspace{-4mm}
  % \caption{Crimes}
  \caption{Maintaining provenance sketches: incremental versus full maintenance on the \crimeds dataset.}
  \label{fig:crimes}
\end{figure}
}

%%% Local Variables:
%%% mode: LaTeX
%%% TeX-master: "../imp"
%%% End:

%% file: sections/experiments.tex
\newcommand{\numfrag}{\#frag\xspace}

%\PL{NEW VERSION of Experiments}
\section{Experiments}%
\label{sec:experiments}
We evaluate \impAbbr against two baselines: \emph{full maintenance} and an approach that does not use \gls{pbds} \ifnottechreport{to demonstrate the effectiveness of \impAbbr to improve the performance of workloads that mix queries and updates.}
\iftechreport{to answer the following questions: Q1: Does \gls{pbds} with \impAbbr improve the performance of workloads that mix queries and updates? Q2: How does incremental maintenance with \impAbbr compare against full maintenance? Q3: Which parameters affect \impAbbr's performance? Q4: How effective are our optimizations? Q5: How do eager and lazy maintenance compare?}
% conducted to 1. compare our incremental
% maintenance against full maintenance (capturing provenance sketches from
% scratch); 2. study how \impAbbr's
% runtime is affected by the characteristics of a workload such as the database
% size, size of the delta, and query structure; and 3. evaluate the effectiveness of the proposed optimizations. \PLI{do we need to say something about lazy and eager?}
%
All experiments use a machine with 2 x 3.3Ghz AMD Opteron 4238 CPUs (12
cores) and 128GB RAM running Ubuntu 20.04 (linux kernel 5.4.0-96-generic) and Postgres 16.2.
\revision{\impAbbr's source code and the experimental setup are available at \cite{availability}.}
% The
% database backend for all experiments was Postgres 16.2.
Experiments are repeated at least 10 times. We report median runtimes. The maximum variance
was $<5\%$ and $<1\%$ in most experiments. % less than $5\%$ in extreme cases and less than 1\% for most experiments.

%%%%%%%%%%%%%%%%%%%%%%%%%%%%%%%%%%%%%%%%%%%%%%%%%%%%%%%%%%%%%%%%%%%%%%%%%%%%%%%%
\parttitle{Datasets and Workloads}
\label{subsec:exp_datasets_workloads}
We use the \emph{\tpchds} benchmark, a real world \emph{\crimeds} dataset and
 synthetic data (tables with 10M rows with at least 11
attributes). Each synthetic table has a key attribute
\texttt{id}. For the other attributes, the values of one attribute ($\att$) are chosen uniform at random.
% by sampling uniformly from a given domain up to $500K$.
The remaining attributes are linearly correlated with $a$ subject to Gaussian noise to
create partially correlated values.
% , each with a scale factor to $\att$, which means
% that attributes $b$ to $j$ are linearly correlated with the attribute $\att$,a
% unique constant multiplier (2 to 8) with a small amount of normally distributed
% noise adding.
%% \BG{What correlation?} % \BG{Maybe only mention in the specific experiments: We vary the following characteristics of the
% generated data: (i) the number of distinct values per for attribute $\att$; (ii)
% % number of groups for group-by queries, b. queries of M-N joins. 2.
% for join
% selectivity experiments, we also control the distrib.}

%%%%%%%%%%%%%%%%%%%%%%%%%%%%%%%%%%%%%%%%%%%%%%%%%%%%%%%%%%%%%%%%%%%%%%%%%%%%%%%%
%%%%%%%%%%%%%%%%%%%%%%%%%%%%%%%%%%%%%%%%%%%%%%%%%%%%%%%%%%%%%%%%%%%%%%%%%%%%%%%%
\subsection{Mixed Workload Performance}
\label{sec:workload-performance}

In this experiment, we measure the end-to-end runtime of \emph{\impAbbr}, full maintenance (\emph{\fullmaintenance}), and non-sketch (\emph{\nonsketch})  on mixed workloads consisting of queries and updates. Both \impAbbr and \fullmaintenance start without any sketches. The cost of maintaining and creating sketches is included in the runtime.
% In this
% In this experiment, we evaluate our approach on mixed workloads consisting of
% queries and updates and measure the end-to-end runtime. \PLI{add: state endtoend
%   will see how IMP behaves for whole workload, if can lead to overall speedup.} This can examine how well our incremental approach performs for a whole workload. We can learn the knowledge whether the \impAbbr engine can lead to an overall speed up for query load.
Each workload consists of
 1000 operations \revision{(each operation is either a query or an update)}. We refer to the ratio between queries and updates (the \emph{query-update ratio}). %  and compare against full maintenance and an approach
% that does not use provenance sketches (non-sketch).
% There are two baselines in this experiment: queries execution without using
% provenance sketches called non-sketch approach and queries using sketches but
% capturing provenance sketches from scratch whenever they are needed called
% full-use approach. And our approach called incremental-use approach where the
% provenance sketches are maintained incrementally using \impAbbr engine.
\revision{We use the
technique from \cite{DBLP:journals/pvldb/NiuGLLGKLP21} to determine whether an
existing sketch for a query $\query'$ can be used to answer the current query
$\query$.}
% sketches for similar queries with the same structure at different parameter
% settings, which allow us to use provenance sketches of $\query_1$ to answer
% $\query_2$.
If an existing sketch can be reused, we maintain the sketch if necessary.
Otherwise, we create a new sketch.
% In this experiment, we apply this technique: using existing
% sketches to answer queries whenever it is applicable, while when there is no
% existing sketches that we can use to answer a query, the provenance sketch of
% this query will be captured.
When an update on relation $R$ is executed, we determine the delta and append it to the delta
table for $R$, % The delta tables cache deltas.
 associating each tuple with a snapshot identifier. This enables us to
fetch only delta tuples of updates that were executed after the sketch was last
maintained.
% version of database and the current version of database.

% \PL{remove figs 1U5Q, and the corresponding content}
\ifnottechreport{
We use a query template \qendtoend (see
\cite{techreport}) which is a
group-by-aggregation-having query over the synthetic data and delta
sizes 1, 20, 200 and 2000. We consider three query-update ratios: 1U1Q (one update
per one query), 1U5Q (one update per five queries) and 5U1Q (five updates per
one query).
For full maintenance, we always recapture sketches to answer the queries when
encountering queries. % For incremental maintenance, we capture the provenance
% sketches the first time a queri are encountered.
\Cref{fig:endtoends}  shows the runtime for several
combinations of
query-update ratio and delta size (\revision{The x-axis indicates the
    total number of operations executed so far}). We present additional combinations
in~\cite{techreport}. \fullmaintenance
has the highest cost: the cost of recapturing sketches
frequently outweighs the benefit of using sketches. % to speed-up queries.
% because each time we need to use sketches to answer queries, we have
% to generate the sketches from scratch since the database updates (there is no
% need generating the provenance sketches for queries that can reuse sketches
% already existing between two updates). In general, non-sketch approach can
% outperform the full-use approach while it is worse then incremental-use
% approach.
\impAbbr outperforms both baselines, except for the extreme case 5U1Q with delta
size 2000 (per update) where 5 updates affecting at least 10k tuples in total are executed between two adjacent queries.
% , each affecting at least 2000 tuples are
% executed between two adjacent queries leading to a minimum 10k tuples for each maintenance run as adjacent queries may not use
% the same sketch.
In the first part of each workload, the cost of creating provenance sketches outweighs the benefits of sketch use. However, once a sufficient set of sketches is available, \gls{pbds} outperforms the \nonsketch baseline.
}
% most of the operations in the workload are updates and when we maintain
% a sketch, each incremental maintenance will process more delta accumulated by
% update operations between two maintenance. The more delta the incremental
% procedure needs to process, the higher cost to complete.
%%%%%%%%%%%%%%%%%%%%%%%%%%%%%%%%%%%%%%%%
\iftechreport{
For this experiment, we utilize a query template \qendtoend (see \Cref{sec:appendix_qlist}) which is a
group-by-aggregation-having query evaluated over the synthetic dataset and delta
sizes 1, 20, 200 and 2000. We consider three query-update ratios: \emph{1U1Q} (one update
per one query), \emph{1U5Q} (one update per five queries) and \emph{5U1Q} (five updates per
one query).
% For \fullmaintenance, we always recapture sketches to answer the queries when
% encountering queries.
% For incremental maintenance, we capture the provenance
% sketches the first time a queri are encountered.
\Cref{fig:endtoends} shows the runtime for several combinations of
query-update ratio and delta size (for this experiment the number of tuples affected by each update in the workload). \fullmaintenance
has the highest cost, as the cost incurred by recapturing sketches
 frequently outweighs the benefit of using these sketches.
% because each time we need to use sketches to answer queries, we have
% to generate the sketches from scratch since the database updates (there is no
% need generating the provenance sketches for queries that can reuse sketches
% already existing between two updates). In general, non-sketch approach can
% outperform the full-use approach while it is worse then incremental-use
% approach.
\impAbbr outperforms both baselines, except for the extreme case 5U1Q with delta
size 2000 (per update) where 5 updates, in total affecting at least 10000 tuples are
executed between two adjacent queries, since the adjacent queries may not use
the same sketch.
}
%%%%%%%%%%%%%%%%%%%%%%%%%%%%%%%%%%%%%%%%
\begin{takeaway}
 \impAbbr\, significantly improves the performance of mixed workloads using \gls{pbds}.
\end{takeaway}
%%%%%%%%%%%%%%%%%%%%%%%%%%%%%%%%%%%%%%%%

%%%%%%%%%%%%%%%%%%%%%%%%%%%%%%%%%%%%%%%%%%%%%%%%%%%%%%%%%%%%%%%%%%%%%%%%%%%%%%%%
\subsection{Incremental Versus Full maintenance}

We now % evaluate the cost of maintenance in more detail,
compare \impAbbr against \fullmaintenance.
% \parttitle{Baselines}
% %
% For sketch maintenance, we evaluate the performance of \impAbbr
% (\emph{incremental maintenance}) against a baseline where sketches are captured
% from scratch (\emph{full maintenance}). Furthermore, in some experiments we
% compare against a baseline (\emph{non-sketch}) that does not use
% provenance-based data skipping.%%%%%%%%%%%%%%%%%%%%%%%%%%%%%%%%%%%%%%%%%%%%%%%%%%%%%%%%%%%%%%%%%%%%%%%%%%%%%%%%
%\parttitle{Parameter}
%
We vary the \emph{delta size} % (the number of tuples that are inserted or
% deleted) in all experiments
focusing on realistic delta sizes: 10, 50,
100, 500 and 1000.

%%%%%%%%%%%%%%%%%%%%%%%%%%%%%%%%%%%%%%%%%%%%%%%%%%%%%%%%%%%%%%%%%%%%%%%%%%%%%%%%
\subsubsection{\tpchds}\label{sec:experiments-tpch}
%%%%%%%%%%%%%%%%%%%%%%%%%%%%%%%%%%%%%%%%%%%%%%%%%%%%%%%%%%%%%%%%%%%%%%%%%%%%%%%%
%
The results for \tpchds~(\url{www.https://www.tpc.org/tpch/}) at SF1 ($\sim
1$ GB) and SF10 ($\sim 10$GB) are shown in
\Cref{fig:tpch-1g,fig:tpch-10g}. \revision{We selected queries that benefit from
 sketches~\cite{DBLP:journals/pvldb/NiuGLLGKLP21} and are sufficiently
complex (multiple joins, aggregation with \lstinline!HAVING! or top-k).}
% properties: 1. Provenance sketches are beneficial to use to answer queries
% \cite{DBLP:journals/pvldb/NiuGLLGKLP21}; 2. Queries have complex structures:
% multi-way joins, large number of aggregation functions and top-k operators.
We turn on the selection push down and join bloom filter
optimizations (see \Cref{sec:optimization}). % \Cref{fig:tpch-1g} and
% \Cref{fig:tpch-10g} show the runtime of incremental vs. full maintenance for SF1
% and SF10, respectively.
The runtime of \fullmaintenance only depends on the current size of the database. Thus,
we do not include results for different delta sizes for this method. \impAbbr outperforms \fullmaintenance by at least a factor of 3.9 and up to a factor of $\sim$2497,
demonstrating the effectiveness of incremental maintenance. Importantly, the
runtime of \impAbbr, while depending on delta size, is mostly unaffected database size as join is the only incremental operator accessing the database.
% by the size of the database as the
% runtime of most incremental operators only depends on de size of the delta.
% In general, the joins in our \impAbbr engine are more expensive. As our engine
% uses the database to execute $\Delta R \join S$, incrementally maintaining joins
% requires data transfer for all delta. The more join operators a query contains,
% the higher cost for \impAbbr engine to complete the incremental procedure.
% The \tpchds $\query$5, $\query$7 and $\query$9 have the least benefit because
% all these three queries contain 6-way join. While $\query$4, $\query$12,
% $\query$14, $\query$19 have a relative better performance than $\query$5,
% $\query$7 and $\query$9 since these queries only contains 2-way joins.
% $\query 1$ and $\query 6$ have the best performance because these two queries
% only work on a single table without join.

\iftechreport{
\Cref{fig:tpch-10g-ins-del} shows the incremental maintenance runtime for both
insertion and deletion for certain amount of delta for 10 GB database size.
}
%%%%%%%%%%%%%%%%%%%%%%%%%%%%%%%%%%%%%%%%%%%%%%%%%%%%%%%%%%%%%%%%%%%%%%%%%%%%%%%%

%%%%%%%%%%%%%%%%%%%%%%%%%%%%%%%%%%%%%%%%%%%%%%%%%%%%%%%%%%%%%%%%%%%%%%%%%%%%%%%%
\subsubsection{\crimeds Dataset}\label{sec:experiments-crime}
%%%%%%%%%%%%%%%%%%%%%%%%%%%%%%%%%%%%%%%%%%%%%%%%%%%%%%%%%%%%%%%%%%%%%%%%%%%%%%%%
\ifnottechreport {
The
Crime\footnote{\url{https://data.cityofchicago.org/Public-Safety/Crimes-2001-to-Present/ijzp-q8t2}}
dataset records consists of a single 1.87GB
table with $~7.3M$ incidents in Chicago. We use two queries (see
 \cite{techreport}): \texttt{CQ1}: The numbers crime of each year in
each beat (geographical location). \texttt{CQ2}: Areas with more than 1000
crimes.  As shown in \Cref{fig:crimes-overview}, \impAbbr
outperforms \fullmaintenance by at least 2 \gls{oom}.
}
\iftechreport{
The
Crime\footnote{\url{https://data.cityofchicago.org/Public-Safety/Crimes-2001-to-Present/ijzp-q8t2}}
dataset records incidents in Chicago and consists of a single 1.87GB
table with $~7.3M$ rows. We use two queries (SQL code for all queries is shown
in \Cref{sec:appendix_qlist}): \texttt{CQ1}: The numbers crime of each year in
each beat (geographical location). \texttt{CQ2}: Areas with more than 1000
crimes. We use realistic delta sizes (10 to 1000). As shown in \Cref{fig:crimes-overview} incremental maintenance
outperform full maintenance by at least 2 \gls{oom}.
\Cref{fig:crimes-ins-del} show the incremental maintenance runtime for
both insertion and deletion under given delta size.
}

%%%%%%%%%%%%%%%%%%%%%%%%%%%%%%%%%%%%%%%%%%%%%%%%%%%%%%%%%%%% %%%%%%%%%%%%%%%%%%%%%
% \subsection{Mixed Workloads}
% \label{sec:end-end-experiments}
% %

%%%%%%%%%%%%%%%%%%%%%%%%%%%%%%%%%%%%%%%%
\begin{takeaway}
For deltas up to 1000 tuples, \impAbbr outperforms \fullmaintenance by several \gls{oom}.
\end{takeaway}
%%%%%%%%%%%%%%%%%%%%%%%%%%%%%%%%%%%%%%%%
% \PL{Takeaways of tpch, crime and endowment}
% \subsubsection{Insights : maybe need to rename}
% \label{sec:insights-tpch-crime-endtoend}
% \textbf{For \tpchds queries, as joins require round trips to the database and
%   evaluation of queries, the number of joins affects the performance of
%   \impAbbr. Nonetheless, even for \tpchds $\query$5, $\query$7 and $\query$9
%   that contain 6-way joins, \impAbbr still outperforms full maintenance by about
%   2 \gls{oom} for SF10. From mixed workloads experiment, \impAbbr
%   outperforms both baselines, except for the extreme case 5U1Q with delta size
%   2000 (per update) where 5 updates, each affecting at least 10000 tuples are
%   executed between two adjacent queries, since the adjacent queries may not use
%   the same sketch. The end-to-end experiment presents that for most realistic
%   delta size updates, incremental maintenance is more cost-effective than full
%   maintenance and using updated sketches to answer queries greatly decreases
%   overall workload runtime. }

%%%%%%%%%%%%%%%%%%%%%%%%%%%%%%%%%%%%%%%%%%%%%%%%%%%%%%%%%%%%%%%%%%%%%%%%%%%%%%%%
\subsection{Microbenchmarks}\label{sec:exp-microbenchmarks}
Next, we evaluate in detail how \impAbbr's performance is affected by various workload parameters using the synthetic dataset.
% use synthetically generated data and vary
% specific aspects of queries to evaluate in detail how our approach behaves.
% For incremental maintenance, we examine realistic delta sizes for both
% insertions and deletions
% (\Cref{fig:micro-small-deltas}) and use large deltas to determine the ``break
% even'' point: for delta sizes larger than that, full maintenance outperforms
% incremental maintenance.
%\parttitle{Baselines}
%
We compare \impAbbr against \fullmaintenance.
% as a baseline varying the delta size
% .
The database size is kept constant, i.e.,
for \fullmaintenance, the runtime is not affected by varying the delta size.

%%%%%%%%%%%%%%%%%%%%%%%%%%%%%%%%%%%%%%%%%%%%%%%%%%%%%%%%%%%%%%%%%%%%%%%%%%%%%%%%
% \parttitle{Parameter}
% %
% We vary the \emph{delta size}: (i) mostly we use realistic
% delta sizes: 10, 50, 100, 500 and 1000; (ii) for determining the break-even
% point between incremental and full maintenance we vary the  the delta size from
% $0.1\%$ of the table to $10\%$ in $0.1\%$ increments.
%

% control the
% properties of the queries. The experiments over synthetic data are in the
% following categories: 1. different numbers of aggregation functions, 2. fixed
% number of  aggregation functions with different number of groups, 3. $M-N$
% joins, 4. joins with different selectivity, 5. different provenance sketch
% sizes.
% We assume that both incremental and full maintenance happen
% after the database updates. Under this scenario,
% For
% incremental maintenance, the time cost will vary for different delta
% sizes(delta is the database change between previous version and current
% version).

% and large delta size: $0.1\%$ of original table size to
% $10\%$ (categories 1 and 2) or $1.5\%$ (categories 3, 4, 5, and 6) of original
% table size with $0.1\%$ of each incremental.

%%%%%%%%%%%%%%%%%%%%%%%%%%%%%%%%%%%%%%%%%%%%%%%%%%%%%%%%%%%%%%%%%%%%%%%%%%%%%%%%
%% [trim={left bottom right top},clip]
%% synthetic figs

%%
%% for fig crop;  inside trim, bottom set 0, and add a \vspce{-7mm} can reduce
%% the sace between caption and bottom of figure. if set some value to bottom
%% inside trim, it will not help too much.
%%

\input{sections/fig-microbenchmarks.tex}

%%%%%%%%%%%%%%%%%%%%%%%%%%%%%%%%%%%%%%%%%%%%%%%%%%%%%%%%%%%%%%%%%%%%%%%%%%%%%%%%
% \parttitle{Number of groups}
\ifnottechreport{
\subsubsection{Aggregation Functions and Groups}
\label{sec:experiment_micro_numofgroups}
}
\iftechreport{
\subsubsection{Number of groups}
\label{sec:experiment_micro_numofgroups}
}
\ifnottechreport{
We use query template \qgroups (see \cite{techreport}) that is a group-by aggregation query with \lstinline!HAVING! over a single table and
 vary the
number of groups: $50$, $1K$, $5K$ and $500K$.
As the state $\statedata$ for
% structures maintained by our approach for
aggregation contains an entry for each group, % and the number of groups affected by a delta also depends on the number of groups,
we expect that runtime will increase when increasing the number of groups. As shown in
\Cref{fig:syn-small-fixagg-diffgroups}, for delta sizes up to 1000 tuples,
\impAbbr outperforms \fullmaintenance by 2 (500k groups) to 3 (50 groups) \gls{oom}. \Cref{fig:syn-large-fixagg-diffgroups} shows that the break even
point %(where \fullmaintenance starts to outperform incremental maintenance)
lies at delta sizes between $\sim$3.5\% (for 50 group) and $\sim 5.5\%$ for
(500k groups).
While the runtime of \impAbbr increases when increasing the number of groups,
% incremental maintenance increases when
% increasing the number of groups,
the effect is more pronounced for \fullmaintenance that calculates results for all groups.}
\iftechreport{
 We use a query \qgroups (SQL code shown in \Cref{sec:appendix_qlist}) that is group-by aggregation with \lstinline!HAVING! over a single table and
 vary the
number of groups: $50$, $1K$, $5K$ and $500K$.
As the data structures maintained by our approach for aggregation stores an
entry for each group and the number of groups affected by a delta also
depends on the number of groups, we expect that runtime will increase when
increasing the number of groups. As shown in
\Cref{fig:syn-small-fixagg-diffgroups}, for delta sizes up to 1000 tuples,
incremental maintenance outperforms \fullmaintenance by 2 (500k groups) to 3 (50 groups) \gls{oom}. \Cref{fig:syn-large-fixagg-diffgroups} shows that the break even
point (where \fullmaintenance starts to outperform incremental maintenance)
lies at delta sizes between $\sim$3.5\% (for 50 group) and $\sim 5.5\%$ for
(500k groups). While the runtime of \impAbbr increases when increasing the number of groups,
% incremental maintenance increases when
% increasing the number of groups,
the effect is more pronounced for \fullmaintenance that calculates results for all groups.
}
\ifnottechreport{
In \Cref{fig:syn-small-aggnum} and \Cref{fig:syn-large-aggnum}
show the runtime of \impAbbr vs. \fullmaintenance for a group-by having query with 5k groups, varying the number of aggregation functions.
% the runtime of \impAbbr is linear in
% the size of the delta and in the number of aggregation functions. For realistic delta sizes,
 \impAbbr outperforms \fullmaintenance by up to $\sim 2$ \gls{oom} for delta sizes  up to $\sim 5\%$ of the database.
}

\iftechreport{
%%%%%%%%%%%%%%%%%%%%%%%%%%%%%%%%%%%%%%%%%%%%%%%%%%%%%%%%%%%%%%%%%%%%%%%%%%%%%%%%
%\parttitle{Number of aggregation functions}
\subsubsection{Number of aggregation functions}
\label{sec:experiment_micro_numofaggs}
In this experiment, we use a query \qhaving (see \Cref{sec:appendix_qlist}) which is an
group-by aggregation on a single table with filtering in the \lstinline!HAVING!
clause on the aggregation function result. The total number of groups is set to
5000 in this experiment. We vary the number of aggregation functions used in the
\lstinline!HAVING! condition as our approach has to maintain the results for
these aggregation. \Cref{fig:syn-small-aggnum} and \Cref{fig:syn-large-aggnum}
show the runtime of incremental vs. \fullmaintenance using both realistic delta
sizes and large deltas. The runtime of \impAbbr is linear in
the size of the delta for this query with a coefficient that
depends on the number of aggregation functions. For realistic delta sizes,
\impAbbr outperforms \fullmaintenance by $\sim 2$ \gls{oom}. As shown in \Cref{fig:syn-large-aggnum}, \impAbbr is
faster than \fullmaintenance for deltas of up to $\sim 5\%$ of the database.
}

%%%%%%%%%%%%%%%%%%%%%%%%%%%%%%%%%%%%%%%%%%%%%%%%%%%%%%%%%%%%%%%%%%%%%%%%%%%%%%%%
% \parttitle{Joins}
\subsubsection{Joins}
\label{sec:experiment_micro_join}
\ifnottechreport{
We evaluate  group-by aggregation queries with
\lstinline!HAVING! over the result of an equi-join using query template \qjoin
(see \cite{techreport}). Both input tables have $10M$ rows. The synthetic tables
are designed as the follows: for an $m-n$ join $R \join S$, the selectivity is
$100\%$ for table $S$, and there are $10^8 / n$ distinct join attribute values
with a multiplicity of $n$; for the other table $R$, there are $m$ tuples that
join with each distinct join attribute value in $S$. For instance, the result
size for $2-2k$ as well as for $2-200k$ is $2 \cdot 10M = 20M$ tuples.

\Cref{fig:syn-small-join-1-n} and \Cref{fig:syn-large-join-1-n} show the
runtime of incremental vs. \fullmaintenance for $1-n$ joins, and
\Cref{fig:syn-small-join-m-n} and \Cref{fig:syn-large-join-m-n} show the
performance of both full and incremental approaches for $m-2K$ joins.
In the 1-n join experiment, the 1-20 join is more expensive than the 1-20k and 1-200k joins
because even through the 1-20 join has less join result tuples, there are more
groups for the aggregation functions above the join operator as the join attribute of
table $S$ is also the group-by attribute for the query. There are $10 ^8 / 20$
distinct groups and each group has a multiplicity of 20. For the m-n join, the
queries all have the same number of groups. 50-2k
has more join results to process and, thus, is slower than
20-2k join.

Recall from \Cref{sec:experiments-tpch}, that \impAbbr computes
$\Delta R \join S$ by running a SQL query. Thus, incrementally maintaining joins
requires sending all delta tuples for the join inputs to the DBMS. That is, the
break even point is lower for \qjoin than for \qhaving. Our bloom-filter
optimization for joins can sometimes avoid an additional round trip to the
database for those tuples that do not have the join partner. We further evaluate
this optimization in \Cref{sec:exp-optimizations}.
}

\iftechreport{
We evaluate the performance of group-by aggregation queries with
\lstinline!HAVING! over the result of an equi-join using query template \qjoin
(see \Cref{sec:appendix_qlist}). Both input tables have $10M$ rows. The synthetic tables
are designed as the follows: for an $m-n$ join $R \join S$, the selectivity is
$100\%$ for table $S$, and there are $10^8 / n$ distinct join attribute values
with a multiplicity of $n$; for the other table $R$, there are $m$ tuples that
join with each distinct join attribute value in $S$. For instance, the result
size for $2-2k$ as well as for $2-200k$ is $2 \cdot 10M = 20M$ tuples.
\Cref{fig:syn-small-join-1-n} and \Cref{fig:syn-large-join-1-n} show the
runtime of \impAbbr vs. \fullmaintenance for $1-n$ joins, and
\Cref{fig:syn-small-join-m-n} and \Cref{fig:syn-large-join-m-n} show the
results for $m-2K$ joins.
In the 1-n join experiment, the 1-20 join is more expensive than the 1-20k and 1-200k joins because even through the 1-20 join produces less results, there are more
groups for the aggregation functions above the join operator as the join attribute of
table $S$ the group-by attribute. There are $10 ^8 / 20$
distinct groups and each group has a multiplicity of 20. For the m-n join, the
queries all have the same number of groups. For 50-2k
 more join results have to be produced and, thus, 50-2k is slower than
20-2k.
Recall from \Cref{sec:experiments-tpch}, that \impAbbr computes
$\Delta R \join S$ by running a SQL query. Thus, incrementally maintaining joins
requires sending delta tuples for the join inputs to the DBMS, resulting in a lower break even point for \qjoin than for \qhaving. % Our bloom-filter
% optimization for joins can sometimes avoid an additional round trip to the
% database for those tuples that do not have the join partner. We further evaluate
% this optimization in \Cref{sec:exp-optimizations}.
}
%%%%%%%%%%%%%%%%%%%%%%%%%%%%%%%%%%%%%%%%%%%%%%%%%%%%%%%%%%%%%%%%%%%%%%%%%%%%%%%%
% \parttitle{Join selectivity}
%
\ifnottechreport{
To evaluate performance of queries with more selective joins, we use a
 group-by-aggregation query over a join: \qjoinsel ($R$ join $S$, see
 \cite{techreport}) and vary join selectivity: $1\%$, $5\%$, and
$10\%$. % We control the join attribute values in table $S$ % (this attribute joins
% % with another attribute in table $R$)
% such that only a certain percentage of tuples
% join.
\Cref{fig:syn-small-join-selectivity} and
\Cref{fig:syn-large-join-selectivity} show the runtime of \impAbbr and \fullmaintenance. For small deltas,
the join selectivity has a smaller impact on \impAbbr
than for larger deltas as for small deltas we are joining a small
 table ($\annoDeltaRel$) with a large table ($\mathscr{S}$), i.e., the bottleneck is
scanning the large table.
}

\iftechreport{
\subsubsection{Join selectivity}
\label{sec:experiment_micro_joinselectivity}
    To evaluate performance of queries with more selective joins, we use query
template group-by-aggregation over join: \qjoinsel ($R$ join $S$) (see
 \Cref{sec:appendix_qlist}), to evaluate different join selectivity: $1\%$, $5\%$, and
$10\%$. We control the join attribute values in table $S$ (this attribute joins
with another attribute in table $R$) such that only certain percentage of tuples
join. \Cref{fig:syn-small-join-selectivity} and
\Cref{fig:syn-large-join-selectivity} show the runtime of incremental and full
maintenance using both realistic delta size and large deltas. For small deltas,
the selectivity of the join has less of an impact on \impAbbr
than for large delta sizes. This is due to the fact that for small deltas we are joining a small
 table ($\annoDeltaRel$) with a large table ($\mathscr{S}$), i.e., the bottleneck is
scanning the large table.
}
% Higher
% selectivity means that it is more expensive for \impAbbr engine to complete the
% maintenance procedure. Because for a fixed delta size, the higher the
% selectivity is, the more the tuple in join result are, which increases time for
% the database engine to compute and for data to transfer.

\input{sections/fig-optimizations.tex}

%%%%%%%%%%%%%%%%%%%%%%%%%%%%%%%%%%%%%%%%%%%%%%%%%%%%%%%%%%%%%%%%%%%%%%%%%%%%%%%%
% \parttitle{Varying Partition Granularity}
\subsubsection{Varying Partition Granularity}
\label{sec:experiment_micro_numPar}
\ifnottechreport{
We now vary \numfrag, the number of fragments in the sketch's partition. We use template \qsketch (see
\cite{techreport}) which is a group-by aggregation query with
\lstinline!HAVING! over the results of a join. We vary the number of fragments
\numfrag from $10$ to $5000$. \Cref{fig:syn-small-sketch-size} and
\Cref{fig:syn-large-sketch-size} show the runtime for \impAbbr and \fullmaintenance. While the cost of \fullmaintenance is impacted by \numfrag, the dominating cost is evaluating the
full capture query, resulting in an insignificant runtime increase when \numfrag
is increased. In contrast, incremental maintenance cost increases linearly in
the delta size. % and the cost paid per tuple is roughly linear
% in \numfrag as intermediate sketches of large size have to be processed.
}

\iftechreport{
In this experiment, we vary \numfrag, the number of fragments of the partition
$\dbranges$ from $10$ to $5,000$. We use query template \qsketch (SQL
code shown in \Cref{sec:appendix_qlist}) which is a group-by aggregation query with
\lstinline!HAVING! over the results of a join. % We vary the number of fragments
% \numfrag from $10$ to $5000$.
\Cref{fig:syn-small-sketch-size} and
\Cref{fig:syn-large-sketch-size} show the runtime for \impAbbr and \fullmaintenance. While the cost of \fullmaintenance is impacted by \numfrag, the dominating cost is evaluating the
full capture query, resulting in an insignificant runtime increase when \numfrag
is increased. In contrast, \impAbbr maintenance cost increases linearly in
the  delta size and the cost paid per tuple is roughly linear
in \numfrag. % as intermediate sketches of large size have to be processed.
}
% In this section, we evaluate the performance of our \impAbbr engine to maintain
% sketches for queries having different properties response to the varying of delta
% size. The results present that the \impAbbr outperform full
% maintenance for delta of realistic size.

%%%%%%%%%%%%%%%%%%%%%%%%%%%%%%%%%%%%%%%%%%%%%%%%%%%%%%%%%%%%%%%%%%%%%%%%%%%%%%%%
%
\subsection{Optimizations}\label{sec:exp-optimizations}
% To understand the impact of our optimizations, we use the synthetic dataset
% to evaluate the performance of \impAbbr on both realistic delta sizes and large
% deltas. We only evaluate the performance of the incremental approach.
% for
% filtering delta inputs based on selection conditions and join using bloom
% filters to filter out join attributes' values that do not have join partners in
% our \impAbbr engine.

%%%%%%%%%%%%%%%%%%%%%%%%%%%%%%%%%%%%%%%%%%%%%%%%%%%%%%%%%%%%%%%%%%%%%%%%%%%%%%%%
% \parttitle{Filtering deltas inputs based on selection conditions}
\subsubsection{Selection push-down for deltas}
\label{sec:experiment_opt_selpd}
\ifnottechreport{
 We evaluate the effectiveness of our delta selection
push-down optimization that filters the delta based on selection conditions
in the query. We use query \qselpd which is a group-by aggregation query (see \cite{techreport}) and vary
the selectivity of the query's \lstinline!WHERE! clause. We fix delta size to $2.5\%$ of the table and vary the
fraction of delta tuples that fulfills the condition from $2\%$ to $100\%$.
% that the x-axis shows the
% percentage of data of total size that can pass the filter conditions.
The results (\Cref{fig:syn-large-selpd}) demonstrate % of this experiment. % Note show that the runtime of filtering deltas increases linearly in the
 % for larger selectivities,
that the cost of filtering delta tuples is amortized by reducing maintenance cost and communication with the DBMS.
% Thus, this optimization should be applied whenever possible.
}
\iftechreport{
In this experiment, we evaluate the effectiveness of our delta selection
push-down optimization that pre-filters the delta based on selection conditions
in the query. We use the query \qselpd which is a group-by aggregation query
without joins (SQL code is shown in \Cref{sec:appendix_qlist}). We evaluate the
performance of \impAbbr with and without delta filtering, varying
the selectivity of the query's selection condition (\lstinline!WHERE! clause).
We fix the delta size to $2.5\%$ of the table. Then we gradually increase the
fraction of the delta that fulfills the selection conditions from $2\%$ to $100\%$.
\Cref{fig:syn-large-selpd} shows results of this experiment. % Note that the x-axis shows the
% percentage of data of total size that can pass the filter conditions.
The results show that the runtime of filtering deltas increases linearly in the
selectivity of the query's selection condition. Even for larger selectivities,
the cost of filtering delta tuples is amortized by reducing the cost of
\impAbbr (and communication between \impAbbr and that database).
Thus, this optimization should be applied whenever possible.
}

%%%%%%%%%%%%%%%%%%%%%%%%%%%%%%%%%%%%%%%%%%%%%%%%%%%%%%%%%%%%%%%%%%%%%%%%%%%%%%%%
% \parttitle{Join optimization using bloom filters}
\subsubsection{Join optimization using bloom filters}
\label{sec:experiment_opt_joinBF}
\ifnottechreport{
Another optimization we applied in \impAbbr engine is to use bloom filters
 to track which tuples potentially have join partners. % This enables us
% to avoid evaluating a join $\annoDeltaRel \join \makeAnno{S}$
% ($\annoRel \join \makeAnnoDelta{S}$) using the database when we can determine
% based on the bloom filter that no tuples in the delta have any join partners or
% at least reduce the delta size.
As in the microbenchmarks for join, we use
query \qjoinsel (see \cite{techreport}). As shown in \Cref{fig:syn-small-bloom,fig:syn-large-bloom},
filtering the delta using bloom filters is effective for all delta sizes, even for larger selectivity, due to (i) the reduction in
data transfer between \impAbbr and the database, (ii) the reduction of the input
size for the query evaluating $\annoDeltaRel \join \makeAnno{S}$ by reducing the
size of \annoDeltaRel, and (iii) reducing the input size for incremental
operators. As show in \Cref{fig:syn-large-bloom}, the bloom filter optimization is effective for both low and high selectivity and across all delta sizes we have tested.
}
\iftechreport{
\impAbbr uses bloom filters for
each join to track which tuples potentially have join partners. This enables us
to avoid evaluating a join $\annoDeltaRel \join \makeAnno{S}$
($\annoRel \join \makeAnnoDelta{S}$) when we can determine
% based on the bloom filter
that no tuples in the delta have any join partners or
at least to reduce delta size. We again use
 \qjoinsel (group-by aggregation with
\lstinline!HAVING! over a join result). \Cref{fig:syn-small-bloom} shows
the runtime of \impAbbr applying bloom filter for joins varying
selectivity and delta size. \Cref{fig:syn-large-bloom} shows runtime
of \impAbbr for large delta sizes. The result shows that
filtering the delta using bloom filters is effective due to (i) the reduction in
data transfer between \impAbbr and the database, (ii) the reduction of the input
size for the query evaluating $\annoDeltaRel \join \makeAnno{S}$ by reducing the
size of \annoDeltaRel, and (iii) reducing the input size for incremental
operators. As show in \Cref{fig:syn-large-bloom}, bloom filter optimization is effective for both low and high selectivity and across all delta sizes we have tested.
}
% In general, both two optimizations are effective in reducing the size of
% deltas for the \impAbbr engine to maintain sketches.
% Both optimizations can
% improve the performance of \impAbbr engine when incrementally maintaining the
% provenance sketches
%
%%%%%%%%%%%%%%%%%%%%%%%%%%%%%%%%%%%%%%%%
\begin{takeaway}
  \impAbbr's performance is mainly impacted by  delta size. As join requires a round trip to the databases, queries with join are typically more expensive. Our bloom filter optimization reduces this cost. Nonetheless, \impAbbr significantly outperforms \fullmaintenance.  % except for very large delta sizes.
  % This set of experiments demonstrate that the runtime of incremental
  % maintenance is affected by the delta size. % From
  % % \Cref{sec:experiment_micro_numofgroups} and
  % % \Cref{sec:experiment_micro_numofaggs} for group-by-aggregation with
  % % \lstinline!HAVING! queries, when the number of aggregation functions is fixed,
  % % increasing the number of groups will increase the runtime of \impAbbr, which is the same
  % % as the number of aggregation functions increase if the number of groups is
  % % fixed.
  % Joins,
  % Joins and join selectivity experiments show that the more join partners
  % (higher selectivity)
  % an equi-join attribute has, the higher cost for incremental engine spends to
  % maintain sketches. Both \Cref{sec:experiments-tpch} and
  % \Cref{sec:experiment_micro_join} show that evaluation of $\deltaRel \join
  % \mathscr{S}$ ($\mathscr{R} \join \Delta \mathscr{S}$) using DBMS is costly,
  % while our \impAbbr engine can still outperforms \fullmaintenance for all realistic
  % delta sizes and some of the larger delta sizes.
\end{takeaway}
%%%%%%%%%%%%%%%%%%%%%%%%%%%%%%%%%%%%%%%%

%%%%%%%%%%%%%%%%%%%%%%%%%%%%%%%%%%%%%%%%%%%%%%%%%%%%%%%%%%%%%%%%%%%%%%%%%%%%%%%%
%%  topk figs
\input{sections/fig-topk-time-and-mem}

%%
\subsubsection{\revision{Top-K operator}}
\label{sec:experiment_opt_topk}
\revision{According to \Cref{sec:op-rule-topk}, for a top-k operator we store all of its inputs in an ordered map to be able to deal with deletions that remove a tuple from the current top-k. In practice this may be overkill as keeping a buffer of top-l for $l > k$ tuples is often sufficient. The potential drawback is that if all tuples from the buffer are deleted, we have to recapture the sketch.}
\ifnottechreport{
 \revision{To evaluate this trade-off we run an experiment varying $l$ (20, 50 and 100). The query we use is a top-10
   query $\qtopk$ \ifnottechreport{(see ~\cite{techreport}).} \iftechreport{(see~\Cref{sec:appendix_qlist}).}
Then evaluate workloads that delete data (20 tuples per
  update) from the table (the table contains 50k tuples and 5k distinct group-by values). We consider two extremes: (i) always delete tuples contributing to the top-k and (ii) randomly deleting tuples.  If there less than k groups remain in the state data structure,
   \impAbbr will have to fully maintain the sketch. \Cref{fig:topk-runtime} shows
   the runtime varying $l$. For the worst case workload (only deleting tuples from the top-k), the additional cost of maintaining a larger state for the top-k operator is amortized by reducing the frequency of recapture. For the other extreme (uniform deletion), recapture is rarely needed and all settings of $l$ exhibit similar performance. Overall, larger buffer sizes $l$ can be recommended.
 Please see~\cite{techreport} for additional workloads and memory usage measurements.
}
}
\iftechreport{
\revision{In this experiment, we evaluate the performance of \impAbbr engine for
top-k operator by storing certain amount of top-k. The query we use is a top-10
query $\qtopk$ \ifnottechreport{(see ~\cite{techreport})} \iftechreport{(see
      ~\Cref{sec:appendix_qlist})}. We control the number of how many top items
  stored in the state data: 20, 50 and 100. Then we delete data (~20 tuples per
  update) from the table (table contains 50000 tuples and 5000 group by values
  with each group has roughly 10 tuples). For deletion, we have several
  strategies: 1. always delete the first 2 minimal groups, 2. always delete
  randomly tuples, 3. control the ratio between random deletion and deleting
  minimal groups called R-M ratio where like query-update ratio: R updates
  containing randomly deleted tuples followed by M updates including tuples of
  deleting 2 minimal groups each update. In our setting, the ratio are 2:1 and 4:1. For
  strategy 1 and 3, we have 500 updates in total while for strategy 2, we delete
  data till the table is empty since table and updates are all uniformly
  distributed and only deleting all tuples can see the effect of how the
  \impAbbr engine performs. If there are less than k groups stored in the state,
  our \impAbbr will fully maintain the sketches. ~\Cref{fig:topk-runtime} shows
  the runtime of different top items stored in the state under different
  deleting strategies. We can observe the following: 1. The more data items
  stored in the top-k state data, the less frequent to fully maintain the
  sketches. 2. The more randomly the delta is generated, the less frequently the
  sketches are maintained. ~\Cref{fig:topk-memory} demonstrates the memory consumption
  of maintaining the \qtopk in \impAbbr
  (the y value is the memory consumption for each the x-th operation). These figures show that: 1. The more data items stored in the state, the more the
  memory consumes. 2. Memory consumption will gradually decrease before sketches
  fully maintained. 3. At the point the memory consumption sudden increasing, it
  means a full maintenance will increase size of the state data of top-k to the
  size as it should be of storing top 20, 50 or 100, while the total memory size
  still decreases due to the size decreasing of aggregation state (less number
  of groups than before).
}
}

%%%%%%%%%%%%%%%%%%%%%%%%%%%%%%%%%%%%%%%%%%%%%%%%%%%%%%%%%%%%%%%%%%%%%%%%%%%%%%%%
%%
%% THIS EXPERIMENT IS NOT REQUIRED ANY MORE FOR BOTH MAIN PAPER AND TECHREPORT
%%
%% techreport
\iftechreport{
For top-k optimization, we evaluate one \tpchds query \qspace (\tpchds Q10, see
~\Cref{sec:appendix_qlist}) as well to examine the state memory consumption
under storing different top number of data items coming into top-k operator. For
1 GB dataset, the total tuples number for top-k is 37293 and for 10 GB dataset,
the number is 371104. We vary the number for $m$ to examine  memory used (MB).
For \tpchds 1GB, we vary the number of $l$ (top l data items stored in top-k state data) in $1K$, $5K$, $10K$, $20K$ and all
tuples. For \tpchds 10GB, these number values are $1K$, $5K$, $10K$, $20K$,
$100K$, $200K$ and all. \Cref{fig:tpch1g-space-opt} and
\Cref{fig:tpch10g-space-opt} show the memory used varying the stored tuple
number $l$. A knowledge learned from the experiment is that memory saving can
be achieved by reducing the number of tuples kept in the state data.
}
\vspace{-8mm}
\ifnottechreport{
    \revision{ %
\paragraph{Maintenance Strategies}
      In \cite{techreport} we also evaluate the impact of batch size on the cost of eager maintenance, demonstrating that batch sizes below 50 tuples should be avoided. In general, lazy maintenance is superior as we delay maintenance as long as possible leading to batching of deltas and avoiding maintenance of sketches that are not used.}
    }
\iftechreport{
%%%%%%%%%%%%%%%%%%%%%%%%%%%%%%%%%%%%%%%%%%%%%%%%%%%%%%%%%%%%%%%%%%%%%%%%%%%%%%%%
    \subsection{Maintenance Strategies}\label{sec:exp:dynamic-adaptation}
Recall the we support two strategies for maintaining sketches: (i) \emph{eager} which batches deltas and maintains sketches proactively once a number of delta tuples equal or exceeding the batch size have been updated and (ii) \emph{lazy} which only maintains stale sketches when a new version is needed to answer a query.
    % In this experiment we evaluate the cost of the eager strategy for incremental maintenance.
We  evaluating the impact of batch size on maintenance cost, by measuring the total maintenance cost for 1000 updates that are processed in  batches of varying sizes using the eager strategy.
% We have two settings: (1) We have 1000 updates in total, evenly divided into
% sets of 1, 5, 10, 50, 100, 500, 1000 updates each comprising $~180$ delta tuples
% that are processed by each maintenance operation.
For this setting, we have two queries: first is a aggregation with \lstinline!HAVING! query template \qendtoend (see ~\Cref{sec:appendix_qlist}) and the
other is join-aggregation with \lstinline!HAVING! $\qjoinsel$ query template with $5\%$ selectivity (see ~\Cref{sec:appendix_qlist}). % We calculate the total runtime of
% completing 1000 maintenance.
 \Cref{fig:evenly_divided_1k_update}
 shows the total runtime of this setting.
The net result is that batch sizes below 50 should be avoided for eager maintenance as they significantly increase the cost of maintenance. Another take-away is that lazy is typically preferable as (i) it leads to larger delta sizes as we delay maintenance as long as possible and (ii) we avoid maintaining provenance sketches that are not used.
  }
\input{sections/fig-dynamic-adaption.tex}

%%%%%%%%%%%%%%%%%%%%%%%%%%%%%%%%%%%%%%%%%%%%%%%%%%%%%%%%%%%%%%%%%%%%%%%%%%%%%%%%
%% MEMory for agg groups and join selectivity
\iftechreport{
\input{sections/fig-mem-agg-join.tex}
\subsection{\revision{Additional Memory Consumption}}
\label{sec:mem_consumption}
\subsubsection{\revision{Memory Usage of Aggregation Function and Join}}
\label{sec:mem_join_agg}
\revision{
   In this experiment, we present the memory usage of queries \qjoinsel (see
   ~\Cref{sec:appendix_qlist}) and \qgroups (see ~\Cref{sec:appendix_qlist}). We
   have exactly the same setting as evaluating for runtime performance.
   ~\Cref{fig:agg-join-memory} shows the memory consumption. For \qgroups, for fixed
   number of groups, the state data size is stable, and the memory consumption
   increases due to the increasing of delta data size. The is the same for \qjoinsel.
}

\subsubsection{\revision{Memory usage of Sketches and Ranges}}
\label{mem_sketch_range}
%%%%%%%%%%%%%%%%%%%%%%%%%%%%%%%%%%%%%%%%
%% table of different size of sketch and range
\input{sections/range-sketch-size-table.tex}
\revision{
  In the ~\Cref{fig:sketch-range-size}, we show the actual sizes of sketches and
  ranges. We encode each sketch as a bitvector. Typically, the size is
  pretty small. For ranges, we store the boundaries of each range in a list. For $n$
  ranges, we record $n + 1$ values in the list i.e., assuming the ranges are
  $[1,4), [4, 9)$, the list will be: $(1, 4, 9)$.
}
}

%% file: sections/fig-microbenchmarks.tex
\begin{figure*}[h]
%%%%%%%%%%%%%%%%%%%%%%%%%%%%%%%%%%%%%%%%
  % row 0 bar legend
  \begin{minipage} {1.0\linewidth}
  \includegraphics[width = 1.0\linewidth]{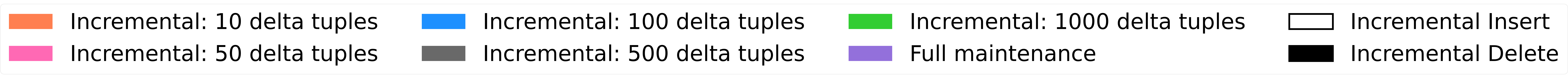}
  \end{minipage}
%%%%%%%%%%%%%%%%%%%%%%%%%%%%%%%%%%%%%%%%
  % row 1
  \begin{minipage}{1.0\linewidth}
    % agg nums
    \begin{subfigure}{0.33\textwidth}
      \includegraphics[width=1.0\textwidth, trim=0 0 0 0, clip]{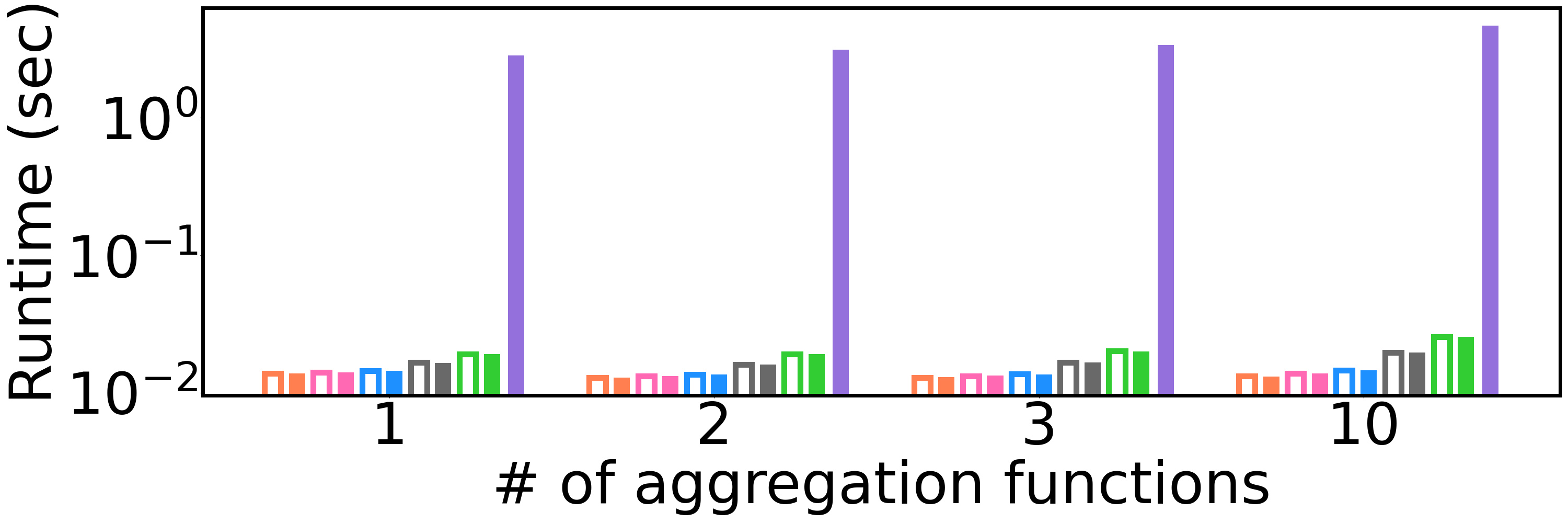}
      \vspace{-5mm} % current figure layout, -7mm is the extreme value.
      \caption{\qhaving number of aggregation functions}
      \label{fig:syn-small-aggnum}
    \end{subfigure}
    % fix agg, groups changes
    \begin{subfigure}{0.33\textwidth}
      \includegraphics[width=1.0\textwidth, trim=0 0 0 0, clip]{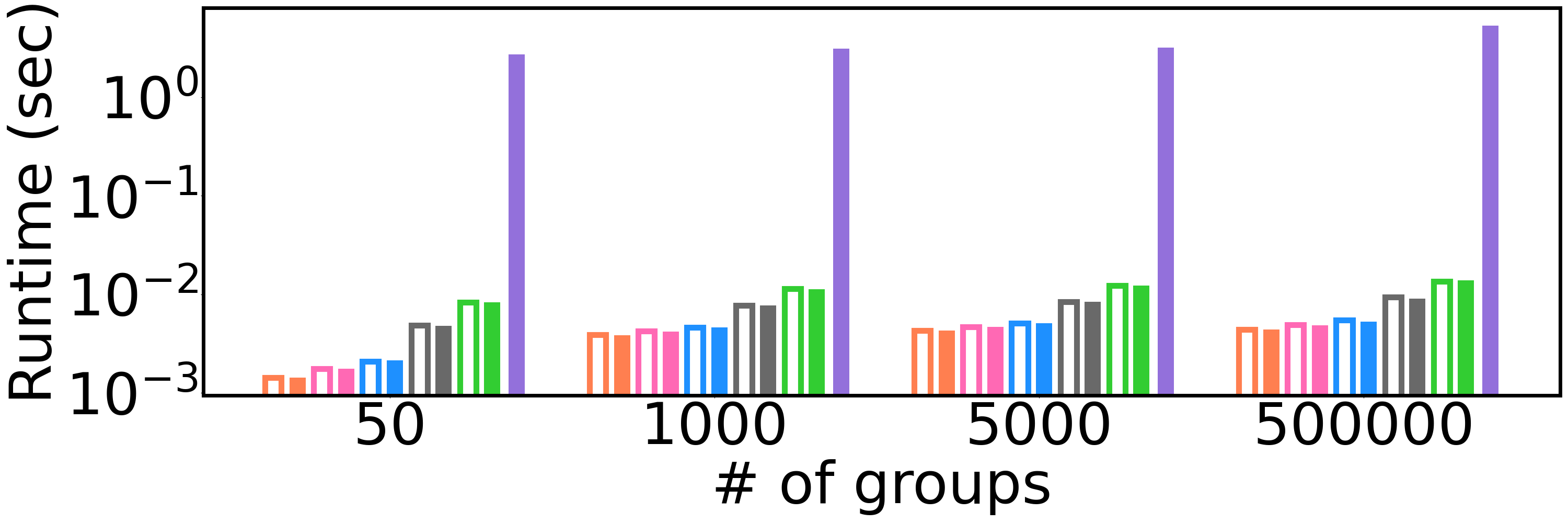}
      \vspace{-5mm}
      \caption{\qgroups number of groups }
      \label{fig:syn-small-fixagg-diffgroups}
    \end{subfigure}
    % join: 1 - n
    \begin{subfigure}{0.33\textwidth}
      \includegraphics[width=1.0\textwidth, trim = 0 0 0 0, clip]{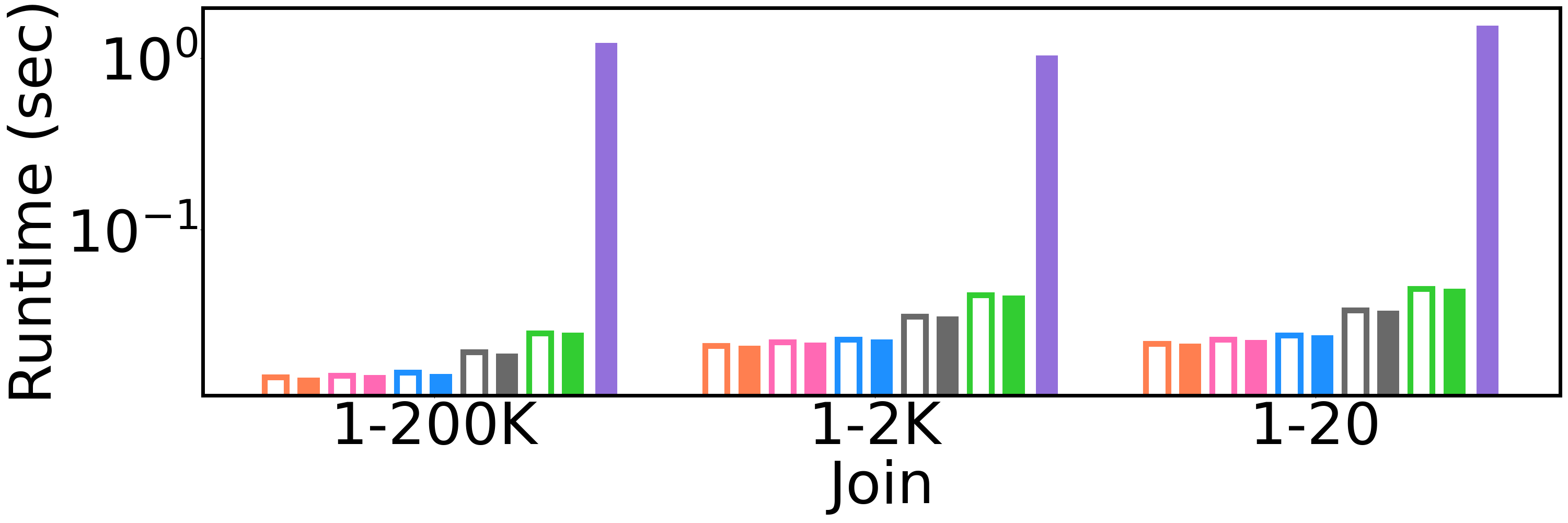}
      \vspace{-5mm}
      \caption{\qjoin 1-n join}
      \label{fig:syn-small-join-1-n}
    \end{subfigure}
  \end{minipage}
%%%%%%%%%%%%%%%%%%%%%%%%%%%%%%%%%%%%%%%%
  % row 2
  \begin{minipage}{1.0\linewidth}
    % join: m - n
    \begin{subfigure}{0.33\textwidth}
      \includegraphics[width=1.0\textwidth, trim=0mm 0mm 0 0mm, clip]{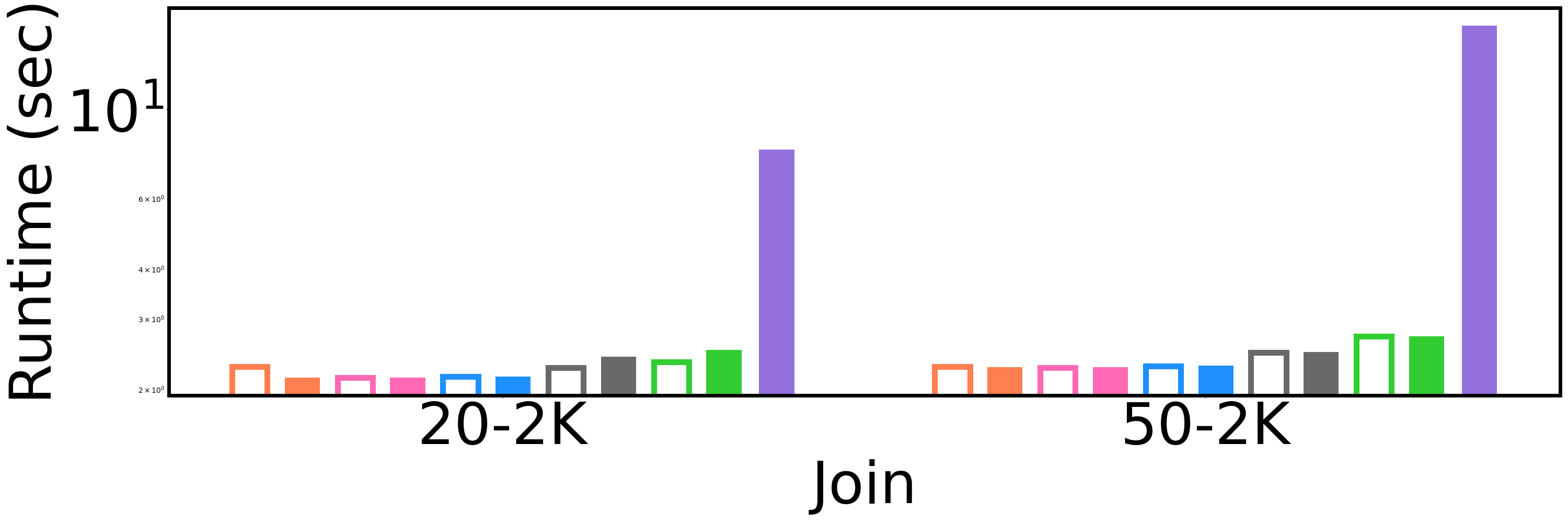}
      \vspace{-5mm}
      \caption{\qjoin m-n join}
      \label{fig:syn-small-join-m-n}
    \end{subfigure}
    % join selectivity
    \begin{subfigure}{0.33\textwidth}
      \includegraphics[width=1.0\textwidth, trim=0mm 0mm 0 0mm, clip]{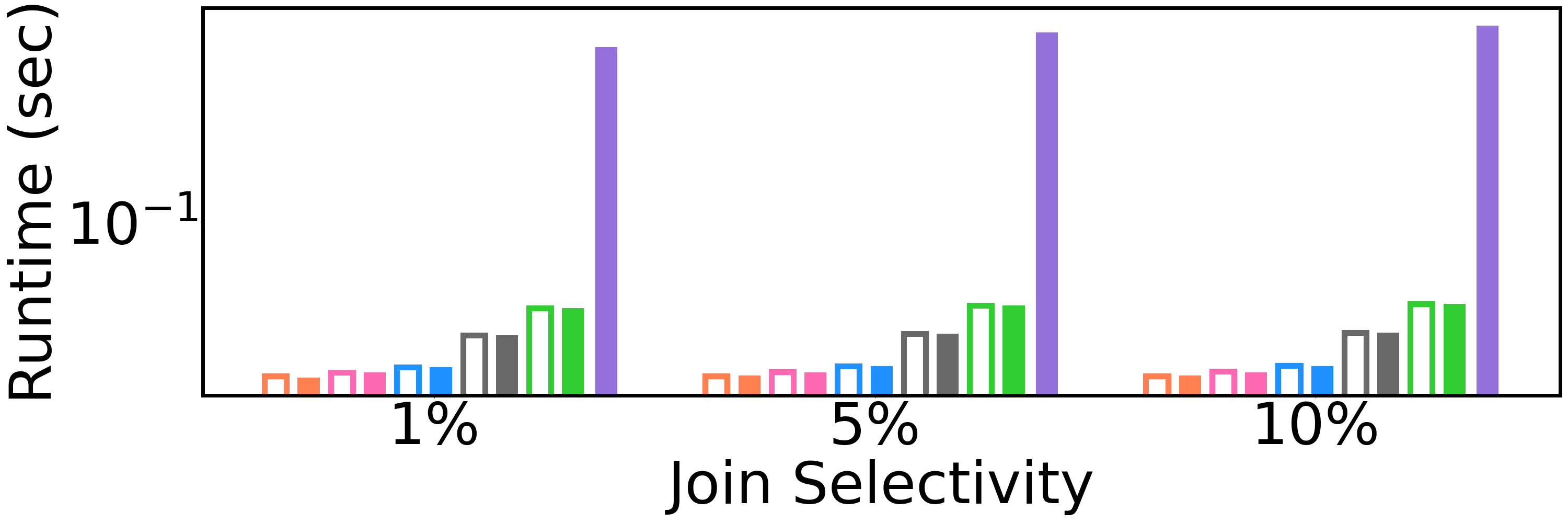}
      \vspace{-5mm}
      \caption{\qjoinsel varying join selectivity}
      \label{fig:syn-small-join-selectivity}
    \end{subfigure}
    % sketch size
    \begin{subfigure}{0.33\textwidth}
      \includegraphics[width=1.0\textwidth, trim=0mm 0mm 0 0mm, clip]{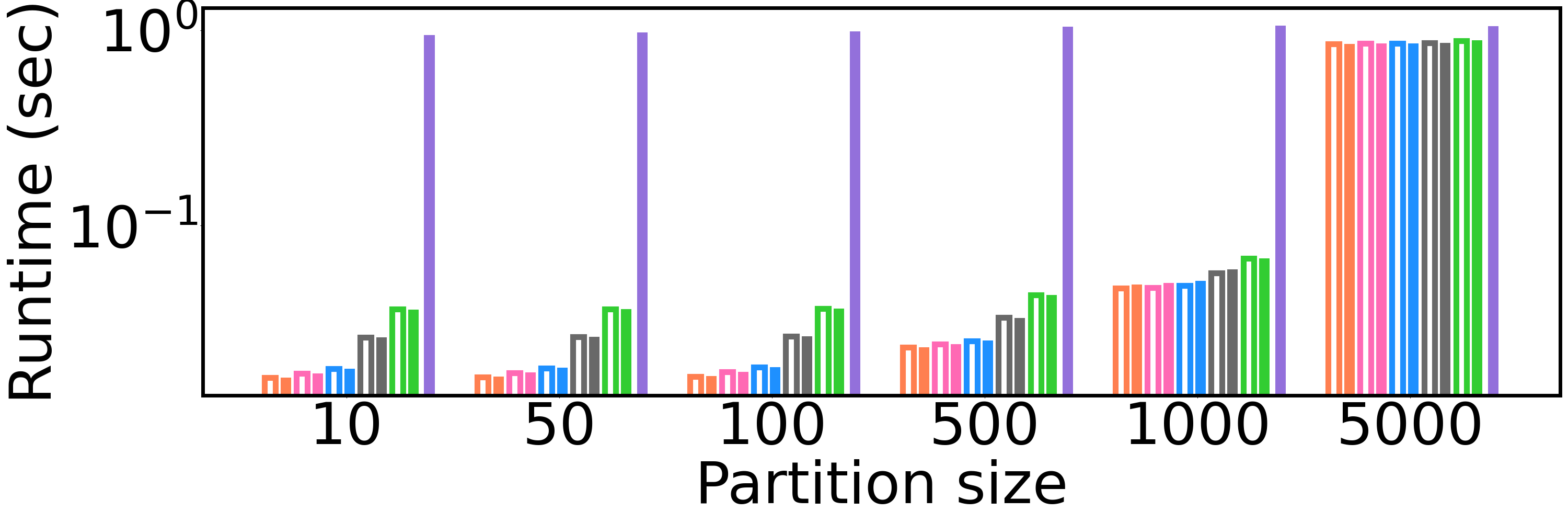}
      \vspace{-5mm}
      \caption{\qsketch, number of fragments of partition}
      \label{fig:syn-small-sketch-size}
    \end{subfigure}
    % \vspace{-7mm}
    \vspace{-6mm}
    \caption{Microbenchmarks (``realistic'' delta size): varying delta size from 10 tuples to 1000 tuples.}\label{fig:micro-small-deltas}
  \end{minipage}
%%%%%%%%%%%%%%%%%%%%%%%%%%%%%%%%%%%%%%%%
  % row 3 sync line legends
  \begin{minipage}{1.0\linewidth}
    \centering
    %\vspace{-2mm}
    \includegraphics[width = 0.5\linewidth,trim=0mm 0mm 0mm 0mm,clip]{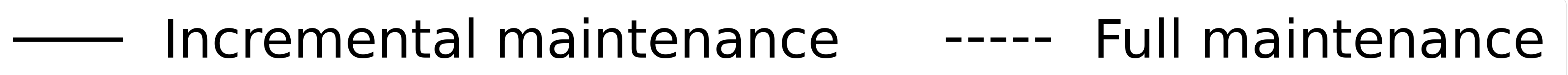}
    %\vspace{-7mm}
  \end{minipage}\\[0mm]
%%%%%%%%%%%%%%%%%%%%%%%%%%%%%%%%%%%%%%%%
  % row 4
  \begin{minipage}{1.0\linewidth}
    % agg num large
    %%%%%%%%%%%%%%%%%%%%%%%%%%%%%%%%%%%%%%%%%%%%%%%%%%%%%%%%%%%%%%%%%%%%%%%%%%%%%%%%
    \begin{minipage}{0.33 \textwidth}
      \centering
      \begin{subfigure}{1\textwidth}
        \includegraphics[width=1.0\textwidth, trim=0mm 0mm 0mm 0mm, clip]{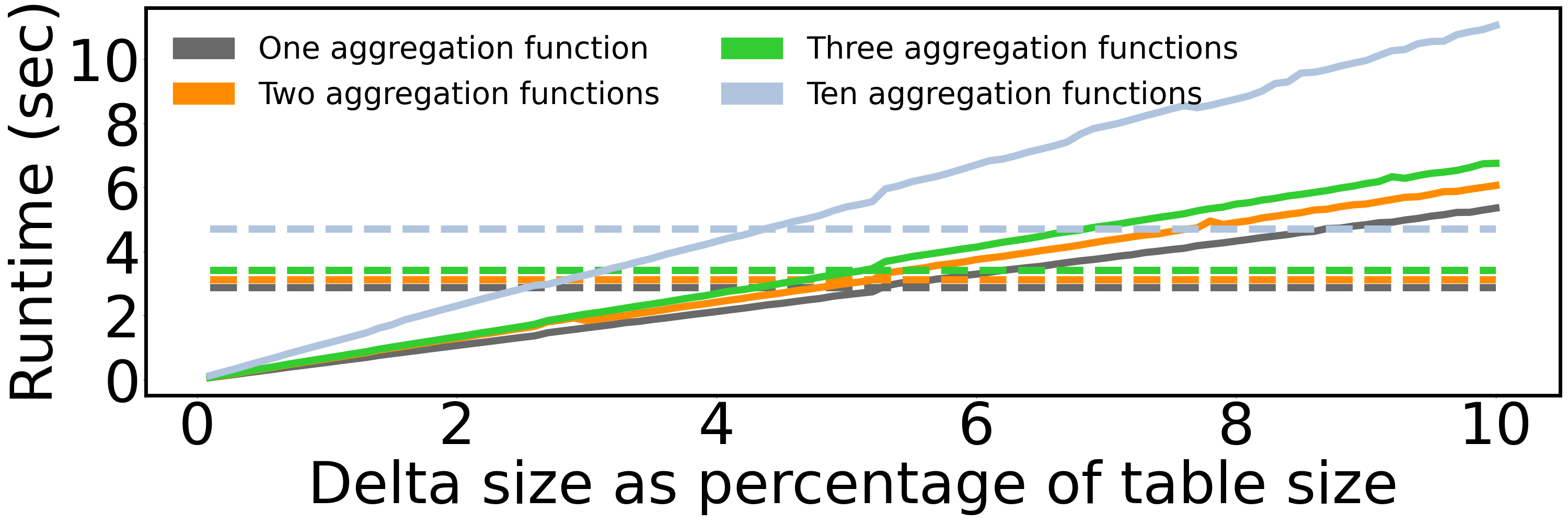}
        \vspace{-5mm}
        \caption{\qhaving}
        \label{fig:syn-large-aggnum}
      \end{subfigure}
    \end{minipage}
    %%%%%%%%%%%%%%%%%%%%%%%%%%%%%%%%%%%%%%%%%%%%%%%%%%%%%%%%%%%%%%%%%%%%%%%%%%%%%%%%
    % groups large
    %%%%%%%%%%%%%%%%%%%%%%%%%%%%%%%%%%%%%%%%%%%%%%%%%%%%%%%%%%%%%%%%%%%%%%%%%%%%%%%%
    \begin{minipage}{0.33 \textwidth}
      \centering
      \begin{subfigure}{1\textwidth}
        \includegraphics[width=1.0\textwidth, trim=0mm 0mm 0mm 0mm, clip]{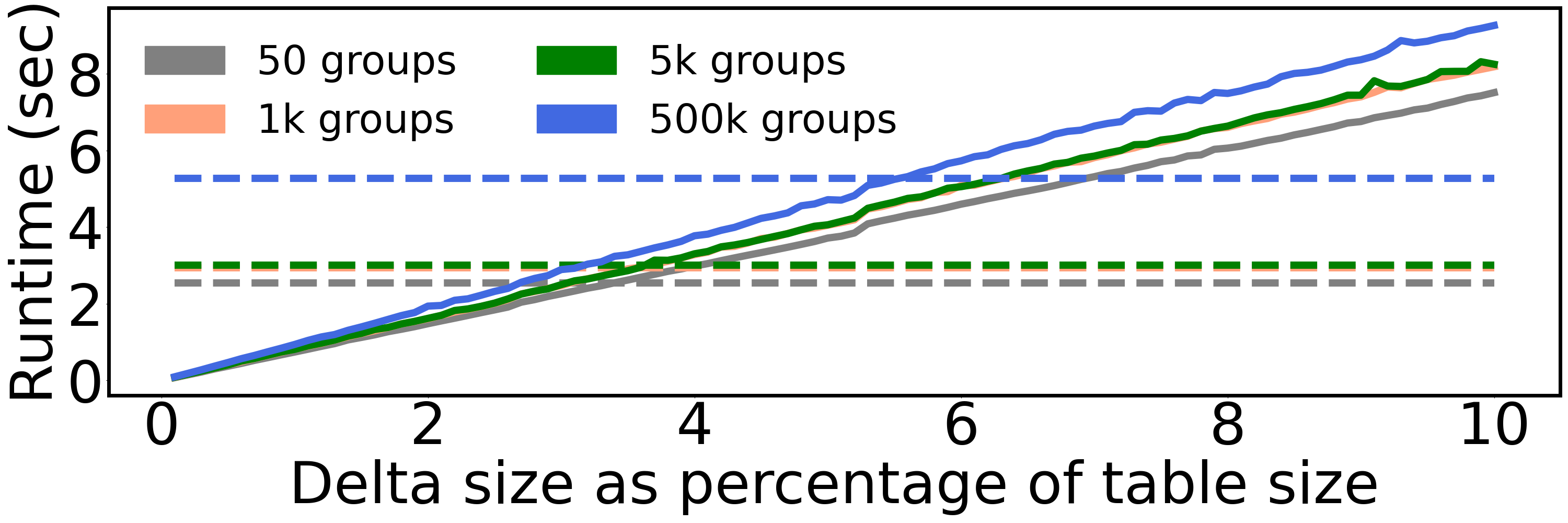}
        \vspace{-5mm}
        \caption{\qgroups}
        \label{fig:syn-large-fixagg-diffgroups}
      \end{subfigure}
    \end{minipage}
    %%%%%%%%%%%%%%%%%%%%%%%%%%%%%%%%%%%%%%%%%%%%%%%%%%%%%%%%%%%%%%%%%%%%%%%%%%%%%%%%
    % join large 1-n
    %%%%%%%%%%%%%%%%%%%%%%%%%%%%%%%%%%%%%%%%%%%%%%%%%%%%%%%%%%%%%%%%%%%%%%%%%%%%%%%%
    \begin{minipage}{0.33 \textwidth}
      \centering
      \begin{subfigure}{1\textwidth}
        \includegraphics[width=1.0\textwidth, trim=0mm 0mm 0 0mm, clip]{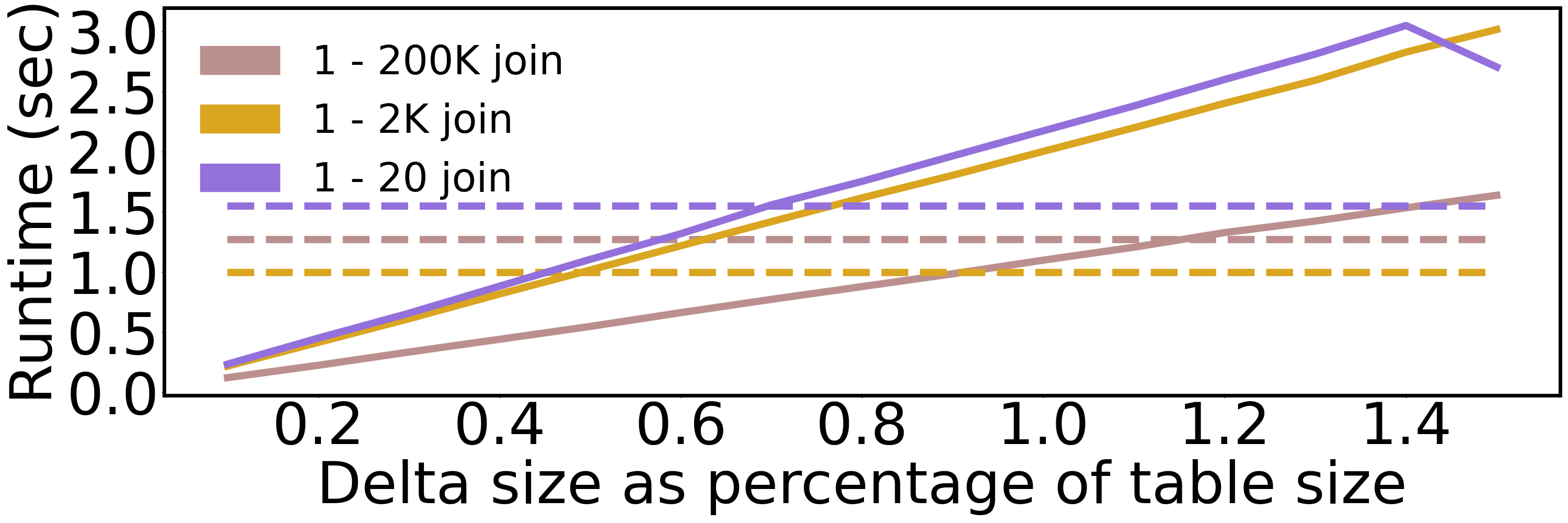}
        \vspace{-5mm}
        \caption{\qjoin (one join partner per tuple in $S$)}
        \label{fig:syn-large-join-1-n}
      \end{subfigure}
    \end{minipage}
    %%%%%%%%%%%%%%%%%%%%%%%%%%%%%%%%%%%%%%%%%%%%%%%%%%%%%%%%%%%%%%%%%%%%%%%%%%%%%%%%
  \end{minipage}
%%%%%%%%%%%%%%%%%%%%%%%%%%%%%%%%%%%%%%%%
  % row 5
  \begin{minipage}{1.0\linewidth}
    % join large m-n
    %%%%%%%%%%%%%%%%%%%%%%%%%%%%%%%%%%%%%%%%%%%%%%%%%%%%%%%%%%%%%%%%%%%%%%%%%%%%%%%%
    \begin{minipage}{0.33 \textwidth}
      \centering
      \begin{subfigure}{1\textwidth}
        \includegraphics[width=1.0\textwidth, trim=0mm 0mm 0mm 0mm, clip]{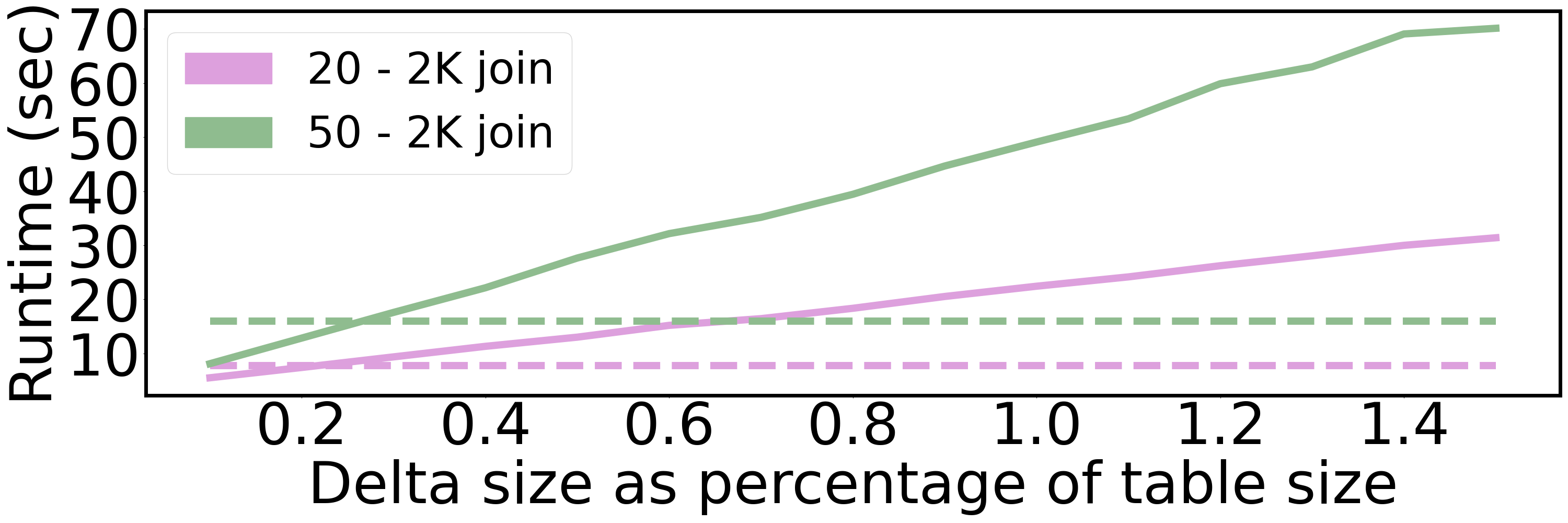}
        \vspace{-5mm}
        \caption{\qjoin (varying \#join partners per tuple in $S$)}
        \label{fig:syn-large-join-m-n}
      \end{subfigure}

    \end{minipage}
    %%%%%%%%%%%%%%%%%%%%%%%%%%%%%%%%%%%%%%%%%%%%%%%%%%%%%%%%%%%%%%%%%%%%%%%%%%%%%%%%
    %%%%%%%%%%%%%%%%%%%%%%%%%%%%%%%%%%%%%%%%%%%%%%%%%%%%%%%%%%%%%%%%%%%%%%%%%%%%%%%%
    \begin{minipage}{0.33 \textwidth}
      \centering
      % join selectivity large
      \begin{subfigure}{1\textwidth}
        \includegraphics[width=1.0\textwidth, trim=0mm 0mm 0mm 0mm, clip]{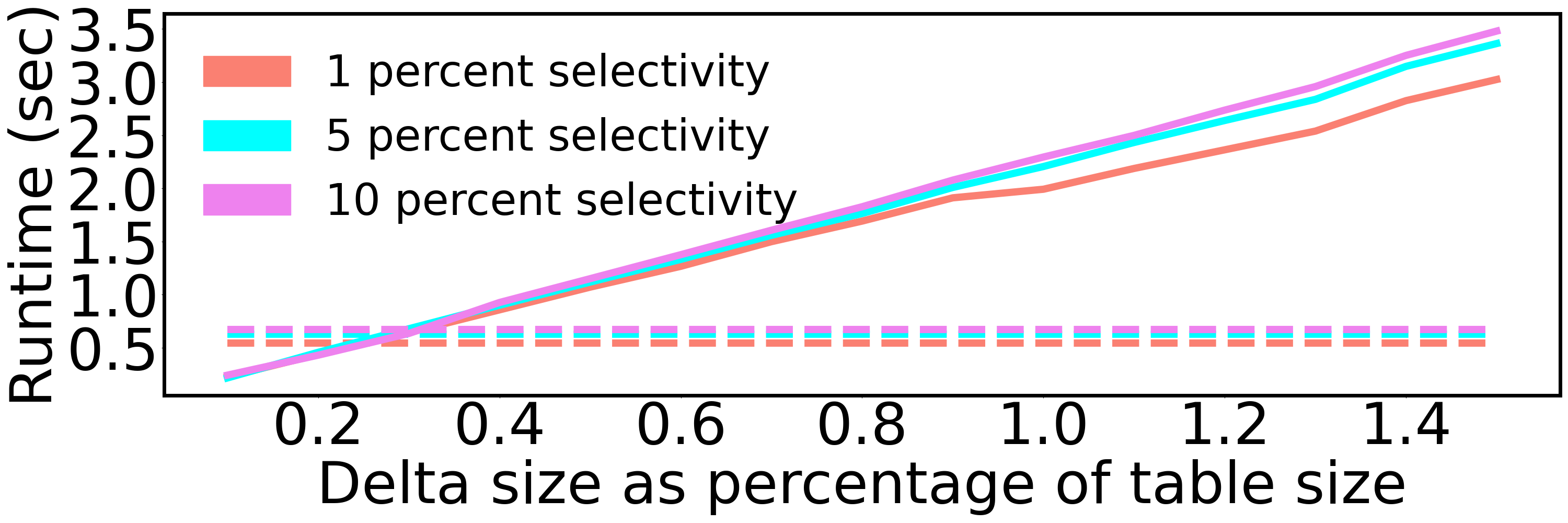}
        \vspace{-5mm}
        \caption{\qjoinsel}
        \label{fig:syn-large-join-selectivity}
      \end{subfigure}
    \end{minipage}
    %%%%%%%%%%%%%%%%%%%%%%%%%%%%%%%%%%%%%%%%%%%%%%%%%%%%%%%%%%%%%%%%%%%%%%%%%%%%%%%%
    %%%%%%%%%%%%%%%%%%%%%%%%%%%%%%%%%%%%%%%%%%%%%%%%%%%%%%%%%%%%%%%%%%%%%%%%%%%%%%%%
    \begin{minipage}{0.33 \textwidth}
      \centering
      % sketch size large
      \begin{subfigure}{1\textwidth}
        \includegraphics[width=1.0\textwidth, trim=0mm 0mm 0mm 0mm, clip]{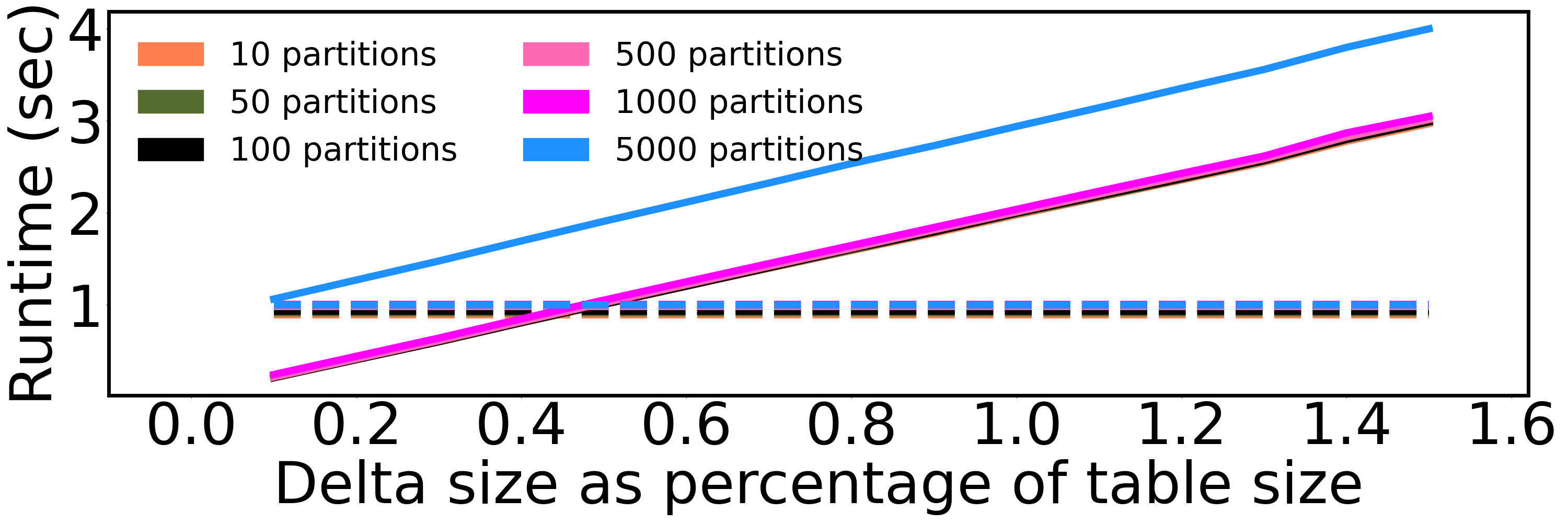}
        \vspace{-5mm}
        \caption{\qsketch}
        \label{fig:syn-large-sketch-size}
      \end{subfigure}
    \end{minipage}
    %%%%%%%%%%%%%%%%%%%%%%%%%%%%%%%%%%%%%%%%%%%%%%%%%%%%%%%%%%%%%%%%%%%%%%%%%%%%%%%%
    \vspace{-4mm}
    \caption{Microbenchmarks: varying delta size to determine the ``break even point'' where \fullmaintenance outperforms \impAbbr.}
    \label{fig:microbench-break-even}
  \end{minipage}
\end{figure*}

%%% Local Variables:
%%% mode: LaTeX
%%% TeX-master: "../imp"
%%% End:

%% file: sections/fig-optimizations.tex
%%%%%%%%%%%%%%%%%%%%%%%%%%%%%%%%%%%%%%%%%%%%%%%%%%%%%%%%%%%%%%%%%%%%%%%%%%%%%%%%
%% techreport
\iftechreport{
\begin{figure*}[h]
  % row 1 selpd small, bloom small, selpd large
  \begin{minipage}{1.0\linewidth}
    \begin{subfigure}{0.33\textwidth}
       \includegraphics[width=1.0\textwidth]{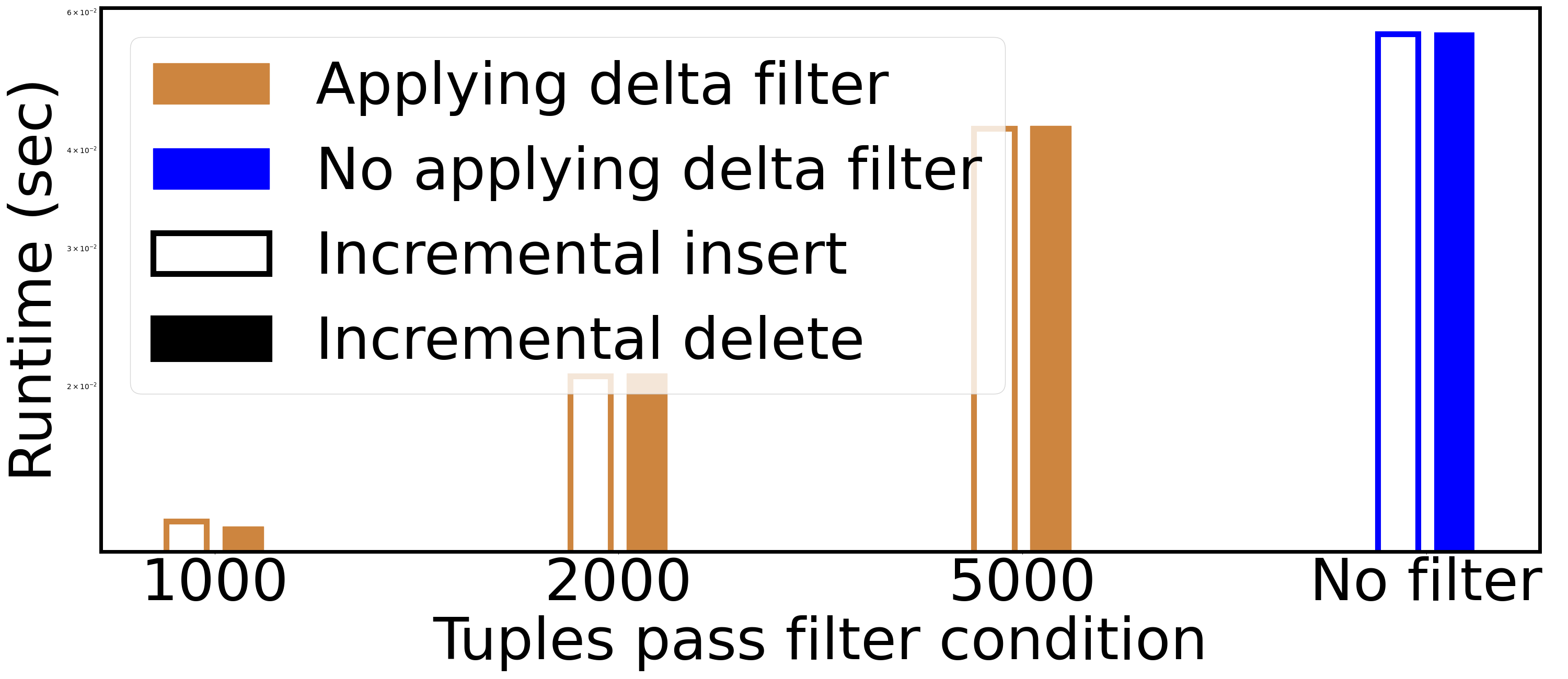}
       \caption{Optimization: filter delta}
       \label{fig:syn-small-selpd}
     \end{subfigure}
    \begin{subfigure}{0.33\textwidth}
      \includegraphics[width=1.0\textwidth]{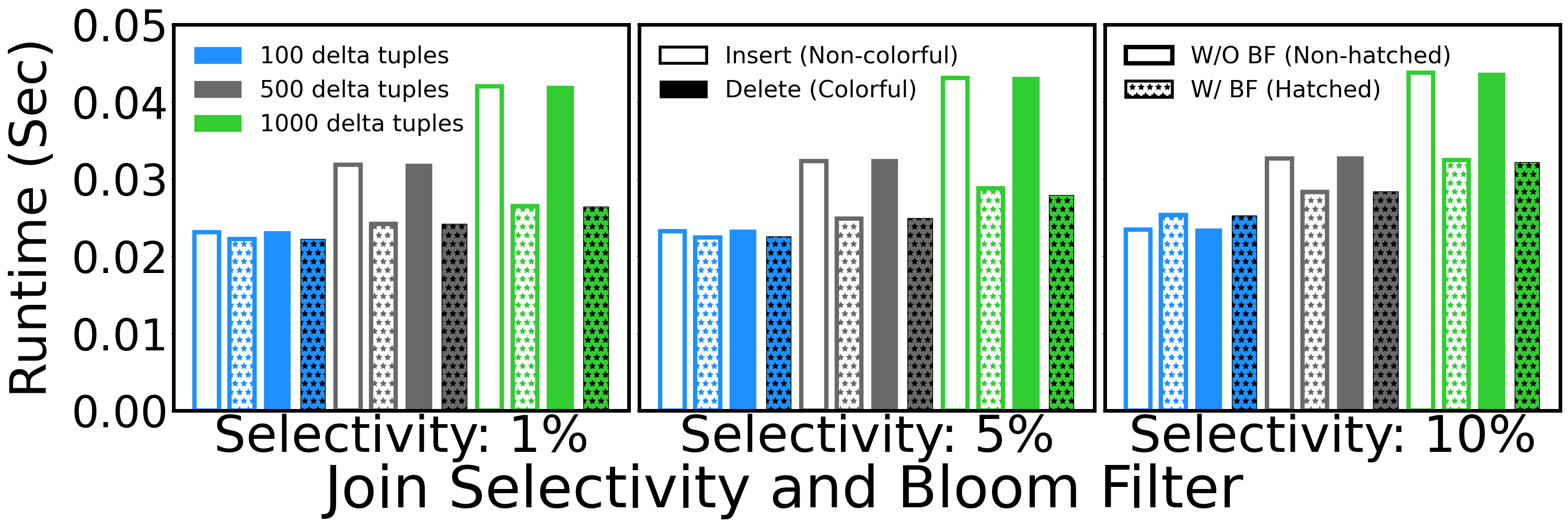}
      \vspace{-6mm}
       \caption{\qbloom overview}
      \label{fig:syn-small-bloom}
     \end{subfigure}
    \begin{subfigure}{0.33\textwidth}
      \includegraphics[width=1.0\textwidth]{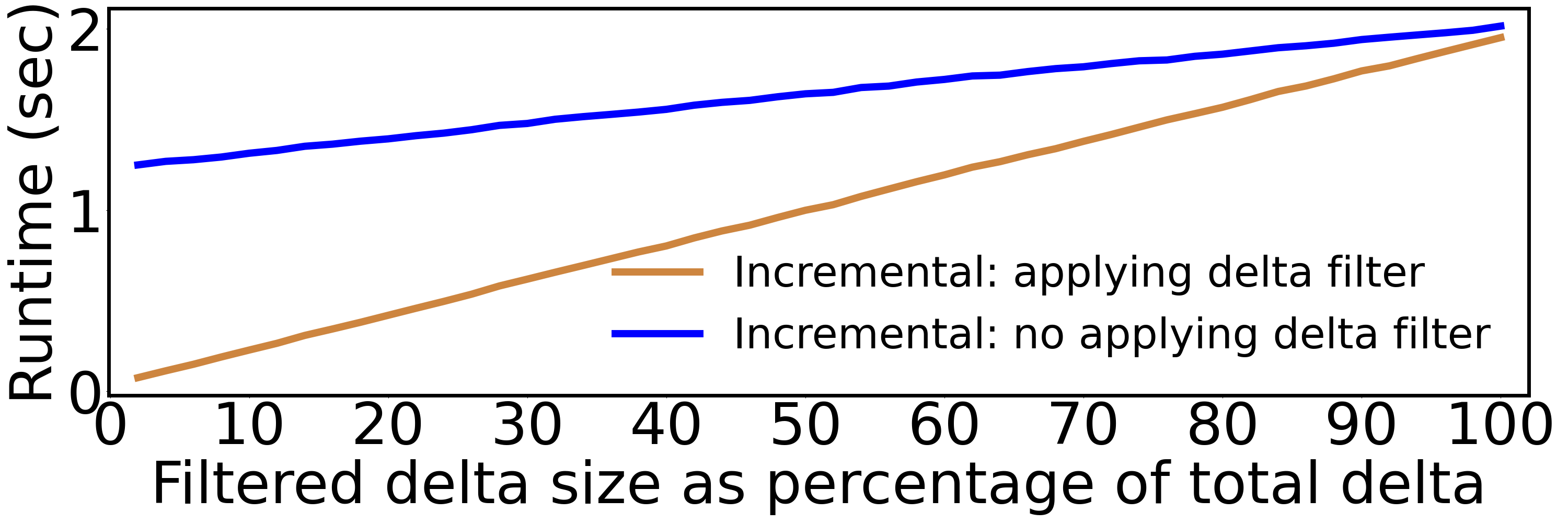}
       \vspace{-6mm}
       \caption{\qselpd}
      \label{fig:syn-large-selpd}
     \end{subfigure}
  \end{minipage}
  % row2: bloom large, space 1gb, space 10gb
  \begin{minipage}{1.0\linewidth}
    \begin{subfigure}{0.33\textwidth}
      \includegraphics[width=1.0\textwidth]{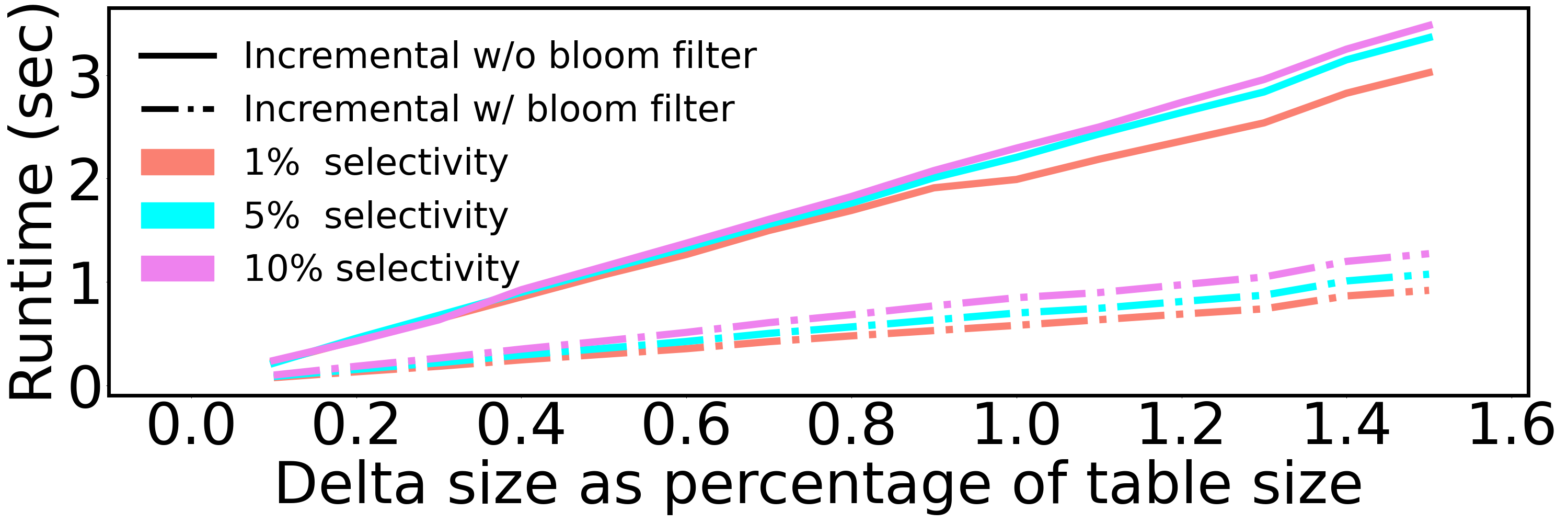}
       \vspace{-6mm}
       \caption{\qbloom}
      \label{fig:syn-large-bloom}
     \end{subfigure}
    \begin{subfigure}{0.33\textwidth}
      \includegraphics[width=1.0\textwidth]{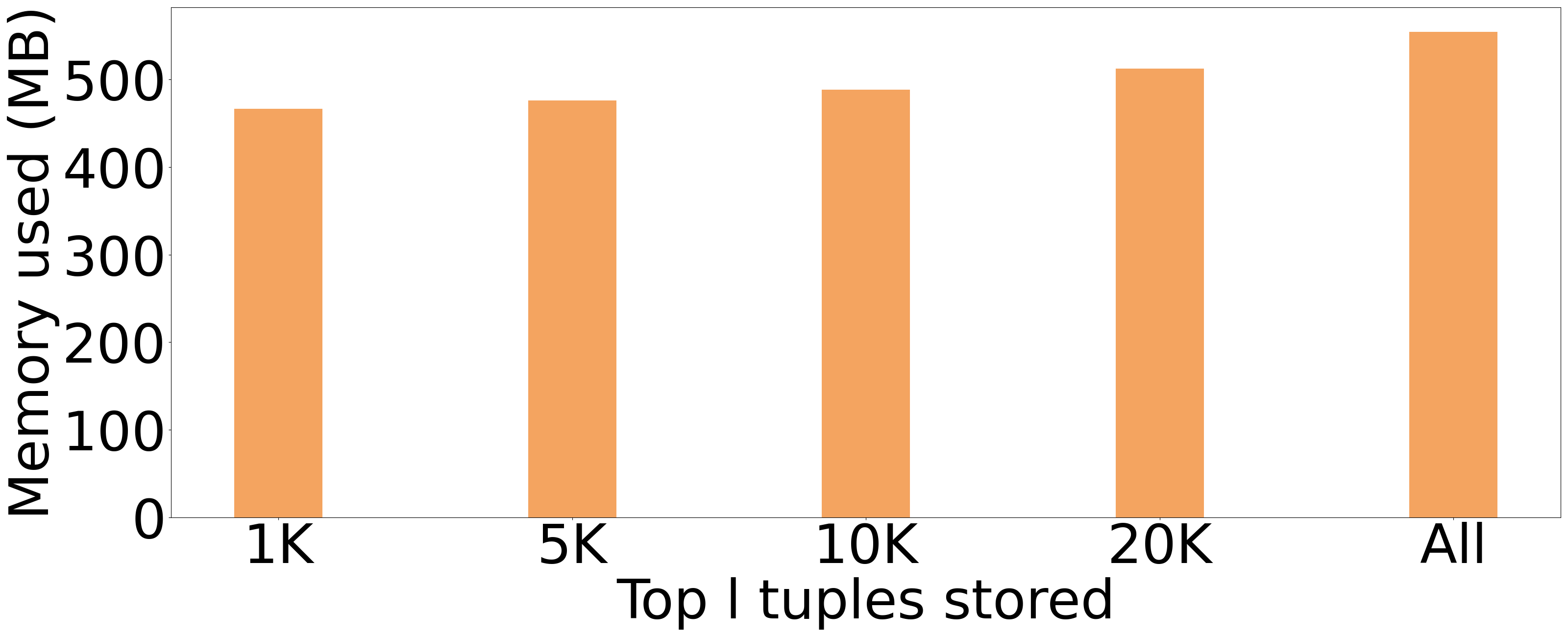}
      \vspace{-6mm}
      \caption{\qspace 1 GB}
      \label{fig:tpch1g-space-opt}
    \end{subfigure}
    \begin{subfigure}{0.33\textwidth}
      \includegraphics[width=1.0\textwidth]{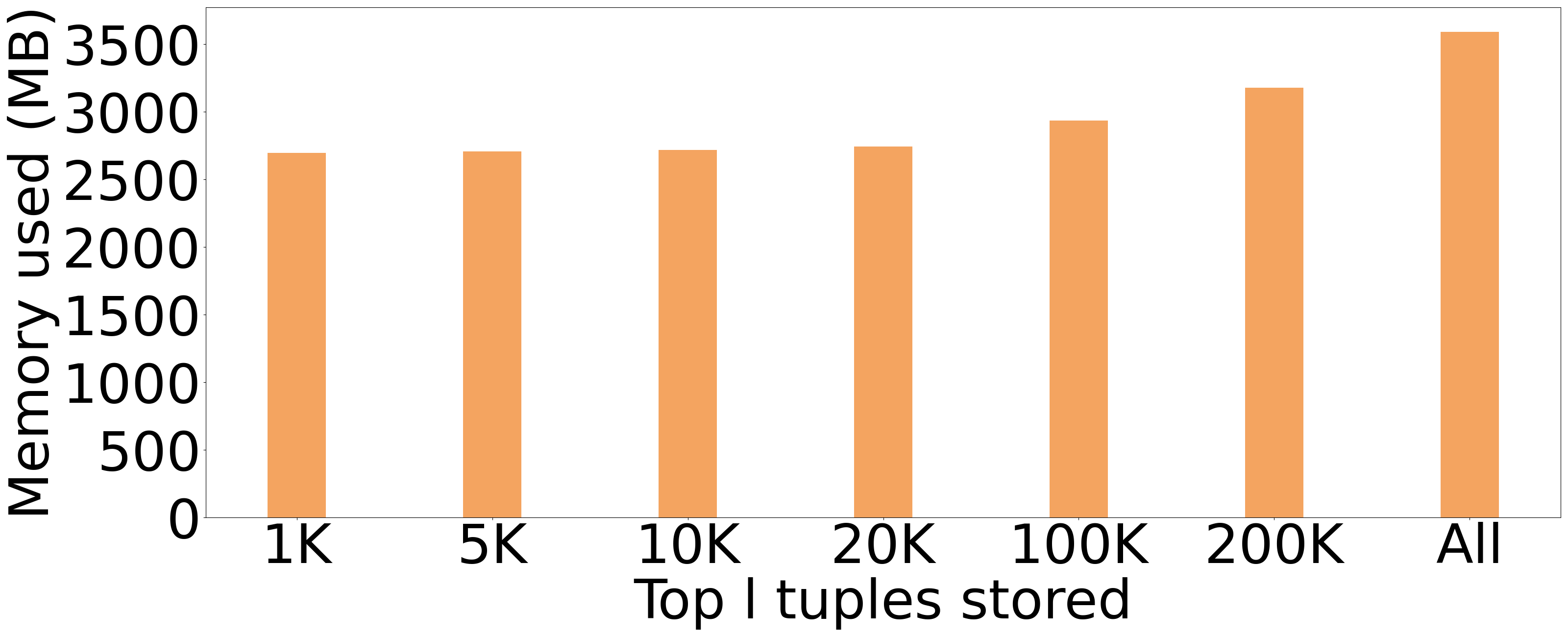}
      \vspace{-6mm}
      \caption{\qspace 10 GB}
      \label{fig:tpch10g-space-opt}
    \end{subfigure}
  \end{minipage}
   \vspace{-5mm}
  % \caption{Optimizations}
  \caption{Optimizations in \impAbbr: filter delta, bloom filter and space optimization}
  \label{fig:optimization}
\end{figure*}
}
%%%%%%%%%%%%%%%%%%%%
%% main paper
\ifnottechreport{
\begin{figure*}[h]
  \begin{minipage}{1.0\linewidth}
    %% selection push down
    \begin{subfigure}{0.33\textwidth}
      \includegraphics[width=1.0\textwidth]{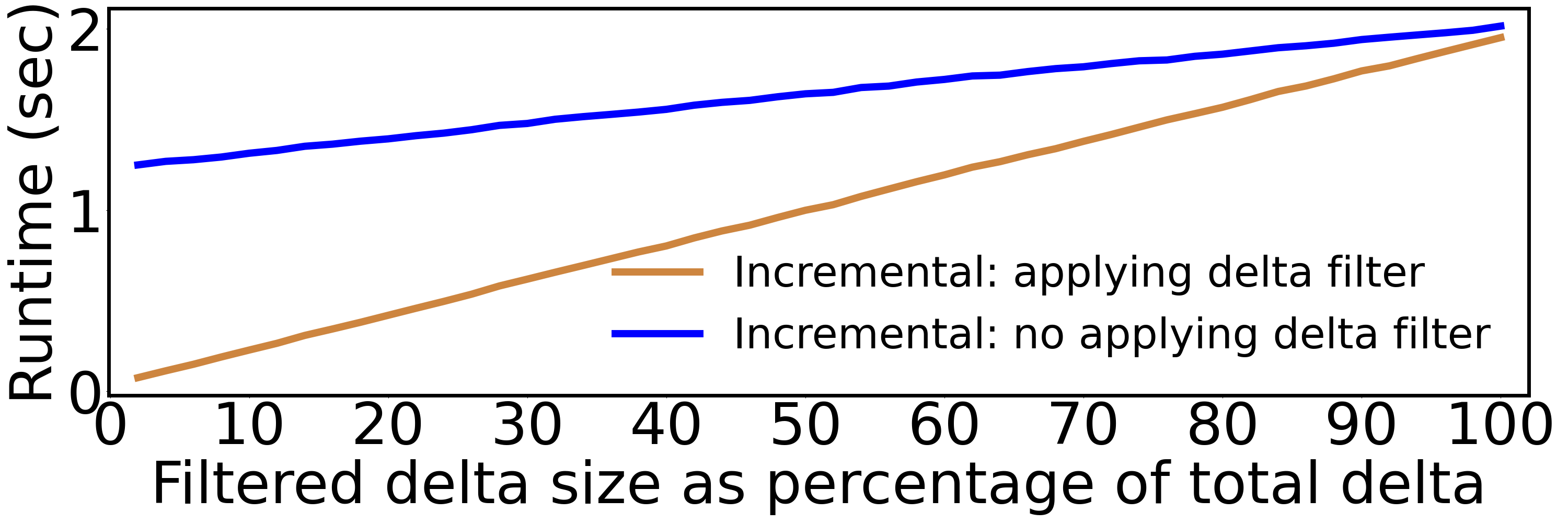}
      \vspace{-5mm}
       \caption{Optimization: filter delta}
      \label{fig:syn-large-selpd}
     \end{subfigure}
    %% bloom filter: small delta
    \begin{subfigure}{0.33\textwidth}
      \includegraphics[width=1.0\textwidth]{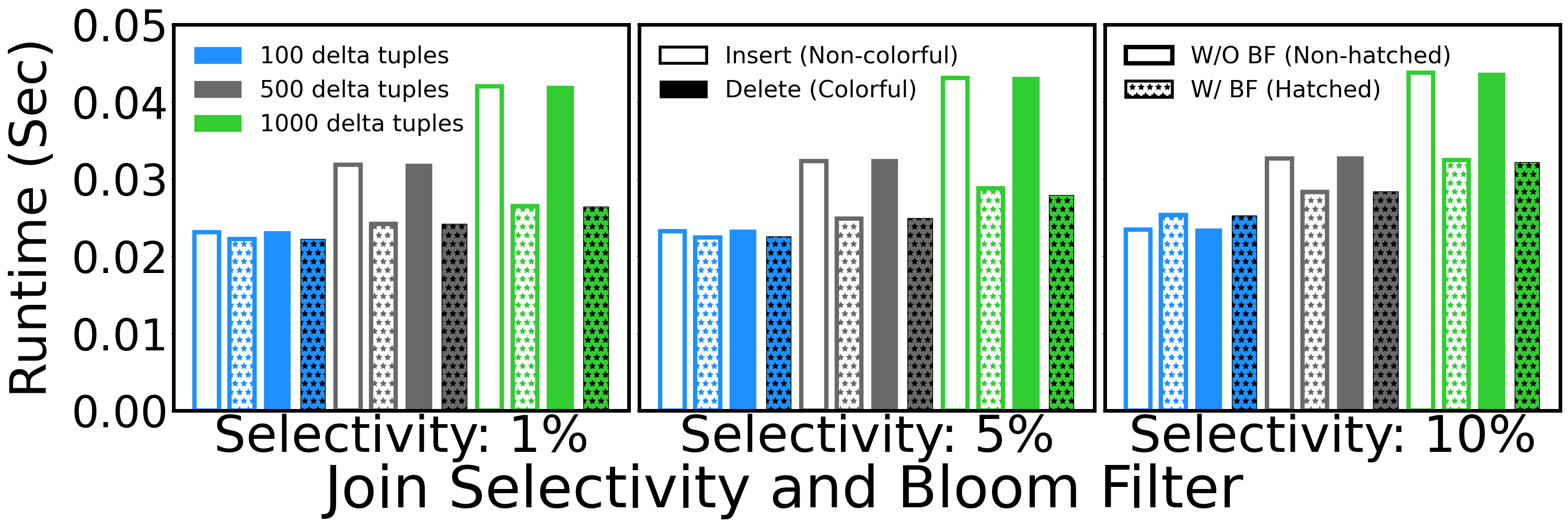}
       \vspace{-5mm}
       \caption{Optimization: bloom filter realistic delta size}
       \label{fig:syn-small-bloom}
     \end{subfigure}
    %% bloom filter: increased delta delta
    \begin{subfigure}{0.33\textwidth}
      \includegraphics[width=1.0\textwidth]{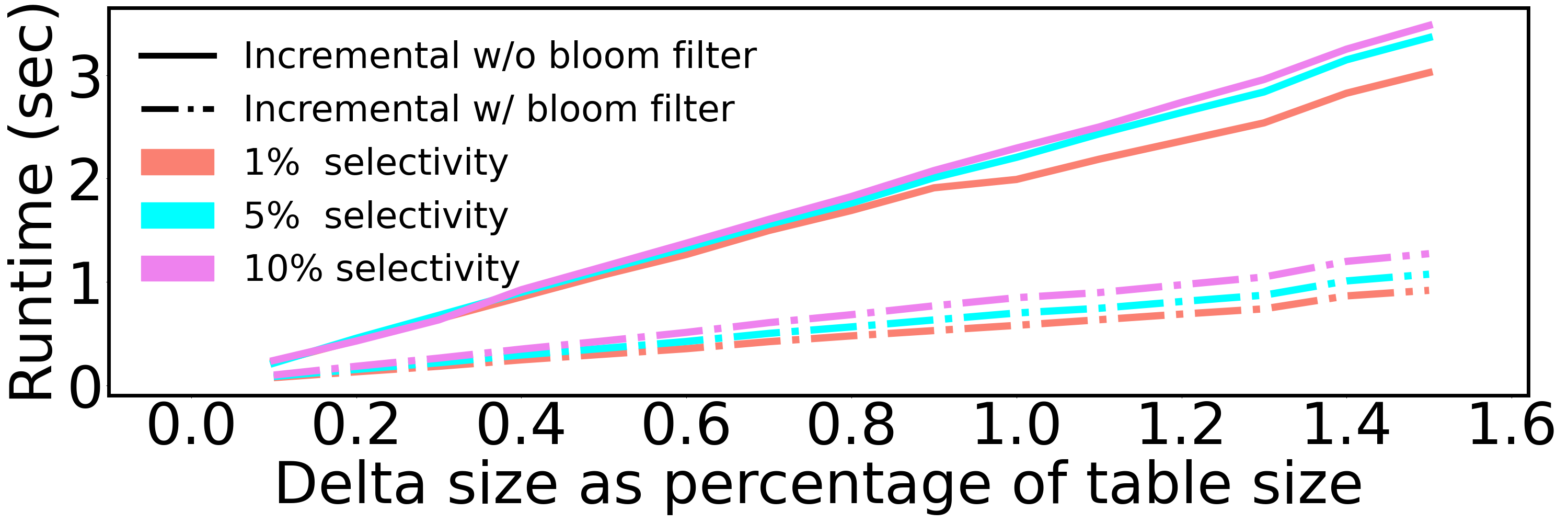}
       \vspace{-5mm}
       \caption{Optimization: bloom filter large delta size}
      \label{fig:syn-large-bloom}
     \end{subfigure}
  \end{minipage}
   \vspace{-3mm}
  \caption{Optimizations in \impAbbr: filtering deltas based on selection conditions  and using bloom filters for joins.}
  \label{fig:optimization}
\end{figure*}
}

%%% Local Variables:
%%% mode: LaTeX
%%% TeX-master: "../imp"
%%% End:

%% file: sections/fig-topk-time-and-mem.tex
\iftechreport{
\begin{figure*}[h]
%%  a general legend
\begin{minipage}{1.0\linewidth}
\includegraphics[width=1.0\textwidth, trim=0  0 0 0,clip, scale=0.2]{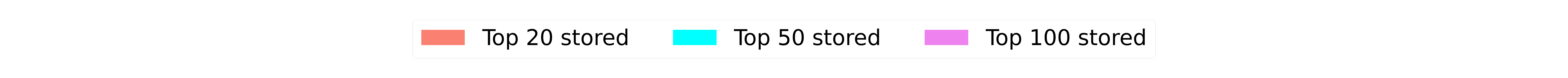}
\vspace{-5mm}
\end{minipage}
%% row 1: runtime of topk
\begin{minipage}{1.0\linewidth}
  %% always delete minimal group
  \begin{subfigure}{0.24\textwidth}
  % \begin{subfigure}{0.24\textwidth}
	\includegraphics[width=1.0\textwidth, trim=0  0 0 0,clip]{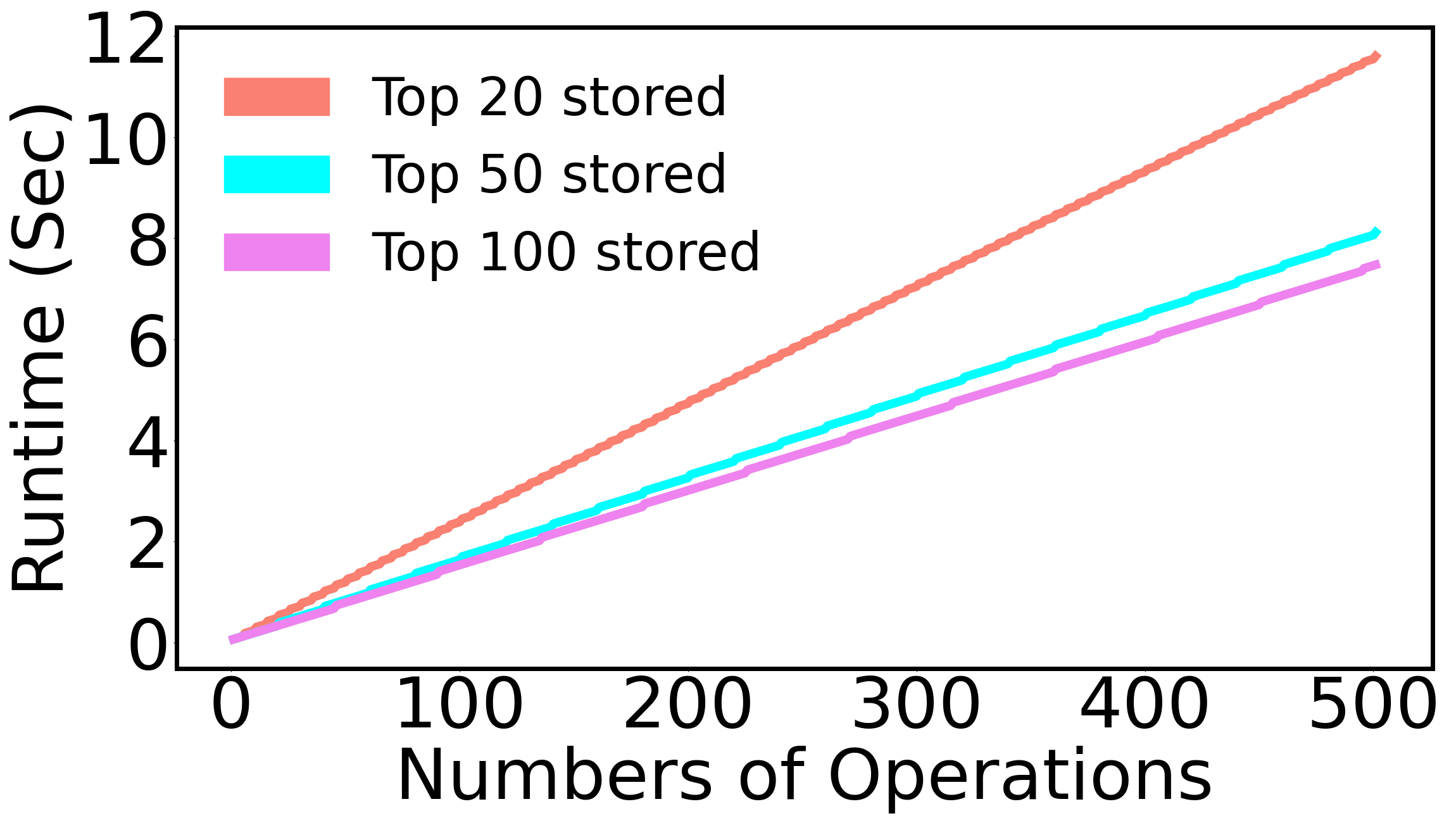}
    \vspace{-5mm}
    \caption{\qtopk delete minimal group}
    \label{fig:tk-del-min-runtime}
  \end{subfigure}
  %% random : minimal = 2:1
  \begin{subfigure}{0.24\textwidth}
	\includegraphics[width=1.0\textwidth, trim=0  0 0 0,clip]{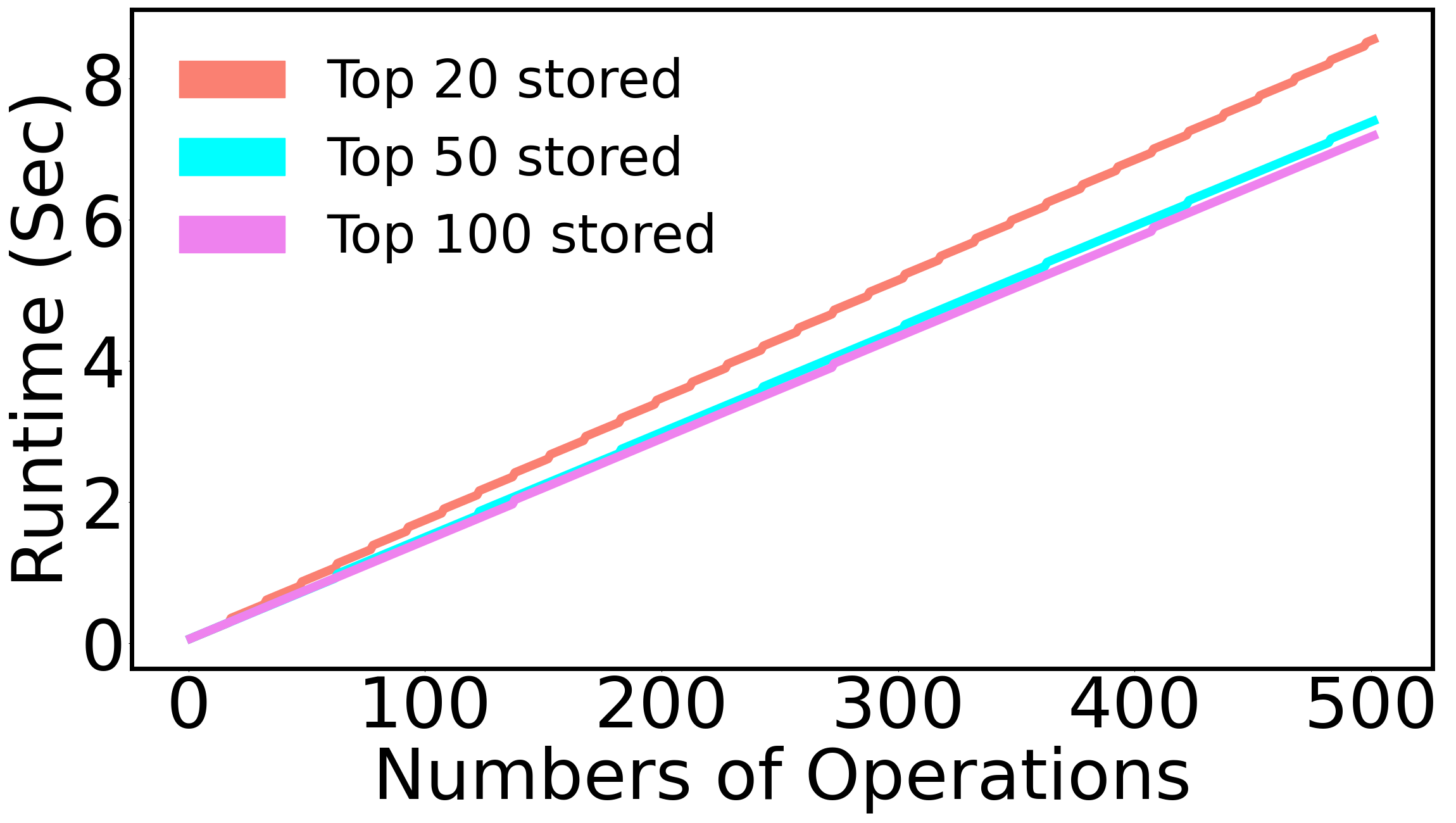}
    \vspace{-5mm}
    \caption{\qtopk ratio 2:1}
    \label{fig:tk-ratio-2-1-runtime}
  \end{subfigure}
  %% random : minimal = 4:1
  \begin{subfigure}{0.24\textwidth}
	\includegraphics[width=1.0\textwidth, trim=0  0 0 0,clip]{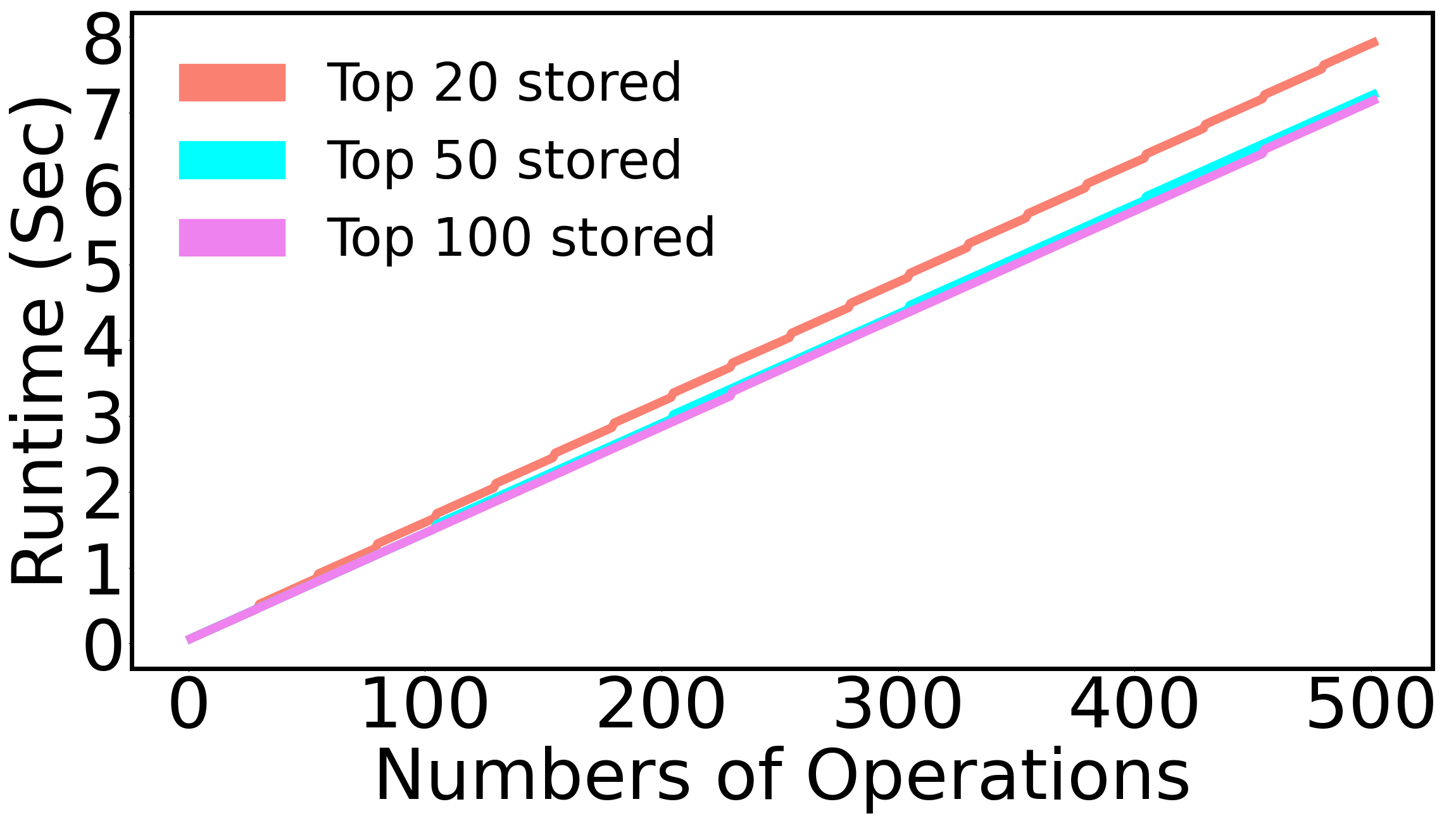}
    \vspace{-5mm}
    \caption{\qtopk ratio 4:1}
    \label{fig:tk-ratio-4-1-runtime}
  \end{subfigure}
  %% alsays delete ramdon value
  \begin{subfigure}{0.24\textwidth}
	\includegraphics[width=1.0\textwidth, trim=0  0 0 0,clip]{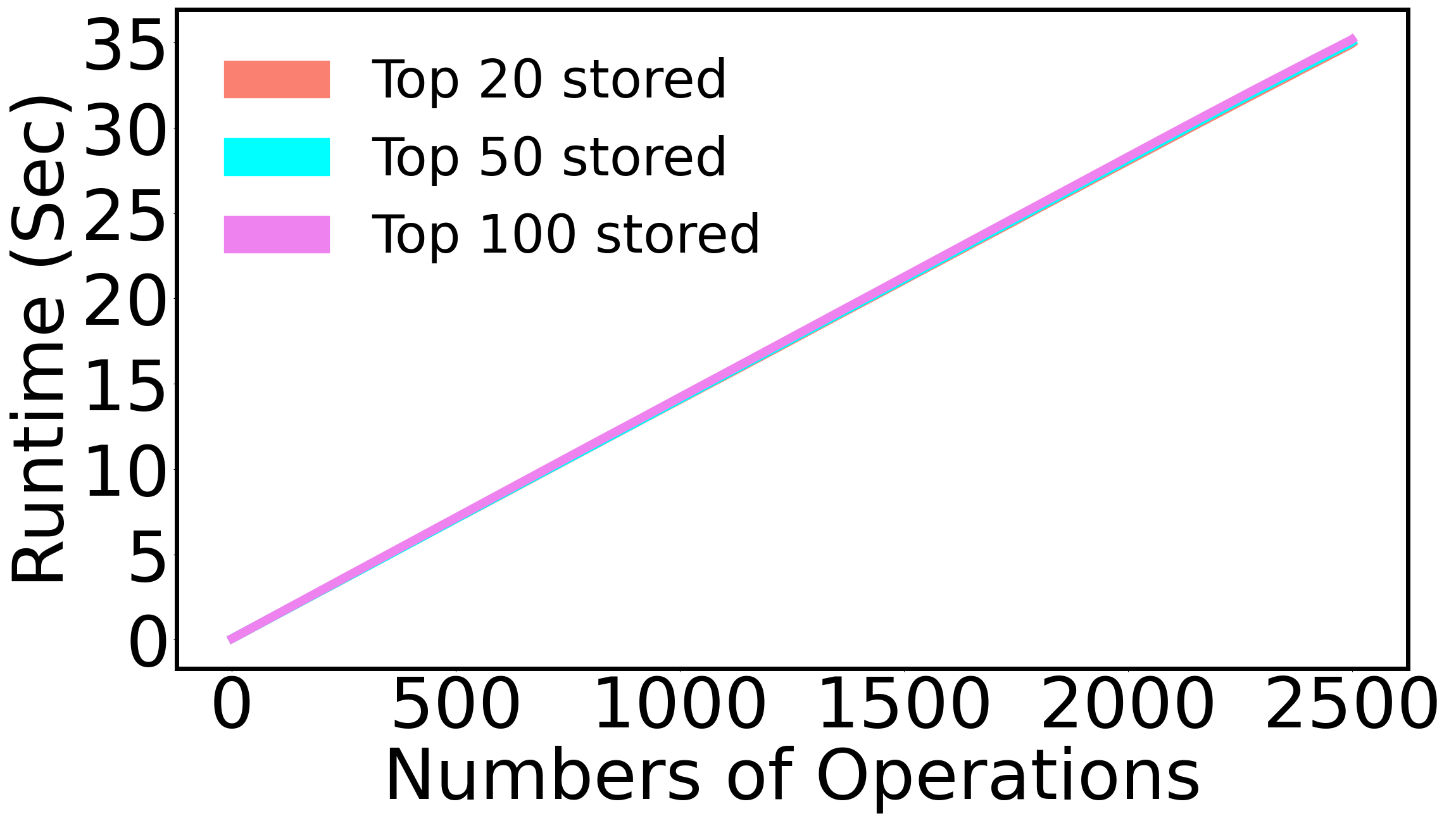}
    \vspace{-5mm}
    \caption{\qtopk delete randomly}
    \label{fig:tk-del-random-runtime}
  \end{subfigure}
  \caption{Top-K: runtime of varying maintenance strategies}\label{fig:topk-runtime}
\end{minipage}
%% row 2: memory of topk
\begin{minipage}{1.0\linewidth}
  %% always delete minimal group
  \begin{subfigure}{0.24\textwidth}
	\includegraphics[width=1.0\textwidth, trim=0  0 0 0,clip]{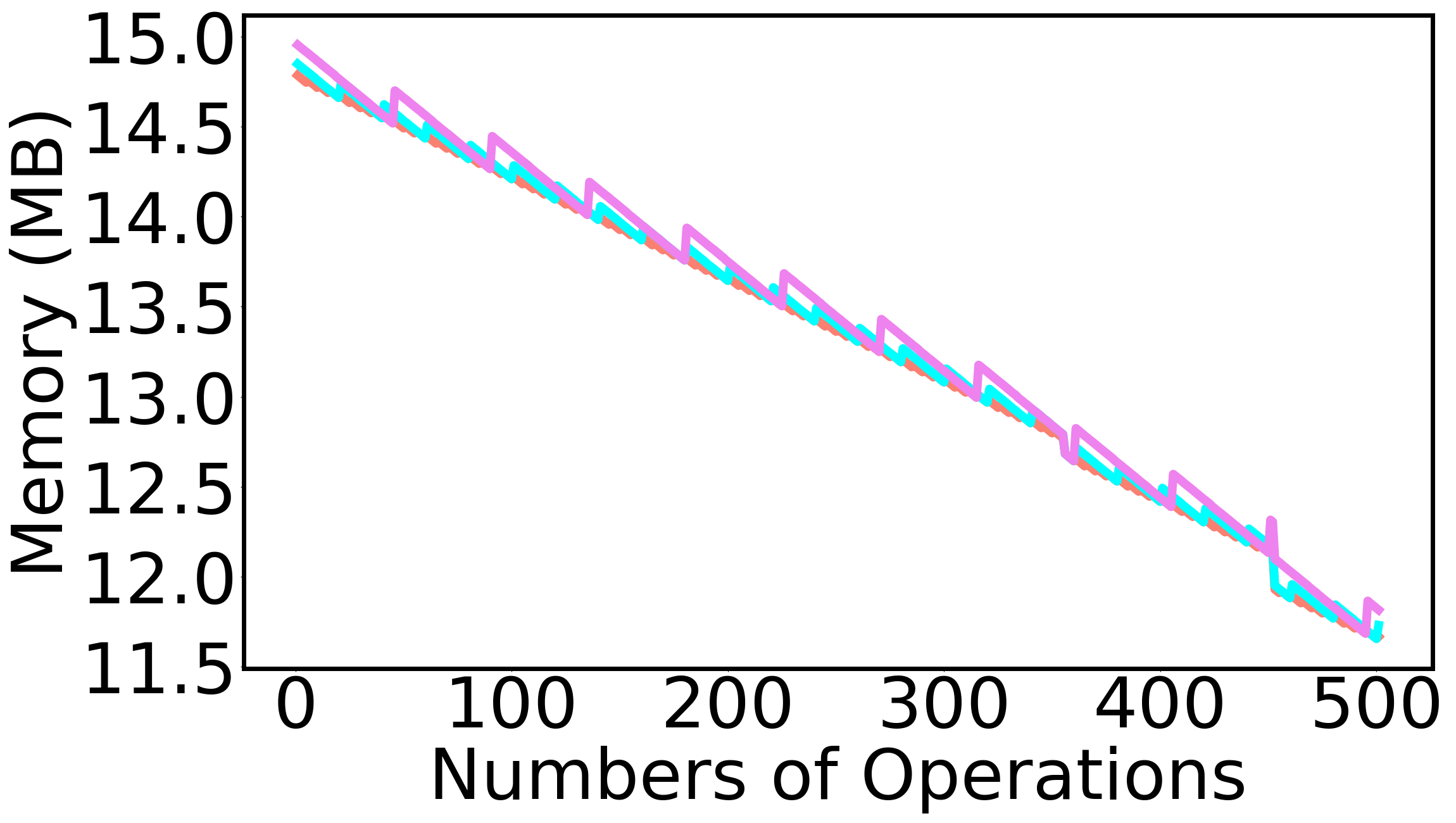}
    \vspace{-5mm}
    \caption{\qtopk delete minimal group}
    \label{fig:tk-del-min-memory}
  \end{subfigure}
  %% random : minimal = 2:1
  \begin{subfigure}{0.24\textwidth}
	\includegraphics[width=1.0\textwidth, trim=0  0 0 0,clip]{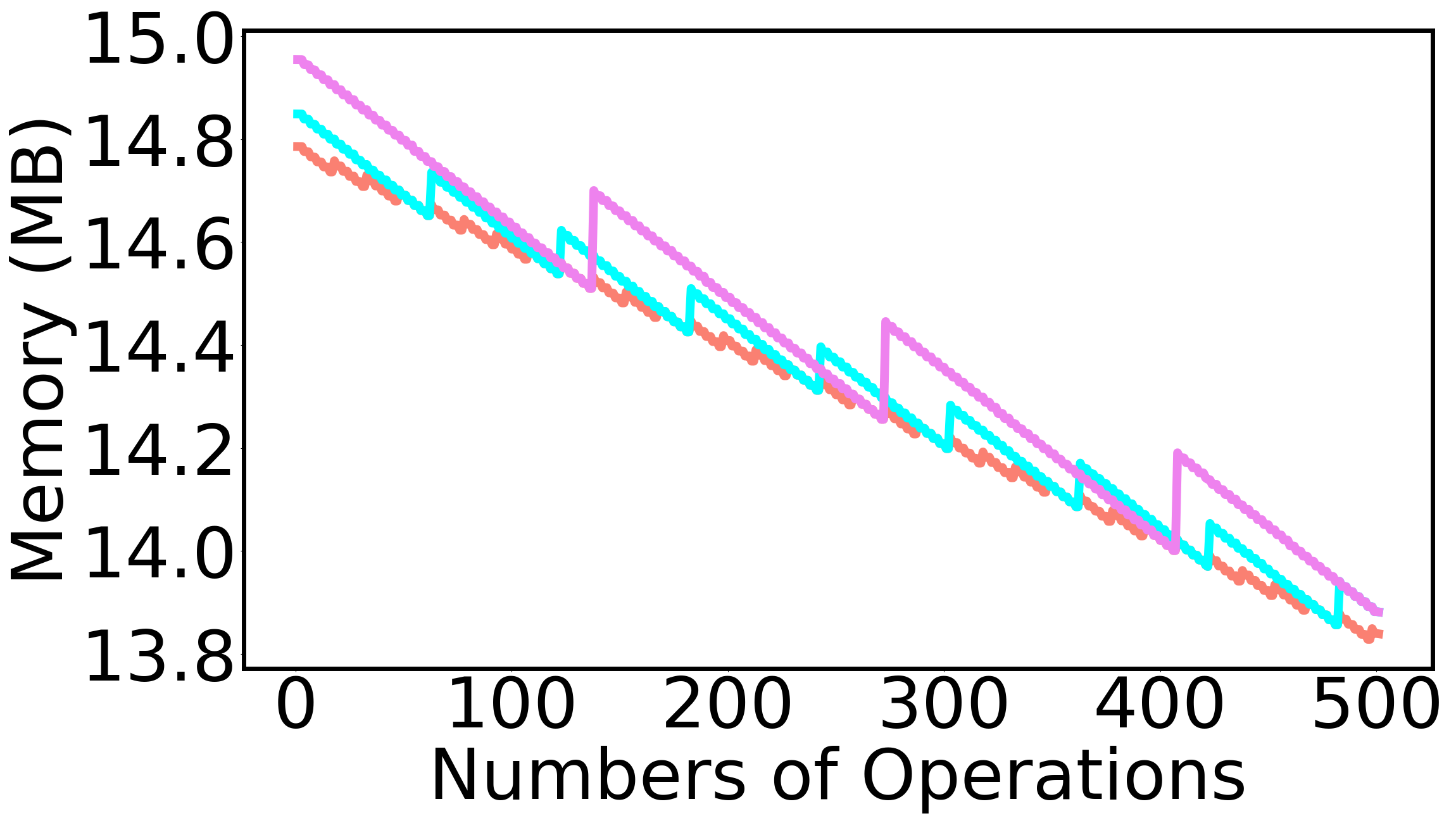}
    \vspace{-5mm}
    \caption{\qtopk ratio 2:1}
    \label{fig:tk-ratio-2-1-memory}
  \end{subfigure}
  %% random : minimal = 4:1
  \begin{subfigure}{0.24\textwidth}
	\includegraphics[width=1.0\textwidth, trim=0  0 0 0,clip]{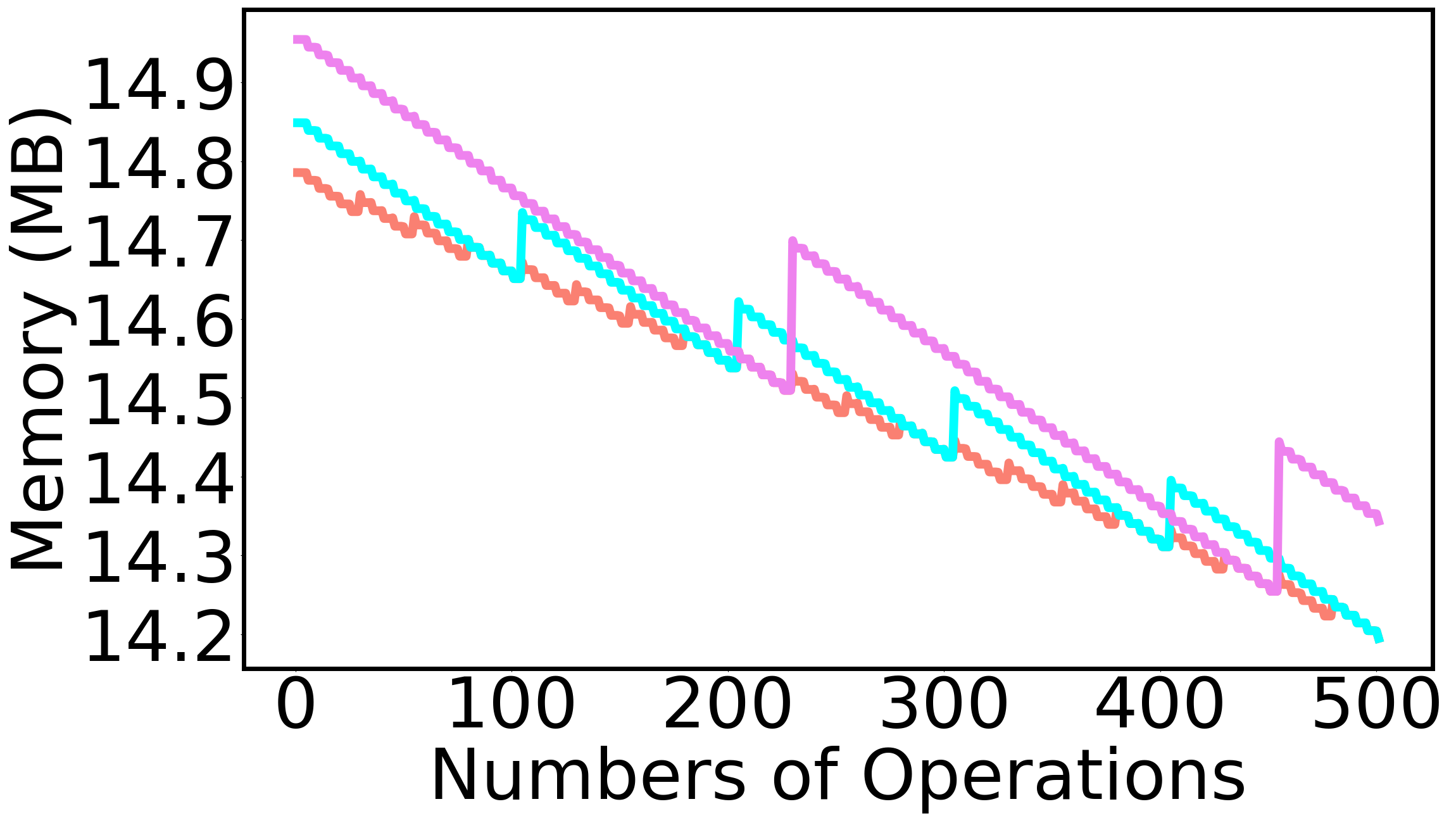}
    \vspace{-5mm}
    \caption{\qtopk ratio 4:1}
    \label{fig:tk-ratio-4-1-memory}
  \end{subfigure}
  %% alsays delete ramdon value
  \begin{subfigure}{0.24\textwidth}
	\includegraphics[width=1.0\textwidth, trim=0  0 0 0,clip]{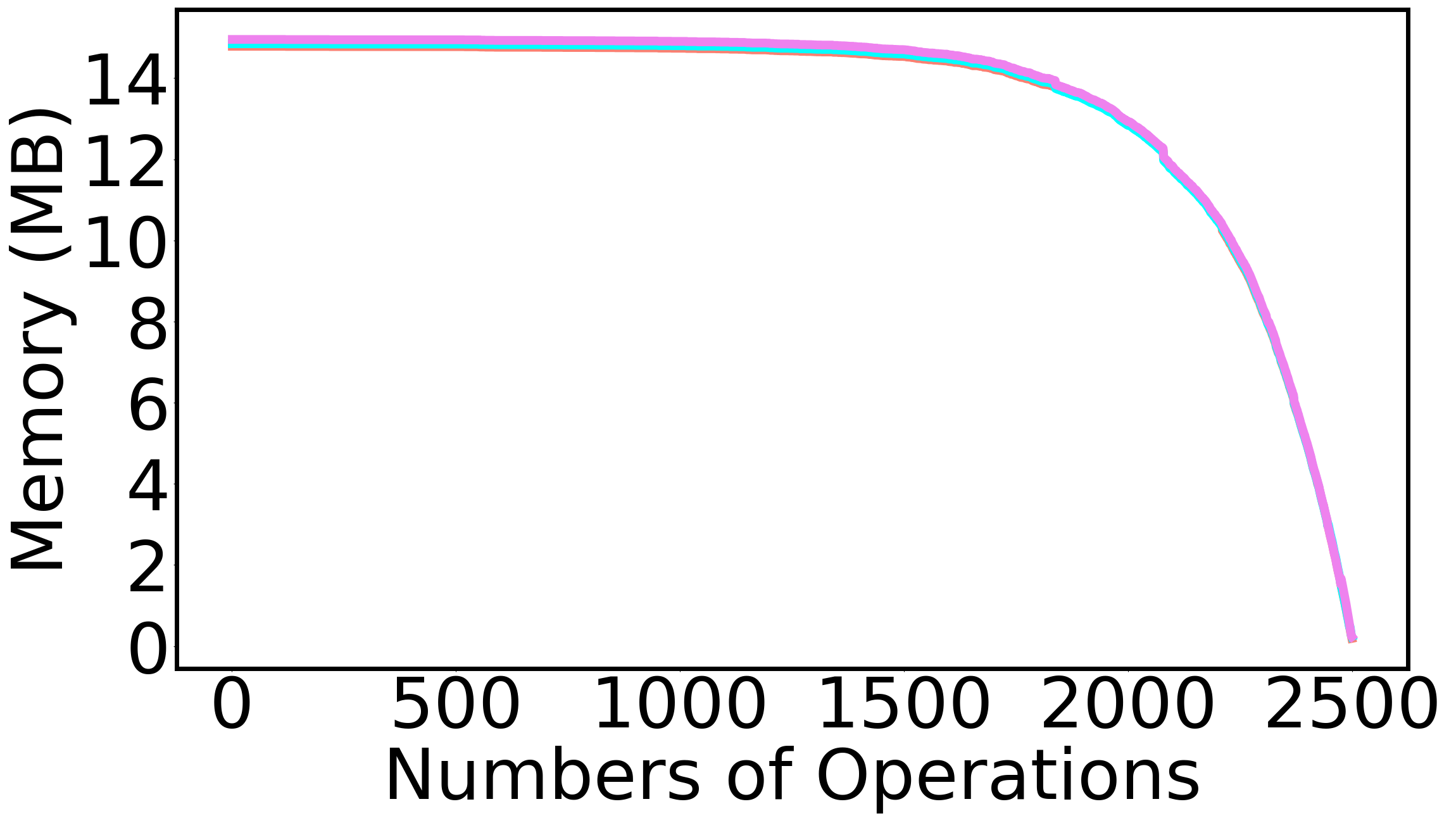}
    \vspace{-5mm}
    \caption{\qtopk delete randomly}
    \label{fig:tk-del-random-memory}
  \end{subfigure}
  \caption{Top-K: memory usage of varying maintenance strategies}\label{fig:topk-memory}
\end{minipage}
\end{figure*}
}

\ifnottechreport{
 \begin{figure}[t]
%%  a general legend
% \begin{minipage}{1.0\linewidth}
% \includegraphics[width=1.0\textwidth, trim=0  0 0 0,clip, scale=0.2]{figs/topk_lines_legends.png}
% \vspace{-5mm}
% \end{minipage}
\begin{minipage}{1.0\linewidth}
  %% always delete minimal group
  \begin{subfigure}{0.49\textwidth}
  % \begin{subfigure}{0.24\textwidth}
	\includegraphics[width=1.0\textwidth, trim=0  0 0 0,clip]{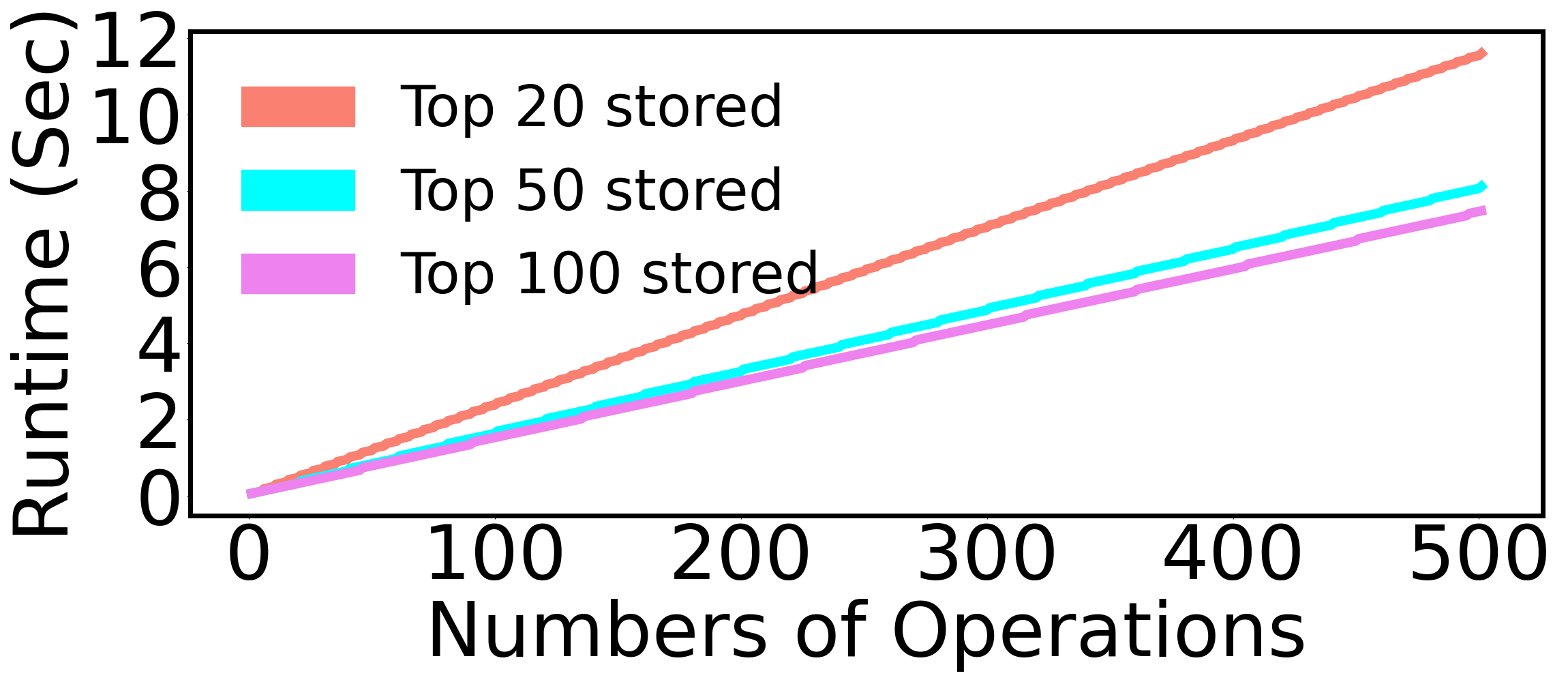}
    \vspace{-5mm}
    \caption{Deleting from top-k groups}
    \label{fig:tk-del-min-runtime}
  \end{subfigure}
  %% alsays delete ramdon value
  \begin{subfigure}{0.49\textwidth}
	\includegraphics[width=1.0\textwidth, trim=0  0 0 0,clip]{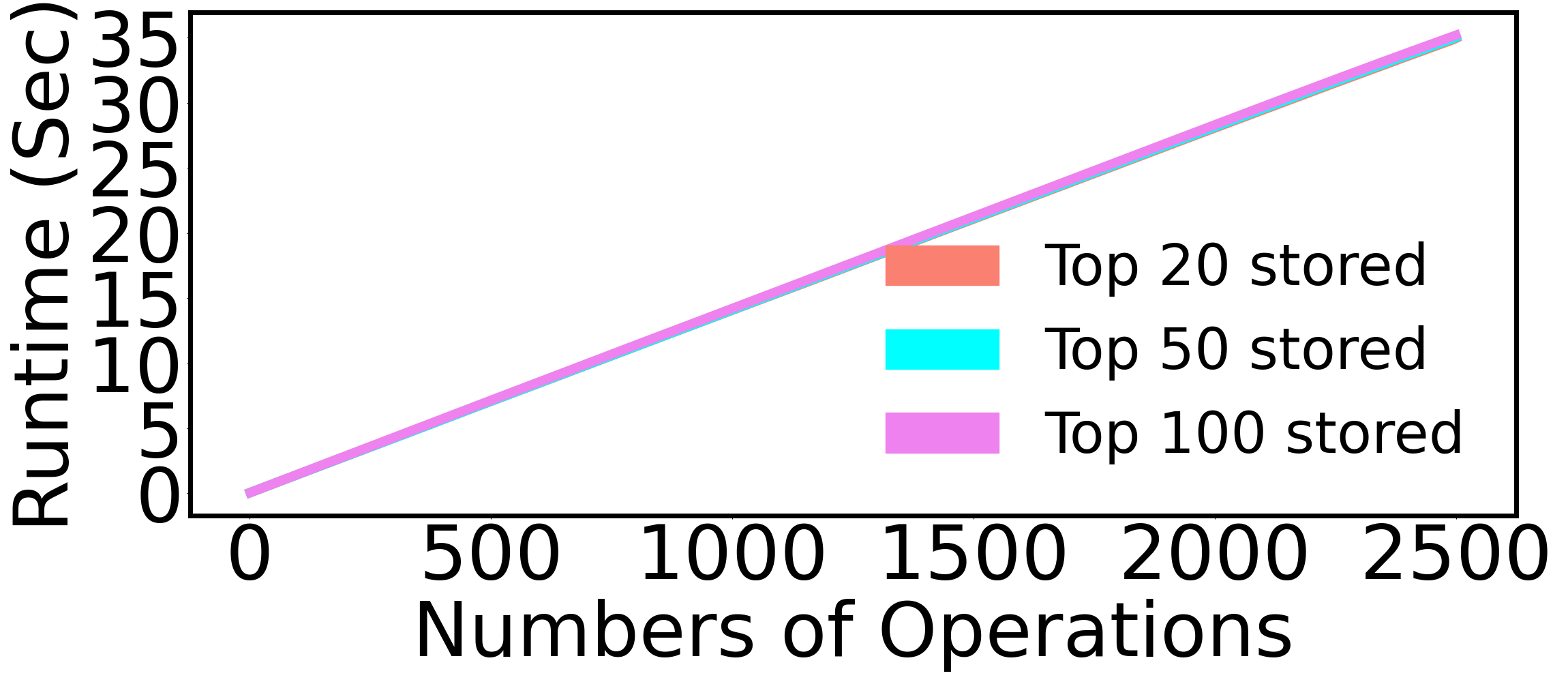}
    \vspace{-5mm}
    \caption{Deleting random tuples}
    \label{fig:tk-del-random-runtime}
  \end{subfigure}
  \vspace{-4mm}
  \caption{Top-k runtime, varying buffer size $l$}\label{fig:topk-runtime}
\end{minipage}
\end{figure}
}

%%% Local Variables:
%%% mode: LaTeX
%%% TeX-master: "../imp"
%%% End:

%% file: sections/fig-dynamic-adaption.tex
% figs for lazy and eager dynamic adaption
\iftechreport{
\begin{figure}[t]
  \begin{minipage}{1.0\linewidth}

%    %% 1U 1Q
%    %% eager vs lazy for 9900 update, eager per 5 update.
%    \begin{subfigure}{0.33\linewidth}
%      \includegraphics[width = 1.0\textwidth]{figs/differentPercentageLazyEager5UpdatePerEager.png}
%      \vspace{-6mm}
%      \caption{Eager: per 5 update}
%      \label{fig:eager_per_5_updates}
%    \end{subfigure}
    %% 5U 1Q
    % \begin{subfigure}{0.49\linewidth}
    %   \includegraphics[width = 1.0\textwidth]{figs/differentPercentageLazyEager100UpdatePerEager.png}
    %   \vspace{-6mm}
    %   \caption{Eager: per 100 update}
    %   \label{fig:eager_per_100_update}
    % \end{subfigure}
%% 1000 update evenly divided into some parts
    % \begin{subfigure}{1.0\linewidth}
      \includegraphics[width = 1.0\textwidth]{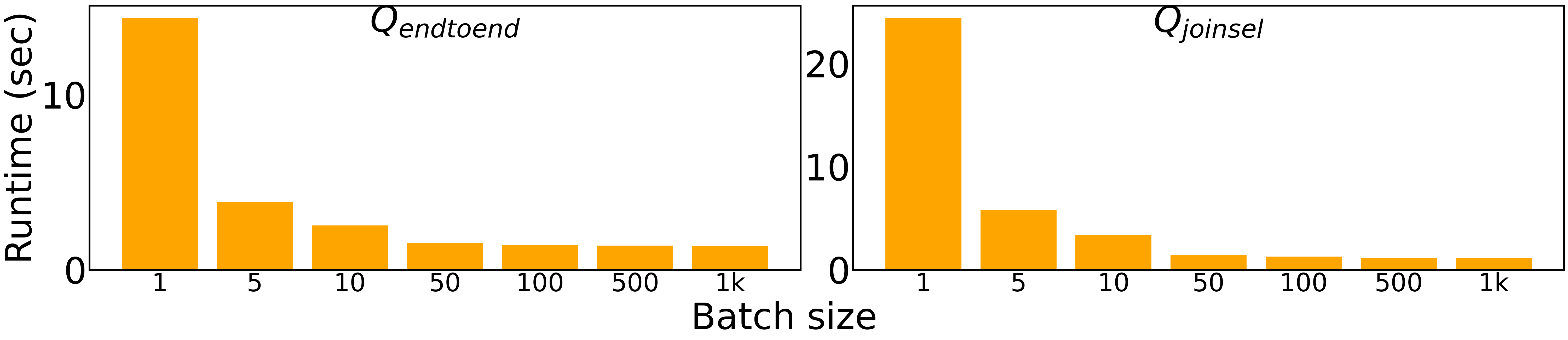}
      \vspace{-6mm}
      % \caption{}
      % \label{fig:evenly_divided_1k_update}
    % \end{subfigure}
  \end{minipage}
  \vspace{-2mm}
  \caption{Cost of maintaining 1000 updates using eager maintenance, varying batch size.}
      \label{fig:evenly_divided_1k_update}
  % \label{fig:dynamic_adaption}
\end{figure}
}
%%% Local Variables:
%%% mode: LaTeX
%%% TeX-master: "../imp"
%%% End:

%% file: sections/fig-mem-agg-join.tex
\iftechreport{
\begin{figure}[h]
\begin{minipage}{1.0\linewidth}
  % agg
  \begin{subfigure}{0.49\textwidth}
  % \begin{subfigure}{0.24\textwidth}
	\includegraphics[width=1.0\textwidth, trim=0  0 0 0,clip]{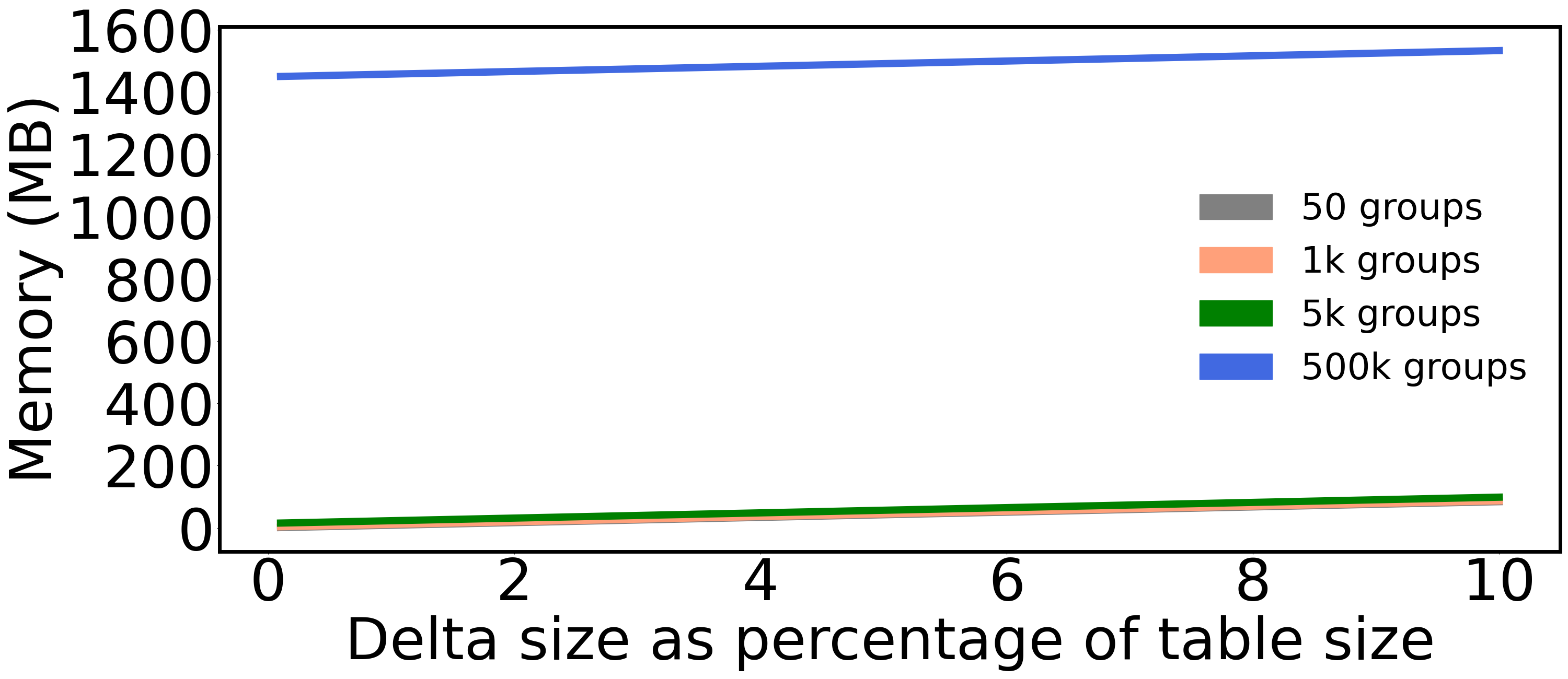}
    \vspace{-5mm}
    \caption{\qgroups Memory usage}
    \label{fig:agggroups_memory}
  \end{subfigure} 
  %% join sel with BF
  \begin{subfigure}{0.49\textwidth}
	\includegraphics[width=1.0\textwidth, trim=0  0 0 0,clip]{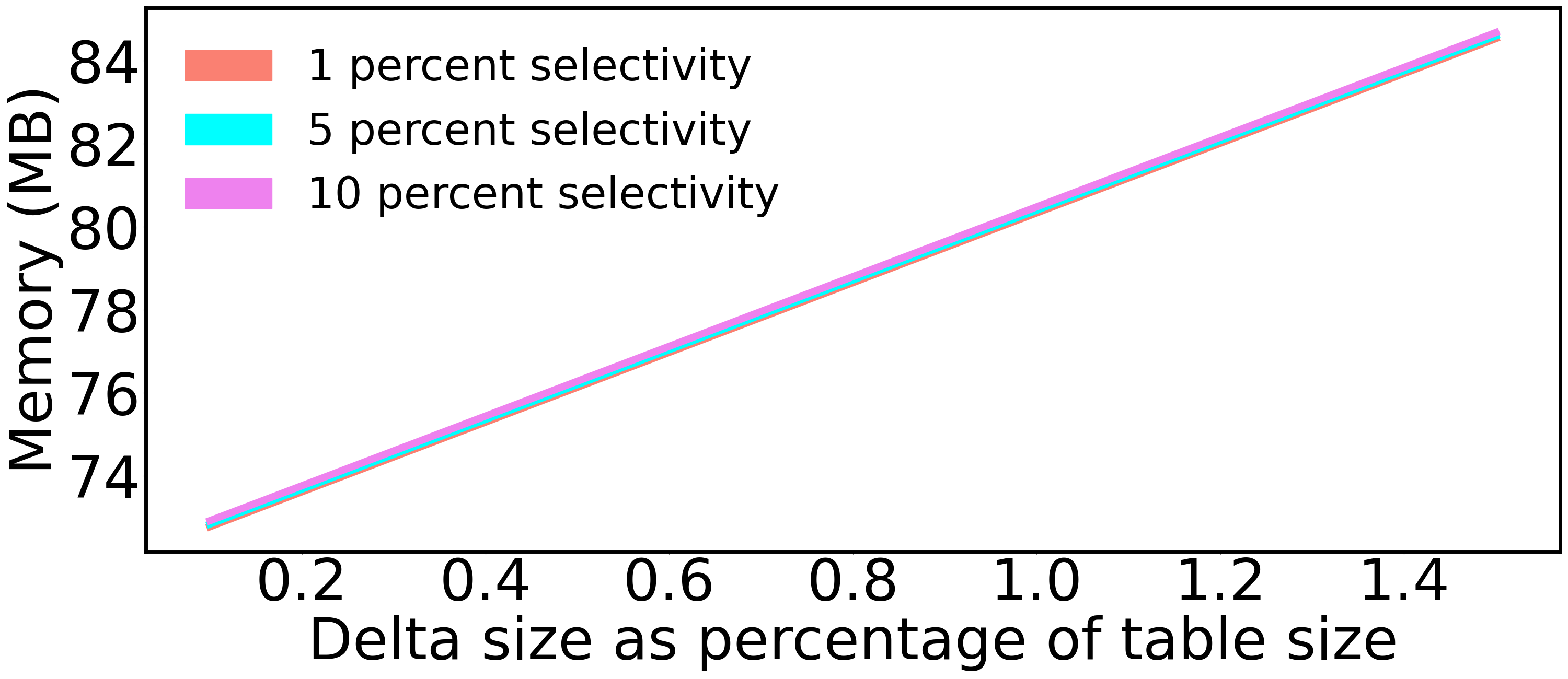}
    \vspace{-5mm}
    \caption{\qjoinsel Memory usage}
    \label{fig:join_memory}
  \end{subfigure}
  \caption{Memory usage of \qgroups}
  \label{fig:agg-join-memory}
\end{minipage}
\end{figure}
}

%%% Local Variables:
%%% mode: LaTeX
%%% TeX-master: "../imp"
%%% End:

%% file: sections/range-sketch-size-table.tex
%%%%%%%%%%%%%%%%%%%%%%%%%%%%%%%%%%%%%%%%%%%%%%%%%%%%%%%%%%%%%%%%%%%%%%%%%%%%%%%%
%%% Sketch and Range size
%% 100	200	500	1000	2000	5000	10000	20000	100000
%% Sketches Memory	
%% Ranges Memory	
\begin{figure*}[t]
{\footnotesize
  \begingroup
  % \color{blue}
  \centering
  \begin{tabular}{llllllllll}
    \hline
     Number of sketches or ranges &  $100$  &  $200$ &  $500$  & $1000$  & $2000$  & $5000$    & $10000$ & $20000$  & $100000$ \\ \hline
     Memory of Sketches (MB)      & 0.000040 & 0.000056 & 0.000088 & 0.000152 & 0.000280  &  0.000656  &  0.001280 &   0.002528 &   0.012528 \\ \hline
     Memory of Ranges (MB)        & 0.004508 & 0.008908 & 0.022108 & 0.044108 & 0.088108  &  0.220108  &  0.440108 &   0.880108 &   4.400108 \\ \hline
  \end{tabular}
  \vspace{-3mm}
  \caption{Sizes different sketches and ranges in memory} 
  \label{fig:sketch-range-size}
  \endgroup
}
\end{figure*}
%%%%%%%%%%%%%%%%%%%%%%%%%%%%%%%%%%%%%%%%%%%%%%%%%%%%%%%%%%%%%%%%%%%%%%%%%%%%%%%%
%% Sketch and Range size
%% 100	200	500	1000	2000	5000	10000	20000	100000
%% Sketches Memory	
%% Ranges Memory	

%%%%%%%%%%%%%%%%%%%%%%%%%%%%%%%%%%%%%%%%%%%%%%%%%%%%%%%
%%% Local Variables:
%%% mode: LaTeX
%%% TeX-master: "../imp"
%%% End:

%% file: sections/conclusion.tex
\section{Conclusions And Future Work}%
\label{sec:conclusion}

% \BG{We write this later. Should contain a short description of what we
% presented in the paper and what the results are as well as what we plan to do
% in future work.}
%%%%%%%%%%%%%%%%%%%%%%%%%%%%%%%%%%%%%%%%%%%%%%%%%%%%%%%%%%%%%%%%%%%%%%%%%%%%%%%%
%% summary of this work
We present the first approach for in-memory incremental maintenance of
provenance sketches. Our \impAbbr  system implement incremental maintenance rules for sketch-annotated
data. % and efficiently produces updates to sketches to correctly reflect which parts of
% the updated database contain provenance information that is needed to answer a
% query and, thus, should belong to the sketch.
Using bloom filters and
selection-push-down, we further improve the performance of the incremental
maintenance process. Our experimental results demonstrate the effectiveness of
our approach and optimizations, outperforming full maintenance by several orders
of magnitude.
%%%%%%%%%%%%%%%%%%%%%%%%%%%%%%%%%%%%%%%%%%%%%%%%%%%%%%%%%%%%%%%%%%%%%%%%%%%%%%%%
%% future work
In future work, in addition to extending \impAbbr with support for more
operators, e.g., set difference and nested or recursive queries, we will
investigate how to integrate provenance-based data skipping and incremental
maintenance of sketches with cost-based query optimization and self-tuning. Another open research question is how  \impAbbr can be extended as a general \gls{ivm} engine for provenance information. While many of the incremental operator semantics could be extended to support that, optimizations that rely on approximation may no longer be acceptable and, in contrast to sketches which are of a fixed size, provenance information may be large. Furthermore, we conjecture that \impAbbr could be extended for maintaining summarizes of provenance~\cite{LL20, AB15a, alomeir-20-spagqrrd} which, like sketches, can tolerate approximation and are typically small. \revision{For example, consider a data integration scenario where we want to track the set of data sources a query result depends on.}

% to maintain the provenance as well. To maintain
% the provenance for queries, the data model will be adapted to provenance
% annotated tuples instead sketch annotated tuples and the data structures used for
% maintaining operators like aggregation and top-k should be the same, while they
% should record the provenance information not the sketch information. We can skip
% the merging step when maintaining the provenance for queries and output a delta
% provenance for queries.
%\PL{add imp for provenance maintenance.}

% we will investigate to improve \impAbbr engine to support wider
% range of queries. Investigate approaches to further optimize operators to
% achieve more efficient maintenance. We will study different cost estimation
% models. Get concrete strategies for online fashion queries running with database
% update between using and not using provenance sketches and how to maintain
% sketches. Since for relational database, the relational algebra
% is not expressive enough and it lacks flexibility to use specialized data
% structures to store operator state. In the future, it is necessary to exploit
% and maintain the sketches in object database as JSON datatype. Last but not the
% least, consider to exploit the provenance related work over vector data and
% develop maintenance for vector data.

%%% Local Variables:
%%% mode: LaTeX
%%% TeX-master: "../imp"
%%% End:

%% file: sections/artifacts.tex
\subsection*{Artifacts}
% \label{sec:artifacts}
%
Supplementary material, \impAbbr's source code (folder \texttt{IMP engine} and \texttt{postgreSQL}), and experiment scripts and data (folder \texttt{dataset})
are available at \cite{availability}.
% \url{https://github.com/IITDBGroup/IMP_EDBT26}. This repository includes a \texttt{IMP
% engine} folder with the source code for \impAbbr, a \verb!postgreSQL! source code folder (we are using Postgres as the database backend for all experiments), a dataset folder including the scripts to generate the data for all experiments and
% a folder that contains all the queries and scripts to run the experiments.

%%% Local Variables:
%%% mode: LaTeX
%%% TeX-master: "../imp"
%%% End:

%% file: sections/appendix_qlist.tex
\section{Appendix I: Query List}%
\label{sec:appendix_qlist}
\subsection{Synthetic dataset query list}
\subsubsection{Different number of aggregation functions \qhaving}
~\\
\parttitle{One aggregation function}
\begin{lstlisting}
    SELECT a, avg(b) AS ab
    FROM r500
    GROUP BY a
\end{lstlisting}
\parttitle{Two aggregation functions}
\begin{lstlisting}
    SELECT a, avg(b) AS ab
    FROM r500
    GROUP BY a
    HAVING avg(c) < 1000
\end{lstlisting}

\parttitle{Three aggregation functions}
\begin{lstlisting}
    SELECT a, avg(b) AS ab
    FROM r500
    GROUP BY a
    HAVING avg(c) < 1000 and avg(d) < 1200
\end{lstlisting}

\parttitle{Ten aggregation functions}
\begin{lstlisting}
    SELECT a, avg(b) AS ab
    FROM r500
    GROUP BY a
    HAVING avg(c) < 1000 and avg(d) < 1200 and avg(e) > 0
        and avg(f) > 0 and avg(g) > 0 and avg(h) > 0
        and avg(i) > 0 and avg(j) > 0
\end{lstlisting}
\subsubsection{Number of groups  \qgroups}
~\\
\parttitle{50 groups}
\begin{lstlisting}
    SELECT a, avg(b) AS ab
    FROM t1gb50g
    GROUP BY a
    HAVING avg(c) < 3
\end{lstlisting}

\parttitle{1K groups}
\begin{lstlisting}
    SELECT a, avg(b) AS ab
    FROM t1gb1000g
    GROUP BY a
    HAVING avg(c) < 320
\end{lstlisting}

\parttitle{5K groups}
\begin{lstlisting}
    SELECT a, avg(b) AS ab
    FROM t1gb5000g
    GROUP BY a
    HAVING avg(c) < 1600
\end{lstlisting}

\parttitle{500K groups}
\begin{lstlisting}
    SELECT a, avg(b) AS ab
    FROM t1gb500000g
    GROUP BY a
    HAVING avg(c) < 1600
\end{lstlisting}

\subsubsection{Join}
~\\
\parttitle{1-200K joins}
\begin{lstlisting}
    SELECT a, avg(b) AS ab
    FROM (
        SELECT a AS a, b AS b, c AS c
        FROM t1gb50g
        WHERE b < 10
    ) tt JOIN tjoinhelp ON (a = ttid)
    GROUP BY a
    HAVING avg(c) < 10
\end{lstlisting}
\parttitle{1-2K joins}
\begin{lstlisting}
    SELECT a, avg(b) AS ab
    FROM(
        SELECT a AS a, b AS b, c AS c
        FROM t1gbjoin WHERE b < 1000
    ) tt JOIN tjoinhelp ON (a = ttid)
    GROUP BY a
    HAVING avg(c) < 1000
\end{lstlisting}

\parttitle{1-20 joins}
\begin{lstlisting}
    SELECT a, avg(b) AS ab
    FROM (
        SELECT a AS a, b AS b, c AS c
        FROM t1gb500000g
        WHERE b < 100000
    ) tt JOIN tjoinhelp ON (a = ttid)
    GROUP BY a
    HAVING avg(c) < 100000
\end{lstlisting}

\subsubsection{Join selectivity: \qjoinsel}
~\\
\parttitle{1\% selectivity}
\begin{lstlisting}
    SELECT a, avg(b) AS ab
    FROM t1gbjoin JOIN tjoinhelppercent1 ON (a = ttid)
    WHERE b < 1000
    GROUP BY a
    HAVING avg(c) < 1000
\end{lstlisting}

\parttitle{5\% selectivity}
\begin{lstlisting}
    SELECT a, avg(b) AS ab
    FROM t1gbjoin JOIN tjoinhelppercent5 ON (a = ttid)
    WHERE b < 1000
    GROUP BY a
    HAVING avg(c) < 1000
\end{lstlisting}

\parttitle{10\% selectivity}
\begin{lstlisting}
    SELECT a, avg(b) AS ab
    FROM t1gbjoin JOIN tjoinhelppercent10 ON (a = ttid)
    WHERE b < 1000
    GROUP BY a
    HAVING avg(c) < 1000
\end{lstlisting}

\subsubsection{Fragment number: \qsketch}
~\\
\begin{lstlisting}
    SELECT a, avg(b) as ab
    FROM (
          SELECT a as a, b as b, c as c
          FROM t1gbjoin
          WHERE b < 1000) tt
    JOIN tjoinhelp on (a = ttid)
    GROUP BY a
    HAVING avg(c) < 1000
\end{lstlisting}

\subsubsection{Delta filter by selection push down \qselpd}
~\\
\begin{lstlisting}
    SELECT a, avg(b) AS ab
    FROM t1gb1000g
    WHERE b < 1000
    GROUP BY a
    HAVING avg(c) < 300
\end{lstlisting}
\subsubsection{End-to-end \qendtoend}
~\\
\begin{lstlisting}
  SELECT a, avg(c) AS ac
  FROM edb1
  GROUP BY a
  HAVING avg(c)> 1684845 AND avg(c) < 1686014;
\end{lstlisting}
\subsection{Crimes dataset query list}
\parttitle{Q1}
\begin{lstlisting}
    SELECT beat, year, count(id) AS crime_count
    FROM crimes
    GROUP BY beat, year
\end{lstlisting}
\parttitle{Q2}
\begin{lstlisting}
    SELECT district, community_area, ward, beat,
        count(beat) AS crime_count
    FROM crimes
    GROUP BY district, community_area, ward, beat
    HAVING count(id) > 1000
\end{lstlisting}
\subsection{Top-K query \qtopk}
\begin{lstlisting}
    SELECT a, avg(b) AS ab
    FROM R
    GROUP BY a
    ORDER BY a
    LIMIT 10
\end{lstlisting}

\subsection{\tpchds Top-K query \qspace}
\begin{lstlisting}
  SELECT c_custkey, c_name,
     sum(l_extendedprice * (1 - l_discount)) AS revenue,
     c_acctbal, n_name, c_address, c_phone, c_comment
  FROM customer, orders, lineitem, nation
  WHERE  c_custkey = o_custkey and l_orderkey = o_orderkey
     AND o_orderdate >= to_date('1994-12-01', 'YYYY-MM-DD')
     AND o_orderdate < to_date('1995-03-01', 'YYYY-MM-DD')
     AND l_returnflag = 'R'
     AND c_nationkey = n_nationkey
  GROUP BY c_custkey, c_name, c_acctbal, c_phone,
     n_name, c_address, c_comment
  ORDER BY revenue
  LIMIT 20
\end{lstlisting}
%%%%%%%%%%%%%%%%%%%%%%%%%%%%%%%%%%%%%%%%%%%%%%%%%%%%%%%%%%%%%%%%%%%%%%%%%%%%%%%%
%%% Local Variables:
%%% mode: latex
%%% TeX-master: "../imp_techreport"
%%% End:

%% file: sections/implementation_sql.tex
\section{SQL-based Strategy}
\label{sec:sql_implementation_appendix}
%%%%%%%%%%%%%%%%%%%%%%%%%%%%%%%%%%%%%%%%%%%%%%%%%%%%%%%%%%%%%%%%%%%%%%%%%%%%%%%%
% \BG{Not sure this deserves a full section. Basically, we acknowledge that this
% is possible, but then argue why this is not we end up doing} 
%\PL{These paragraph has been merged to in-memory implementation section at the
%beginning of that section. This section can be removed.} 
A pure SQL-based strategy is an option to update the provenance sketches
incrementally. Accordingly, the state data required for operators is persisted
as database tables. Then the incremental maintenance is modeled as running a
series of queries over tables. For each operator of a query, queries are
executed to process its input and to handle state data if required for both
utilizing and updating the data.

This SQL-based approach has the advantages: first, there is no need to transfer
data between the client and DBMS. What is more, DBMS can offer good plans and
apply optimizations for query executions. Furthermore, data-heavy operations are
executed inside the database, which reduces potential bottlenecks arising from
data transfer across systems. All these advantages can enhance the efficiency of
the incremental procedure.

However, this approach introduces challenges: first, it lacks flexibility to use
specialized data structures to store operator state. All state data will be
maintained in tables. But different operators' state data can be kept in
different data structures in-memory: all groups' average can be stored in a map
and order and limit operators together can use binary search trees to fast
access Top-K elements. In addition, state data storing in tables is less
efficient compared to in in-memory data structures. For example, to update the
average value of a group, the state data table should be scanned twice:
accessing and updating data. While, these two operations can be done more faster
if the state data is keep in a map in-memory. Moreover, incremental operations
cannot always be efficiently expressed in SQL.

%%% Local Variables:
%%% mode: latex
%%% TeX-master: "../imp_techreport"
%%% End:

%% file: main.bbl

\begin{thebibliography}{47}

%%% the Web address.

\ifx \showCODEN    \undefined \def \showCODEN     #1{\unskip}     \fi
\ifx \showISBNx    \undefined \def \showISBNx     #1{\unskip}     \fi
\ifx \showISBNxiii \undefined \def \showISBNxiii  #1{\unskip}     \fi
\ifx \showISSN     \undefined \def \showISSN      #1{\unskip}     \fi
\ifx \showLCCN     \undefined \def \showLCCN      #1{\unskip}     \fi
\ifx \shownote     \undefined \def \shownote      #1{#1}          \fi
\ifx \showarticletitle \undefined \def \showarticletitle #1{#1}   \fi
\ifx \showURL      \undefined \def \showURL       {\relax}        \fi
% invisible to TeX
\providecommand\bibfield[2]{#2}
\providecommand\bibinfo[2]{#2}
\providecommand\natexlab[1]{#1}
\providecommand\showeprint[2][]{arXiv:#2}

\bibitem[Abadi et~al\mbox{.}(2015)]%
        {abadi-15-fdd}
\bibfield{author}{\bibinfo{person}{Martín Abadi}, \bibinfo{person}{Frank
  McSherry}, {and} \bibinfo{person}{Gordon~D. Plotkin}.}
  \bibinfo{year}{2015}\natexlab{}.
\newblock \showarticletitle{Foundations of Differential Dataflow}. In
  \bibinfo{booktitle}{\emph{ETAPS}}, Vol.~\bibinfo{volume}{9034}.
  \bibinfo{pages}{71--83}.
\newblock


\bibitem[Agrawal et~al\mbox{.}(2000)]%
        {DBLP:conf/vldb/AgrawalCN00}
\bibfield{author}{\bibinfo{person}{Sanjay Agrawal}, \bibinfo{person}{Surajit
  Chaudhuri}, {and} \bibinfo{person}{Vivek~R. Narasayya}.}
  \bibinfo{year}{2000}\natexlab{}.
\newblock \showarticletitle{Automated Selection of Materialized Views and
  Indexes in SQL Databases}. In \bibinfo{booktitle}{\emph{VLDB}}.
  \bibinfo{pages}{496--505}.
\newblock


\bibitem[Ainy et~al\mbox{.}(2015)]%
        {AB15a}
\bibfield{author}{\bibinfo{person}{Eleanor Ainy}, \bibinfo{person}{Pierre
  Bourhis}, \bibinfo{person}{Susan~B. Davidson}, \bibinfo{person}{Daniel
  Deutch}, {and} \bibinfo{person}{Tova Milo}.} \bibinfo{year}{2015}\natexlab{}.
\newblock \showarticletitle{Approximated Summarization of Data Provenance}. In
  \bibinfo{booktitle}{\emph{CIKM}}, \bibfield{editor}{\bibinfo{person}{James
  Bailey}, \bibinfo{person}{Alistair Moffat}, \bibinfo{person}{Charu~C.
  Aggarwal}, \bibinfo{person}{Maarten de~Rijke}, \bibinfo{person}{Ravi Kumar},
  \bibinfo{person}{Vanessa Murdock}, \bibinfo{person}{Timos~K. Sellis}, {and}
  \bibinfo{person}{Jeffrey~Xu Yu}} (Eds.). \bibinfo{pages}{483--492}.
\newblock


\bibitem[AlOmeir et~al\mbox{.}(2020)]%
        {alomeir-20-spagqrrd}
\bibfield{author}{\bibinfo{person}{Omar AlOmeir},
  \bibinfo{person}{Eugenie~Yujing Lai}, \bibinfo{person}{Mostafa Milani}, {and}
  \bibinfo{person}{Rachel Pottinger}.} \bibinfo{year}{2020}\natexlab{}.
\newblock \bibinfo{booktitle}{\emph{Summarizing Provenance of Aggregation Query
  Results in Relational Databases}}.
\newblock \bibinfo{type}{{T}echnical {R}eport}.
\newblock


\bibitem[Arab et~al\mbox{.}(2016)]%
        {AG17}
\bibfield{author}{\bibinfo{person}{Bahareh Arab}, \bibinfo{person}{Dieter
  Gawlick}, \bibinfo{person}{Vasudha Krishnaswamy}, \bibinfo{person}{Venkatesh
  Radhakrishnan}, {and} \bibinfo{person}{Boris Glavic}.}
  \bibinfo{year}{2016}\natexlab{}.
\newblock \showarticletitle{Reenactment for Read-Committed Snapshot Isolation}.
  In \bibinfo{booktitle}{\emph{CIKM}}. \bibinfo{pages}{841--850}.
\newblock


\bibitem[Arab et~al\mbox{.}(2018b)]%
        {AG17c}
\bibfield{author}{\bibinfo{person}{Bahareh Arab}, \bibinfo{person}{Dieter
  Gawlick}, \bibinfo{person}{Vasudha Krishnaswamy}, \bibinfo{person}{Venkatesh
  Radhakrishnan}, {and} \bibinfo{person}{Boris Glavic}.}
  \bibinfo{year}{2018}\natexlab{b}.
\newblock \showarticletitle{Using Reenactment to Retroactively Capture
  Provenance for Transactions}.
\newblock \bibinfo{journal}{\emph{TKDE}} \bibinfo{volume}{30},
  \bibinfo{number}{3} (\bibinfo{year}{2018}), \bibinfo{pages}{599--612}.
\newblock


\bibitem[Arab et~al\mbox{.}(2018a)]%
        {DBLP:journals/debu/ArabFGLNZ17}
\bibfield{author}{\bibinfo{person}{Bahareh~Sadat Arab}, \bibinfo{person}{Su
  Feng}, \bibinfo{person}{Boris Glavic}, \bibinfo{person}{Seokki Lee},
  \bibinfo{person}{Xing Niu}, {and} \bibinfo{person}{Qitian Zeng}.}
  \bibinfo{year}{2018}\natexlab{a}.
\newblock \showarticletitle{GProM - A Swiss Army Knife for Your Provenance
  Needs}.
\newblock \bibinfo{journal}{\emph{IEEE Data Eng. Bull.}} \bibinfo{volume}{41},
  \bibinfo{number}{1} (\bibinfo{year}{2018}), \bibinfo{pages}{51--62}.
\newblock


\bibitem[Bhagwat et~al\mbox{.}(2005)]%
        {DBLP:journals/vldb/BhagwatCTV05}
\bibfield{author}{\bibinfo{person}{Deepavali Bhagwat}, \bibinfo{person}{Laura
  Chiticariu}, \bibinfo{person}{Wang~Chiew Tan}, {and} \bibinfo{person}{Gaurav
  Vijayvargiya}.} \bibinfo{year}{2005}\natexlab{}.
\newblock \showarticletitle{An annotation management system for relational
  databases}.
\newblock \bibinfo{journal}{\emph{VLDBJ}} \bibinfo{volume}{14},
  \bibinfo{number}{4} (\bibinfo{year}{2005}), \bibinfo{pages}{373--396}.
\newblock


\bibitem[Blakeley et~al\mbox{.}(1986)]%
        {DBLP:conf/sigmod/BlakeleyLT86}
\bibfield{author}{\bibinfo{person}{José~A. Blakeley}, \bibinfo{person}{Per
  Åke Larson}, {and} \bibinfo{person}{Frank~Wm. Tompa}.}
  \bibinfo{year}{1986}\natexlab{}.
\newblock \showarticletitle{Efficiently Updating Materialized Views}. In
  \bibinfo{booktitle}{\emph{SIGMOD}}. \bibinfo{pages}{61--71}.
\newblock


\bibitem[Budiu et~al\mbox{.}(2023)]%
        {budiu-23-d}
\bibfield{author}{\bibinfo{person}{Mihai Budiu}, \bibinfo{person}{Tej Chajed},
  \bibinfo{person}{Frank McSherry}, \bibinfo{person}{Leonid Ryzhyk}, {and}
  \bibinfo{person}{Val Tannen}.} \bibinfo{year}{2023}\natexlab{}.
\newblock \showarticletitle{DBSP: Automatic Incremental View Maintenance for
  Rich Query Languages}.
\newblock \bibinfo{journal}{\emph{PVLDB}} \bibinfo{volume}{16},
  \bibinfo{number}{7} (\bibinfo{year}{2023}), \bibinfo{pages}{1601--1614}.
\newblock


\bibitem[Buneman and Clemons(1979)]%
        {DBLP:journals/tods/BunemanC79}
\bibfield{author}{\bibinfo{person}{Peter Buneman} {and}
  \bibinfo{person}{Eric~K. Clemons}.} \bibinfo{year}{1979}\natexlab{}.
\newblock \showarticletitle{Efficiently Monitoring Relational Databases}.
\newblock \bibinfo{journal}{\emph{TODS}} \bibinfo{volume}{4},
  \bibinfo{number}{3} (\bibinfo{year}{1979}), \bibinfo{pages}{368--382}.
\newblock


\bibitem[Ceri and Widom(1991)]%
        {DBLP:conf/vldb/CeriW91}
\bibfield{author}{\bibinfo{person}{Stefano Ceri} {and}
  \bibinfo{person}{Jennifer Widom}.} \bibinfo{year}{1991}\natexlab{}.
\newblock \showarticletitle{Deriving Production Rules for Incremental View
  Maintenance}. In \bibinfo{booktitle}{\emph{VLDB}}. \bibinfo{pages}{577--589}.
\newblock


\bibitem[Chaudhuri et~al\mbox{.}(1995)]%
        {DBLP:conf/icde/ChaudhuriKPS95}
\bibfield{author}{\bibinfo{person}{Surajit Chaudhuri}, \bibinfo{person}{Ravi
  Krishnamurthy}, \bibinfo{person}{Spyros Potamianos}, {and}
  \bibinfo{person}{Kyuseok Shim}.} \bibinfo{year}{1995}\natexlab{}.
\newblock \showarticletitle{Optimizing Queries with Materialized Views}. In
  \bibinfo{booktitle}{\emph{ICDE}}. \bibinfo{pages}{190--200}.
\newblock


\bibitem[Colby et~al\mbox{.}(1996)]%
        {DBLP:conf/sigmod/ColbyGLMT96}
\bibfield{author}{\bibinfo{person}{Latha~S. Colby}, \bibinfo{person}{Timothy
  Griffin}, \bibinfo{person}{Leonid Libkin}, \bibinfo{person}{Inderpal~Singh
  Mumick}, {and} \bibinfo{person}{Howard Trickey}.}
  \bibinfo{year}{1996}\natexlab{}.
\newblock \showarticletitle{Algorithms for Deferred View Maintenance}. In
  \bibinfo{booktitle}{\emph{SIGMOD}}. \bibinfo{pages}{469--480}.
\newblock


\bibitem[Cormen et~al\mbox{.}(2009)]%
        {cormen-09-inaled}
\bibfield{author}{\bibinfo{person}{Thomas~H. Cormen},
  \bibinfo{person}{Charles~E. Leiserson}, \bibinfo{person}{Ronald~L. Rivest},
  {and} \bibinfo{person}{Clifford Stein}.} \bibinfo{year}{2009}\natexlab{}.
\newblock \bibinfo{booktitle}{\emph{Introduction to Algorithms, 3rd Edition}}.
\newblock \bibinfo{publisher}{MIT Press}.
\newblock


\bibitem[Fehrenbach and Cheney(2018)]%
        {DBLP:journals/scp/FehrenbachC18}
\bibfield{author}{\bibinfo{person}{Stefan Fehrenbach} {and}
  \bibinfo{person}{James Cheney}.} \bibinfo{year}{2018}\natexlab{}.
\newblock \showarticletitle{Language-integrated provenance}.
\newblock \bibinfo{journal}{\emph{Sci. Comput. Program.}}
  \bibinfo{volume}{155} (\bibinfo{year}{2018}), \bibinfo{pages}{103--145}.
\newblock


\bibitem[Ghandeharizadeh et~al\mbox{.}(1992)]%
        {DBLP:conf/edbt/GhandeharizadehHJ92}
\bibfield{author}{\bibinfo{person}{Shahram Ghandeharizadeh},
  \bibinfo{person}{Richard Hull}, {and} \bibinfo{person}{Dean Jacobs}.}
  \bibinfo{year}{1992}\natexlab{}.
\newblock \showarticletitle{Implementation of Delayed Updates in Heraclitus}.
  In \bibinfo{booktitle}{\emph{EDBT}}, Vol.~\bibinfo{volume}{580}.
  \bibinfo{pages}{261--276}.
\newblock


\bibitem[Glavic(2021)]%
        {G21}
\bibfield{author}{\bibinfo{person}{Boris Glavic}.}
  \bibinfo{year}{2021}\natexlab{}.
\newblock \showarticletitle{Data Provenance - Origins, Applications,
  Algorithms, and Models}.
\newblock \bibinfo{journal}{\emph{Foundations and Trends® in Databases}}
  \bibinfo{volume}{9}, \bibinfo{number}{3-4} (\bibinfo{year}{2021}),
  \bibinfo{pages}{209--441}.
\newblock


\bibitem[Glavic et~al\mbox{.}(2013)]%
        {DBLP:conf/birthday/GlavicMA13}
\bibfield{author}{\bibinfo{person}{Boris Glavic}, \bibinfo{person}{Renée~J.
  Miller}, {and} \bibinfo{person}{Gustavo Alonso}.}
  \bibinfo{year}{2013}\natexlab{}.
\newblock \showarticletitle{Using SQL for Efficient Generation and Querying of
  Provenance Information}. In \bibinfo{booktitle}{\emph{In Search of Elegance
  in the Theory and Practice of Computation - Essays Dedicated to Peter
  Buneman}}, Vol.~\bibinfo{volume}{8000}. \bibinfo{pages}{291--320}.
\newblock


\bibitem[Griffin and Libkin(1995)]%
        {DBLP:conf/sigmod/GriffinL95}
\bibfield{author}{\bibinfo{person}{Timothy Griffin} {and}
  \bibinfo{person}{Leonid Libkin}.} \bibinfo{year}{1995}\natexlab{}.
\newblock \showarticletitle{Incremental Maintenance of Views with Duplicates}.
  In \bibinfo{booktitle}{\emph{SIGMOD}}. \bibinfo{pages}{328--339}.
\newblock


\bibitem[Gupta et~al\mbox{.}(1992)]%
        {DBLP:conf/deductive/GuptaKM92}
\bibfield{author}{\bibinfo{person}{Ashish Gupta}, \bibinfo{person}{Dinesh
  Katiyar}, {and} \bibinfo{person}{Inderpal~Singh Mumick}.}
  \bibinfo{year}{1992}\natexlab{}.
\newblock \showarticletitle{Counting solutions to the View Maintenance
  Problem}. In \bibinfo{booktitle}{\emph{Workshop on Deductive Databases}},
  Vol.~\bibinfo{volume}{CITRI/TR-92-65}. \bibinfo{pages}{185--194}.
\newblock


\bibitem[Gupta and Mumick(1995)]%
        {DBLP:journals/debu/GuptaM95}
\bibfield{author}{\bibinfo{person}{Ashish Gupta} {and}
  \bibinfo{person}{Inderpal~Singh Mumick}.} \bibinfo{year}{1995}\natexlab{}.
\newblock \showarticletitle{Maintenance of Materialized Views: Problems,
  Techniques, and Applications}.
\newblock \bibinfo{journal}{\emph{IEEE Data Eng. Bull.}} \bibinfo{volume}{18},
  \bibinfo{number}{2} (\bibinfo{year}{1995}), \bibinfo{pages}{3--18}.
\newblock


\bibitem[Gupta et~al\mbox{.}(1993)]%
        {DBLP:conf/sigmod/GuptaMS93}
\bibfield{author}{\bibinfo{person}{Ashish Gupta},
  \bibinfo{person}{Inderpal~Singh Mumick}, {and} \bibinfo{person}{V.~S.
  Subrahmanian}.} \bibinfo{year}{1993}\natexlab{}.
\newblock \showarticletitle{Maintaining Views Incrementally}. In
  \bibinfo{booktitle}{\emph{SIGMOD}}. \bibinfo{pages}{157--166}.
\newblock


\bibitem[Hu and Sintos(2024)]%
        {hu-24-fswcq}
\bibfield{author}{\bibinfo{person}{Xiao Hu} {and} \bibinfo{person}{Stavros
  Sintos}.} \bibinfo{year}{2024}\natexlab{}.
\newblock \showarticletitle{Finding Smallest Witnesses for Conjunctive
  Queries}. In \bibinfo{booktitle}{\emph{ICDT}}.
\newblock


\bibitem[Karvounarakis and Green(2012)]%
        {DBLP:journals/sigmod/KarvounarakisG12}
\bibfield{author}{\bibinfo{person}{Grigoris Karvounarakis} {and}
  \bibinfo{person}{Todd~J. Green}.} \bibinfo{year}{2012}\natexlab{}.
\newblock \showarticletitle{Semiring-annotated data: queries and provenance?}
\newblock \bibinfo{journal}{\emph{SIGMOD Rec.}} \bibinfo{volume}{41},
  \bibinfo{number}{3} (\bibinfo{year}{2012}), \bibinfo{pages}{5--14}.
\newblock


\bibitem[Koch et~al\mbox{.}(2014)]%
        {DBLP:journals/vldb/KochAKNNLS14}
\bibfield{author}{\bibinfo{person}{Christoph Koch}, \bibinfo{person}{Yanif
  Ahmad}, \bibinfo{person}{Oliver Kennedy}, \bibinfo{person}{Milos Nikolic},
  \bibinfo{person}{Andres Nötzli}, \bibinfo{person}{Daniel Lupei}, {and}
  \bibinfo{person}{Amir Shaikhha}.} \bibinfo{year}{2014}\natexlab{}.
\newblock \showarticletitle{DBToaster: higher-order delta processing for
  dynamic, frequently fresh views}.
\newblock \bibinfo{journal}{\emph{VLDBJ}} \bibinfo{volume}{23},
  \bibinfo{number}{2} (\bibinfo{year}{2014}), \bibinfo{pages}{253--278}.
\newblock


\bibitem[Küchenhoff(1991)]%
        {DBLP:conf/dood/Kuchenhoff91}
\bibfield{author}{\bibinfo{person}{Volker Küchenhoff}.}
  \bibinfo{year}{1991}\natexlab{}.
\newblock \showarticletitle{On the Efficient Computation of the Difference
  Between Concecutive Database States}. In \bibinfo{booktitle}{\emph{DOOD}},
  Vol.~\bibinfo{volume}{566}. \bibinfo{pages}{478--502}.
\newblock


\bibitem[Lee et~al\mbox{.}(2020)]%
        {LL20}
\bibfield{author}{\bibinfo{person}{Seokki Lee}, \bibinfo{person}{Bertram
  Ludäscher}, {and} \bibinfo{person}{Boris Glavic}.}
  \bibinfo{year}{2020}\natexlab{}.
\newblock \showarticletitle{Approximate Summaries for Why and Why-not
  Provenance}.
\newblock \bibinfo{journal}{\emph{PVLDB}} \bibinfo{volume}{13},
  \bibinfo{number}{6} (\bibinfo{year}{2020}), \bibinfo{pages}{912--924}.
\newblock


\bibitem[Li et~al\mbox{.}(2025)]%
        {availability}
\bibfield{author}{\bibinfo{person}{Pengyuan Li}, \bibinfo{person}{Boris
  Glavic}, \bibinfo{person}{Dieter Gawlick}, \bibinfo{person}{Vasudha
  Krishnaswamy}, \bibinfo{person}{Zhen~Hua Liu}, \bibinfo{person}{Danica
  Porobic}, {and} \bibinfo{person}{Xing Niu}.} \bibinfo{year}{2025}\natexlab{}.
\newblock \bibinfo{title}{experiments and source code repository}.
\newblock
  \bibinfo{howpublished}{\url{https://github.com/IITDBGroup/IMP\_EDBT26}}.
\newblock


\bibitem[Liu and Glavic(2025)]%
        {costarxiv}
\bibfield{author}{\bibinfo{person}{Ziyu Liu} {and} \bibinfo{person}{Boris
  Glavic}.} \bibinfo{year}{2025}\natexlab{}.
\newblock \showarticletitle{Cost-based Selection of Provenance Sketches for
  Data Skipping}.
\newblock \bibinfo{journal}{\emph{CoRR}}  \bibinfo{volume}{abs/2504.19252}
  (\bibinfo{year}{2025}).
\newblock
\showeprint[arxiv]{2504.19252}


\bibitem[Mistry et~al\mbox{.}(2001)]%
        {DBLP:conf/sigmod/MistryRSR01}
\bibfield{author}{\bibinfo{person}{Hoshi Mistry}, \bibinfo{person}{Prasan Roy},
  \bibinfo{person}{S. Sudarshan}, {and} \bibinfo{person}{Krithi Ramamritham}.}
  \bibinfo{year}{2001}\natexlab{}.
\newblock \showarticletitle{Materialized View Selection and Maintenance Using
  Multi-Query Optimization}. In \bibinfo{booktitle}{\emph{SIGMOD}}.
  \bibinfo{pages}{307--318}.
\newblock


\bibitem[Moerkotte(1998)]%
        {DBLP:conf/vldb/Moerkotte98}
\bibfield{author}{\bibinfo{person}{Guido Moerkotte}.}
  \bibinfo{year}{1998}\natexlab{}.
\newblock \showarticletitle{Small Materialized Aggregates: A Light Weight Index
  Structure for Data Warehousing}. In \bibinfo{booktitle}{\emph{VLDB}}.
  \bibinfo{pages}{476--487}.
\newblock


\bibitem[Mohammed et~al\mbox{.}(2023)]%
        {mohammed-23-sd}
\bibfield{author}{\bibinfo{person}{Haneen Mohammed}, \bibinfo{person}{Charlie
  Summers}, \bibinfo{person}{Sughosh Kaushik}, {and} \bibinfo{person}{Eugene
  Wu}.} \bibinfo{year}{2023}\natexlab{}.
\newblock \showarticletitle{SmokedDuck Demonstration: SQLStepper}. In
  \bibinfo{booktitle}{\emph{SIGMOD}},
  \bibfield{editor}{\bibinfo{person}{Sudipto Das}, \bibinfo{person}{Ippokratis
  Pandis}, \bibinfo{person}{K.~Selçuk Candan}, {and} \bibinfo{person}{Sihem
  Amer-Yahia}} (Eds.). \bibinfo{pages}{183--186}.
\newblock


\bibitem[Motik et~al\mbox{.}(2015)]%
        {DBLP:conf/aaai/MotikNPH15a}
\bibfield{author}{\bibinfo{person}{Boris Motik}, \bibinfo{person}{Yavor Nenov},
  \bibinfo{person}{Robert Edgar~Felix Piro}, {and} \bibinfo{person}{Ian
  Horrocks}.} \bibinfo{year}{2015}\natexlab{}.
\newblock \showarticletitle{Incremental Update of Datalog Materialisation: the
  Backward/Forward Algorithm}. In \bibinfo{booktitle}{\emph{AAAI}}.
  \bibinfo{pages}{1560--1568}.
\newblock


\bibitem[Murray et~al\mbox{.}(2013)]%
        {MM13c}
\bibfield{author}{\bibinfo{person}{Derek~G Murray}, \bibinfo{person}{Frank
  McSherry}, \bibinfo{person}{Rebecca Isaacs}, \bibinfo{person}{Michael Isard},
  \bibinfo{person}{Paul Barham}, {and} \bibinfo{person}{Martin Abadi}.}
  \bibinfo{year}{2013}\natexlab{}.
\newblock \showarticletitle{Naiad: a timely dataflow system}. In
  \bibinfo{booktitle}{\emph{Proceedings of the Twenty-Fourth ACM Symposium on
  Operating Systems Principles}}. \bibinfo{pages}{439--455}.
\newblock


\bibitem[Murray et~al\mbox{.}(2016)]%
        {murray-16-initdpwtd}
\bibfield{author}{\bibinfo{person}{Derek~Gordon Murray}, \bibinfo{person}{Frank
  McSherry}, \bibinfo{person}{Michael Isard}, \bibinfo{person}{Rebecca Isaacs},
  \bibinfo{person}{Paul Barham}, {and} \bibinfo{person}{Martín Abadi}.}
  \bibinfo{year}{2016}\natexlab{}.
\newblock \showarticletitle{Incremental, Iterative Data Processing With Timely
  Dataflow}.
\newblock \bibinfo{journal}{\emph{Commun. ACM}} \bibinfo{volume}{59},
  \bibinfo{number}{10} (\bibinfo{year}{2016}), \bibinfo{pages}{75--83}.
\newblock


\bibitem[Niu et~al\mbox{.}(2021)]%
        {DBLP:journals/pvldb/NiuGLLGKLP21}
\bibfield{author}{\bibinfo{person}{Xing Niu}, \bibinfo{person}{Boris Glavic},
  \bibinfo{person}{Ziyu Liu}, \bibinfo{person}{Pengyuan Li},
  \bibinfo{person}{Dieter Gawlick}, \bibinfo{person}{Vasudha Krishnaswamy},
  \bibinfo{person}{Zhen~Hua Liu}, {and} \bibinfo{person}{Danica Porobic}.}
  \bibinfo{year}{2021}\natexlab{}.
\newblock \showarticletitle{Provenance-based Data Skipping}.
\newblock \bibinfo{journal}{\emph{PVLDB}} \bibinfo{volume}{15},
  \bibinfo{number}{3} (\bibinfo{year}{2021}), \bibinfo{pages}{451--464}.
\newblock


\bibitem[Niu et~al\mbox{.}(2019)]%
        {DBLP:journals/tkde/NiuKGGLKR19}
\bibfield{author}{\bibinfo{person}{Xing Niu}, \bibinfo{person}{Raghav Kapoor},
  \bibinfo{person}{Boris Glavic}, \bibinfo{person}{Dieter Gawlick},
  \bibinfo{person}{Zhen~Hua Liu}, \bibinfo{person}{Vasudha Krishnaswamy}, {and}
  \bibinfo{person}{Venkatesh Radhakrishnan}.} \bibinfo{year}{2019}\natexlab{}.
\newblock \showarticletitle{Heuristic and Cost-Based Optimization for Diverse
  Provenance Tasks}.
\newblock \bibinfo{journal}{\emph{TKDE}} \bibinfo{volume}{31},
  \bibinfo{number}{7} (\bibinfo{year}{2019}), \bibinfo{pages}{1267--1280}.
\newblock


\bibitem[Niu et~al\mbox{.}(2017)]%
        {DBLP:conf/icde/NiuKGGLR17}
\bibfield{author}{\bibinfo{person}{Xing Niu}, \bibinfo{person}{Raghav Kapoor},
  \bibinfo{person}{Boris Glavic}, \bibinfo{person}{Dieter Gawlick},
  \bibinfo{person}{Zhen~Hua Liu}, {and} \bibinfo{person}{Venkatesh
  Radhakrishnan}.} \bibinfo{year}{2017}\natexlab{}.
\newblock \showarticletitle{Provenance-Aware Query Optimization}. In
  \bibinfo{booktitle}{\emph{ICDE}}. \bibinfo{pages}{473--484}.
\newblock


\bibitem[Palpanas et~al\mbox{.}(2002)]%
        {DBLP:conf/vldb/PalpanasSCP02}
\bibfield{author}{\bibinfo{person}{Themistoklis Palpanas},
  \bibinfo{person}{Richard Sidle}, \bibinfo{person}{Roberta Cochrane}, {and}
  \bibinfo{person}{Hamid Pirahesh}.} \bibinfo{year}{2002}\natexlab{}.
\newblock \showarticletitle{Incremental Maintenance for Non-Distributive
  Aggregate Functions}. In \bibinfo{booktitle}{\emph{VLDB}}.
  \bibinfo{pages}{802--813}.
\newblock


\bibitem[Psallidas and Wu(2018)]%
        {DBLP:journals/corr/abs-1801-07237}
\bibfield{author}{\bibinfo{person}{Fotis Psallidas} {and}
  \bibinfo{person}{Eugene Wu}.} \bibinfo{year}{2018}\natexlab{}.
\newblock \showarticletitle{Smoke: Fine-grained Lineage at Interactive Speed}.
\newblock \bibinfo{journal}{\emph{CoRR}}  \bibinfo{volume}{abs/1801.07237}
  (\bibinfo{year}{2018}).
\newblock
\showeprint[arXiv]{1801.07237}


\bibitem[Senellart et~al\mbox{.}(2018)]%
        {DBLP:journals/pvldb/SenellartJMR18}
\bibfield{author}{\bibinfo{person}{Pierre Senellart}, \bibinfo{person}{Louis
  Jachiet}, \bibinfo{person}{Silviu Maniu}, {and} \bibinfo{person}{Yann
  Ramusat}.} \bibinfo{year}{2018}\natexlab{}.
\newblock \showarticletitle{ProvSQL: Provenance and Probability Management in
  PostgreSQL}.
\newblock \bibinfo{journal}{\emph{PVLDB}} \bibinfo{volume}{11},
  \bibinfo{number}{12} (\bibinfo{year}{2018}), \bibinfo{pages}{2034--2037}.
\newblock


\bibitem[Shmueli and Itai(1984)]%
        {DBLP:conf/sigmod/ShmueliI84}
\bibfield{author}{\bibinfo{person}{Oded Shmueli} {and} \bibinfo{person}{Alon
  Itai}.} \bibinfo{year}{1984}\natexlab{}.
\newblock \showarticletitle{Maintenance of Views}. In
  \bibinfo{booktitle}{\emph{SIGMOD}}. \bibinfo{pages}{240--255}.
\newblock


\bibitem[Vista(1994)]%
        {DBLP:conf/cascon/Vista94}
\bibfield{author}{\bibinfo{person}{Dimitra Vista}.}
  \bibinfo{year}{1994}\natexlab{}.
\newblock \showarticletitle{View maintenance in relational and deductive
  databases by incremental query evaluation}. In
  \bibinfo{booktitle}{\emph{Proceedings of the 1994 Conference of the Centre
  for Advanced Studies on Collaborative Research, October 31 - November 3,
  1994, Toronto, Ontario, Canada}}. \bibinfo{pages}{70}.
\newblock


\bibitem[Yang and Widom(2003)]%
        {DBLP:journals/vldb/YangW03}
\bibfield{author}{\bibinfo{person}{Jun Yang} {and} \bibinfo{person}{Jennifer
  Widom}.} \bibinfo{year}{2003}\natexlab{}.
\newblock \showarticletitle{Incremental computation and maintenance of temporal
  aggregates}.
\newblock \bibinfo{journal}{\emph{VLDBJ}} \bibinfo{volume}{12},
  \bibinfo{number}{3} (\bibinfo{year}{2003}), \bibinfo{pages}{262--283}.
\newblock


\bibitem[Zhou et~al\mbox{.}(2010)]%
        {DBLP:conf/sigmod/ZhouSTLLM10}
\bibfield{author}{\bibinfo{person}{Wenchao Zhou}, \bibinfo{person}{Micah
  Sherr}, \bibinfo{person}{Tao Tao}, \bibinfo{person}{Xiaozhou Li},
  \bibinfo{person}{Boon~Thau Loo}, {and} \bibinfo{person}{Yun Mao}.}
  \bibinfo{year}{2010}\natexlab{}.
\newblock \showarticletitle{Efficient querying and maintenance of network
  provenance at internet-scale}. In \bibinfo{booktitle}{\emph{SIGMOD}}.
  \bibinfo{pages}{615--626}.
\newblock


\bibitem[Zilio et~al\mbox{.}(2004)]%
        {DBLP:conf/icac/ZilioZLMLCPCGALV04}
\bibfield{author}{\bibinfo{person}{Daniel~C. Zilio}, \bibinfo{person}{Calisto
  Zuzarte}, \bibinfo{person}{Sam Lightstone}, \bibinfo{person}{Wenbin Ma},
  \bibinfo{person}{Guy~M. Lohman}, \bibinfo{person}{Roberta Cochrane},
  \bibinfo{person}{Hamid Pirahesh}, \bibinfo{person}{Latha~S. Colby},
  \bibinfo{person}{Jarek Gryz}, \bibinfo{person}{Eric Alton},
  \bibinfo{person}{Dongming Liang}, {and} \bibinfo{person}{Gary Valentin}.}
  \bibinfo{year}{2004}\natexlab{}.
\newblock \showarticletitle{Recommending Materialized Views and Indexes with
  IBM DB2 Design Advisor}. In \bibinfo{booktitle}{\emph{ICAC}}.
  \bibinfo{pages}{180--188}.
\newblock


\end{thebibliography}
